\documentclass[10pt]{article}
\usepackage{amsmath}
\usepackage{amsthm}
\usepackage{graphicx}
\usepackage{latexsym}
\usepackage{amssymb}
\usepackage{amsbsy}
\usepackage[]{authblk}

\DeclareSymbolFont{msbm}{U}{msb}{m}{n}
\DeclareMathSymbol{\C}{\mathalpha}{msbm}{'103}
\DeclareMathSymbol{\R}{\mathalpha}{msbm}{'122}
\DeclareMathSymbol{\Z}{\mathalpha}{msbm}{'132}
\DeclareMathSymbol{\N}{\mathalpha}{msbm}{'116}

\newcommand{\G}{\mathcal{G}}

\newtheorem{remark}{Remark}
\newtheorem{proposition}{Proposition}

\def\vr{\boldsymbol{\varrho}}

\def\be{\begin{equation}}
\def\ee{\end{equation}}
\def\bea{\begin{eqnarray}}
\def\ba{\begin{array}{l}\displaystyle}
\def\eea{\end{eqnarray}}
\def\ea{\end{array}}
\def\E{{\cal E}}

\begin{document}
\title{\small \hfill 
    \Large
A MULTISCALE KINETIC-FLUID SOLVER WITH DYNAMIC LOCALIZATION OF
KINETIC EFFECTS
\thanks{Acknowledgements: This work was supported by the Marie Curie Actions of the European
Commission in the frame of the DEASE project (MEST-CT-2005-021122)
and by the French Commisariat \`{a} l'\'{E}nergie Atomique (CEA) in
the frame of the contract ASTRE (SAV 34160) }}

\author[1,2]{Pierre Degond}
\author[1,2,3]{Giacomo Dimarco\footnote
{Corresponding author address: Institut de Math\'{e}matiques de
Toulouse, UMR 5219 Universit\'{e} Paul Sabatier, 118, route de
Narbonne 31062 TOULOUSE Cedex, FRANCE. \\
\emph{E-mail}:pierre.degond@math.univ-toulouse.fr,giacomo.dimarco@unife.it,luc.mieussens@math.u-bordeaux1.fr}}
\author[4]{Luc Mieussens}

\affil[1]{Universit\'{e} de Toulouse; UPS, INSA, UT1, UTM ;

Institut de Math\'{e}matiques de Toulouse ; F-31062 Toulouse,
France.} \affil[2]{CNRS; Institut de Math\'{e}matiques de Toulouse
UMR 5219;

F-31062 Toulouse, France.} \affil[3]{Commissariat  \`a l'Energie
Atomique CEA-Saclay DM2S-SFME;

91191 Gif-sur-Yvette, France.} \affil[4]{Universit\'{e} de Bordeaux;
Institut de Math\'{e}matiques; 33405 Talence cedex Bordeaux,
France.} \maketitle

\date{}

\maketitle

\begin{abstract}
This paper collects the efforts done in our previous works
\cite{degond1},\cite{dimarco1},\cite{degond3} to build a robust
multiscale kinetic-fluid solver. Our scope is to efficiently solve
fluid dynamic problems which present non equilibrium localized
regions that can move, merge, appear or disappear in time. The main
ingredients of the present work are the followings ones: a fluid
model is solved in the whole domain together with a localized
kinetic upscaling term that corrects the fluid model wherever it is
necessary; this multiscale description of the flow is obtained by
using a micro-macro decomposition of the distribution function
\cite{degond3}; the dynamic transition between fluid and kinetic
descriptions is obtained by using a time and space dependent
transition function; to efficiently define the breakdown conditions
of fluid models we propose a new criterion based on the distribution
function itself. Several numerical examples are presented to
validate the method and measure its computational efficiency.
\end{abstract}


\section{Introduction}

Many engineering problems involve fluids in transitional regimes
such as hypersonic flows or micro-electro-mechanical devices. In
these cases, usual fluid models (like Euler or Navier-Stokes
equations) break down in localized regions of the computational
domain (typically in shock and boundary layers). For such problems,
using classical fluid models is generally not sufficient for an
accurate description of the flow in non-equilibrium regions.
However, it is not necessary to solve the Boltzmann equation---which
is computationally more expensive than continuum solvers by several
orders of magnitude---especially in situations where the flow is
close to thermodynamical equilibrium.

For the above reasons, it is important to develop hybrid techniques
which can reduce the use of kinetic solvers to the regions where
they are strictly necessary (kinetic regions), leaving the
simulation in the rest of the domain to a continuum or fluid solver
(fluid regions). The construction of these methods involve two main
problems. The first one is  how to accurately identify the different
regions. We refer, for instance, to the works of Wijesinghe and
Hadjiconstantinou~\cite{he}, Levermore, Morokoff, and
Nadiga~\cite{Levermore}, and Wang and Boyd~\cite{Wang}, in which
various breakdown criteria are proposed. The second main problem is
how to efficiently and correctly match the two models at the
interfaces. Most of the recent methods are based on domain
decomposition techniques, such as in the works of Bourgat, LeTallec,
Perthame, and Qiu \cite{BLPQ}, Bourgat, LeTallec and Tidriri
\cite{BLT}, LeTallec and Mallinger \cite{Letallec}, Aktas and Aluru
\cite{AA}, Roveda, Goldstein and Varghese \cite{RGV1}, Sun, Boyd and
Candler \cite{Boyd2}, Wadsworth and Erwin \cite{Wad}, and Wijesinghe
et al.~\cite{Wij}. The same domain decomposition approach has been
also used in many others fields, such as, for instance, in molecular
dynamics \cite{WE1}, in epitaxial growth \cite{SSE} or for problems
involving diffusive scalings \cite{Klar} instead of hydrodynamic
ones. We also mention the use of decompositions in velocity instead
of physical space done by Crouseilles, Degond and Lemou
\cite{degond} and by Dimarco and Pareschi \cite{dimarco2}.

It is important to stress that most of the mentioned methods use a
static interface between kinetic and fluid regions that is chosen
once for all at the beginning of the computation. However, for
unsteady problems, this approach appears as somehow inadequate and
inefficient, and for this reason, some automatic domain
decomposition methods have also been proposed, see for example
Kolobov et al.~\cite{Kobolov}, or Tiwari~\cite{tiwari_JCP,tiwari_1}
and Dimarco and Pareschi~\cite{dimarco1}. We have also proposed a
similar approach in~\cite{dimarco1}.

In this paper, we propose a method that has similar features as the
methods mentioned above: we solve the Boltzmann-BGK equation coupled
with the compressible Euler equations through an adaptive domain
decomposition technique. With this technique, it is possible to
achieve considerable computational speedup, as compared to steady
interface coupling strategies, without losing accuracy in the
solution. This method is somehow an extension of our previous
work~\cite{dimarco1}, but several important differences must be
noted.  First, we introduce a new breakdown criterion which is based
on a careful inspection of the distribution function. This criterion
can be defined by using the macroscopic variables only, at least in
fluid regions, and thus does not introduce additional expensive
computations. This allows us to define kinetic regions that are
small as possible. Second, we use a decomposition of the
distribution function that has better properties than the one used
in~\cite{dimarco1}. In fact, while it has been proved by Degond,
Jin, and Mieussens~\cite{degond1} that the decomposition used
in~\cite{dimarco1} preserves uniform flows at the continuous level,
we show in this paper that this is not true in general at the
discrete level, except if a quite specific and very expensive
kinetic scheme is used. For this reason in the present work we use
the decomposition proposed by Degond, Liu, and
Mieussens~\cite{degond3}, since it perfectly preserves uniform
flows, both for the continuous and discrete cases. As
in~\cite{degond3}, we decompose the distribution function into an
equilibrium part, that can be described by macroscopic fluid
variables, and a perturbative non-equilibrium part. We obtain a
micro-macro fluid model in which the macroscopic variables are
determined by solving a fluid equation with a kinetic upscaling.
This kinetic upscaling is determined by solving a kinetic equation,
and is dynamically and automatically localized wherever it is
necessary, by using our breakdown criterion and the transition
function idea~\cite{degond2,degond1,degond3,dimarco1}. Third, we
propose an efficient numerical scheme for discretizing our
micro-macro fluid model: we use a time splitting approach that has
several advantages. In particular it is shown to preserve the
positivity of the distribution function.

The outline of the article is the following. In section
\ref{sec_Boltzmann}, we introduce the BGK equation and its
properties. In section \ref{sec_coupling}, we present the coupling
strategy, while in section \ref{sec_num_approx} the numerical scheme
is described and positivity properties are analyzed. In section
\ref{sec:moving}, we derive our breakdown criterion, and the final
algorithm is presented. Several numerical tests are presented in
section \ref{sec_tests} to illustrate the properties of our method
and to demonstrate its efficiency. A short conclusion is given in
section \ref{sec_conclu}. In appendix A, the differences between the
present coupling strategy and the decomposition used in
\cite{degond1,dimarco1} are analyzed in some details.


\section{The Boltzmann-BGK model}
\label{sec_Boltzmann}

We consider the kinetic equation
\be
\partial_t f + v\cdot\nabla_{x}f = Q(f),
\label{eq:B}
\ee
with the initial data
\begin{equation*}
  f(x,v,t=0)=f_{init},
\end{equation*}
where $f=f(x,v,t)$ is the density of particles that have velocity $v
\in \R^3 $ and position $x \in \Omega \subset \R^3$ at time $ t> 0$.
The collision operator $Q$ locally acts in space and time and takes
into account interactions between particles. It is assumed to
satisfy local conservation properties \be \langle m Q(f)\rangle=0\ee
for every $f$, where we denote weighted integrals of $f$ over the
velocity space by \be\langle \phi f\rangle=\int_{\R^3}
\phi(v)f(v)dv,\ee where $\phi(v)$ is any function of $v$, and
$m(v)=(1,v,|v|^2)$ are the so-called collisional invariants. It
follows that the multiplication of (\ref{eq:B}) by $m(v)$ and the
integration in velocity space leads to the system of local
conservation laws \be
\partial_t \langle m f\rangle +\nabla_{x}\langle v m f\rangle = 0. \label{consl}\ee
We also assume that the functions satisfying $Q(f)=0$, referred to
as local equilibrium distributions and denoted by $E[\vr]$, are
defined implicitly through their moments $\varrho$ by \be
\vr=\langle m E[\boldsymbol{\varrho}]\rangle\ee In the present paper
we will work with the BGK model of the Boltzmann collision operator
that reads \be Q(f)=\nu (E[\boldsymbol{\varrho}]-f).\ee With this
operator, collisions are modelled by a relaxation towards the local
Maxwellian equilibrium: \be
 E[\boldsymbol{\varrho]}(v)=\frac{\varrho}{(2\pi \theta)^{3/2}}\exp\left(\frac{-|u-v|^{2}}{2\theta}\right) ,
\label{eq:M} \ee where $\varrho$ and $u$ are the density and mean
velocity while $\theta=RT$ with $T$ the temperature of the gas and
$R$ the gas constant. The macroscopic values $\varrho$, $u$ and $T$
are related to $f$ by: \be \varrho=\int_{\R^3} fdv, \qquad \varrho
u=\int_{\R^3} vfdv, \qquad  \theta=\frac{1}{3\varrho}
\int_{\R^3}|v-u|^{2}fdv, \label{eq:Mo} \ee while the internal energy
$e$ is defined as \be e=\frac{1}{2\varrho} \int_{\R^3}|v|^{2}fdv =
\frac{1}{2}|u|^2 + \frac{3}{2}  \theta. \label{eq:E} \ee The
parameter $\nu> 0$ is the relaxation frequency. In this paper, we
use the classical choice $\nu=\mu/p$ where $\mu=\mu_{ref}(\theta/
\theta_{ref})^{\omega}$ is the viscosity and $p$ is the pressure. We
refer to section \ref{sec_tests} for numerical values of
$\mu_{ref}$, $\theta_{ref}$ and $\omega$.

Boundary conditions have to be specified for equation
(\ref{eq:B}). Different type of conditions are used in applications: inflow,
outflow, specular reflection or total accomodation. We will specify the conditions
we use for every numerical test in section \ref{sec_tests}.

When the mean free path between particles is very small compared to
the size of the computational domain, the space and time variables
can be rescaled to \be x'=\varepsilon x, \ t'=\varepsilon
t\label{eq:scaling}\ee where $\varepsilon$ is the ratio between the
microscopic and the macroscopic scale (the so-called Knudsen
number). Using these new variables in~(\ref{eq:B}), we get \be
\partial_{t'} f^\varepsilon + v\cdot\nabla_{x'}f^\varepsilon =
\frac{\nu}{\varepsilon}(E^\varepsilon[\boldsymbol{\varrho}]-f^{\varepsilon}).\ee
If the Knudsen number $\varepsilon$ tends to zero, this equation shows
that the
distribution function converges towards the local
Maxwellian equilibrium $E^\varepsilon[\boldsymbol{\varrho}]$. Using this relation into the
conservation laws~(\ref{consl}) gives the Euler equations for $\boldsymbol{\varrho}$:
\be \partial_{t'} \boldsymbol{\varrho}
+\nabla_{x'}F(\boldsymbol{\varrho}) = 0,\ee  where
$F(\boldsymbol{\varrho})=\langle v m
E^{\varepsilon}[\boldsymbol{\varrho}]\rangle=(\varrho,\varrho u,
\varrho e)$.

In the sequel, to give a simple description of our approach, all
schemes and all algorithms are shown for the one dimensional case in
velocity and physical space. The extension to the multidimensional
case does not introduce any additional difficulty in the
mathematical setting. We will also omit the primes wherever they are unnecessary.


\section{The coupling method}
\label{sec_coupling}

In this section, we follow the work of~\cite{degond3} and extend the
micro-macro fluid model to allow for dynamic localization of the
kinetic upscaling.


\subsection{Decomposition of the kinetic equation}
\label{subsec_decom_kin}

Our method is based on the micro-macro decomposition of the
distribution function: it is decomposed in its local Maxwellian
equilibrium and the deviation part as \be f=E[\vr]+g.\ee Because the
equilibrium distribution has the same first three moments as $f$ we
have \be \langle m g\rangle=0.\ee Then it can be easily proved that
the following coupled system
\begin{align}
& \partial_t \vr +\partial_x F(\vr)+\partial_x \langle v m g
\rangle=0\label{eq:s1} \\
&  \partial_t g +v \partial_x g=-\nu g -(\partial_t+v \partial_x
)E[\vr]\label{eq:s2}
\end{align}
 is satisfied by $\vr=\langle m f\rangle$ and
$g=f-E[\vr]$, where $F(\vr)=\langle v m E[\vr]\rangle$ is the flux
associated to the equilibrium state. The corresponding initial data
are \be \vr_{t=0}=\vr_{init}=\langle m f_{init}\rangle, \ \ \ \
g_{t=0}=f_{init}-E[\vr_{init}].\label{eq:idata} \ee The converse
statement is also true: if $\vr$ and $g$ satisfy system
(\ref{eq:s1}) and (\ref{eq:s2}) with initial data~(\ref{eq:idata}),
then $f=E[\vr]+g$ satisfies the kinetic equation (\ref{eq:B})
(see~\cite{degond1} for details). In the following section, starting
from this decomposition, we introduce the set of equations that
define the domain decomposition technique we are proposing.

\subsection{Transition function}
\label{subsec_decom_kin}

Let $\Omega_1$, $\Omega_2$, and $\Omega_3$ be three disjointed sets
such that $\Omega_1\cup \Omega_2 \cup \Omega_3=\R^3$. The first set
$\Omega_1$ is supposed to be a domain in which the flow is far from
the equilibrium (the "kinetic zone"), while the flow is supposed to
be close to the equilibrium in $\Omega_2$ (the "fluid zone") and
also in $\Omega_3$ (the "buffer zone"). We define a function
$h(x,t)$ such that
\begin{equation}
h(x,t)=\left\lbrace
\begin{array}{lll}
\displaystyle 1, & \mbox{for} & x \in \Omega_1, \\
0, & \mbox{for} & x \in \Omega_2, \\
0 \leq h(x,t) \leq 1, & \mbox{for} & x \in \Omega_3. \\
\end{array}
\right.
\end{equation}
Note that the time dependence of $h$ means that we account for
possibly dynamically changing fluid and kinetic zones. The topology
and geometry of these zones is directly encoded in $h$ and may
change dynamically as well.

Next, we split the perturbation
term in two distribution functions $g_{K}=hg$ and
$g_{F}=(1-h)g$. The time derivatives of these functions then are
\begin{align*}
& \partial_t g_K = \partial_t(hg) = g \, \partial_t h + h \partial_t g, \\
& \partial_t g_F = \partial_t((1-h)g) = - g \, \partial_t h +
(1-h)
\partial_t g,
\end{align*}
and it is therefore easy to derive the following coupled system of equations
\begin{align}
 &\partial_t \vr +\partial_x F(\vr)+\partial_x \langle v m g_K
\rangle+\partial_x \langle v m g_F \rangle=0\label{eq:s3}\\  &
\partial_t g_{K} + hv\partial_xg_{K}+hv\partial_xg_{F} =-\nu g_K
-h(\partial_t+v \partial_x )E[\vr]+\frac{g_K}{h}\partial_{t}h ,
\label{eq:gk} \\
 & \partial_t g_{F} + (1-h)v\partial_xg_{K} +
(1-h)v\partial_xg_{F}=
-\nu g_F -(1-h)(\partial_t+v
\partial_x )E[\vr] -\frac{g_F}{1-h}\partial_{t}h ,
\label{eq:gF}
\end{align}
with initial data
\begin{equation}
 g_{K,t=0}=h_{t=0}g_{t=0}\, , \quad  g_{F,t=0}=(1-h_{t=0})g_{t=0}\, , \quad
\vr_{t=0}=\vr_{init} \label{eq:Bg}
\end{equation}
and with $h_{t=0}=h_{init}$ and
$g_{t=0}=f_{init}-E[\vr_{init}]$. Again,
system~(\ref{eq:s3}--\ref{eq:gF}) with initial data~(\ref{eq:Bg}) is
equivalent to system~(\ref{eq:s1}--\ref{eq:s2}) with initial data~(\ref{eq:idata}) (see~\cite{degond1} for details).

Now assume that the flow is very close to equilibrium in
$\Omega_2\cup\Omega_3$. This means that $g$ is very small in these
domains and can be set to zero. Since $g=g_F$ in $\Omega_2$, we set
$g_F=0$ in this domain. In $\Omega_3$, we also set $g_F=0$, which
means that we approximate $g$ by $g_K$. Consequently, $g_F$ can be
eliminated from~(\ref{eq:s3}--\ref{eq:gF}) to get:
\begin{align}
&\partial_t \vr +\partial_x
F(\vr)+\partial_x \langle
v m g_K \rangle=0\label{eq:s5}\\
&\partial_t g_{K} + hv\partial_xg_{K} =-\frac{\nu}{\varepsilon} g_K
-h(\partial_t+v \partial_x
)E[\vr]+\frac{g_K}{h}\partial_{t}h ,
\label{eq:gk2}
\end{align}
with initial data
\begin{equation}
 g_{K,t=0}=h_{t=0}g_{t=0}=h_{init}(f_{init}-E[\vr_{init}])\, , \quad \vr_{t=0}=\vr_{init}.
\label{eq:Bg1}
\end{equation}

Note that since by definition $g_K$ is zero in the fluid zone
$\Omega_2$, the kinetic equation equation~(\ref{eq:gk2}) is solved in the kinetic and buffer
zones $\Omega_1$ and $\Omega_3$ only. Indeed, in the fluid zone, we only
solve~(\ref{eq:s5}) with $g_K=0$, which is nothing but the Euler
equations. In the kinetic zone, we have $g_K=g$ and hence
system~(\ref{eq:s5}--\ref{eq:gk2}) is nothing but
system~(\ref{eq:s1}--\ref{eq:s2}), which is equivalent to the original
BGK equation. System~(\ref{eq:s5}--\ref{eq:gk2}) is our micro-macro fluid model
with dynamically localized kinetic effects which will be used to solve multiscale
kinetic problems. With this system, the distribution function $f$ is
approximated by $E[\vr]+g_K$.

In the next section we describe and analyze the numerical scheme we
use to discretize this system, and we compare this new model to the
model used in our previous work~\cite{dimarco1}.

\begin{remark}
  We mention here a slightly different derivation that leads to a
  different micro-macro model. In~(\ref{eq:gk}), the term
  $\frac{g_K}{h}$ can be equivalently replaced by $g_K+g_F$, since $g_K=hg$
  by definition and also $g=g_K+g_F$. In this case, the approximation
  $g_F=0$ in $\Omega_2$ and $\Omega_3$ leads to the model:
\begin{align}
&\partial_t \vr +\partial_x
F(\vr)+\partial_x \langle
v m g_K \rangle=0\label{eq-s5bis}   \\
&\partial_t g_{K} + hv\partial_xg_{K} =-\frac{\nu}{\varepsilon} g_K
-h(\partial_t+v \partial_x
)E[\vr]+g_K\partial_{t}h\label{eq-gK_bis}
\end{align}
Note that, surprisingly, this model is different
from~(\ref{eq:s5}--\ref{eq:gk2}): indeed, the factor of $\partial_t h$
is $g_K$ in~(\ref{eq-s5bis}--\ref{eq-gK_bis}), while it is
$\frac{g_K}{h}$ in~(\ref{eq:s5}--\ref{eq:gk2}). However, we only use
system~(\ref{eq:s5}--\ref{eq:gk2}) in the sequel, since it can
be proved to have very good
properties (like positivity preservation).
\end{remark}


\section{Numerical approximation of the coupled model and its properties}
\label{sec_num_approx}

First, we briefly describe a velocity discretization of the kinetic
BGK equation. Then, we propose a second order in space numerical
scheme for the perturbation term $g_K$ and for the macroscopic fluid
equations. Then, we introduce a time splitting method between the
transition function term and the rest of the system to compute the
evolution of the perturbation function $g_K$. Finally, positivity
property for the distribution function $f$ is analyzed in details.

\subsection{Discrete velocity model for kinetic equations}
\label{subsec_vel_disc}

Here, we replace the continuous velocity space by a bounded
Cartesian grid $\mathcal{V}$ of $N$ nodes $v_{j}=j\Delta v+a$, where
$j$ is a bounded index, $\Delta v$ is the grid step, and $a$ is a
constant. The collisional invariants $m(v)$ are replaced by
$m_{j}=(1,v_{j},\frac{1}{2}|v_{j}|^{2})$. The distribution function
$f$ is approximated on the grid by $(f_{j}(t,x))_{j}$, where
$f_{j}(t,x) \approx f(x,v_{j},t)$, while the fluid quantities are
obtained from $f_{j}$ through discrete summations on $\mathcal{V}$:
\be \boldsymbol{\varrho}=\sum_{j}m_{j}f_j\, \Delta v . \label{eq:DM}
\ee The BGK model is then replaced by the following system of $N$
hyperbolic equations with a stiff relaxation term:   \be
\partial_t f_{j} + v_{j} \partial_xf_{j} = \frac{\nu}{\varepsilon} (\E_{j}[\boldsymbol{\varrho}]-f_{j}),
\label{eq:DM1} \ee where $\E_{j}[\boldsymbol{\varrho}]$ is the
approximation of the continuous Maxwellian $E[\varrho]$. Note that
this approximation is not the evaluation of $E[\vr]$ on the grid
points: in fact, to ensure conservation of macroscopic quantities
and entropy decay at the discrete level, the approximated Maxwellian
$\E[\vr_j]$ is defined through an entropy minimization problem that
can be solved by computing the solution of a small non-linear system
(we refer the reader to~\cite{Mieussens,Mieussens2} for details
about this approximation).

Finally, a micro-macro system with localized upscaling can be
derived from the discrete-velocity BGK equation~(\ref{eq:DM1}),
exactly as in the continous case, and
we find:
\begin{align}
&\partial_t \vr +\partial_x
F(\vr)+\partial_x \langle
v m g_K \rangle=0\label{eq:rhod}\\
&\partial_t g_{K,j} + hv_j\partial_xg_{K,j}
=-\frac{\nu}{\varepsilon} g_{K,j} -h(\partial_t+v_j \partial_x
)\E_j[\vr]+\frac{g_{K_j}}{h}\partial_{t}h , \label{eq:gd}
\end{align}
with initial data
\begin{equation*}
 g_{K,j,t=0}=h_{t=0}g_{t=0,j}=h_{init}(f_{init,j}-\E_j[\vr_{init}])\, , \quad \vr_{t=0}=\vr_{init},
\end{equation*}
where $\langle. \rangle$ now stands for $\sum_j . \Delta v$.


\subsection{Numerical schemes}
\label{subsec_num scheme I}

\subsubsection{Non-splitting scheme}

For the space discretization, we consider a grid of step $\Delta x$
and nodes $x_i$, while for the time discretization, we consider the
step $\Delta t$ and times $t_n=n\Delta t$. The unknowns $\vr$ and
$g_{K,j}$ are approximated by $\vr^n_i\approx\vr(t_n,x_i)$ and
$g^n_{K,i,j}\approx g_j(t_n,x_i)$. Now, the space and time
discretization of the discrete velocity micro-macro
system~(\ref{eq:rhod}--\ref{eq:gd}) is:
\begin{eqnarray}
              & &\frac{g^{n+1}_{K,i}-g^{n}_{K,i}}{\Delta
t}+h_i^n\left(\frac{\phi_{i+1/2}(g_K^{n})-\phi_{i-1/2}(g^{n}_K)}{\Delta
x}\right)=-\frac{\nu}{\varepsilon}g^{n+1}_{K,i} \nonumber\\
              &&-h^{n}_i\left(\frac{\E[\vr^{n+1}_i]-\E[\vr^{n}_i]}{\Delta
t}+\frac{\phi_{i+1/2}(\E[\vr^n])-\phi_{i-1/2}(\E[\vr^n])}{\Delta
x}\right)+\frac{g^{n}_{K,i}}{h_i^n}\frac{h_i^{n+1}-h_i^{n}}{\Delta
t}
         \label{eq:discg}  \end{eqnarray}
where the second order numerical fluxes are defined by \be
\phi_{i+1/2}(g^{n}_K)
=v^{-}g^{n}_{K,i+1}+v^{+}g^{n}_{K,i}+\frac{1}{2}|v_j|\min\hspace{-0.5cm}
\mod(g^n_{K,i}-g^n_{K,i-1},g^n_{K,i+1}-g^n_{K,i},g^n_{K,i+2}-g^n_{K,i})\nonumber\ee
with $v^-=v_j$ if $v_j<0$ and $v^-=0$ in other cases, while
$v^+=v_j$ if $v_j\geq0$ and $v^+=0$ if $v_j$ is negative. The same
numerical flux is used for $\phi_{i+1/2}(E[\vr^n])$. Note that for
simplicity, in~(\ref{eq:discg}) and all what follows, the discrete-velocity index $j$ is omitted,
as well as the space and time dependency of $\nu$.

Note that in~(\ref{eq:discg}), $\vr^{n+1}$ is computed by using a
discrete version of~(\ref{eq:rhod}) which is explained
below. Moreover, note that the last term of the right-hand side
of~(\ref{eq:discg}) models the evolution of the transition function
$h$: the new value $h^{n+1}$ depends on the equilibrium/non
equilibrium state of the gas in a way that will be described in
section~\ref{sec:moving}. In addition, we point out that when $h_i^n=0$, equation (\ref{eq:discg})
is not solved, thus the term $g^n_{K,i}/h_i^n$ does not lead to any computational difficulties.
Finally, note that the stiff relaxation term of~(\ref{eq:discg})
is implicit. This allows us to use a time step which is independent of $\varepsilon$.

Now, we describe the numerical scheme for the macroscopic
equation~(\ref{eq:rhod}). This equation is discretized according to
\be \frac{\vr^{n+1}_i-\vr^{n}_i}{\Delta t}
+\frac{\psi_{i+1/2}(\vr^n,g^n_K)-\psi_{i-1/2}(\vr^n,g^{n}_K)}{\Delta
x}=0\label{eq:discmac}\ee where the numerical flux is an
approximation of the total flux $F(\vr,g_K)=F(\vr)+\langle vm
g_K\rangle$ obtained by the second order MUSCL extension of a
Lax-Friedrichs like scheme: \be  \label{eq:flux}
   \psi_{i+1/2}(\vr^n,g^n_K)=\frac{1}{2}(F(\vr^n_{i},g^n_{K,i})+F(\vr^n_{i+1},g^n_{K,i+1}))-\frac{1}{2}\alpha(\vr^n_{i+1}-\vr^n_{i})+\frac{1}{4}(\sigma^{n,+}_i-\sigma^{n,-}_{i+1})
\ee In this relation, we set \be
\sigma^{n,\pm}_i=\left(F(\vr^n_{i+1},g^n_{K,i+1})\pm
\alpha\vr^n_{i+1}-F(\vr^n_{i},g^n_{K,i})\mp
\alpha\vr^n_{i}\right)\varphi(\chi^{n,\pm}_i)\ee where $\varphi$ is
a slope limiter, $\alpha$ is the largest eigenvalue of the Euler
system and \be \chi^{n,\pm}_i=\frac{F(\vr^n_{i},g^n_{K,i})\pm
\alpha\vr^n_{i}-F(\vr^n_{i-1},g^n_{K,i-1})\mp
\alpha\vr^n_{i-1}}{F(\vr^n_{i+1},g^n_{K,i+1})\pm
\alpha\vr^n_{i+1}-F(\vr^n_{i},g^n_{K,i})\mp \alpha\vr^n_{i}}\ee
where the above vectors ratios are defined componentwise.


\subsubsection{Time splitting scheme}

Here, we propose an alternative scheme based on a time splitting
between the $\partial_t h$ term and the other terms in the kinetic
equation for $g_K$~(\ref{eq:gd}). We will show in the next section
that this method preserves the positivity of the distribution
function $f=E[\vr]+g_K$ under a suitable CFL condition.

First, we solve the macroscopic equation using (\ref{eq:discmac}) as
in the previous scheme, where the numerical fluxes are defined
in~(\ref{eq:flux}). Now, for the kinetic equation on $g_K$, the time
variation of $h$ only is taken into account in~(\ref{eq:gd}) to get
the second step:
\begin{equation*}
g^{n+\frac{1}{2}}_{K,i}=g^{n}_{K,i}+\frac{g_{K,i}^{n}}{h_i^n}(h_{i}^{n+1}-h_{i}^{n}).
\end{equation*}
Note that this relation can be readily simplified in
\begin{equation}
g^{n+\frac{1}{2}}_{K,i}=g_{K,i}^{n}\frac{h_{i}^{n+1}}{h_i^n},
  \label{eq:step1bis}
\end{equation}
where, again, we point out that this equation is solved only if $h^n_i\neq
0$.

In a third step,~(\ref{eq:gd}) is discretized without the $\partial_t
h$ term, by using the same approximation as for the non-splitting
scheme. We get:
\begin{eqnarray}
              & &\frac{g^{n+1}_{K,i}-g^{n+1/2}_{K,i}}{\Delta
t}+h_i^{n+1}\left(\frac{\phi_{i+1/2}(g_K^{n+1/2})-\phi_{i-1/2}(g^{n+1/2}_K)}{\Delta
x}\right)=-\frac{\nu}{\varepsilon}g^{n+1}_{K,i} \nonumber\\
              &&-h^{n+1}_i\left(\frac{\E[\vr^{n+1}_i]-\E[\vr^{n}_i]}{\Delta
t}+\frac{\phi_{i+1/2}(\E[\vr^n])-\phi_{i-1/2}(\E[\vr^n])}{\Delta
x}\right).
         \label{eq:step2}  \end{eqnarray}

Note that for the moment, we did not mention how the new value of
the transition function $h^{n+1}$ is defined. This is done by using
some criteria that are introduced in section~\ref{sec:moving}.
Independently of this problem, we analyze in the following section
the positivity property for the distribution function.

\subsection{Positivity of the distribution function for the
discretized equations}

In this section, we prove that the splitting
scheme~(\ref{eq:step1bis}-\ref{eq:step2}) preserves the positivity
of $f$ under a suitable CFL condition. Another interesting property
of the model here proposed, the preservation of uniform flows, will
be analyzed in the appendix in comparison with different coupling
strategies proposed in the recent past
\cite{degond1,degond2,dimarco1}.
\begin{proposition}
  If $f^0\geq 0$ and $g^{0}_K=h^{0}(f^{0}-\E[\vr^{0}])$, where
  $\vr^0=\langle m E[\boldsymbol{\varrho}]\rangle$ and $0\leq h^0\leq
  1$, then scheme~(\ref{eq:discmac}--\ref{eq:step2}) satisfies
\begin{equation*}
f^n_i=\E[\vr^{n}_i]+g^{n}_{K,i}\geq 0
\end{equation*}
for every $n$ and $i$, provided that $\Delta t$ satisfies the
following CFL condition:
\begin{equation}\label{eq-CFL}
  \Delta t\leq \frac{\Delta
x}{\max(v_j)}\min_{i,v_j}\left(\frac{g^{n+1/2}_{K,i}+h^{n+1}_i\E[\vr^{n}_i]}{h^{n+1}_i(g^{n+1/2}_{K,i}+\E[\vr^{n}_i])}\right).
\end{equation}
\end{proposition}
\begin{proof}
  The idea is in fact to prove a stronger property: indeed, we can
  prove, by induction, that the positivity of
  $h^{n}_i\E[\vr^{n_i}]+g^{n}_{K,i}$ is preserved at any time.

  First, note that this relation holds at $n=0$: from the definition of
  $g^{0}_K$, we have
  $h^{0}_i\E[\vr^{0}_i]+g^{0}_{K,i}=h^{0}_if^{0}_i\geq 0$.

  Then, we assume that this relation is satisfied for some $n$, and we
  prove that it is true for $n+1$. This is done in the following three
  steps.

{\it Step 1.} \\
We first use~(\ref{eq:step2}) (where the numerical fluxes
$\phi_{i+1/2}$ are computed by the first order upwind scheme) to
explicitely compute $g^{n+1}_{K,i}$ and then to obtain:
\begin{equation}  \label{eq-hEg}
\begin{split}
 g^{n+1}_{K,i}+h^{n+1}_i\E[\vr^{n+1}_i]=
 \frac{1}{1+\nu\Delta t/\varepsilon}
 &  \biggl(
   (g^{n+1/2}_{K,i}+h^{n+1}_i\E[\vr^{n}_i])
 -\frac{|v|\Delta t}{\Delta x}h^{n+1}_i(g^{n+1/2}_{K,i}+\E[\vr^{n}_i])
 \\
&
+\frac{v^+\Delta t}{\Delta  x}h^{n+1}_{i-1}(g^{n+1/2}_{K,{i-1}}+\E[\vr^{n}_{i-1}]) \\
 & -\frac{v^-\Delta t}{\Delta x}h^{n+1}_{i+1}(g^{n+1/2}_{K,{i+1}}+\E[\vr^{n}_{i+1}])\biggr) \\
 &
+\frac{1}{1+\varepsilon/(\nu\Delta t)}h^{n+1}_i\E[\vr^{n+1}_i]
\end{split}
\end{equation}
Now, it is clear that the sign of the left-hand side depends on the sign
of $g^{n+1/2}_{K,i}+h^{n+1}_i\E[\vr^{n}_i]$ and
$g^{n+1/2}_{K,i}+\E[\vr^{n}_i]$. These two terms are studied in step 2.

\bigskip

{\it Step 2.} \\
Here, we use the definition of $g^{n+1/2}_{K,i}$
(see~(\ref{eq:step1bis})) to obtain
$g^{n+1/2}_{K,i}+h^{n+1}_i\E[\vr^{n}_i] =
\frac{h^{n+1}_i}{h^{n}_i}(h^n_i\E[\vr^{n}_i]+g^{n}_{K,i})$ which is
non-negative (due to the induction assumption). Consequently,
$g^{n+1/2}_{K,i}+h^{n+1}_i\E[\vr^{n}_i]$ is non-negative. Since
$\E[\vr^{n}_i]\geq 0$ and $ 0\leq h^{n+1}_i\leq 1$, then we also
have that $g^{n+1/2}_{K,i}+\E[\vr^{n}_i]$ is non-negative.

\bigskip

{\it Step 3.} \\
Note that step 2 shows that the last three terms of the right-hand
side of~(\ref{eq-hEg}) are non-negative. Consequently,~(\ref{eq-hEg}) shows that
$g^{n+1}_{K,i}+h^{n+1}_i\E[\vr^{n+1}_i]$ is non-negative if $\Delta t$
satisfies the CFL condition~(\ref{eq-CFL}). By induction,
$h^{n}_i\E[\vr^{n}_i]+g^{n}_{K,i}$ is non-negative for every $n$, and
for every $i$
and $v$.

Finally, using again that $\E[\vr^{n}_i]\geq 0$ and $ 0\leq h^{n}_i\leq
1$, we easily deduce that $\E[\vr^{n}_i]+g^{n}_{K,i}$ is also
non-negative.
\end{proof}

\begin{remark}
If $g^{n}_{K,i}\geq 0$, condition~(\ref{eq-CFL}) is less restrictive
than the CFL condition $ \Delta t\leq \frac{\Delta x}{\max(v_j)}$
obtained with a classical semi-implicit scheme for the original BGK
equation. On the contrary, condition~(\ref{eq-CFL}) becomes more
restrictive if the perturbation term $g^{n}_{K,i}$ is negative.
However, if we assume that $g$ is small enough (e.g.
$\E[\vr^{n}_i]\gg g^{n+1/2}_{K,i}$), then the factor $\frac{\Delta
x}{\max(v_j)}$ in~(\ref{eq-CFL}) is close to 1. Indeed: \be
\left(\frac{g^{n+1/2}_{K,i}+h^{n+1}_i\E[\vr^{n}_i]}{h^{n+1}_i(g^{n+1/2}_{K,i}+\E[\vr^{n}_i])}\right)=
\left(1+\frac{(1-h_i^{n+1})g^{n+1/2}_{K,i}}{h^{n+1}_i(g^{n+1/2}_{K,i}+\E[\vr^{n}_i])}\right)\simeq
1, \ee and~(\ref{eq-CFL}) reduces to the classical CFL for transport
$\Delta t\leq \frac{\Delta x}{\max(v_j)}$.
\end{remark}

By contrast, we justify below why we think that the non-splitting
scheme~(\ref{eq:discg}--\ref{eq:discmac}) cannot preserve the
positivity of $f$. Indeed, by using similar computations as the ones
we did for the splitting scheme, we find:
\begin{eqnarray}
& &f^{n+1}_i\geq \frac{\varepsilon/\nu}{\varepsilon/\nu+\Delta
t}\left((g^{n}_{K,i}+h^{n}_i\E[\vr^{n}_i])-\frac{v\Delta t}{\Delta
x}h^{n}_i(g^{n}_{K,i}+\E[\vr^{n}_i])+g_{K,i}^{n}\frac{h^{n+1}_{i}-h^{n}_{i}}{h^{n}_{i}}\right)\nonumber\\
&&+\frac{\varepsilon/\nu}{\varepsilon/\nu+\Delta
t}\frac{h^{n}_i\Delta t v}{\Delta
x}(g^{n}_{K,i-1}+\E[\vr^{n}_{i-1}])+\frac{\Delta
t}{\varepsilon/\nu+\Delta t}\E[\vr^{n+1}_i].\nonumber
\end{eqnarray}
Now, it is clear that the sign of $f^{n+1}$ depends on the signs of
$h^{n+1}_{i}-h^{n}_{i}$ and $g_K^{n}$, and hence cannot be controlled by a CFL
condition on $\Delta t$.


\section{Localization of Fluid-Kinetic Transitions and the Dynamic Coupling Technique}
\label{sec:moving}

One of the key points in a domain decomposition technique for gas
dynamics problems is to efficiently localize the regions where the
state of the gas departs from equilibrium, so as to describe the
solution with the appropriate microscopic model. In other words we
look for an accurate criterion the evaluation of which is
computationally inexpensive.

Here, we propose three different criteria based on the information
which can be retrieved from either the kinetic distribution function
or from the macroscopic variables. The way the localization of the
equilibrium and non-equilibrium regions evolves is described at the
end of this section.


\subsection{Analysis of Microscopic and Macroscopic Criteria}
\label{subsec:micro_criteria}

\subsubsection{Microscopic Criteria}

In regions where the kinetic or coupled kinetic/fluid models are
solved, we can use the distribution function to measure the fraction
of gas particles which are not distributed according to a Maxwellian
(as in~\cite{dimarco1}). In the same way, the fractions of momentum,
energy, and heat flux due to the non-equilibrium flux can be
measured. Consequently, in every cell where $h\neq 0$, it is
possible to evaluate the parameters \be\label{eq:lambda}
\lambda_{1,K}=\langle |g_K(v)|\rangle, \ \lambda_{2,K}=\langle v
|g_K(v)| \rangle, \
\lambda_{3,K}=\langle\frac{|v|^2}{2}|g_K(v)|\rangle, \
\lambda_{4,K}=|\langle v\frac{|v|^2}{2}g_K(v)\rangle|. \ee Note that
the first three values above will be zero if we use $g_K$ instead of
$|g_K|$, while the last one is in general different from zero. For
compatibility with the macroscopic criterion introduced in the
sequel, the definition of $\lambda_4$ in~(\ref{eq:lambda}) has been
preferred to the alternate definition $\lambda_{4,K}=\langle|
v\frac{|v|^2}{2}g_K(v)|\rangle$. The four parameters are computed in
our code with the following quadrature formula
\begin{eqnarray}
\lambda_{1,K}=\sum_{j}|g_{K,j}| \Delta v, \ \lambda_{2,K}=\sum_{j
}v_j
|g_{K,j}|\Delta v,\nonumber\\
\lambda_{3,K}=\sum_{j}\frac{|v_j|^2}{2} |g_{K,j}|\Delta v, \
\lambda_{4,K}=|\sum_{j}v\frac{|v_j|^2}{2} g_{K,j}\Delta v|,
\end{eqnarray} where $g_{K,j}(t,x) \approx g_K(x,v_{j},t)$.
In kinetic and buffer zones $\Omega_1\cup \Omega_3$ (where $h\neq
0$), the discrepancy between the fluid and kinetic models can be
measured by the following parameters \be
\beta^{n}_{1,i,K}=\frac{\lambda^{n}_{1,i,K}}{\varrho^n_i}, \
\beta^{n}_{2,i,K}=\frac{\lambda^{n}_{2,i,K}}{\varrho_i^n u_i^n}, \
\beta^{n}_{3,i,K}=\frac{\lambda^{n}_{3,i,K}}{\varrho_i^n
e_i^n}\label{eq:beta}, \ee or alternatively, we can use the value of
the heat flux relatively to the value of the equilibrium energy flux
\be \beta_{4,i,K}^{n}=\frac{\lambda_{4,i,K}^{n}}{|\langle
v\frac{1}{2}|v|^2
  E[\rho_i^n]  \rangle|} .\label{eq:crit1}
\ee In order to define a unique variable which permits to switch
from one model to the other one in every regime and in every region
($\Omega_i, \ i=1,2,3$), we choose $\beta_4$ as the breakdown
parameter. Indeed, as shown below, it is possible to estimate this
quantity also in fluid regimes. By using this criterion, the
transition function can be defined by an appropriate function that
maps $\beta_4$ to the interval $[0,1]$:
$h_i^n=\mathfrak{f}(\beta^{n}_{4,i,K})$. Such a mapping is defined
in section~\ref{subsec:adaptive_algo}.

\subsubsection{Macroscopic criteria}

The previous analysis is quite efficient and does not induce
expensive additional computations, since the perturbation term $g_K$
is already known in regions where the parameter $\beta_4$ has to be
computed. However, if we decide to use this criterion in the whole
domain, the cost will be equivalent to the cost of computing the
solution of the kinetic model in the whole domain. For this reason,
it is necessary to look for others indicators that are based on the
equilibrium values only. The most obvious one is the local Knudsen
number $\varepsilon$ which is defined as the ratio of the mean free
path of the particles $\lambda_{path}$ to a reference length $L$:
\be \varepsilon=\lambda_{path}/L, \label{eq:knudsen}\ee where the
mean free path is defined by
$$
\lambda_{path}=\frac{kT}{\sqrt{2}\pi p \sigma_{c}^{2}} ,
$$
with $k$ the Boltzmann constant equal to $1.380062\times 10^{-23}
JK^{-1}$, $p$ the pressure and $\pi \sigma_{c}^2$ the collision
cross section of the molecules. The Knudsen number is determined
through macroscopic quantities and can be computed in the whole
domain with a minimum additional cost. Now, in order to take into
account the elementary fact that, even in extremely rarefied
situations, the flow can be in thermodynamic equilibrium, according
to Bird \cite{bird}, the reference length is defined as \be
L=\min\left(\frac{\varrho}{\partial\varrho/\partial x},\frac{\varrho
u}{\partial\varrho u/\partial x},\frac{\varrho e}{\partial \varrho
e/\partial x}\right).\ee  According to~\cite{Levermore}
and~\cite{Kobolov}, the fluid model is accurate enough if the local
Knudsen number is lower than the threshold value $0.05$. It is
argued that, in this way, the error between a macroscopic and a
microscopic model is less than 5\% \cite{Wang}. This parameter has
been extensively used in many works and is now considered in the
rarefied gas dynamic community as an acceptable indicator. We notice
that the local Knudsen number takes into account both the physics
(with the measure of the mean free path and the identification of
large gradients) and the numerics (through the approximation of
derivatives on the mesh).

In the present work we propose an alternative criterion, based on
the analysis of the micro-macro decomposition. We will apply this
criterion only in fluid regions. For this reason, we consider the
equation for the {\it non-localized} perturbation $g$
(see~(\ref{eq:s2})) in its discretized form
\begin{eqnarray}
              & &\frac{g^{n+1}_{i}-g^{n}_{i}}{\Delta
t}+\left(\frac{\phi_{i+1/2}(g^{n})-\phi_{i-1/2}(g^{n})}{\Delta
x}\right)=-\frac{\nu}{\varepsilon}g^{n+1}_{i} +\nonumber\\
              &&-\left(\frac{\E[\vr^{n+1}_i]-\E[\vr^{n}_i]}{\Delta
t}+\frac{\phi_{i+1/2}(\E[\vr^n])-\phi_{i-1/2}(\E[\vr^n])}{\Delta
x}\right).\label{eq:discg1}  \end{eqnarray} Let us consider a point
$x_i$ which lies in the macroscopic region at time $t^n$ , i.e.
$g^{n}_{i}\equiv 0$. If, in addition, we assume that $g^n$ is close
to zero in the neighboring cells (which should be true if the
transition function is smooth enough), then we obtain \be
g^{n+1}_{i}=-\frac{\varepsilon/\nu}{\varepsilon/\nu+\Delta
t}\left(\E[\vr_i^{n+1}]-\E[\vr_i^{n}]\right)-\frac{\varepsilon/\nu\Delta
t}{(\varepsilon/\nu+\Delta t)\Delta x}\left(
\phi_{i+1/2}(\E[\vr^{n}])-\phi_{i-1/2}(
\E[\vr^{n}])\right)\label{eq:pert} \ee Using this relation, we are
able to evaluate the mismatch between the macroscopic fluid
equations and the kinetic equation. In fact, note that, in the
macroscopic equation (\ref{eq:s1}) we do not know how to evaluate
the kinetic term $\partial_x <vmg>$, except if we solve, at the same
time, the kinetic and the macroscopic problem
(\ref{eq:s1}-\ref{eq:s2}). However, at point $x_i$, thanks to
(\ref{eq:pert}), the perturbation term only depends on the
Maxwellian distribution which in turn only depends on the
macroscopic variables at the previous time step. Then,
integrating (\ref{eq:pert}) over the velocity space we get:
\begin{eqnarray}\nonumber
  &&\int_{\R^3}v m g^{n+1}_{i}dv = -\frac{\varepsilon/\nu}{\varepsilon/\nu+\Delta
t}\left(\int_{\R^3}v m \E[\vr_i^{n+1}]
dv-\int_{\R^3}v m \E[\vr_i^{n}] dv\right)+ \\
  &&-\frac{\varepsilon/\nu\Delta
t}{(\varepsilon/\nu+\Delta t)\Delta
x}\left[\phi_{i+1/2}\left(\int_{\R^3}v v m \E[\vr^{n}]
dv\right)-\phi_{i-1/2}\left(\int_{\R^3}v v m \E[\vr^{n}]
dv\right)\right]\label{eq:equil}
\end{eqnarray}
where in one space dimension we have \be \int_{\R^3}v E[\vr]
dv=\varrho u, \ \int_{\R^3}v^2 E[\vr] dv=\varrho (u^{2}+3\theta)\ee
and \be \int_{\R^3}v^3 E[\vr] dv=\varrho u(u^{2}+5\theta), \
\int_{\R^3}v^4 E[\vr] dv=\varrho (u^{4}+8u^2\theta+5\theta^2)\ee

Now, the last step is to measure the ratio of the non-equilibrium
fraction to the equilibrium one. Observe that in the one dimensional
case the only non-zero moment of the non-equilibrium term $g$ is the
heat flux. Thus, like for the microscopic criterion
(\ref{eq:crit1}), we measure the ratio of the heat flux to the
energy flux: \be
\beta_{4,i}^{n}=\frac{\lambda_{4,i}^{n}}{|F_3(\vr^{n}_i)|}, \
\lambda_{4,i}=\int_{\R^3}v \frac{|v|^2}{2}
g^{n+1}_{i}dv\label{eq:crit2}\ee and define, as before, the value of
the transition function $h_i^n$ at this point as an appropriate
function of $\beta_4$:\be h^{n}_{i}=\mathfrak{f}(\beta^{n}_{4,i}), \
0\leq h_i^n\leq 1\ee In practice, we use the same function
$\mathfrak{f}$ to evaluate $h_i^n$ in all the computational domain,
but while $\beta_{4,i}^{n}$ is defined by (\ref{eq:crit1}) in kinetic
regions, it is defined by (\ref{eq:crit2}) in fluid regions.

The quantities (\ref{eq:crit1})-(\ref{eq:crit2}) (which will be
referred to as breakdown parameters) furnish a \emph{true} measure
of the model error, while the local Knudsen criterion is rather a
physical-based criterion. In the numerical test section, we will
compare these two strategies.

\begin{remark}~
In the above analysis we have discarded the convection term $(v
\partial_x g)$. This can be justified by the
hypothesis of smoothly varying transitions, which means that this term
is supposed to be small. Anyway, it is possible to take it into
account. For example, through an upwind discretization, we obtain
$$
v\partial_x g^{n}_{i}=\left\lbrace
\begin{array}{ll}
\displaystyle v
\frac{g^{n}_{i}-g^n_{i-1}}{\Delta x} \ if \ v\geq 0\\
\displaystyle v \frac{g^n_{i}-g^n_{i-1}}{\Delta x} \ if
\ v<0\\
\end{array}
\right.
$$
Now, as before, $g_{i}^{n}\equiv 0$, while $g_{i+1}^{n}$ or
$g_{i-1}^{n}$ assume known values, which can be different from zero
if the transition function $h$ appears to be greater than zero in
these cells ($h^{n}_{i+1}\neq 0$ or $h^{n}_{i-1}\neq 0$).
\end{remark}


\subsection{Kinetic/Fluid Coupling Algorithm}
\label{subsec:adaptive_algo}

We now describe the kinetic/fluid coupling algorithm.

Define $\beta_{thr}$ and $\beta^*_{thr}\leq\beta_{thr}$ as the
maximum errors that we can afford by using the fluid model instead
of the kinetic one. Then:
\medskip

Assume $\vr^n, g^n_K, h^n$ are known in the whole space domain at
time $n$.
\begin{enumerate}
\item Advance the macroscopic equation in time by using
scheme~(\ref{eq:discmac}) and obtain $\vr^{n+1}$.
\item Compute the equilibrium parameter $\beta^n_{4,i}$
in every space cell for which $h=0$ through relation
(\ref{eq:crit2}).
\item If $\beta^n_{4,i}\geq \beta_{thr}$ then set $h_i^{n+1}=1$ which means that
$x_i$ at time $t^{n+1}$ belongs to the kinetic region; if
$\beta^n_{4,i}< \beta_{thr}^*$ then set $h_i^{n+1}=0$, which means
that $x_i$ at time $t^{n+1}$ belongs to the fluid region.
\item If $\beta^*_{thr}\leq \beta_{4,i}^n\leq \beta_{thr}$ then set $h_i^{n+1}=\frac{\beta_{4,i}^n-\beta^*_{thr}}{\beta_{thr}-\beta^*_{thr}}$,
which means that $x_i$ at time $t^{n+1}$ belongs to the buffer
region.
\item Advance the kinetic equation in time by using
scheme~(\ref{eq:step1bis})--(\ref{eq:step2}) and obtain $g^{n+1}_K$.
\item Compute the equilibrium parameter $\beta^n_{4,i,K}$
in every space cell for which $h\neq 0$ through relation
(\ref{eq:crit1}).
\item If $\beta^n_{4,i,K}\geq \beta_{thr}$ then set $h_i^{n+1}=1$ which means that
$x_i$ at time $t^{n+1}$ belongs to the kinetic region, if
$\beta^n_{4,i,K}< \beta_{thr}^*$ then set $h_i^{n+1}=0$, which means
that $x_i$ at time $t^{n+1}$ belongs to the fluid region.
\item If $\beta^*_{thr}\leq \beta_{4,i,K}^n\leq \beta_{thr}$ then set $h_i^{n+1}=\frac{\beta_{4,i,K}^n-\beta^*_{thr}}{\beta_{thr}-\beta^*_{thr}}$,
which means that $x_i$ at time $t^{n+1}$ belongs to the buffer
region.
\item Re-project the non equilibrium part of the distribution
function $g^{n}_K$ through the relation (\ref{eq:step1bis}).
\end{enumerate}

\begin{remark}~
\begin{itemize}
\item In the above algorithm the steps 2-3-4 can be substituted
by equivalent steps in which the breakdown criterion is the local
Knudsen number, with convenient threshold values.
\begin{enumerate}
\setcounter{enumi}{1}
\item Compute the local Knudsen number $\varepsilon^n_i$ in every space cell for which $h=0$ through
relation (\ref{eq:knudsen}).
\item If $\varepsilon^n_i\geq \varepsilon_{trh}$ then set $h_i^{n+1}=1$ which means that
$x_i$ at time $t^{n+1}$ belongs to the kinetic region, if
$\varepsilon^n_i< \varepsilon_{trh}^*$ then set $h_i^{n+1}=0$, which
means that $x_i$ at time $t^{n+1}$ belongs to the fluid region.
\item If $\varepsilon^*_{trh}\leq \varepsilon_i^n\leq \varepsilon_{trh}$ then set $h_i^{n+1}=\frac{\varepsilon_i^n-\varepsilon^*_{trh}}{\varepsilon_{trh}-\varepsilon^*_{trh}}$,
which means that $x_i$ at time $t^{n+1}$ belongs to the buffer
region.
\end{enumerate}
In the next section, we report comparisons between using the Knudsen
number and the new breakdown parameter.
\item At the beginning of the computation the full domain is
supposed in thermodynamical equilibrium. During the computation,
kinetic regions are created. Some of these regions can become even
one cell thick, merge or split. The transition function can also
pass from 0 to 1 and vice-versa in a single time step and with jumps
in space. Every step is completely automatic in each simulation, no
additional parameters are used.
\item Compared to the previous work \cite{dimarco1} in
which the buffer regions where fixed once for all at the beginning
of the computation, the present strategy consists in making the
buffer regions dependent on the current state of the gas through the
functions (\ref{eq:crit1})-(\ref{eq:crit2}) and the thresholds
values. This modification leads to a considerable improvement in
accuracy, flexibility and usability of the method.
\end{itemize}
\end{remark}


\section{Numerical tests}
\label{sec_tests}


\subsection{General setting}
\label{subsec_gen_set}

In this section, we present several numerical results to highlight
the performances of the method. By using unsteady test problems, we
emphasize the deficiencies of the static decomposition method. As in
\cite{dimarco1} we start with an unsteady shock test problem. Even
in this simple situation, a standard static domain decomposition
fails in its scope. Indeed, the shock moves in time. Thus in
rarefied regimes, it is necessary to use a kinetic solver in the
full domain, which turns to be a quite inefficient method. On the
other hand, with our algorithm, we reduce the computationally
expensive regions to a small zone compared to the full domain.

Next,
we use our scheme to compute the solution of the Sod test. Here some
new difficulties arise. Indeed, contact discontinuities and
rarefaction waves appear but, as described below, the method
efficiently deals with such situations.

Finally, in the third test, a
blast wave simulation is performed. In spite of the complexity of
the solution, the algorithm shows a very good behavior, creating,
deleting or merging zones together and obtaining fast and precise
results.

In order to obtain the correct equation of state with only one
velocity-space dimension, we use the following model:
$$
\begin{array}{l}
\partial_{t}\left(
\begin{array}{ll}
F \\G
\end{array}
\right) \ + \ v\partial_{x}\left(
\begin{array}{ll}
F \\G
\end{array}
\right) \ = \ \nu\left(\begin{array}{ll} M_{F}-F
\\TM_{F}-G
\end{array}\right)
\end{array} .
\label{eq:BGK3}
$$
It is obtained from the full three-dimensional Boltzmann-BGK
system by means of a reduction technique \cite{Huang}.
In this model, the fluid energy is given by
$$
E=\int_{\R} (\frac{1}{2} v^{2} F + G) dv .
$$
This model permits to recover the correct hydrodynamic limit given
by the standard Euler system even with a one-dimensional velocity
space.

The collision frequency is given by
$\nu=\tau^{-1}=(\frac{\mu}{p})^{-1}$ where
$\mu=C\cdot\theta^{\omega}$. We choose gas hydrogen for our
simulations. Thus $C=1.99\times 10^{-3}$, $\omega=0.81$ and
$R=208.24$ \cite{bird}.

In all tests the time step is given by the minimum of the maximum
time steps allowed by the kinetic and fluid schemes. This means that
no attempts have been made to try to increase the computational
efficiency by means of a reduction of the number of necessary
effective steps for the less restrictive scheme, by, e.g. freezing
one model in time, while the other one is advanced. Indeed, such a
reduction still requires further investigations. Thus, the global
speed-up is only due to the reduction of the sizes of the kinetic
and buffer regions inside the domain. This reduction is achieved through a correct prediction of the evolution of the
transition function and the use of efficient criteria for the determination of the equilibrium regions. We point out that no a-priori choices on the dimension of
the buffer and position of the different regions are done in all the
tests. Instead all the procedure is automatic and determined by the step by step algorithm presented in the previous section. For all
tests, we set $\beta_{thr}=10^{-3}$ and $\beta^*_{thr}=10^{-4}$.


\subsection{Unsteady Shock Tests}
\label{subsec_unst_shock}

We consider an unsteady shock that propagates from left to right in
the computational domain $x=-20 \ m$, $x=20 \ m$. The shock is
produced pushing the hydrogen gas against a wall which is located on
the left boundary. We consider that the particles are specularly
reflected and that the wall instantaneously adopts the temperature
of the gas. This effect is numerically simulated by setting
macroscopic variables in ghost cells (two cells for a second order
scheme) beyond the left boundary with parameters $\varrho, \ T$
equal to the values of the first cell while the momentum is set
opposite. In the non equilibrium case $g_K$ is also different from
zero in the ghost cell, and is equal to a specularly reflected copy
of $g_K$ in the first cell. At the right boundary, we mimic the
introduction of the gas by adding two ghost cells where, at each
time step, the initial values for density, momentum and energy are
fixed. The gas is supposed in thermodynamic equilibrium, which
implies that $g_K=0$. The computation is stopped at the final time
$t=0.04$ s. There are $300$ cells in physical space and $40$ cells
in velocity space. We do not use a finer mesh because the scheme is
second order. Symmetric artificial boundaries in velocity space are
fixed at the beginning of the computation through the relation
$v_b=\pm C_1\max(\sqrt{RT_W})$, where $R$ is the gas constant, $C_1$
is a parameter normally fixed equal to $4$ and $T_W$ is a
temperature set equal to the maximum attainable temperature, which
is obtained by supposing that all the kinetic energy is transformed
into thermal energy. The transition function $h$ is initialized as
$h=0$ (fluid region) everywhere.

In the first test the initial conditions are $\varrho=10^{-6}$
kg/m$^{3}$ for the mass density, $u=-900$ m/s for the mean velocity
and $T=273$ K for the temperature. In Figure \ref{UST1.1} we have
reported the mass density on the left and the velocity on the right.
>From top to bottom, time increases from $t=10^{-2} \ s$ (top) to
$t=4\times 10^{-2} \ s$ (bottom), with $t=2 \times 10^{-2} \ s$
middle top and $t=3\times 10^{-2} \ s$ middle bottom. In Figure
\ref{UST1.2} we have reported the temperature on the left, the
transition function, the heat flux and the local Knudsen number on
the right. From top to bottom the same instants of time as for the
previous Figure are shown. In the plots regarding the macroscopic
variables we reported the solution computed with our algorithm
(mic-mac in the legend of the Figures), the solution computed with a
kinetic solver and the solution computed with a macroscopic fluid
solver. Magnifications of the solutions close to non equilibrium
regions are given for clarity.

As soon as the simulation starts on the left boundary, the
transition function $h$ increases from zero to one, which means that
the solution is computed with the kinetic scheme, while in the rest
of the domain the solution is still computed with the fluid scheme
($h=0$). When the shock starts to move towards the right, we notice
a splitting of the kinetic region. One very narrow region still
continues to follow the shock and one remains close to the left
boundary. Once that the threshold values of the breakdown parameters
$\beta$ and $\beta_K$ are fixed, the procedure automatically
determines the sizes of the kinetic and buffer regions.

We repeat the simulation with a lower initial density
$\varrho=10^{-7}$ kg/m$^{3}$. This yields different results which
are reported in Figure \ref{UST2.1} for the density and velocity and
in Figure \ref{UST2.2} for the temperature, local Knudsen, heat flux
and transition function. The results are obtained with the same
criteria as in the previous test case. Again at the beginning $h$ is
set equal to zero (fluid), but now the shock is much less sharp and
the non equilibrium region becomes larger.

We observe that in this first test the discrepancy between the fully
kinetic solver and the coupling strategy is not perceivable, while
the computational time is reduced in proportion to the ratio between
the areas of the kinetic and buffer regions compared to the entire
domain. Thus, in the first case we have a reduction of approximately
$70\%$ of the computational time while in the second one the
reduction is only the $20\%$. Further reductions are possible
through code optimization.


\subsection{Sod Tests}
\label{subsec_sod}

In these second series of tests, we consider the classical Sod
initial data in a domain which ranges from $-20 \ m$ to $20 \ m$.
The numerical parameters are the same as for the previous examples.
Thus respectively $300$ and $40$ mesh points in physical and
velocity space are used. Symmetric artificial boundary condition are
fixed in velocity space at the same points $v_b=\pm
C_1\max(\sqrt{RT_W})$, while Neumann boundary condition are chosen
both for the kinetic (if necessary) and fluid models. Finally the
simulations are initialized with a thermodynamic equilibrium with
$h=0$ and $g_K=0$ everywhere.

In the first case, we take the following initial conditions: mass
density $\varrho_{L}=2\times 10^{-5}$ kg/m$^{3}$, mean velocity
$u_{L}=0$ m/s and temperature $T_{L}=273.15$ K if $-20 \leq x \leq 0
\ m$, while $\varrho_{R}=0.25\times 10^{-5}$ kg/m$^{3}$, $u_{R}=0$
m/s, $T_{R}=218.4$ K if $0 \leq x \leq 20 \ m$. The results are
reported in Figure \ref{sod1.1} for the density (left) and velocity
(right) and in Figure \ref{sod1.2} for the temperature (left),
Knudsen number, heat flux and transition function (right). Snapshot
at increasing times are displayed top to down, corresponding
successively to $t=6\times 10^{-3} \ s$, then $t=1.2\times 10^{-2} \
s$, $t=1.8\times 10^{-2} \ s$ and finally $t=2.4\times 10^{-2} \ s$.
Again we provide magnifications of the solution close to non
equilibrium zones in order to highlight the differences between the
three different schemes: the macroscopic and kinetic ones and the
coupling strategy (mic-mac in the legend). We observe that due to
the initial shock, a kinetic region appears immediately and starts
to grow in time, but as soon as the different non equilibrium
regions separate, the kinetic region itself splits into three: one
around the rarefaction wave, one around the contact discontinuity,
and one around the shock. Even with magnifications it is not
possible to perceive differences between the kinetic model and the
coupling strategy, even though the kinetic regions remain very tiny,
permitting a fast computation.

The simulation is repeated, with a lower initial density
$\varrho_{L}=5\times 10^{-6}$ kg/m$^{3}$ and $\varrho_{R}=0.75\times
10^{-6}$ kg/m$^{3}$, and the results are displayed in Figure
\ref{sod2.1} for the density and velocity and in Figure \ref{sod2.2}
for the temperature, heat flux, local Knudsen and transition
function. The same qualitative features as in the previous test can
be observed, the only difference being that the kinetic regions are
thicker, which means that the non equilibrium zone is larger. This
is clearly visible from the plots of the macroscopic quantities: the
difference between the kinetic and fluid models is significant in a
large portion of the domain.

We finally observe that compared to \cite{dimarco1}, in which a
similar scheme was developed, we are able to capture small
discrepancies between the kinetic and macroscopic models in very
tiny regions. This turns to be a much more efficient use of the
domain decomposition technique. The computational time reduction is
of the order of $70\%$ and $60\%$ respectively for the two tests
compared to a kinetic solver.


\subsection{Blast Wave Tests}
\label{subsec_Blast_wave}

In this paragraph we present two interacting symmetric blast waves
in hydrogen gas. The domain ranges from $x=0 \ m$ to $x=1 \ m$, while
the numerical parameters in terms of mesh and velocity space
boundaries are the same as in the previous tests. The physical
boundaries are represented by two specularly reflecting walls, on
which impinging particles are re-emitted in the opposite direction
with the same velocity (in magnitude). Mass is conserved at the walls
which additionally are supposed to adopt the gas temperature
instantaneously. These effects are obtained like in the unsteady
shock test with two ghost cells (four for a second order scheme), in
which the same macroscopic values as those of the first and last
cells are imposed, except for momentum which changes sign. The
perturbation function $g_K$ can in general be different from zero
and assumes the same values of its corresponding counterpart in the
first and last cell, with a sign change in the velocity variable. At
the beginning, we set $h=0$ and $g_K=0$ everywhere.

In the first test, the initial data are: \be \nonumber
\varrho=10^{-3} kg/m^{3} \ \ u=200 \ m/s \ \ T=10000 \ K \ if \
x\leq 0.1\ee \be \nonumber \varrho=10^{-3} kg/m^{3} \ \ u=-200 \ m/s
\ \ T=10000 \ K \ if \ x\geq 0.9\ee \be \nonumber\varrho=10^{-3}
kg/m^{3} \  \ \ u=0 \ m/s \ \ \ \ \ \ T=50 \ K \ if \ 0.1\leq x\leq
0.9\ee The results in terms of the density and mean velocity are
reported in Figure \ref{blast1.1}, while the temperature, local
Knudsen, heat flux and transition function are reported in Figure
\ref{blast1.2}. The displayed results are for increasing times
$t=10^{-4} \ s$ to $t=4\times 10^{-4} \ s$ from top to bottom. Again
we plot the kinetic, the fluid schemes and the coupling strategy
(micmac scheme) in each Figure and magnifications close to non
equilibrium regions are provided.

In the second test the initial density is decreased, and we use: \be
\nonumber \varrho=10^{-4} kg/m^{3} \ \ u=200 \ m/s \ \ T=10000 \ K \
if \ x\leq 0.1\ee \be \nonumber \varrho=10^{-4} kg/m^{3} \ \ u=-200
\ m/s \ \ T=10000 \ K \ if \ x\geq 0.9\ee \be
\nonumber\varrho=10^{-4} kg/m^{3} \  \ \ u=0 \ m/s \ \ \ \ \ \ T=50
\ K \ if \ 0.1\leq x\leq 0.9\ee The results obtained with this
second set of data are reported in Figure \ref{blast2.1} and
\ref{blast2.2}. We observe that starting from a situation where the
fluid model is used almost everywhere we end up in the opposite
situation where the kinetic model is used in the whole domain($h=1 \
\forall x$). Thus, while a static domain decomposition technique
leads to similar computational times as a fully kinetic resolution,
the coupling strategy leads to a speed up of about $40\%$ for the
first test and $30\%$ for the second test, compared to a fully
kinetic solver.


\section{Conclusion}
\label{sec_conclu}

In this paper we have presented a moving domain decomposition method
which provides an efficient way to deal with multiscale fluid
dynamic problems. Regions far from thermodynamical equilibrium are
treated with a kinetic solver. The method is based on the
micro-macro decomposition technique developed by
Degond-Liu-Mieussens \cite{degond3} in which macroscopic fluid
equations are coupled with a kinetic equation which describes the
time evolution of the perturbation from equilibrium. The method
consists in splitting the distribution function into an equilibrium
part and a non-equilibrium part, together with the introduction of
buffer zones and transition functions as proposed in \cite{degond3},
\cite{degond1} and \cite{dimarco1} to smoothly pass from the
macroscopic model to the kinetic model and vice versa.

In order to build up an efficient method that can be used in a wide
spectrum of situations, and by contrast to \cite{degond3}, we
consider the possibility of moving the different domains like in
\cite{dimarco1}, using however, a decomposition technique which
shows enhanced performances. Moreover, we have developed a scheme
which is able to automatically create, cancel and move as many
kinetic, fluid or buffer regions as necessary. The method relies on
the proper combination of the two equilibrium criteria we have
identified and on a priori tolerance that we decide to accept. An
important point also resides in the introduction of a new criterion
for the breakdown of the fluid model, which is able to measure the
discrepancy between the kinetic and the fluid model in a much more
accurate way then the mere Knudsen number. Finally, we have proved
that the coupling strategy preserves positivity under a CFL
condition, and the uniform flows, also in the fully discrete case.

The last part of the work is devoted to several numerical tests. The
results clearly demonstrate the advantages of this method over
existing ones. The method captures the correct kinetic behaviors even
in transition regions and provides significant improvements in terms
of computational speedup while maintaining the accuracy of a kinetic
solver.

In the future we will extend the method to make it consistent with the
Navier-Stokes equations instead of the Euler model. This step can further
improve the technique allowing very narrow kinetic zones and
providing considerable speed-up while maintaining accuracy. We will
also explore the use of Monte Carlo techniques for the full
Boltzmann equation and that of time sub-cycling for the two models.
We finally observe that the computational speed-up will
significantly increase for two or three dimensional simulations,
which we intend to carry out in the future. To conclude, because
multiscale effects are very important also in many others fields we
plan to extend our method to other models.

\appendix
\section{Preservation of uniform flows}
\label{sec_preservation of uniform flows}

Preservation of uniform flows is a very important property, which
prevents oscillations to appear when the transition region is
located in a domain where the flow is smooth. In this appendix we
compare the properties of the present micro-macro coupling strategy to those of the
decomposition methods of \cite{degond1,degond2,dimarco1} regarding
preservation of uniform flows.

In \cite{degond3}, it has been demonstrated that the micro-macro
model is able to preserve uniform flows in the continuous case. This
property is also true for the decomposition used in
\cite{degond1,degond2,dimarco1}, but only when the collision
operator has specific properties (which are satisfied by the
Boltzmann and BGK operator). In this appendix, we show that the present
micro-macro coupling strategy is able to preserve uniform flows also in the
discrete case independently of the choice of the discretization
scheme. We observe that this property does not hold in
the general case for the decompositions used in
\cite{degond1,degond2,dimarco1}. The satisfaction of this property by the present micro-macro coupling strategy constitutes a very big advantadge of this method over the previous ones \cite{degond1,degond2,dimarco1}.

As an example, we consider the decomposition used in \cite{dimarco1}, which reads \be
\frac{\partial \vr_{L}}{\partial t} + (1-h)
\partial_x F(\vr_L)+ (1-h) \partial_x \langle v m f_R\rangle= -
\vr \partial_{t} h \label{eq:old1} \ee
\begin{equation}
\partial_t f_{R} + h v \partial_x f_{R} + h v \partial_x E[\vr] = h\frac{\nu}{\varepsilon} (E[\vr]-f) + f \partial_{t} h ,
\label{eq:old2}
\end{equation}
where the distribution function is defined by $f=f_R+E[\vr_L]$,
$f_R=hf$ and $E[\vr_L]=(1-h)E[\vr]$. In this model the solution of
the full kinetic problem is given by $f_{R}$ if $x\in \Omega_1$, by
$E[\vr_L]$ if $x\in \Omega_2$ and by $E[\vr_L]+f_{R}$ if $x \in
\Omega_3$. This is due to the fact that $f_R=0$ for $x\in \Omega_2$,
$E[\vr_L]=0$ for $x\in \Omega_1$, while in $\Omega_3$ they are both
different from zero and the global solution is obtained as a sum of
the two partial solutions. We refer to the above cited papers for
details.

Here we recall that, in \cite{degond1,degond2,dimarco1}, small oscillations  appear inside the transition regions except when the scheme used for the fluid part is an exact discrete velocity integration of the scheme used for the kinetic part (in this case, we say that the two schemes are 'compatible', and are 'incompatible' otherwise). These oscillations appear even in situations where preservation of uniform flows is true for the continuous model. To circumvent this problem, in \cite{degond1,degond2,dimarco1}, we
used a standard shock-capturing scheme (such as e.g. the Godunov scheme) for the Euler equations in the pure fluid
region (i.e. $h=0$), but we converted to a compatible scheme with the discretization of the kinetic model  inside
the buffer zones (see \cite{dimarco1} for details). However, this
choice introduces some implementation difficulties and reduces the
performances. Indeed, a compatible scheme with the discretization of the kinetic model has intrinsically the same cost as the full kinetic solver, and the coupling strategy is twice more costly than the mere kinetic model in all the buffer region.

In order to prove the property that uniform flows are not
preserved by the decomposition (\ref{eq:old1}-\ref{eq:old2}) in the
discrete case, we focus on the first one of the two equations. The same
considerations hold for the other one. If the initial data is such
that $f=E[\vr]$ we have
\begin{eqnarray}
   && \frac{\partial \vr_{L}}{\partial t} + (1-h)
\partial_x F(\vr_L)+ (1-h) \partial_x \langle v m f_R\rangle+
\vr \partial_{t} h= \label{eq:old3} \nonumber\\
   &&= (1-h)\left(\frac{\partial \vr}{\partial t} + (1-h)
\partial_x F(\vr)+ h \partial_x \langle v m E[\vr]\rangle-h'F(\vr)+h'\langle v m
E[\vr]\rangle\right)= \nonumber\\
&&=(1-h)\left(\frac{\partial \vr}{\partial t} + (1-h)
\partial_x F(\vr)+ h \partial_x \langle v m
E[\vr]\rangle\right).\nonumber
\end{eqnarray}
In the above equation, the time derivative with respect to $h$
disappears and so does the flux, using that $F(\vr)=\langle v m
E[\vr]\rangle$. In the continuous case it is also true that \be(1-h)
\partial_x F(\vr)+ h \partial_x \langle v m
E[\vr]\rangle=\partial_x F(\vr)=\partial_x \langle v m
E[\vr]\rangle\label{eq:fluxdisc}\ee and so, uniform flows are
preserved. However, when we discretize the fluxes according to $\partial_x
F(\vr)=(\phi_{i+1/2}(\vr)-\phi_{i-1/2}(\vr))/\Delta x$ and
$\partial_x \langle v m
\E(\vr)\rangle=(\psi_{i+1/2}(\E[\vr])-\psi_{i-1/2}(\E[\vr]))/\Delta
x$, the equality (\ref{eq:fluxdisc}) does not hold anymore. In fact,
in the general case we have
\begin{eqnarray}
  && (1-h)\left( \frac{\phi_{i+1/2}(\vr)-\phi_{i-1/2}(\vr)}{\Delta
x}\right)+ h\left(
\frac{\psi_{i+1/2}(\E[\vr])-\psi_{i-1/2}(\E[\vr])}{\Delta
x}\right)\neq \nonumber\\
  &&\neq\left( \frac{\phi_{i+1/2}(\vr)-\phi_{i-1/2}(\vr)}{\Delta
x}\right)\neq\left(
\frac{\psi_{i+1/2}(\E[\vr])-\psi_{i-1/2}(\E[\vr])}{\Delta
x}\right).\nonumber
\end{eqnarray}
Thus, if two incompatible numerical schemes are used to discretize
the kinetic and fluid fluxes, oscillations in the solution can
appear as documented in \cite{dimarco1}. However, we observe that,
using the same numerical flux is not sufficient to ensure
preservation of uniform flows through the decomposition
(\ref{eq:old1}-\ref{eq:old2}). To that aim, consider a generic
discretization of the coupled system (\ref{eq:old1}-\ref{eq:old2}):
\begin{align}
\nonumber
\boldsymbol{\varrho}^{n+1}_{i,L}&=\boldsymbol{\varrho}^{n}_{i,L}
-(1-h_{i})\frac{\Delta t}{\Delta
x}\left(\psi_{i+1/2}(\boldsymbol{\varrho}_{L}^n)-\psi_{i-1/2}(\boldsymbol{\varrho}
_{L}^{n})\right)
\\&
-(1-h_{i})\frac{\Delta t}{\Delta
x}\sum_{k}\mathbf{m}_{k}\left(\psi_{i+1/2}(f_{k,R}^{n})-\psi_{i-1/2}(f_{k,R}^{n})\right)\Delta
v \label{eq:H1scheme},
\end{align}
\begin{align}
\nonumber f^{n+1}_{k,i,R} = f_{k,i,R}^{n} & -h_{i} \frac{\Delta
t}{\Delta x} \left( \phi_{i+1/2} (f^n_{k,R}) - \phi_{i-1/2}
(f^n_{k,R}) \right)
\\ \label{eq:Kscheme} &-h_{i} \frac{\Delta t}{\Delta x} \left( \phi_{i+1/2}
(\E_k [\boldsymbol{\varrho}^n_{L}] ) - \phi_{i-1/2} (\E_k
[\boldsymbol{\varrho}^n_{L}] ) \right)
\\ \nonumber & + h_{i} \frac{\Delta t \nu} {\varepsilon} \left(
\E_{k} [\boldsymbol{\varrho}^n_{i}] - f^{n}_{k,i} \right),
\end{align}
where a discrete velocity model has been used to resolve the kinetic
equation (\ref{eq:Kscheme}) with $\mathbf{m}_{k}$ the discretized
collision invariants. The function $\phi_{i\pm 1/2}, \psi_{i\pm
1/2}$ are two different generic numerical fluxes while, for
simplicity, the function $h$ is considered constant in time. The
initial data are $\vr^{0}_i=\vr_0$, $f^{0}_i=E[\vr_{0}]$,
$f^{0}_{R,i}=h_if_{0}$, $\vr^{0}_{L,i}=(1-h_i)\vr_{0}$ and
$E[\vr^{0}_{L,i}]=(1-h_i)E[\vr_{0}]$ $\forall i$. Now, supposing the
flow uniform at time $n$ we will see that not every numerical scheme
ensures a uniform flow at time $n+1$. To this aim, if we integrate
equation (\ref{eq:Kscheme}) multiplied by
the collision invariants $\textbf{m}_{k}$ over the velocity space, we can rewrite the
coupled numerical schemes (\ref{eq:H1scheme}-\ref{eq:Kscheme}) as
follows:
\begin{equation}
\boldsymbol{\varrho}^{n+1}_{i,L}=\sum_{j=i}^{i\pm I}
\boldsymbol{A}_j\boldsymbol{\varrho}^{n}_{j,L}\label{eq:H1scheme1},
\end{equation}
\begin{equation}
\boldsymbol{\varrho}^{n+1}_{i,R}=\sum_{j=i}^{i\pm I1}\sum_k^{K}
\boldsymbol{B}_{j,k}\mathbf{m}_{k}f^{n}_{j,R} \Delta
v\label{eq:H1scheme1},
\end{equation}
with $\vr^{n+1}_{j,R}=\sum_{k}\mathbf{m}_{k}f^{n+1}_{j,R} \Delta v$,
$I$ and $I1$ the length of the stencils in physical space and
$K$ in velocity space. The symbols $\boldsymbol{A}_j$ and
$\boldsymbol{B}_{j,k}$ represent the weights determined by the
particular choice of the numerical schemes. Without loss of
generality, suppose that $I=I1$. Then, we have that
\begin{eqnarray}
\vr^{n+1}_i&=&\vr^{n+1}_{i,L}+\vr^{n+1}_{i,R}=\sum_{j}\boldsymbol{A}_j\vr^{n}_{j,L}
+\sum_{j}\sum_k \boldsymbol{B}_{j,k}\mathbf{m}_{k}f^{n}_{j,R} \Delta
v =\nonumber\\
&&=\sum_{j}\boldsymbol{A}_j(\vr^{n}_{j,L}+\vr^{n}_{j,R})+\sum_{j}\left[\left(\sum_k
\boldsymbol{B}_{j,k}\mathbf{m}_{k}f^{n}_{j,R}\right)-\boldsymbol{A}_j\vr^{n}_{j,R}\right]= \nonumber\\
&&=\sum_{j} \boldsymbol{A}_j\vr^{n}_j+\sum_{j}\left[\left(\sum_k
\boldsymbol{B}_{j,k}\mathbf{m}_{k}f^{n}_{j,R}\right)-\boldsymbol{A}_j\vr^{n}_{j,R}\right]=\nonumber\\
&&=\vr^{n}_i+\sum_{j}\left[\left(\sum_k
\boldsymbol{B}_{j,k}\mathbf{m}_{k}f^{n}_{j,R}\right)-\boldsymbol{A}_j\vr^{n}_{j,R}\right]
\end{eqnarray}
which means that we do not have preservation of uniform flows except
in some particular cases, such as, for instance, when the numerical schemes
used to discretize the two equations
(\ref{eq:H1scheme}-\ref{eq:Kscheme}) are compatible. Therefore, such compatible schemes are needed in all buffer zones to make sure that
oscillations in the solutions will be avoided.

Instead if we consider the present micro-macro coupling strategy with the same
initial data $f=E[\vr]$, the following property holds:
\begin{proposition}
If the initial condition $f^0\geq 0$ is a constant equilibrium
$E[\vr^{0}]$, then $\vr = \vr^{0}$ and $g_K=h(f-E[\vr])=0$ are
solutions of the micro-macro model (\ref{eq:s5}-\ref{eq:gk2}), and
$E[\vr] + g_K = E[\vr^{0}]$. In other words, the kinetic/fluid solution of
the micro-macro model is exactly the solution of the original
kinetic model.\end{proposition} Indeed, for the micro-macro
decomposition, the total flux is independent of $h$ in equilibrium
regimes and is not obtained as a sum of two complementary terms
weighted by the function $h$. It follows directly that the
micro-macro method preserves uniform flows even in the discrete case
independently of the choice of the numerical scheme.



\begin{figure}
\begin{center}
\includegraphics[scale=0.34]{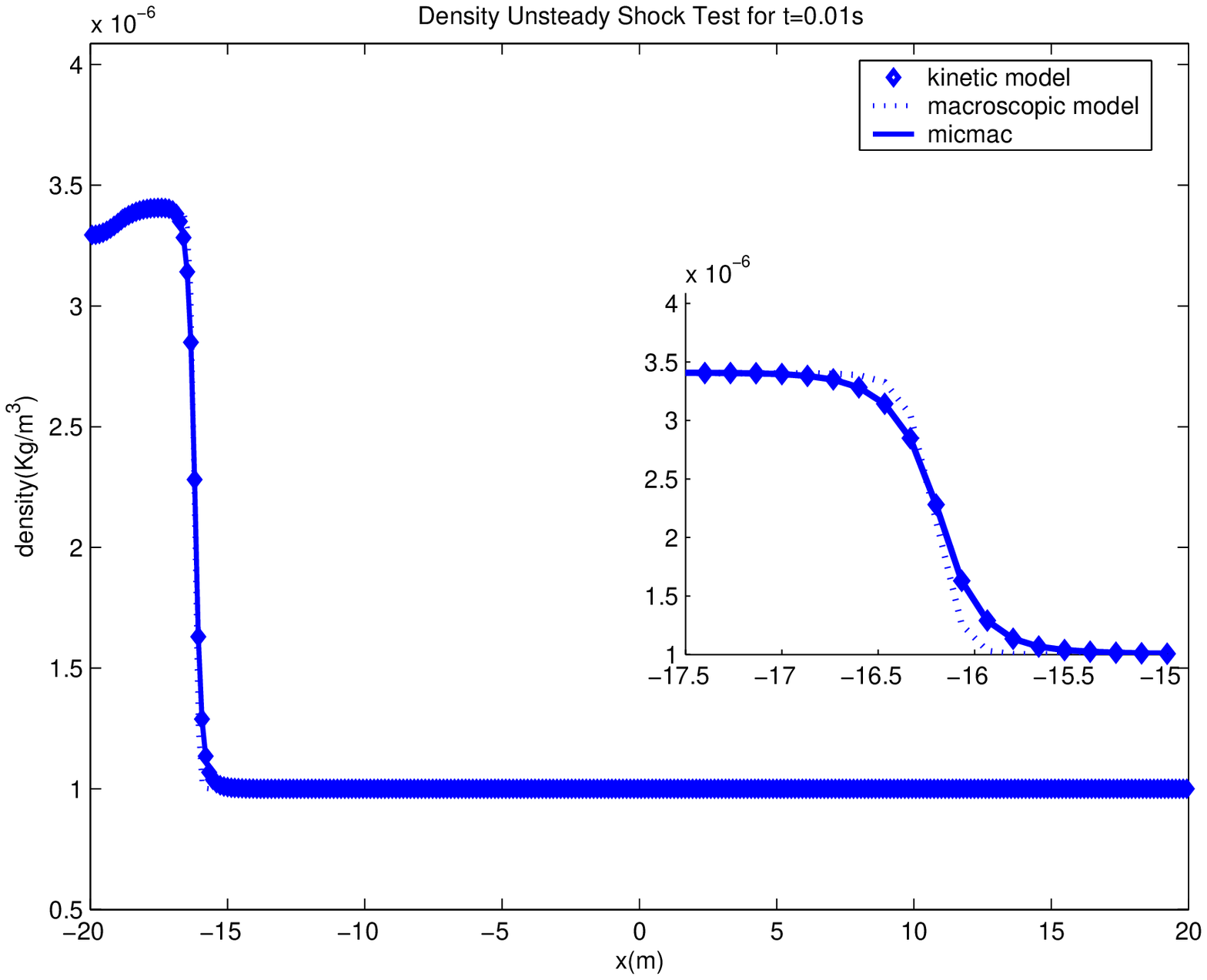}
\includegraphics[scale=0.34]{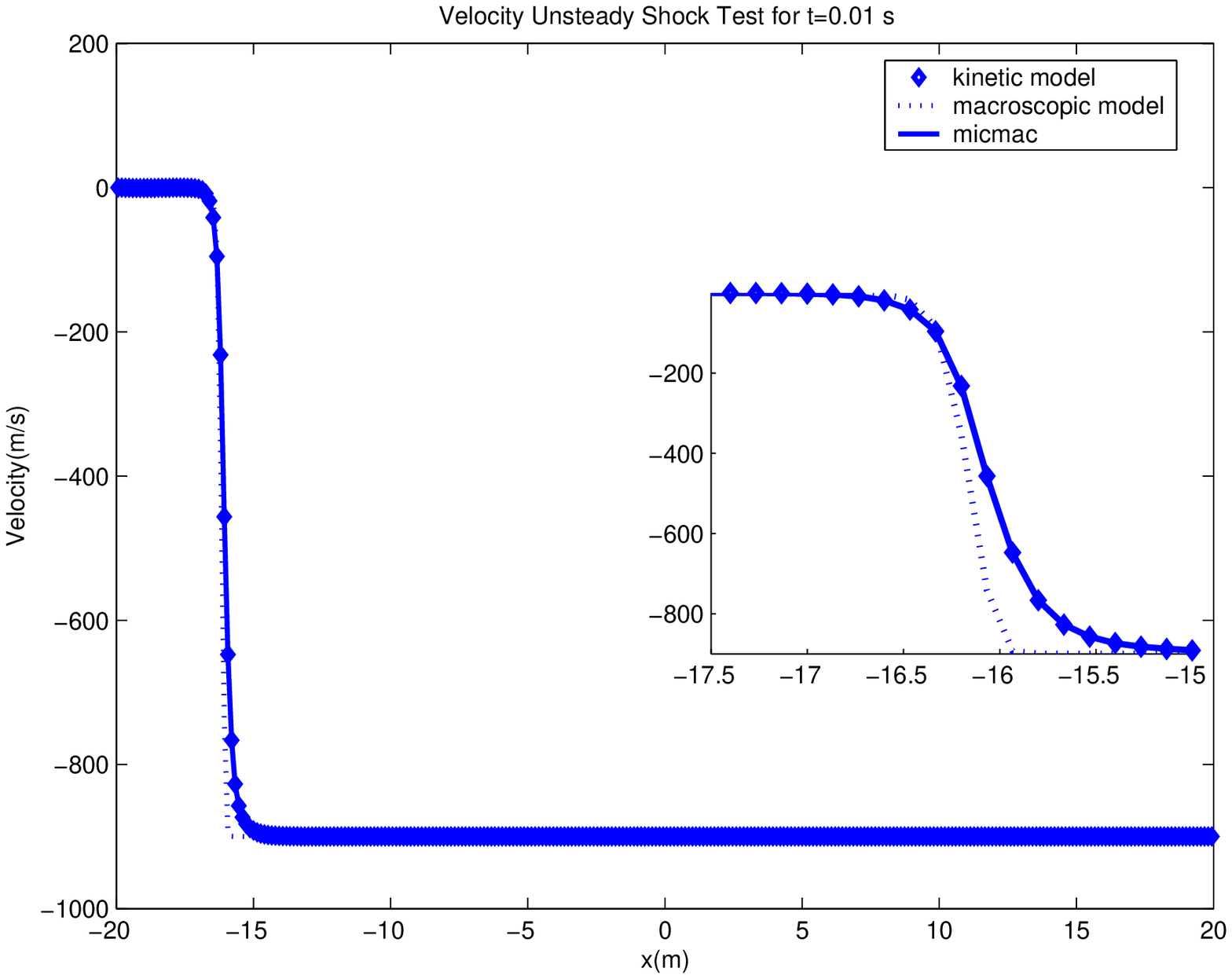}
\includegraphics[scale=0.34]{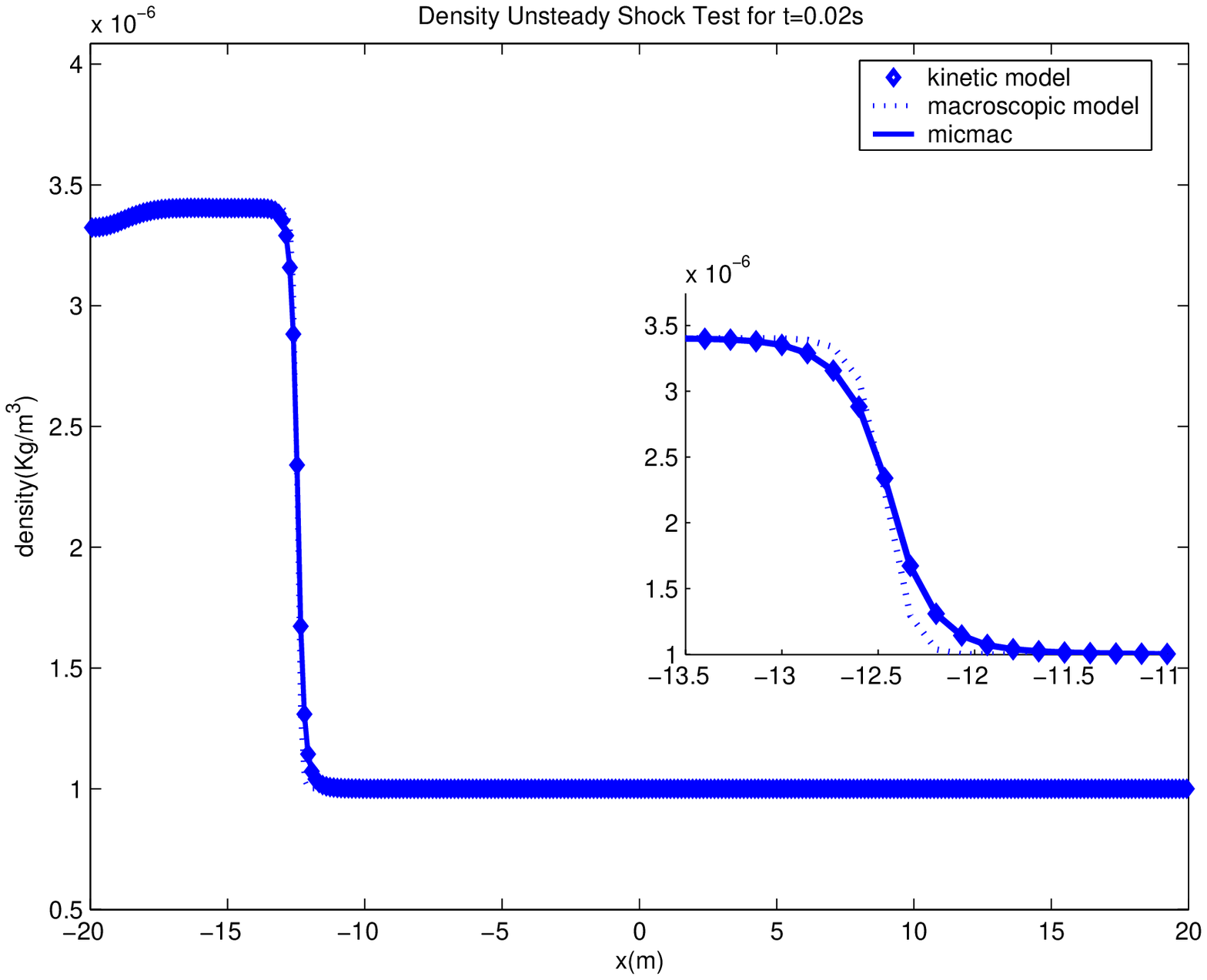}
\includegraphics[scale=0.34]{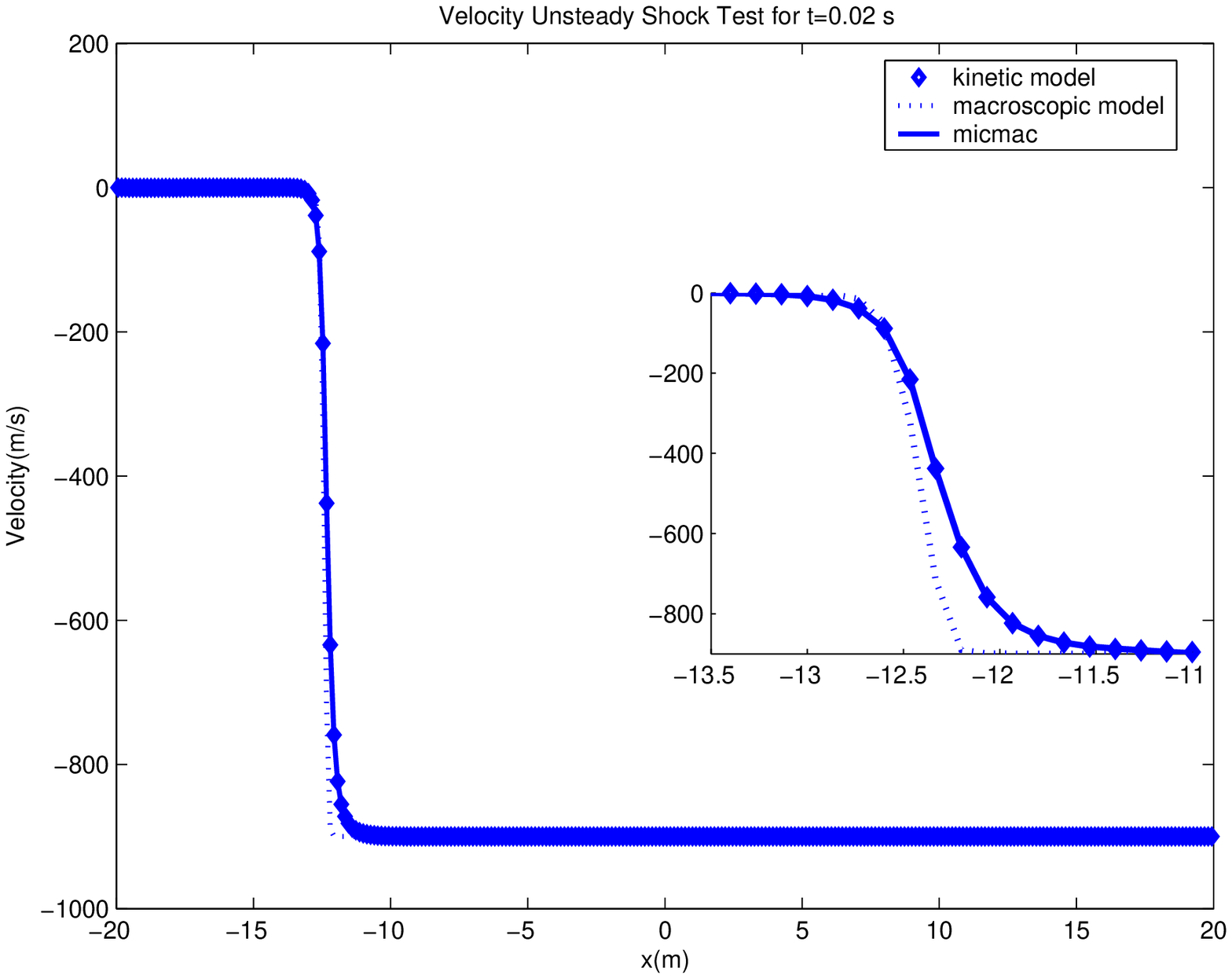}
\includegraphics[scale=0.34]{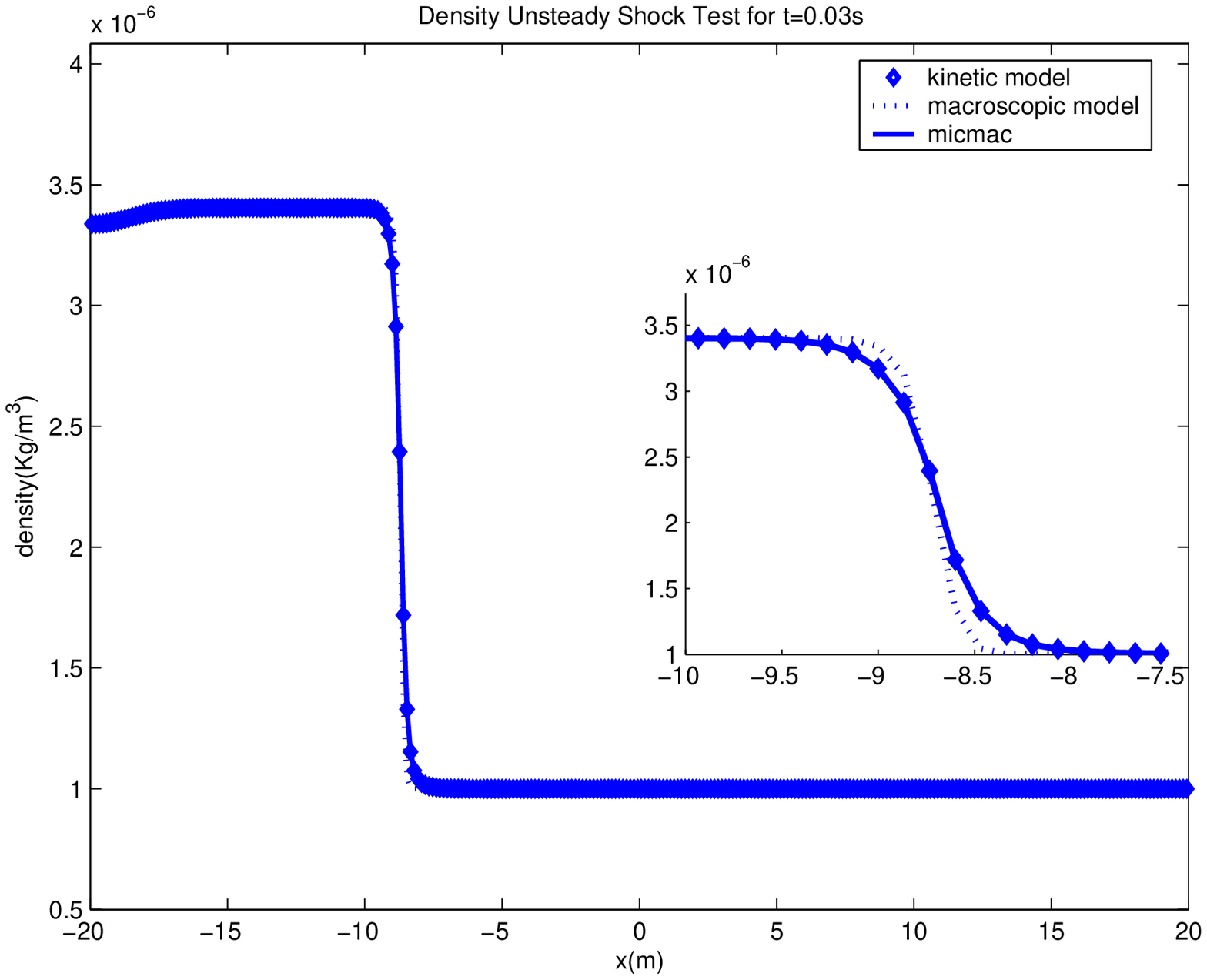}
\includegraphics[scale=0.34]{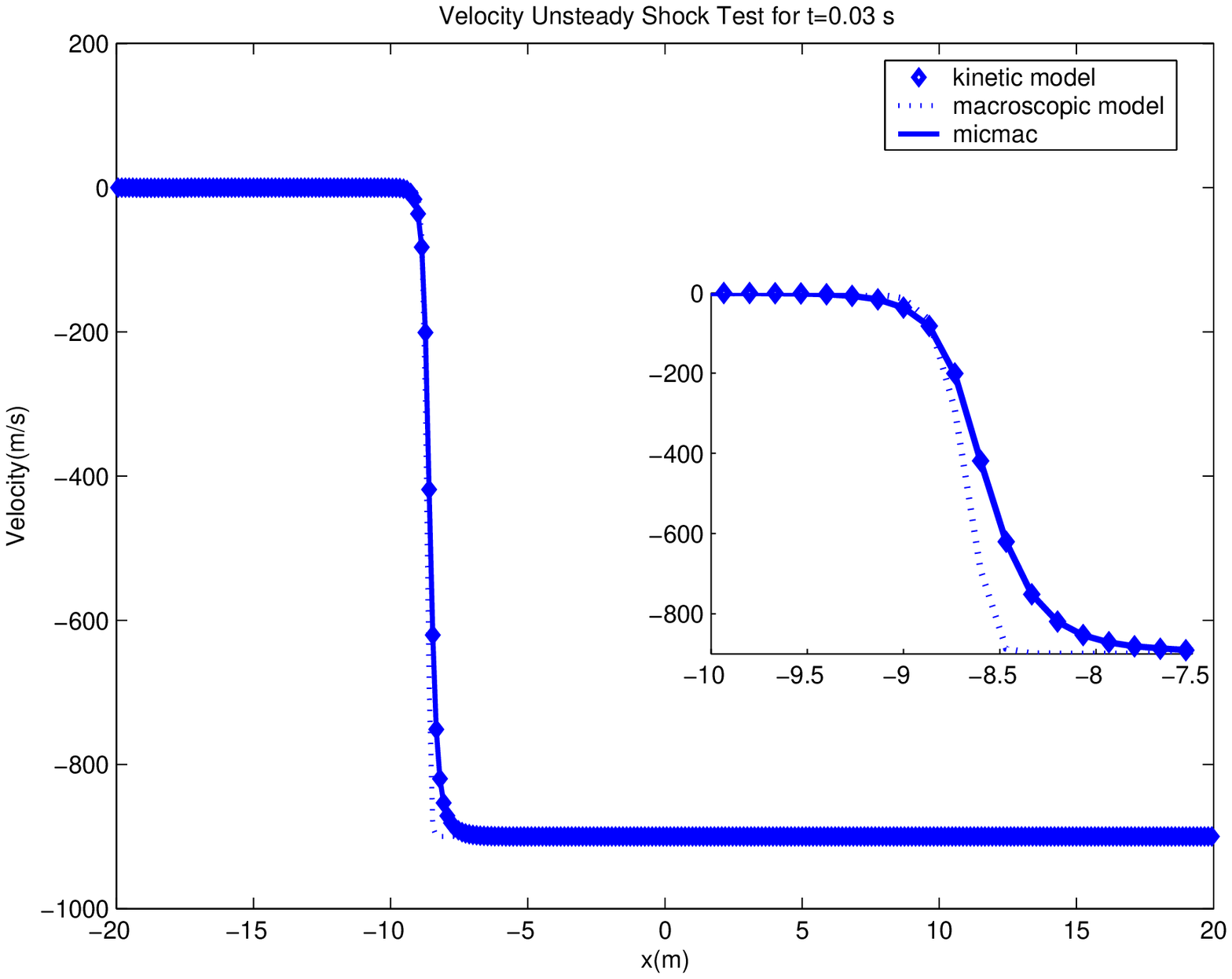}
\includegraphics[scale=0.34]{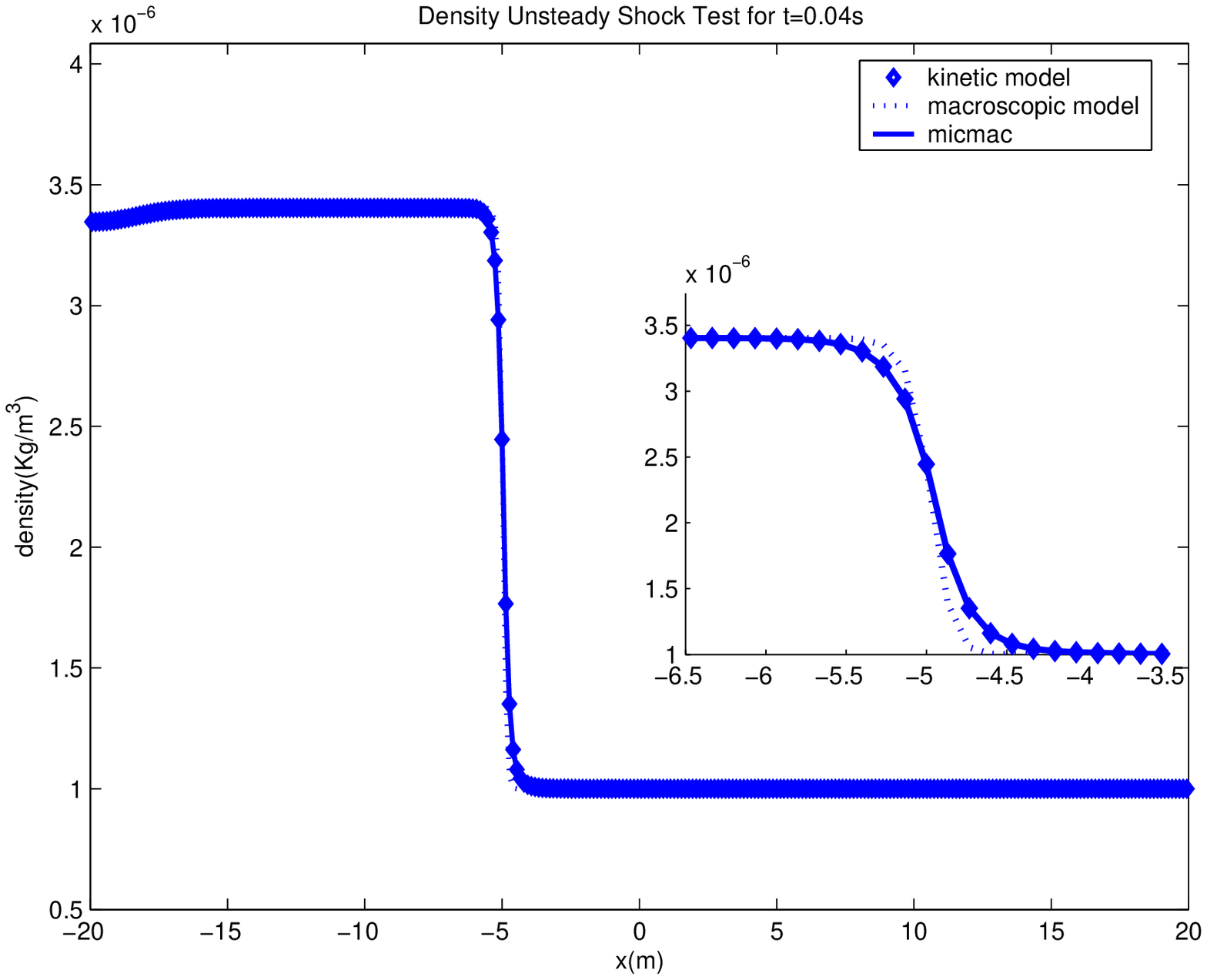}
\includegraphics[scale=0.34]{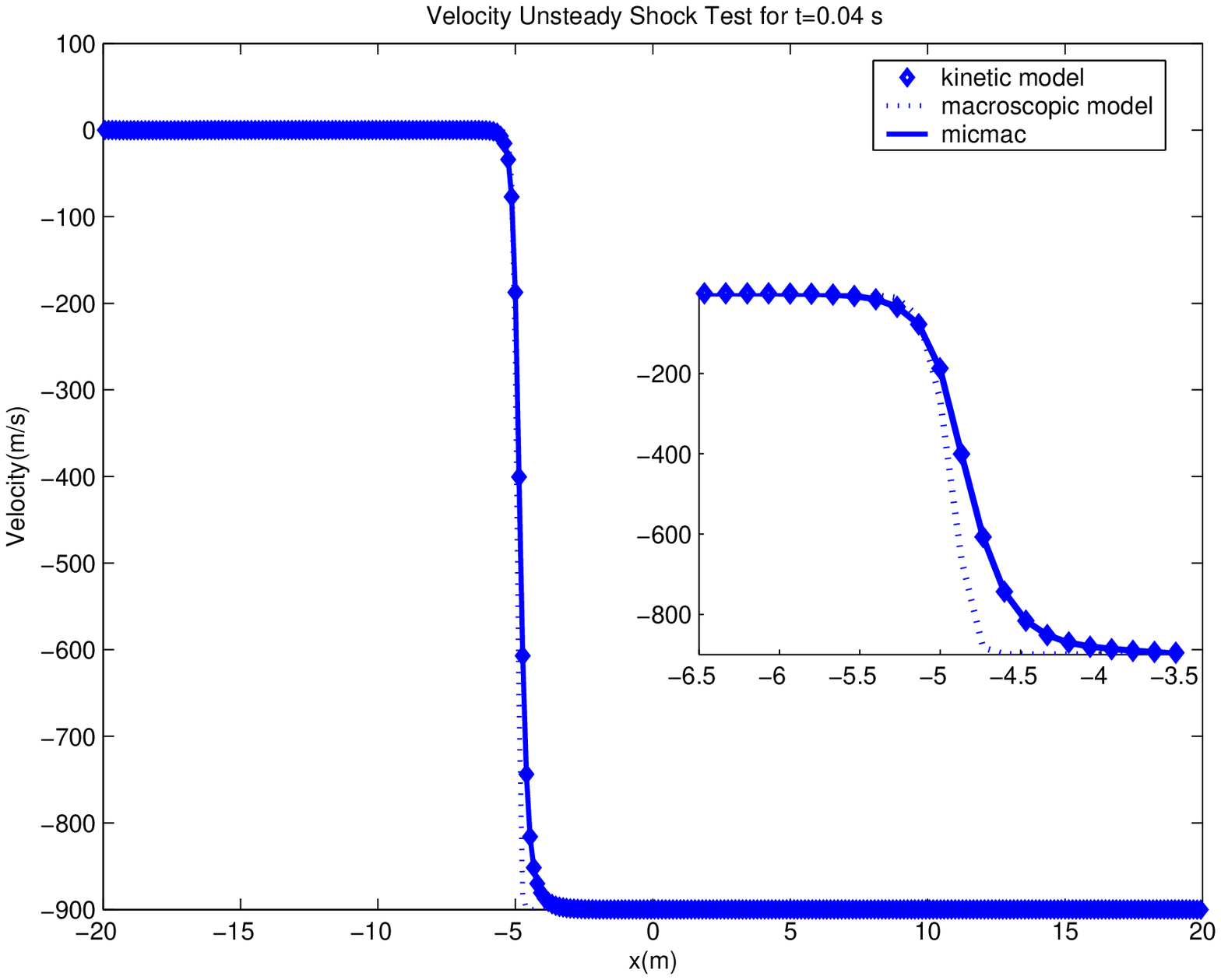}
\caption{Unsteady Shock 1: Solution at $t=1\times 10^{-2}$ top,
$t=2\times 10^{-2}$ middle top, $t=3\times 10^{-2}$ middle bottom,
$t=4\times 10^{-2}$ bottom, density left, velocity  right. The small
panels are a magnification of the solution close to the shock.
\label{UST1.1}}
\end{center}
\end{figure}

\begin{figure}
\begin{center}
\includegraphics[scale=0.34]{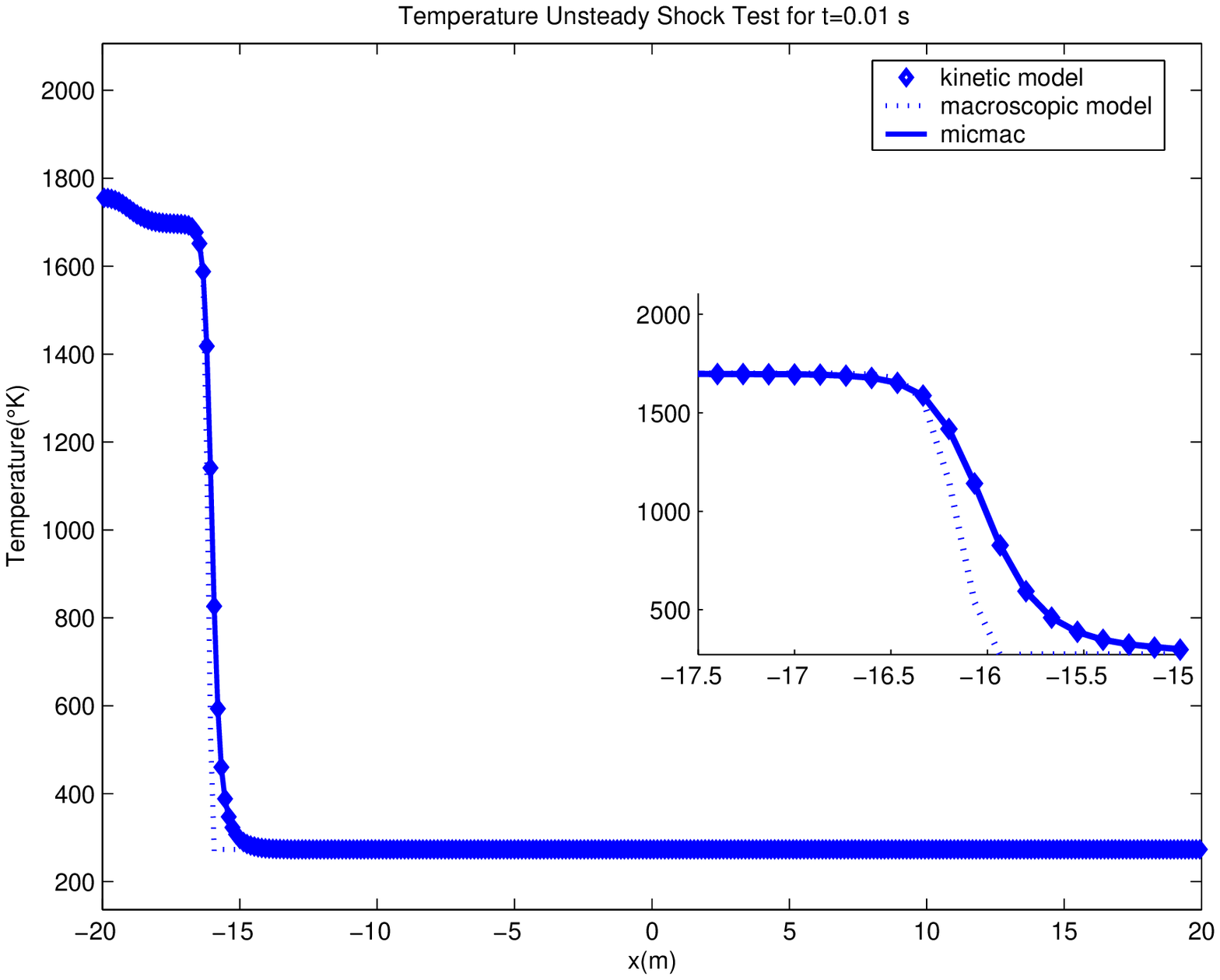}
\includegraphics[scale=0.34]{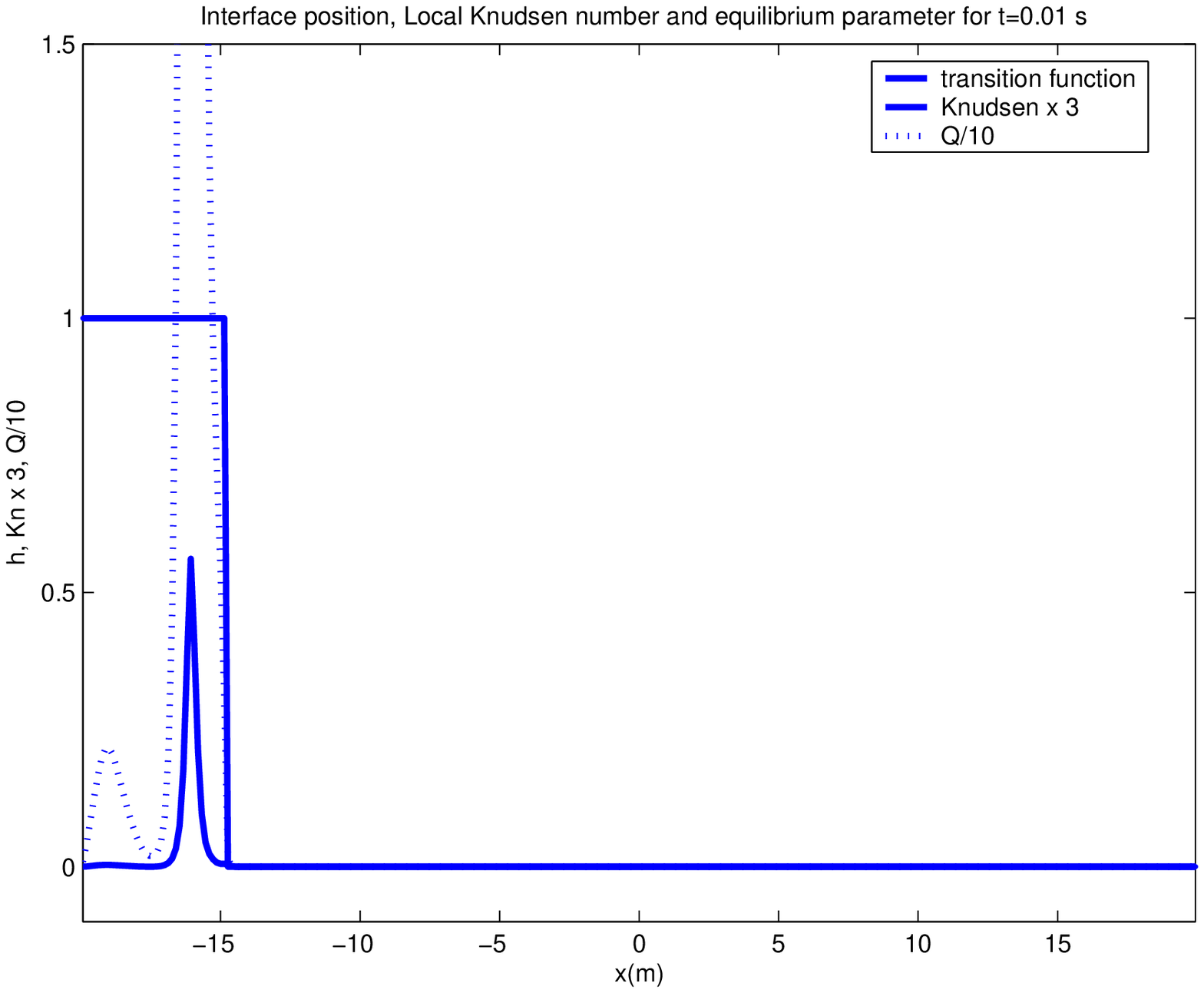}
\includegraphics[scale=0.34]{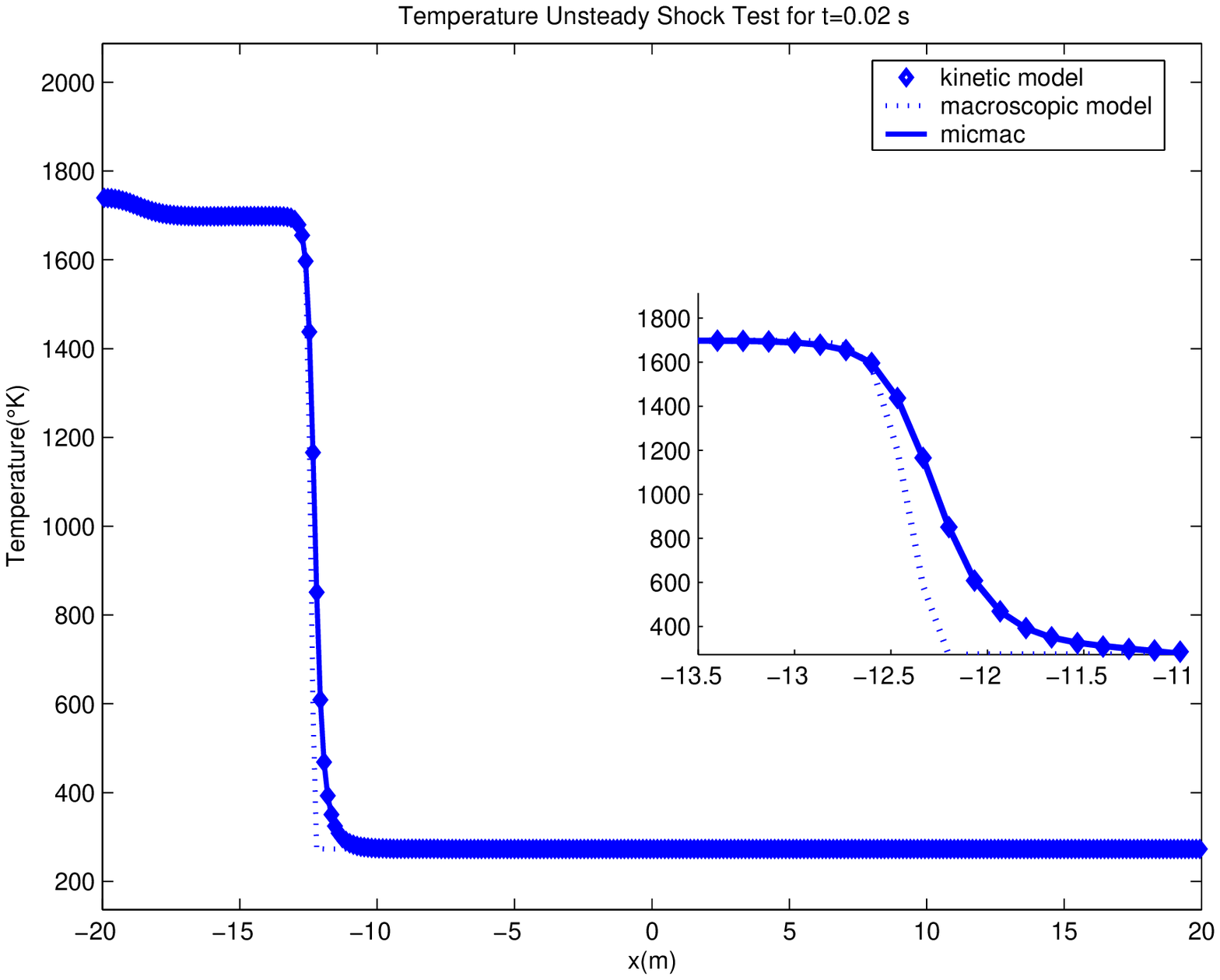}
\includegraphics[scale=0.34]{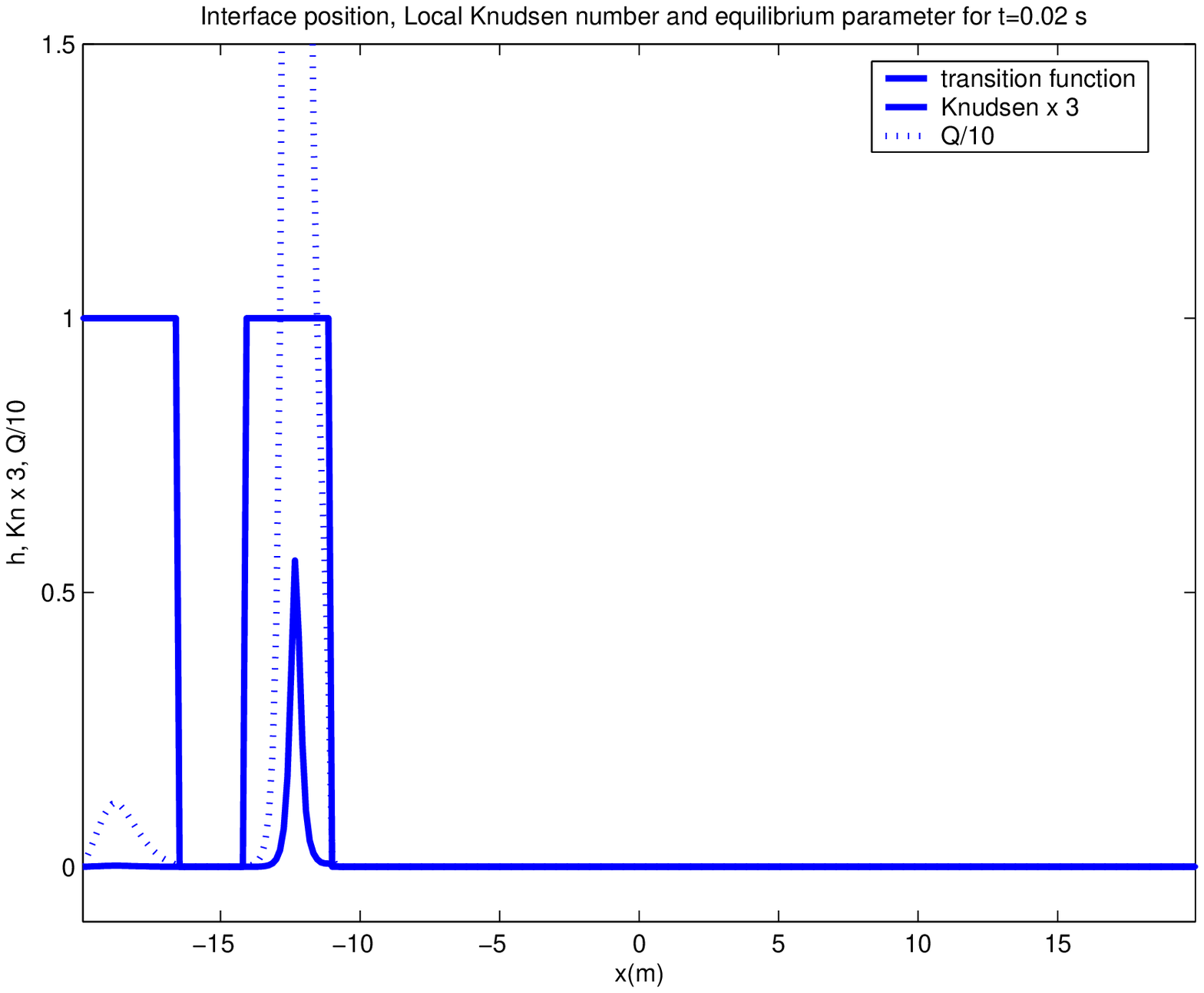}
\includegraphics[scale=0.34]{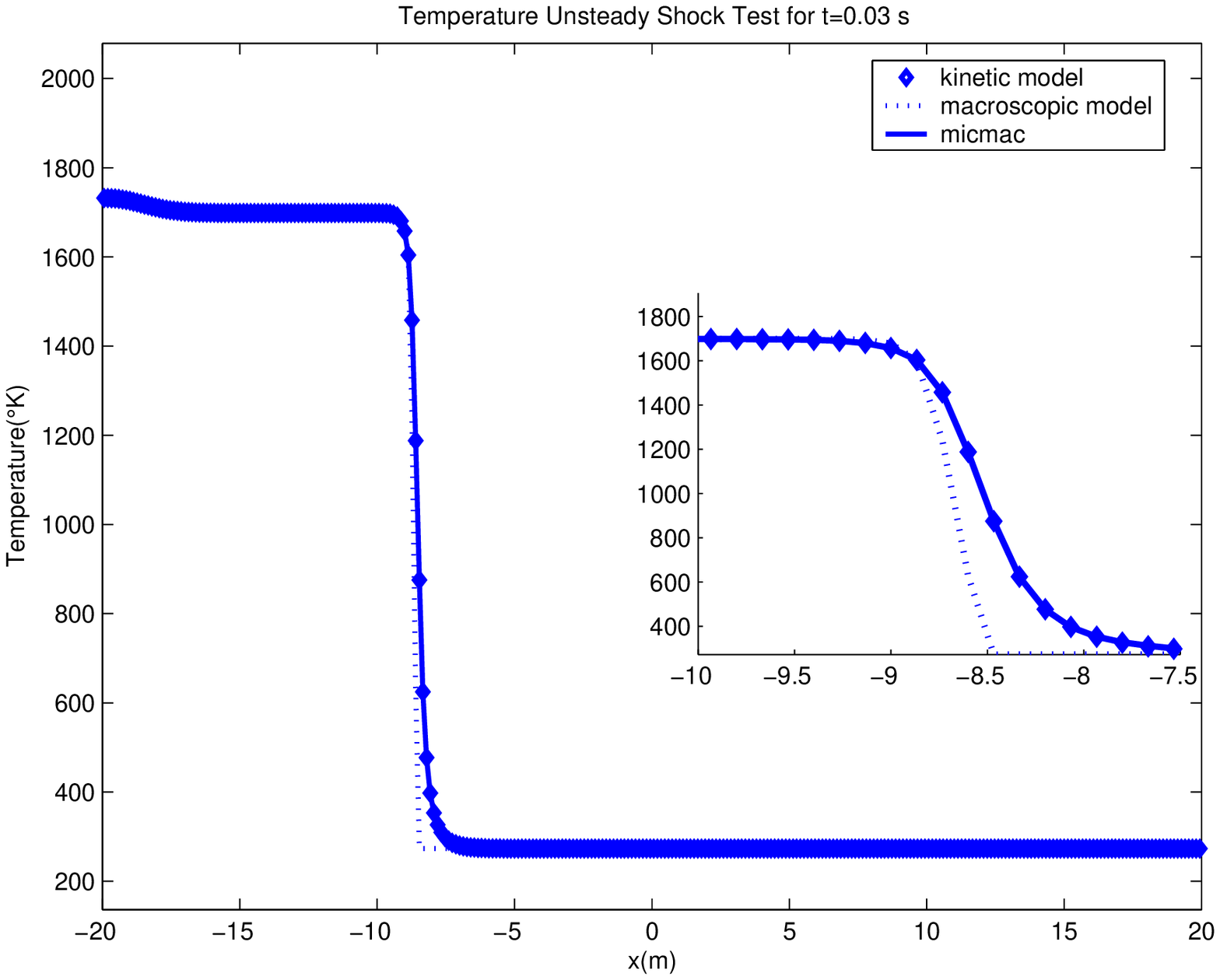}
\includegraphics[scale=0.34]{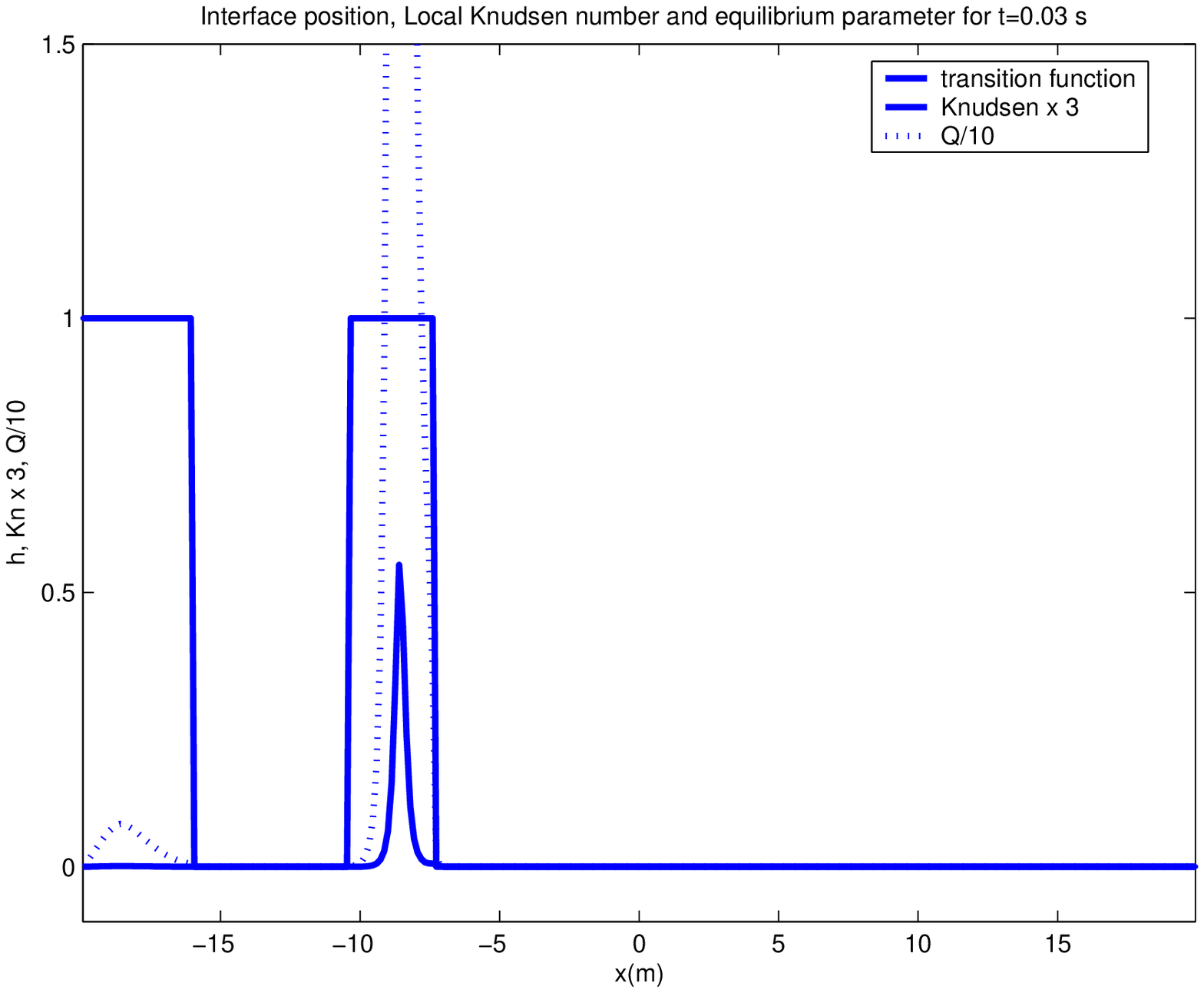}
\includegraphics[scale=0.34]{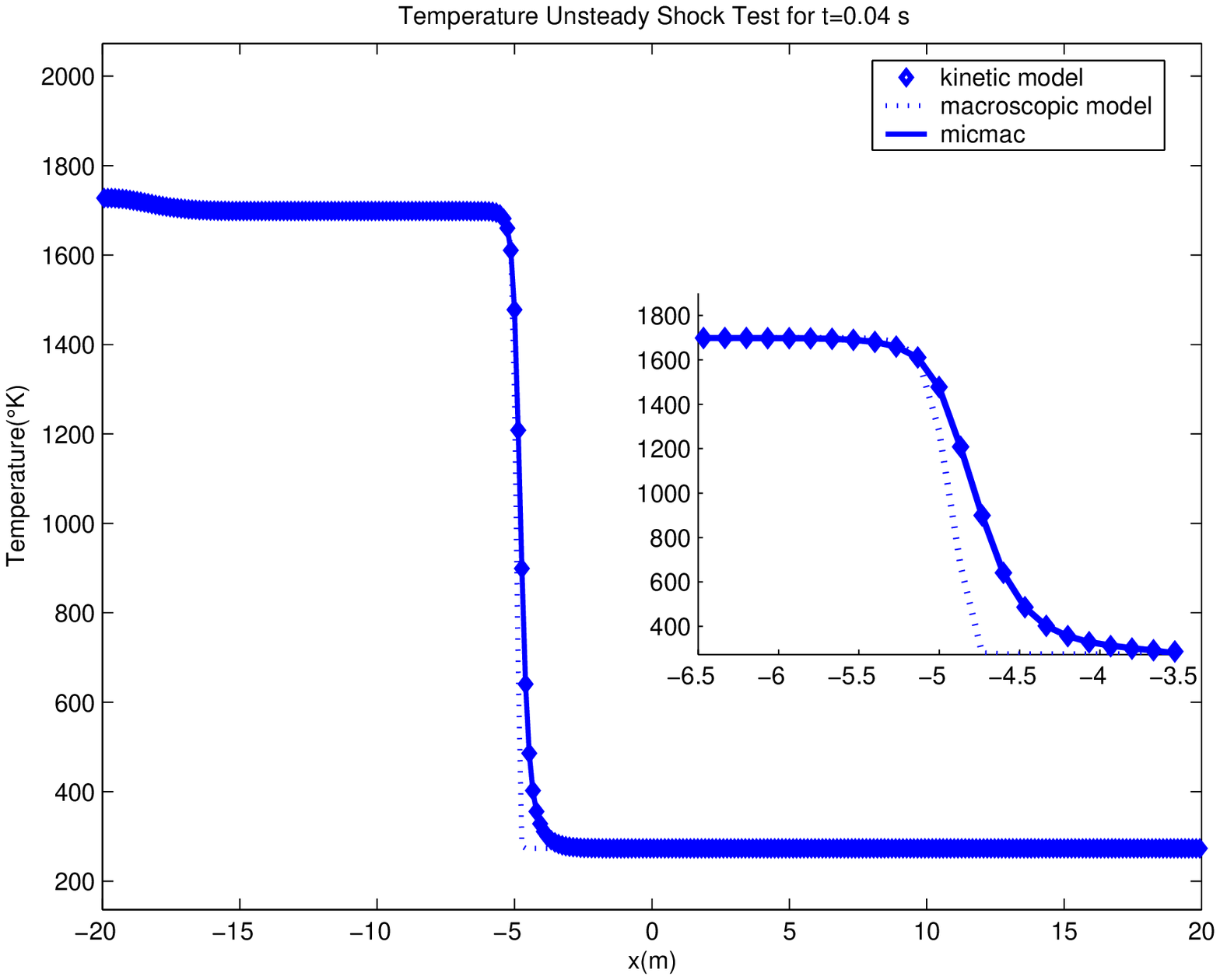}
\includegraphics[scale=0.34]{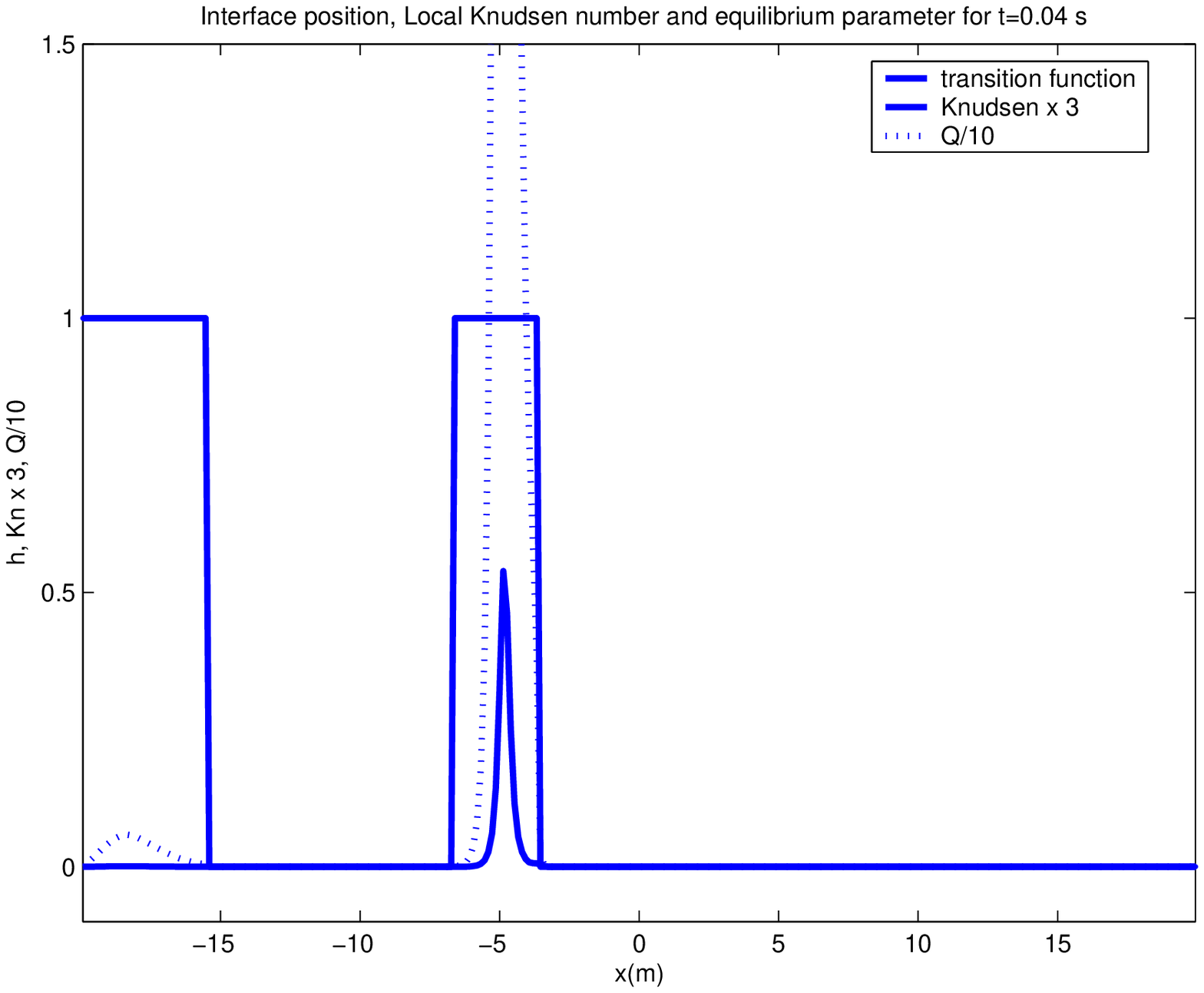}
\caption{Unsteady Shock 1: Solution at $t=1\times 10^{-2}$ top,
$t=2\times 10^{-2}$ middle top, $t=3\times 10^{-2}$ middle bottom,
$t=4\times 10^{-2}$ bottom, temperature left, transition function,
Knudsen number and heat flux right. The small panels are a
magnification of the solution close to the shock. \label{UST1.2}}
\end{center}
\end{figure}

\begin{figure}
\begin{center}
\includegraphics[scale=0.34]{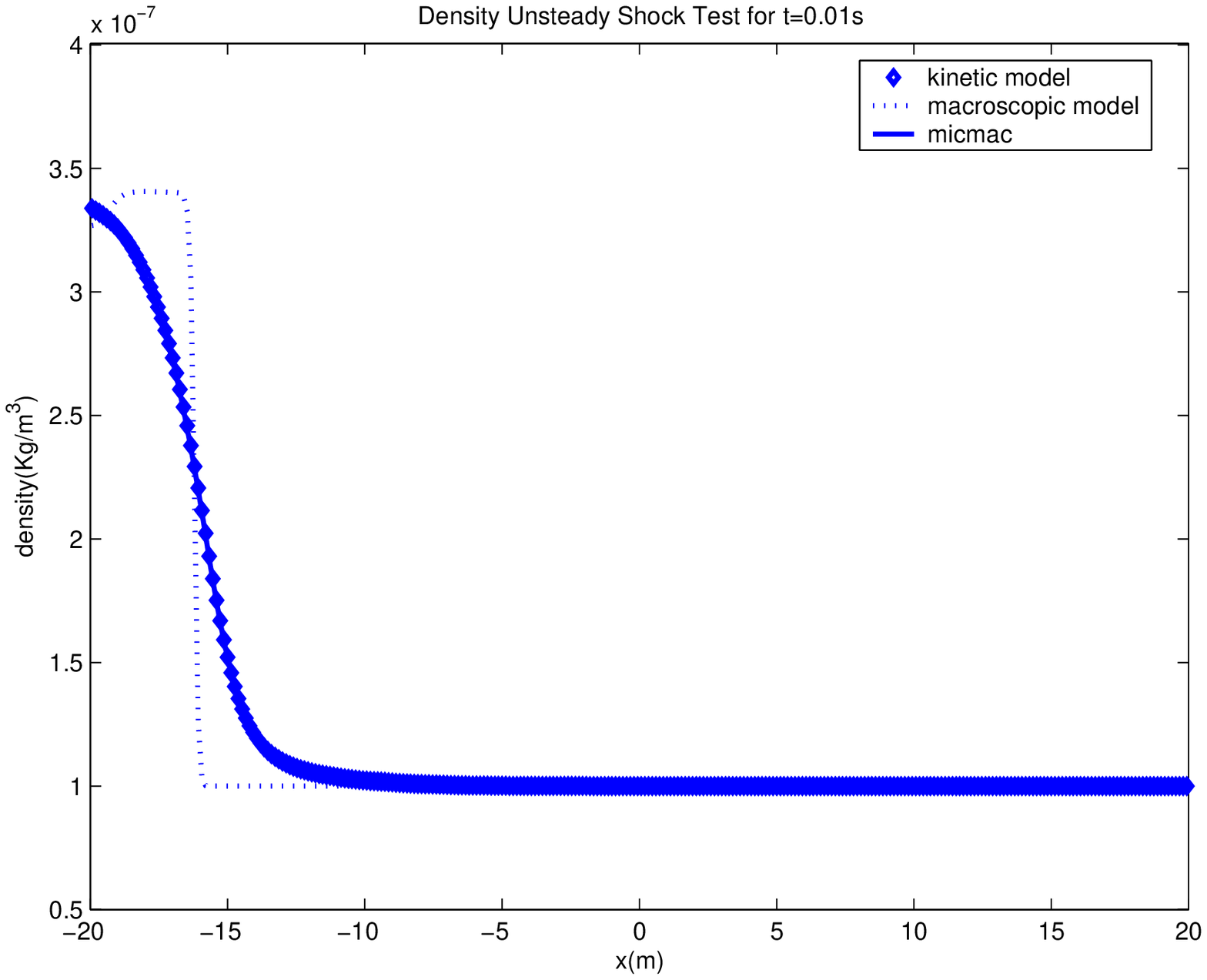}
\includegraphics[scale=0.34]{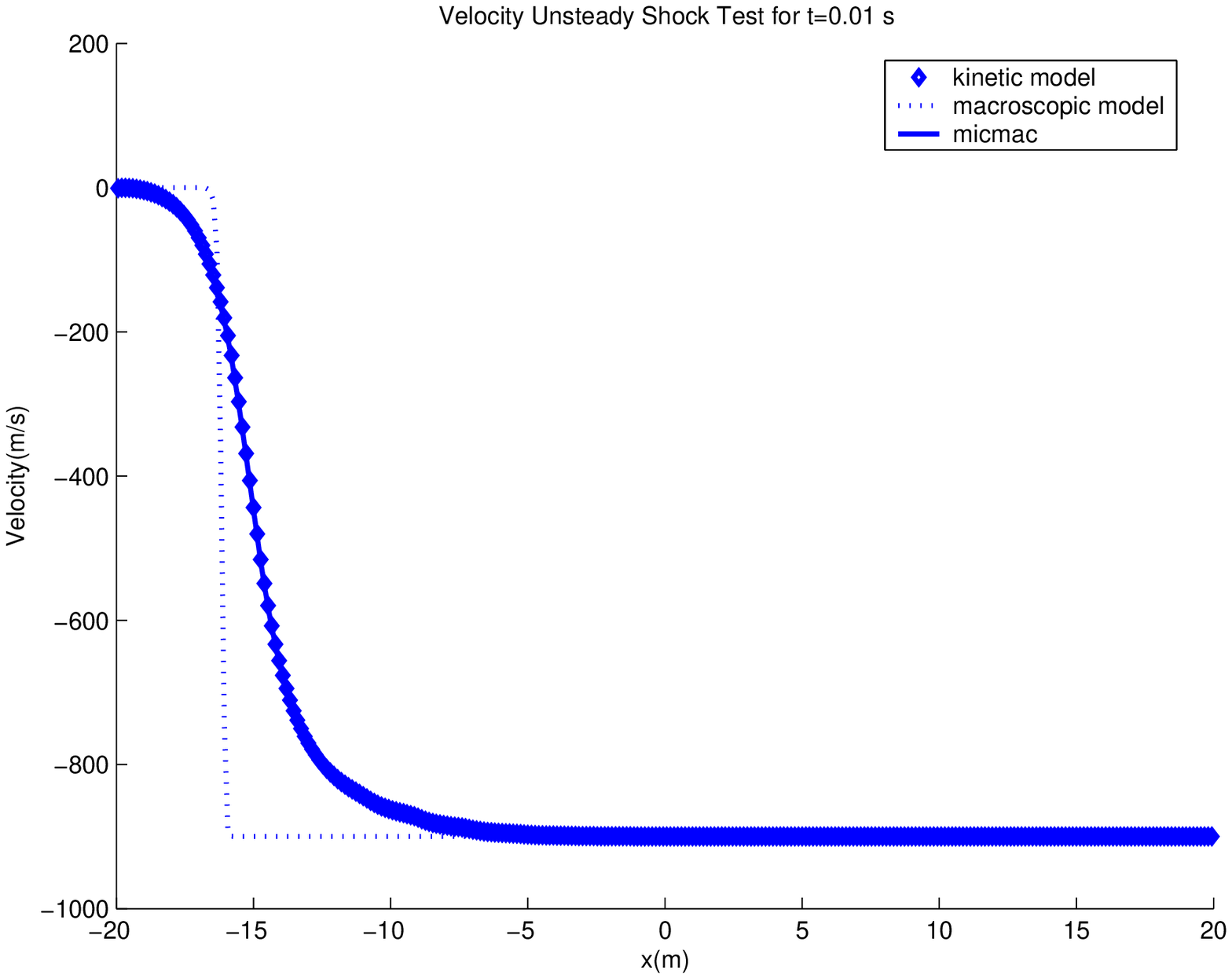}
\includegraphics[scale=0.34]{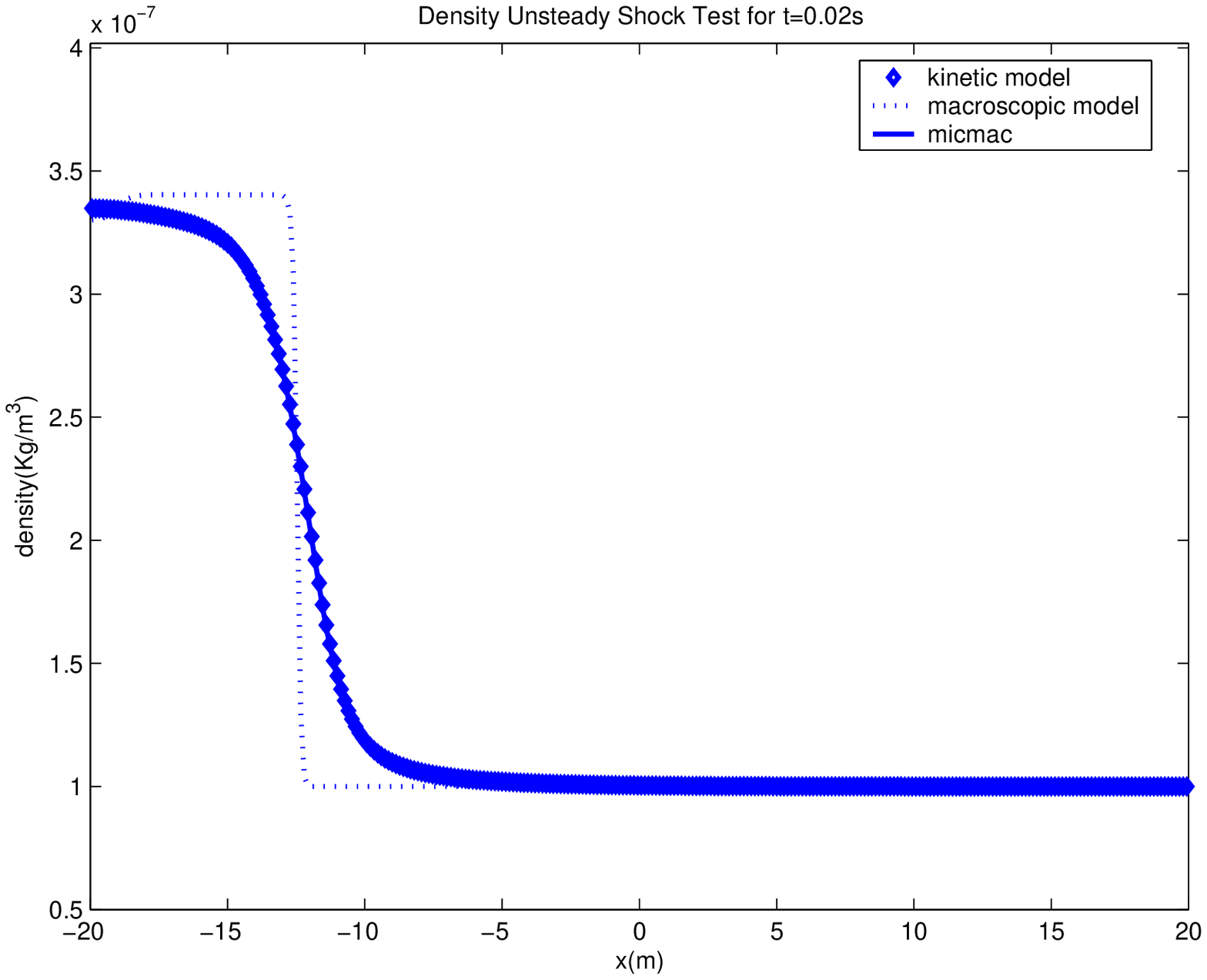}
\includegraphics[scale=0.34]{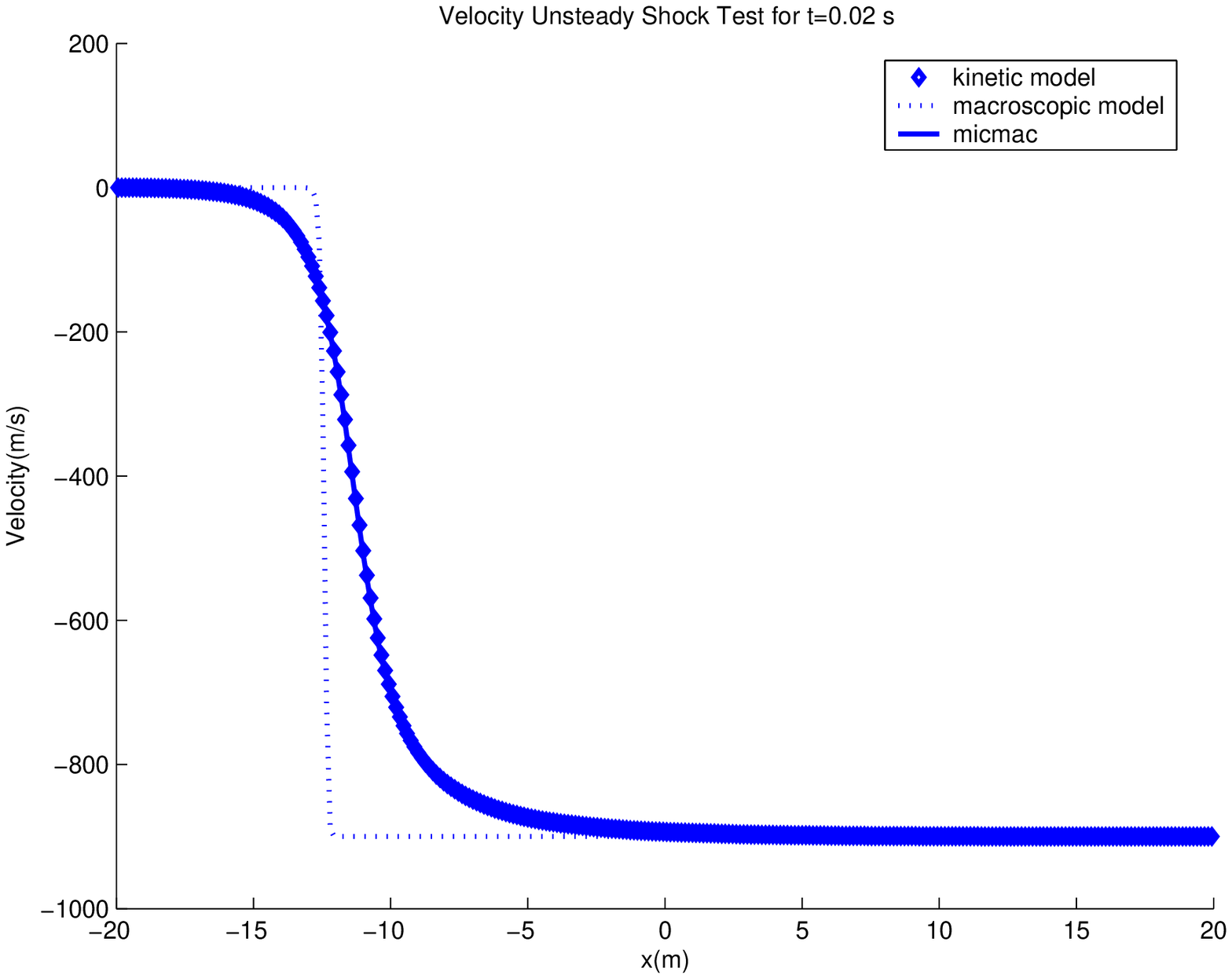}
\includegraphics[scale=0.34]{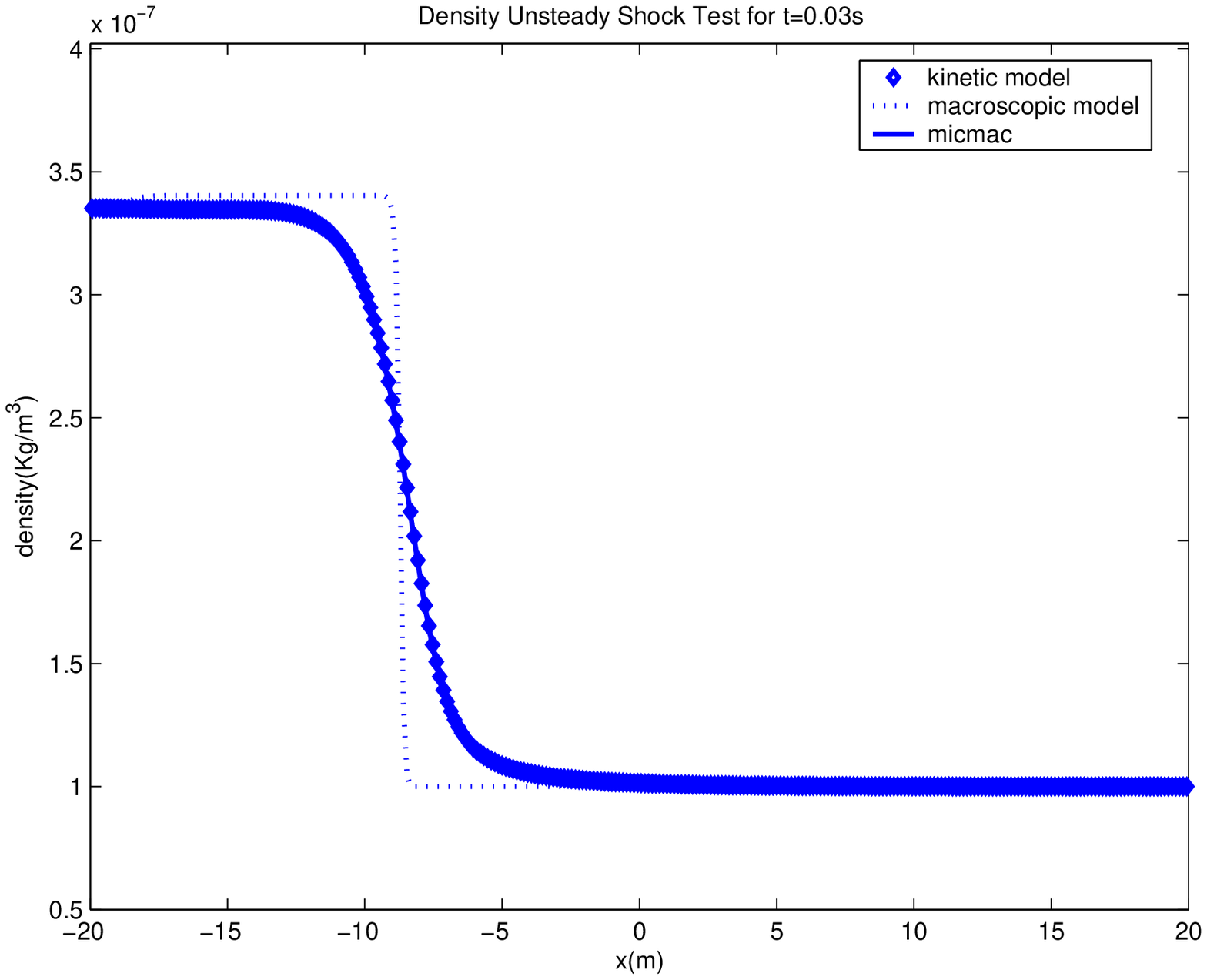}
\includegraphics[scale=0.34]{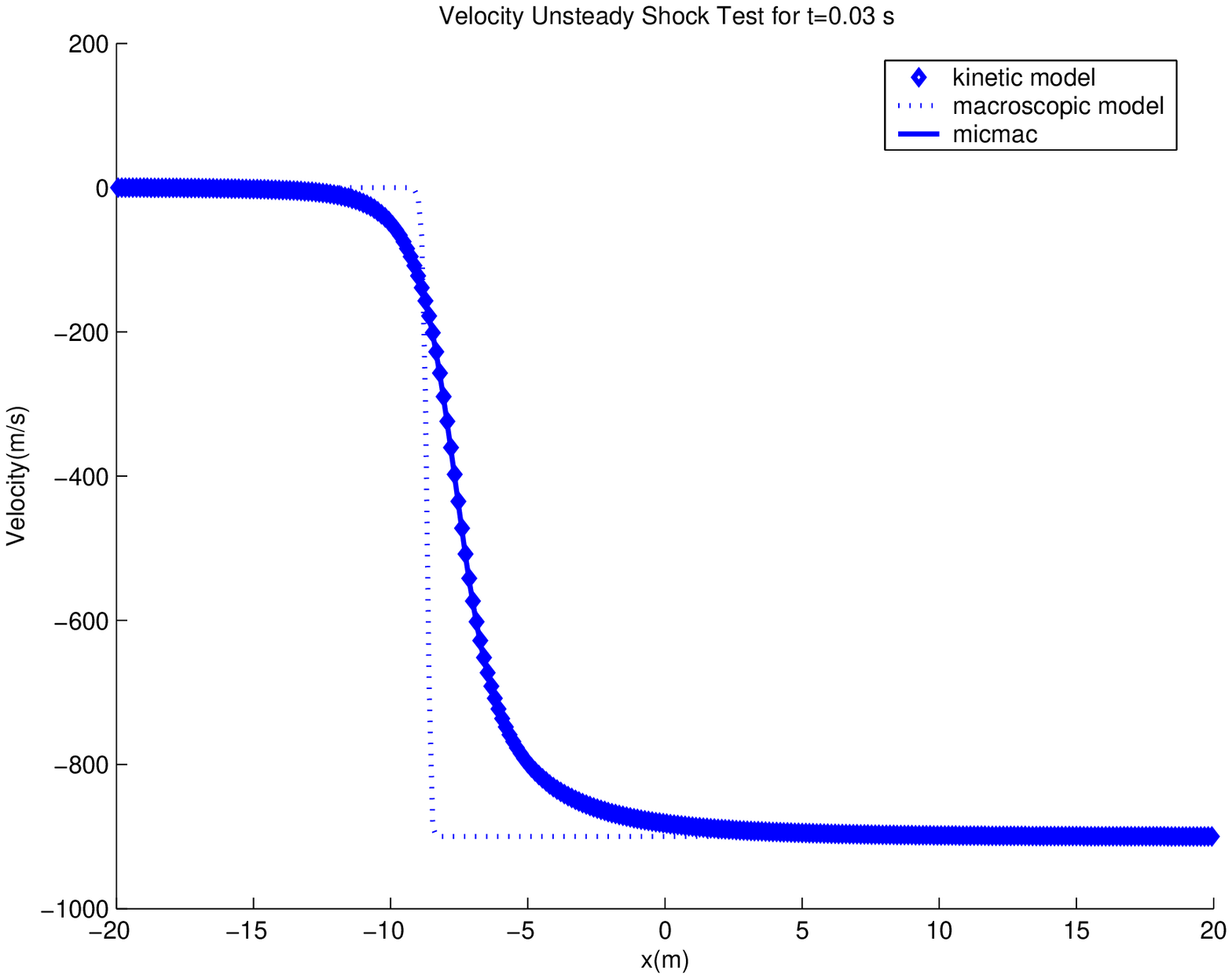}
\includegraphics[scale=0.34]{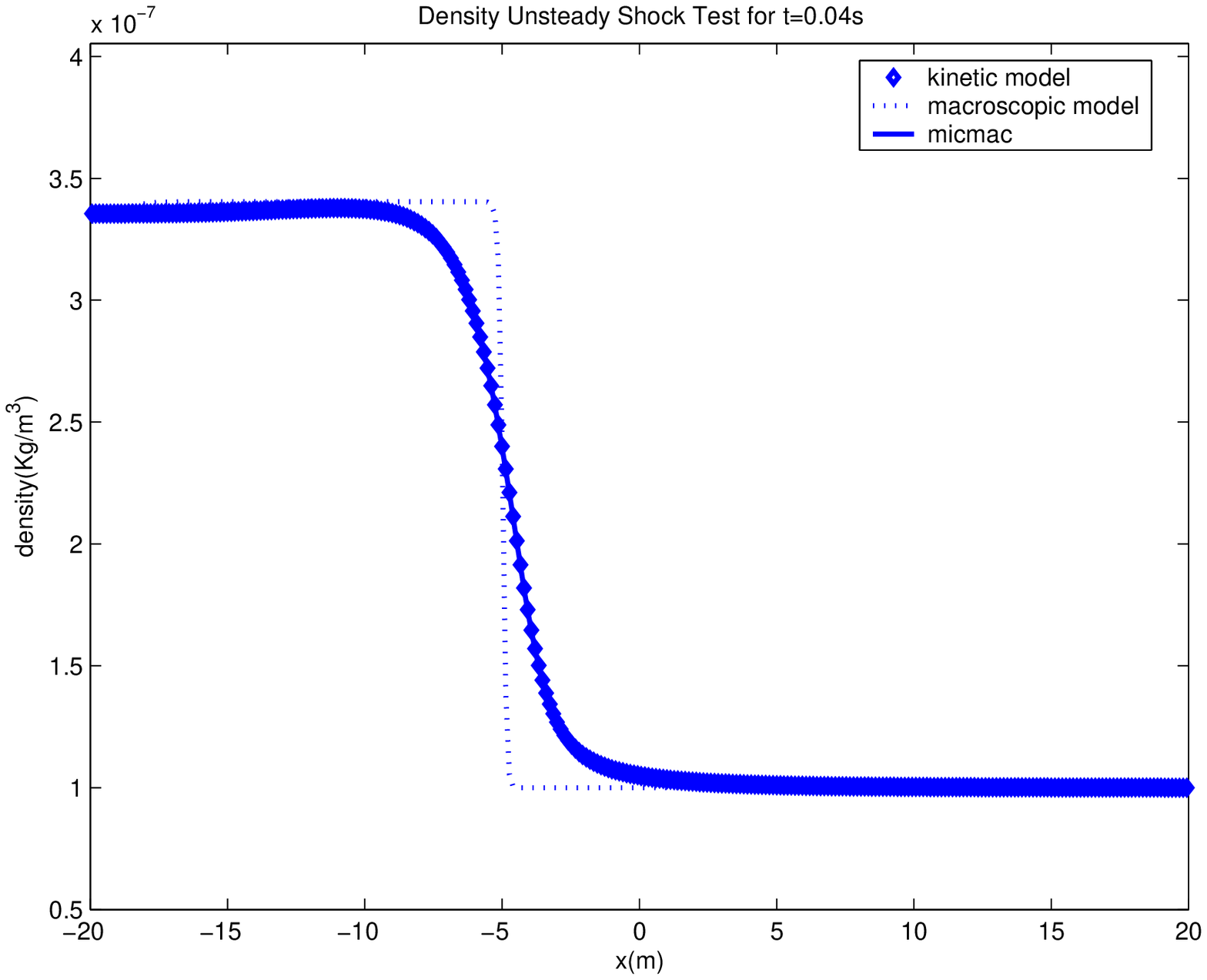}
\includegraphics[scale=0.34]{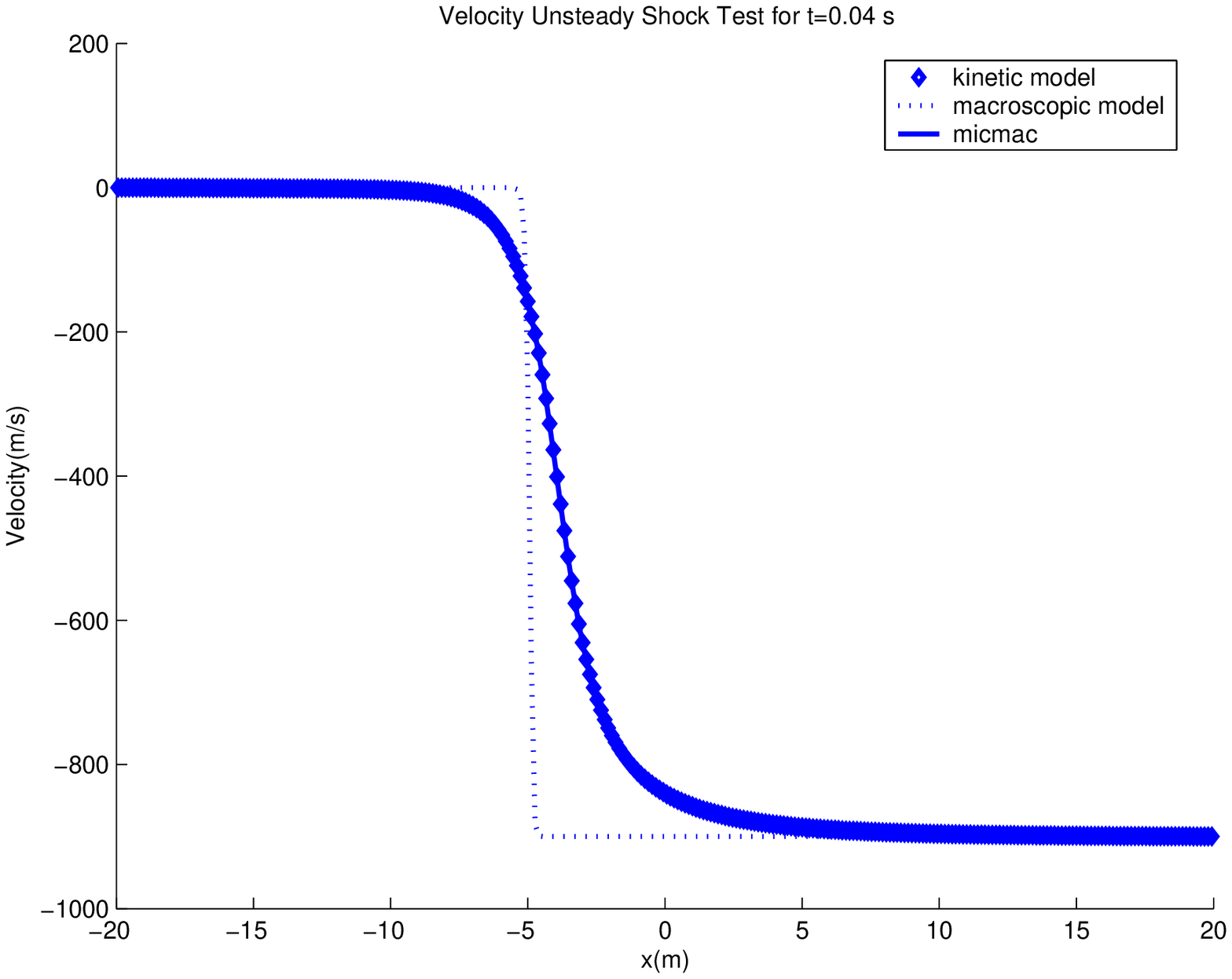}
\caption{Unsteady Shock 2: Solution at $t=1\times 10^{-2}$ top,
$t=2\times 10^{-2}$ middle top, $t=3\times 10^{-2}$ middle bottom,
$t=4\times 10^{-2}$ bottom, density left, velocity  right.
\label{UST2.1}}
\end{center}
\end{figure}

\begin{figure}
\begin{center}
\includegraphics[scale=0.34]{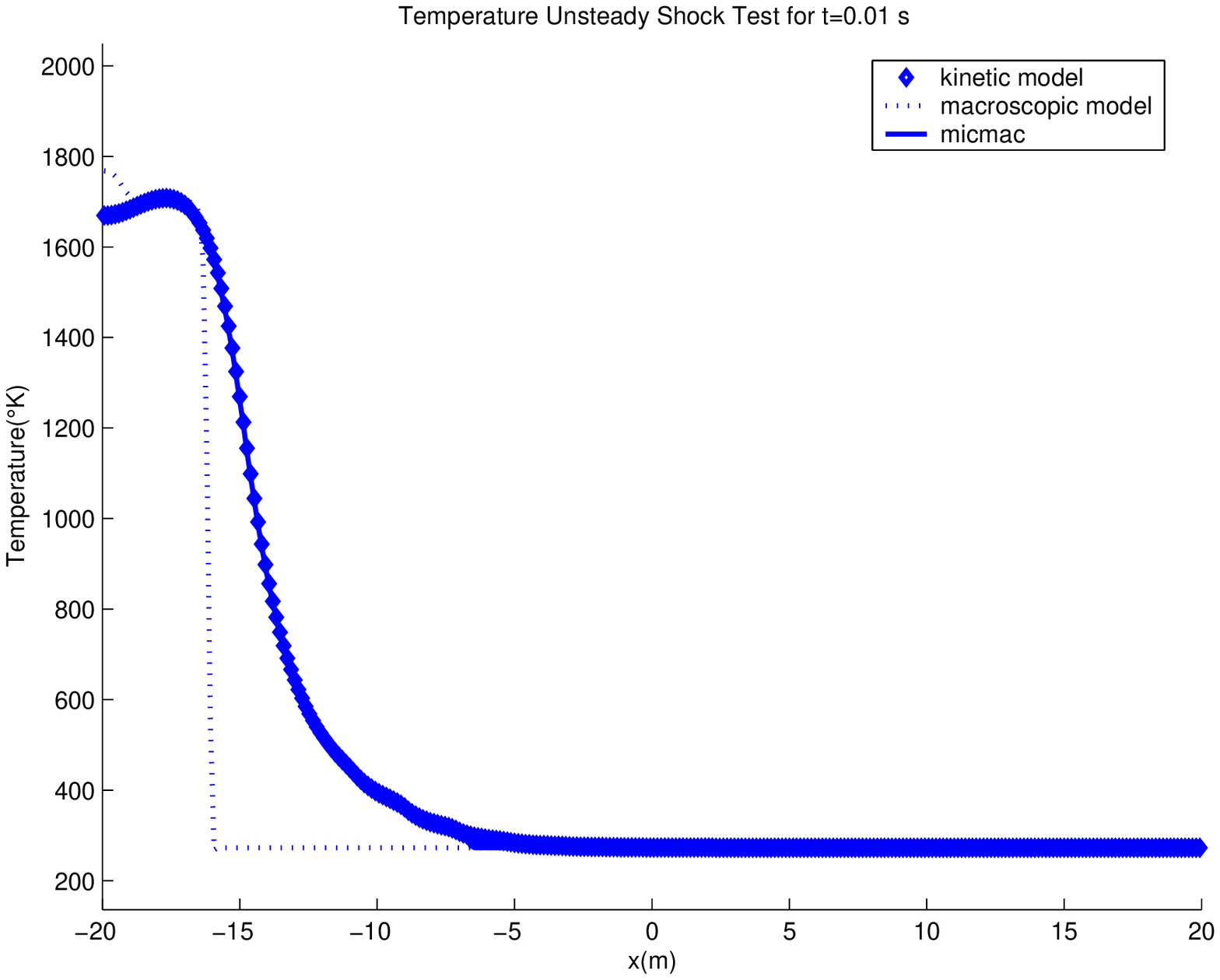}
\includegraphics[scale=0.34]{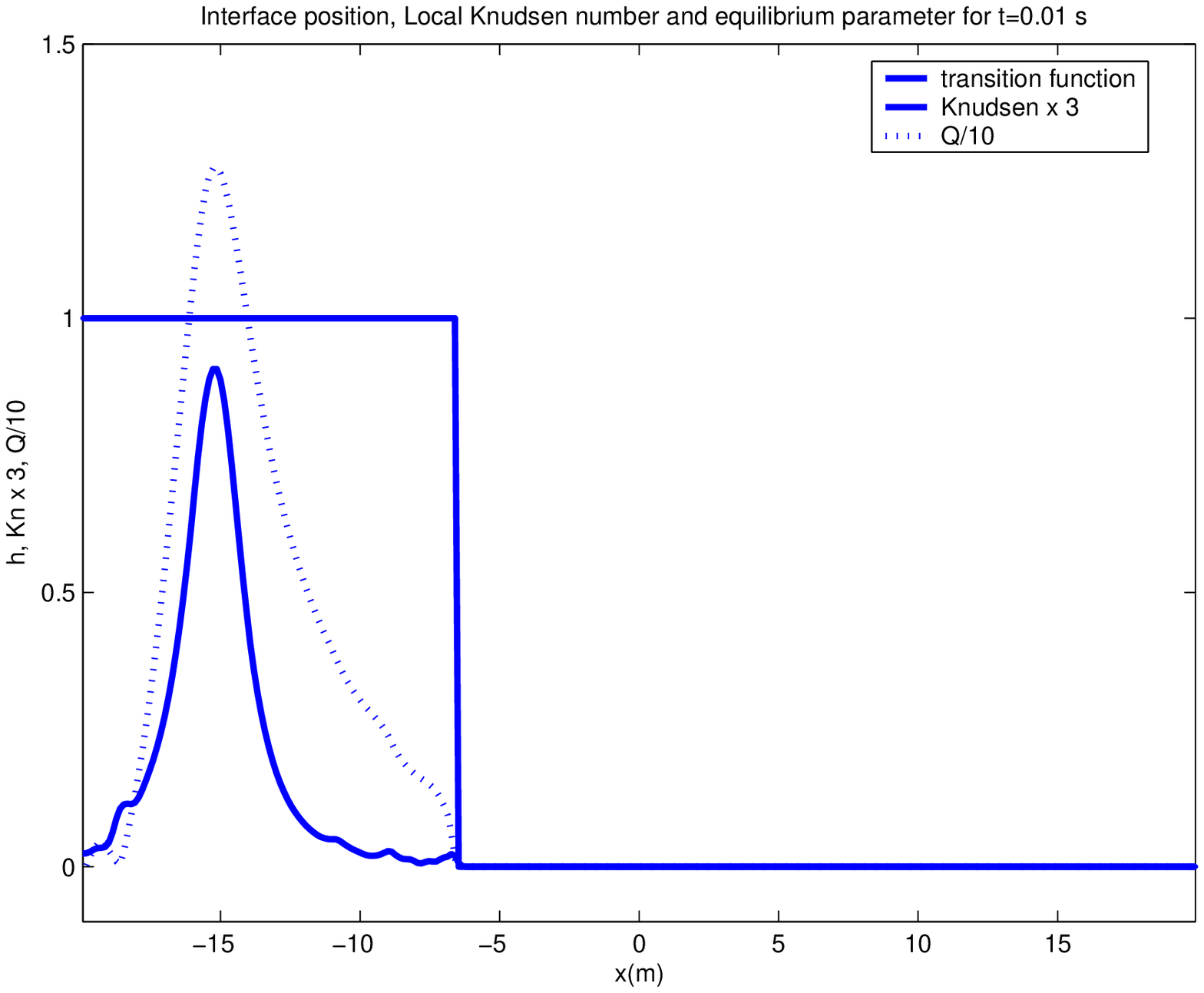}
\includegraphics[scale=0.34]{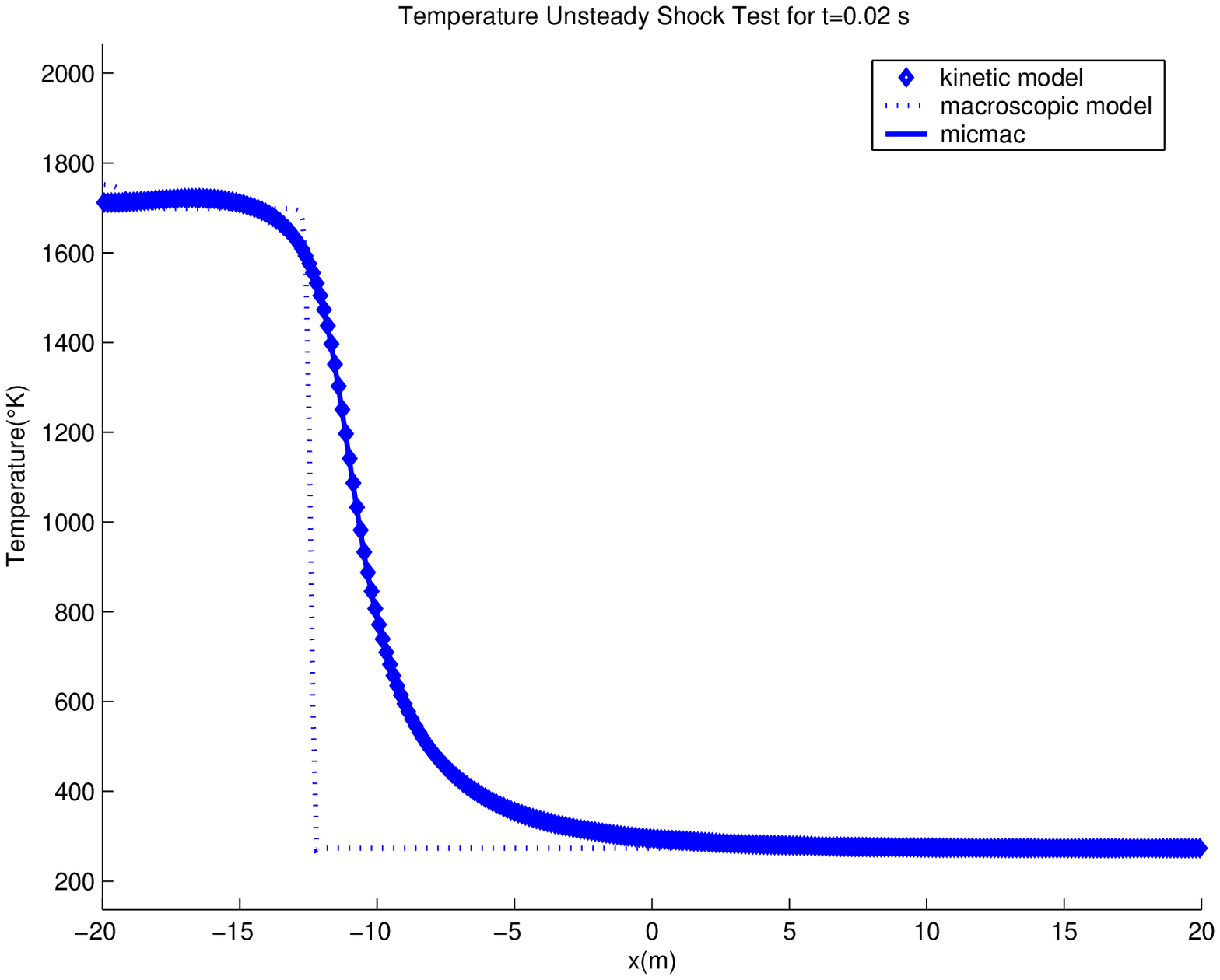}
\includegraphics[scale=0.34]{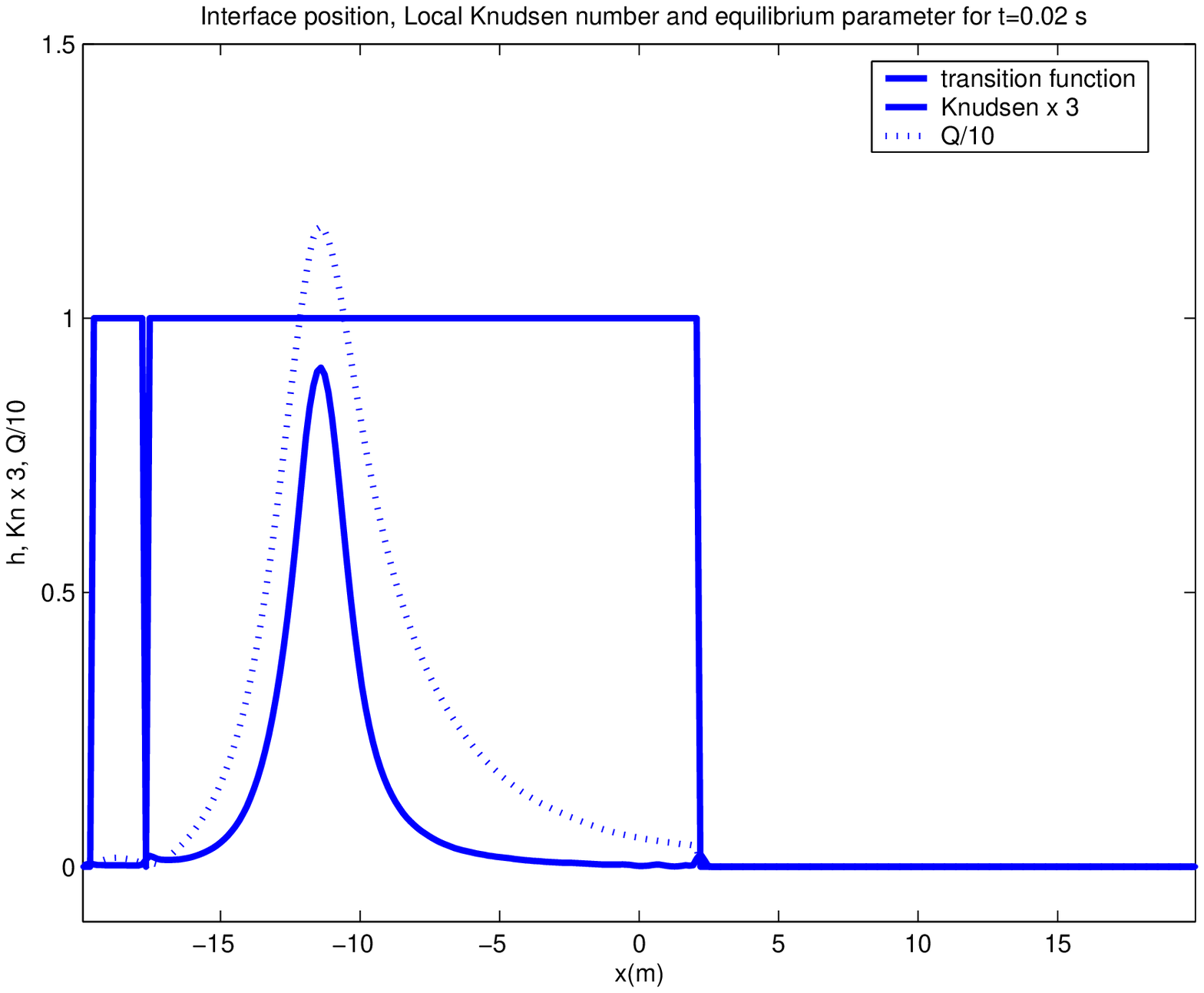}
\includegraphics[scale=0.34]{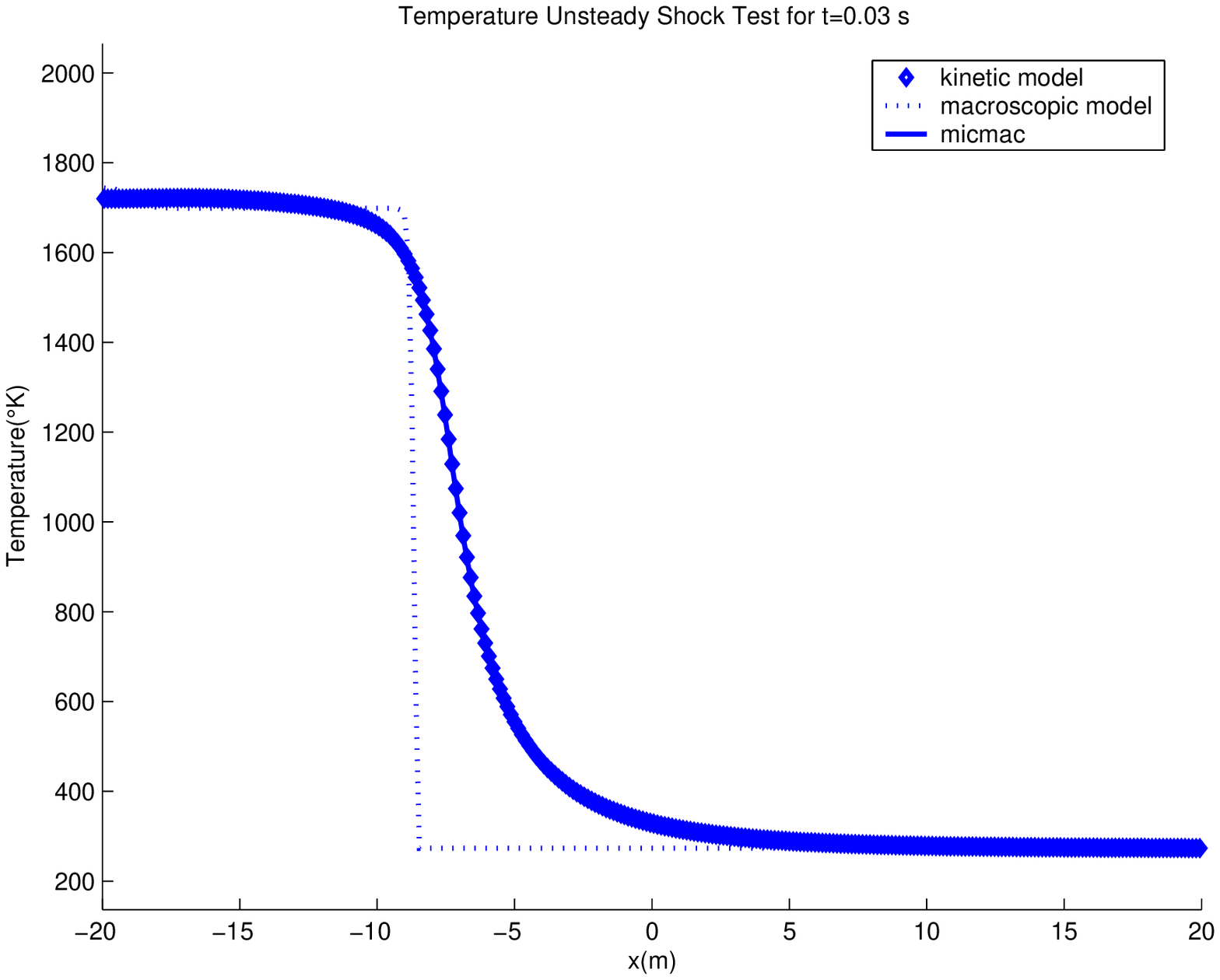}
\includegraphics[scale=0.34]{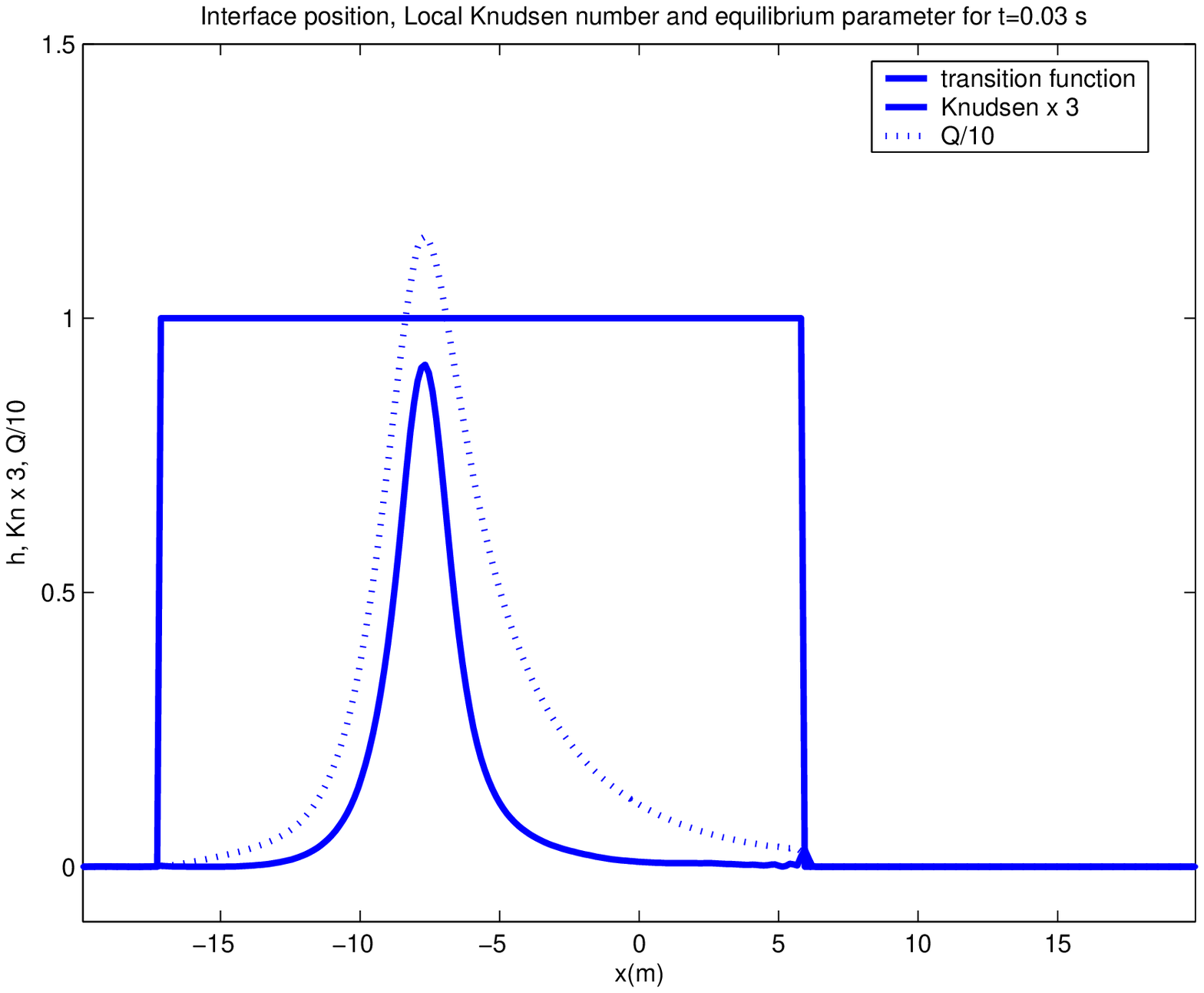}
\includegraphics[scale=0.34]{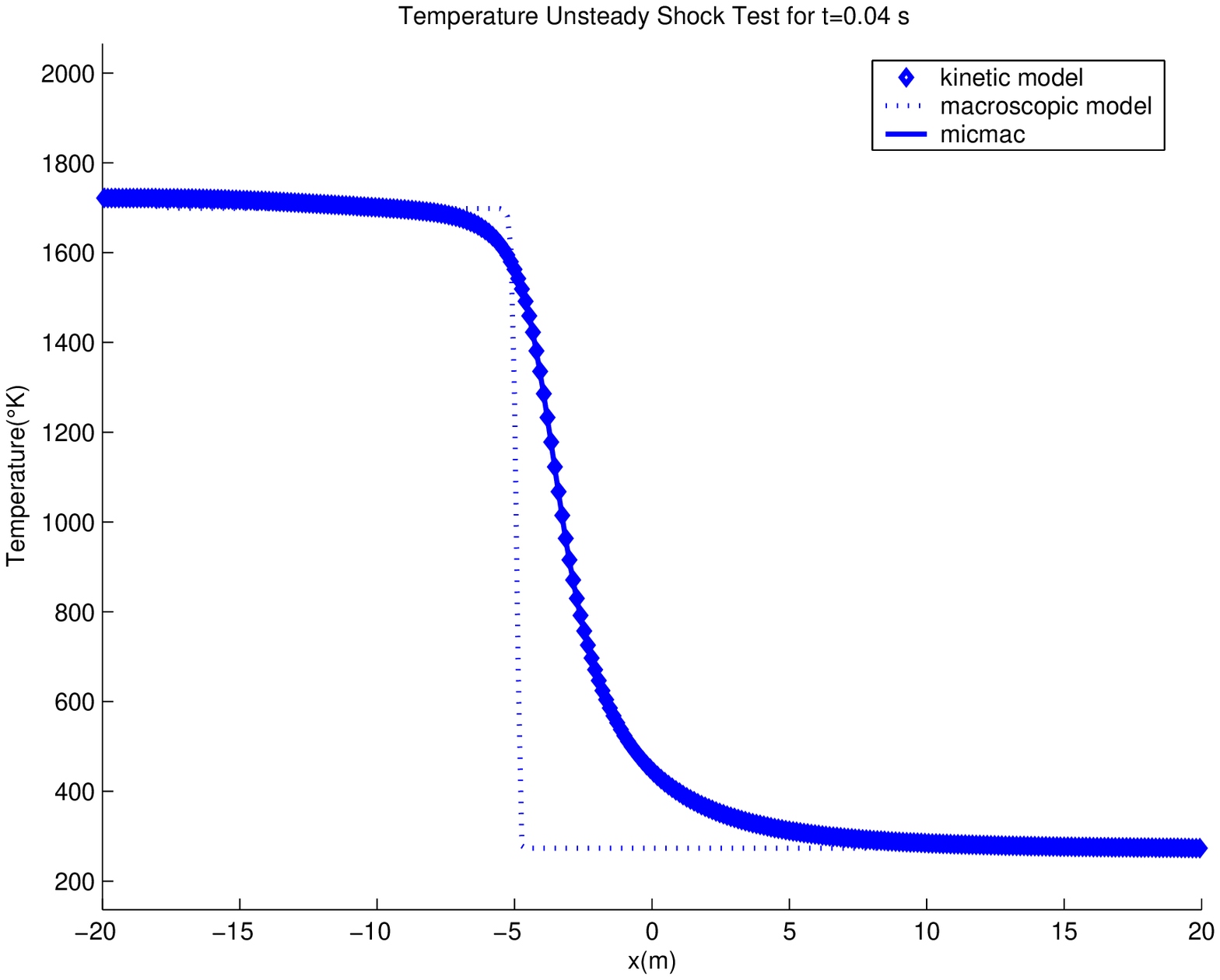}
\includegraphics[scale=0.34]{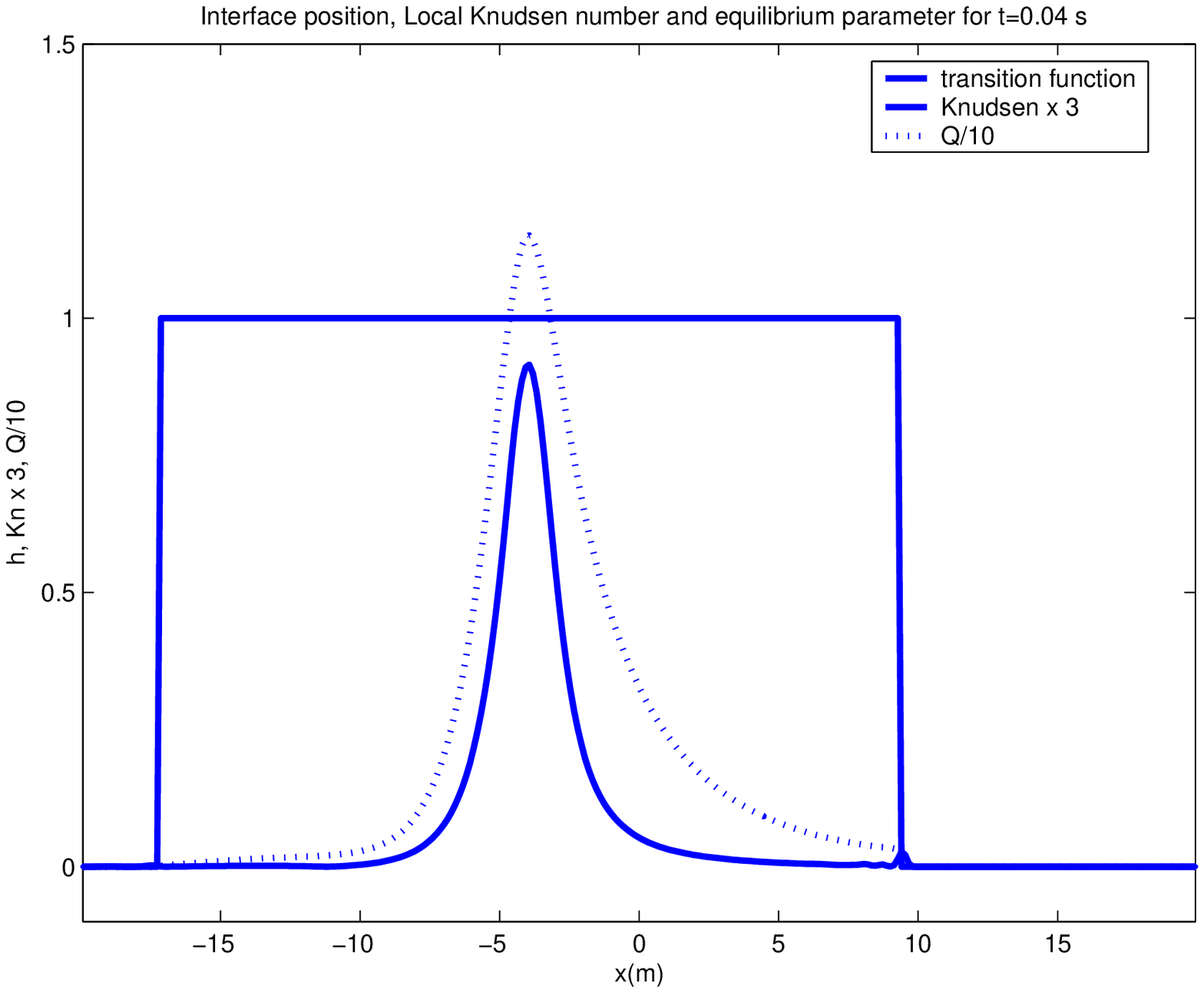}
\caption{Unsteady Shock 2: Solution at $t=1\times 10^{-2}$ top,
$t=2\times 10^{-2}$ middle top, $t=3\times 10^{-2}$ middle bottom,
$t=4\times 10^{-2}$ bottom, temperature left, transition function,
Knudsen number and heat flux right. \label{UST2.2}}
\end{center}
\end{figure}


\begin{figure}
\begin{center}
\includegraphics[scale=0.34]{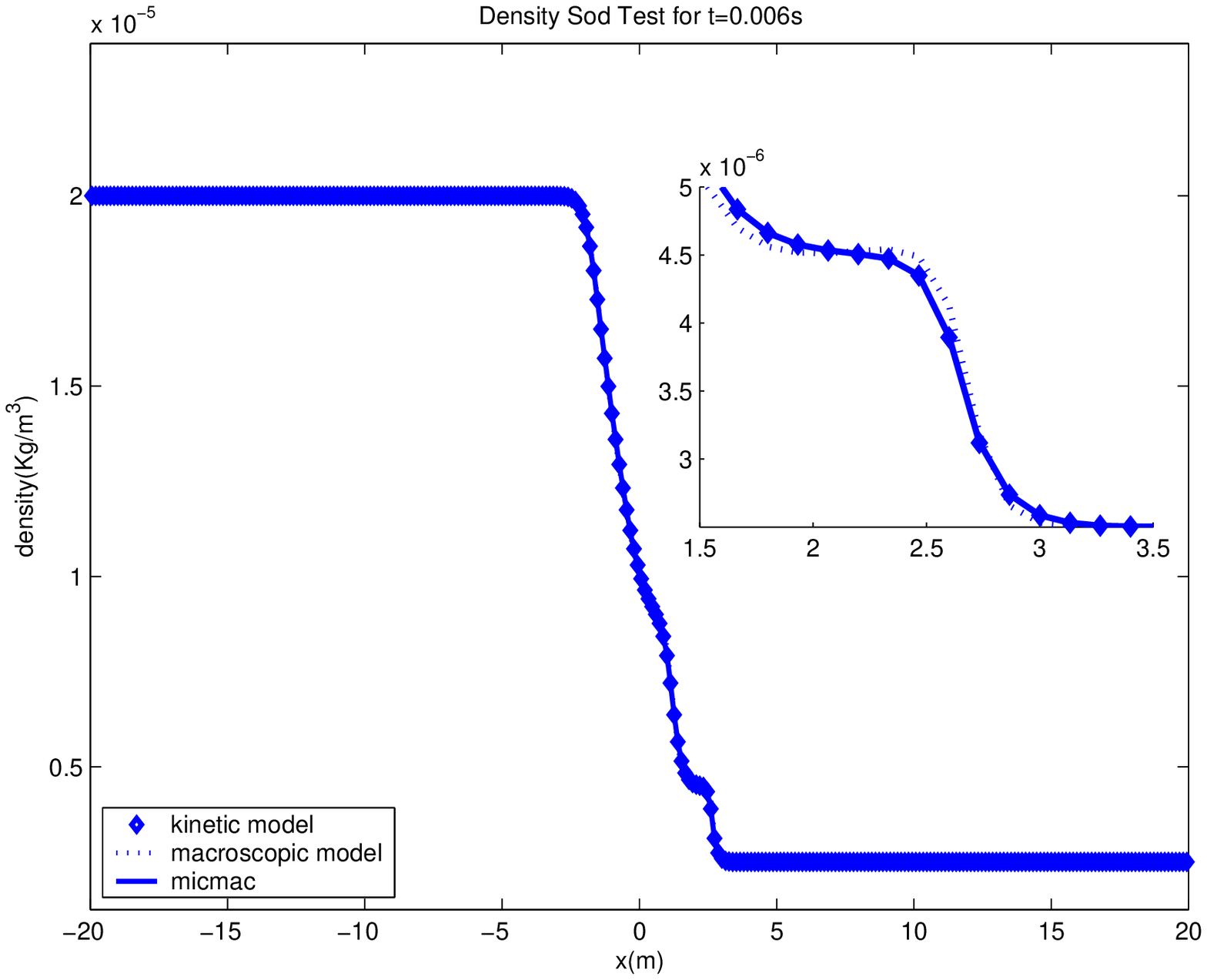}
\includegraphics[scale=0.34]{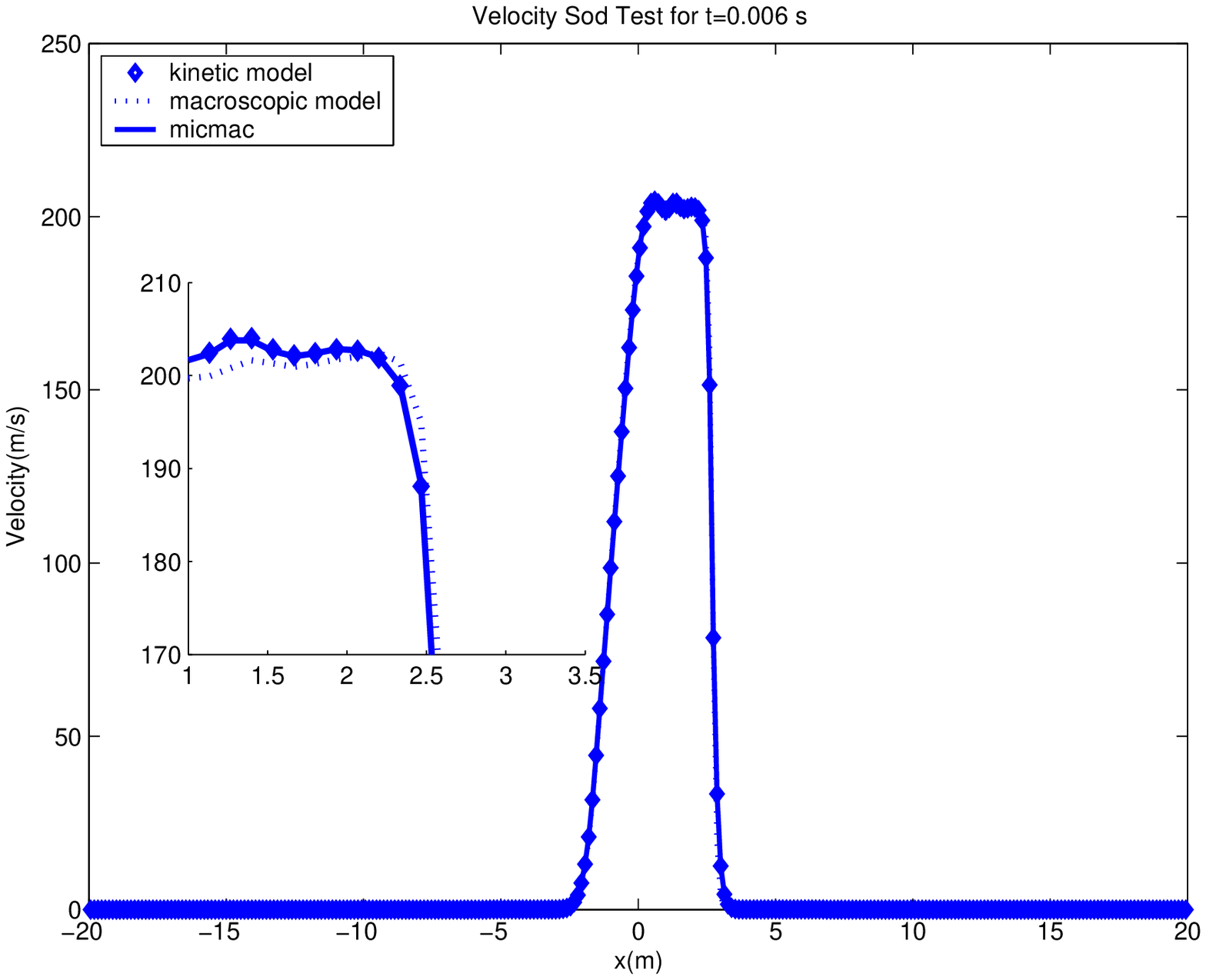}
\includegraphics[scale=0.34]{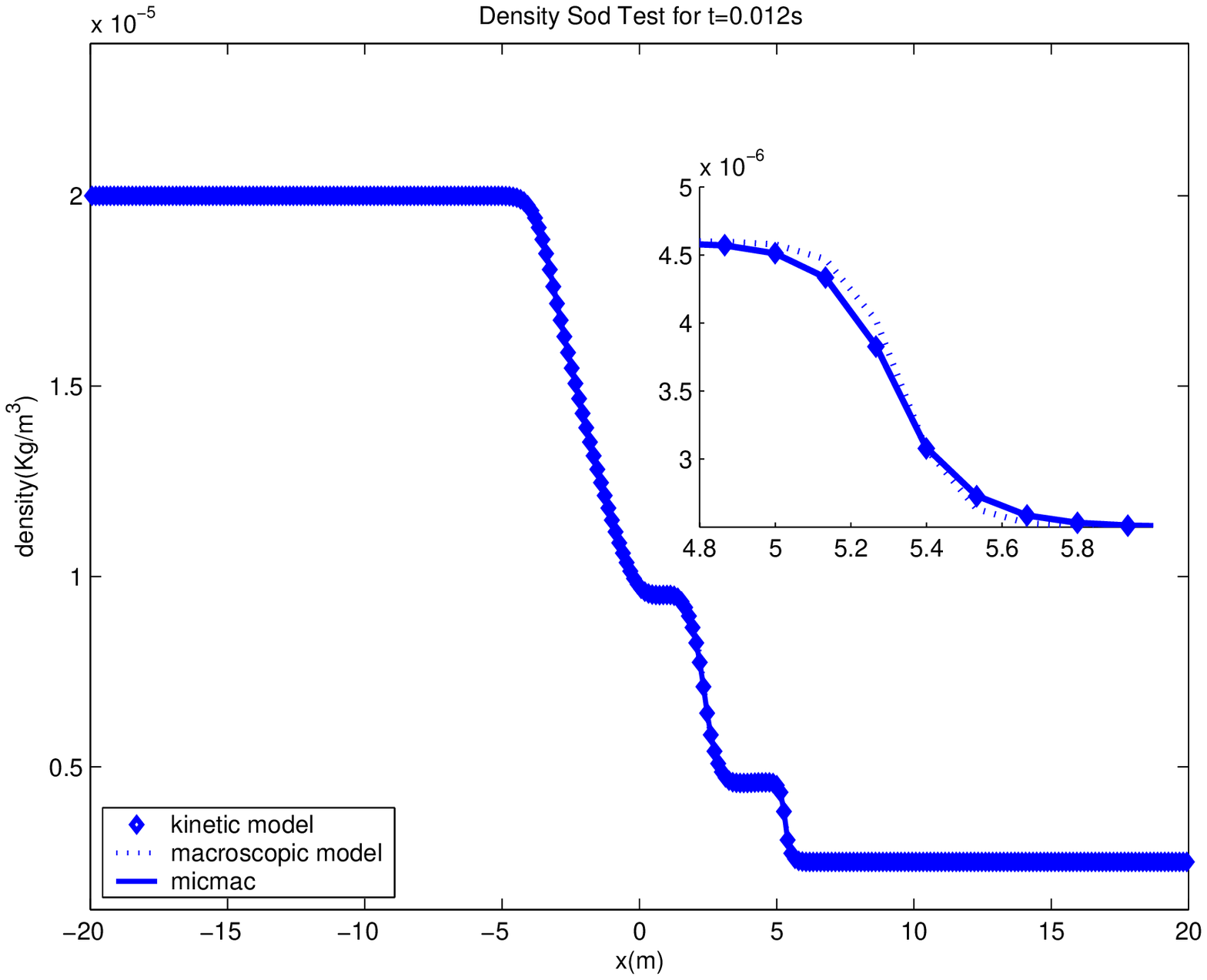}
\includegraphics[scale=0.34]{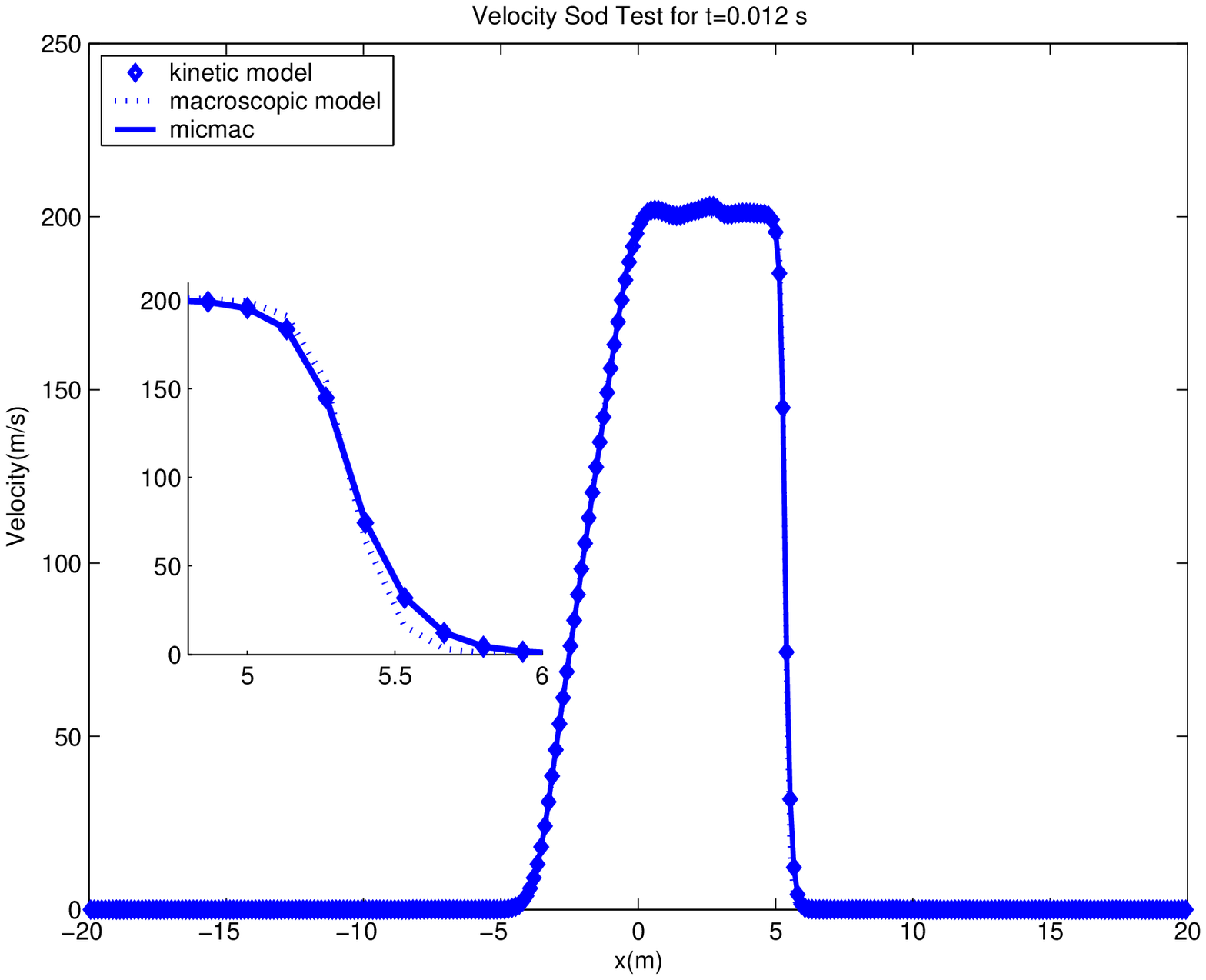}
\includegraphics[scale=0.34]{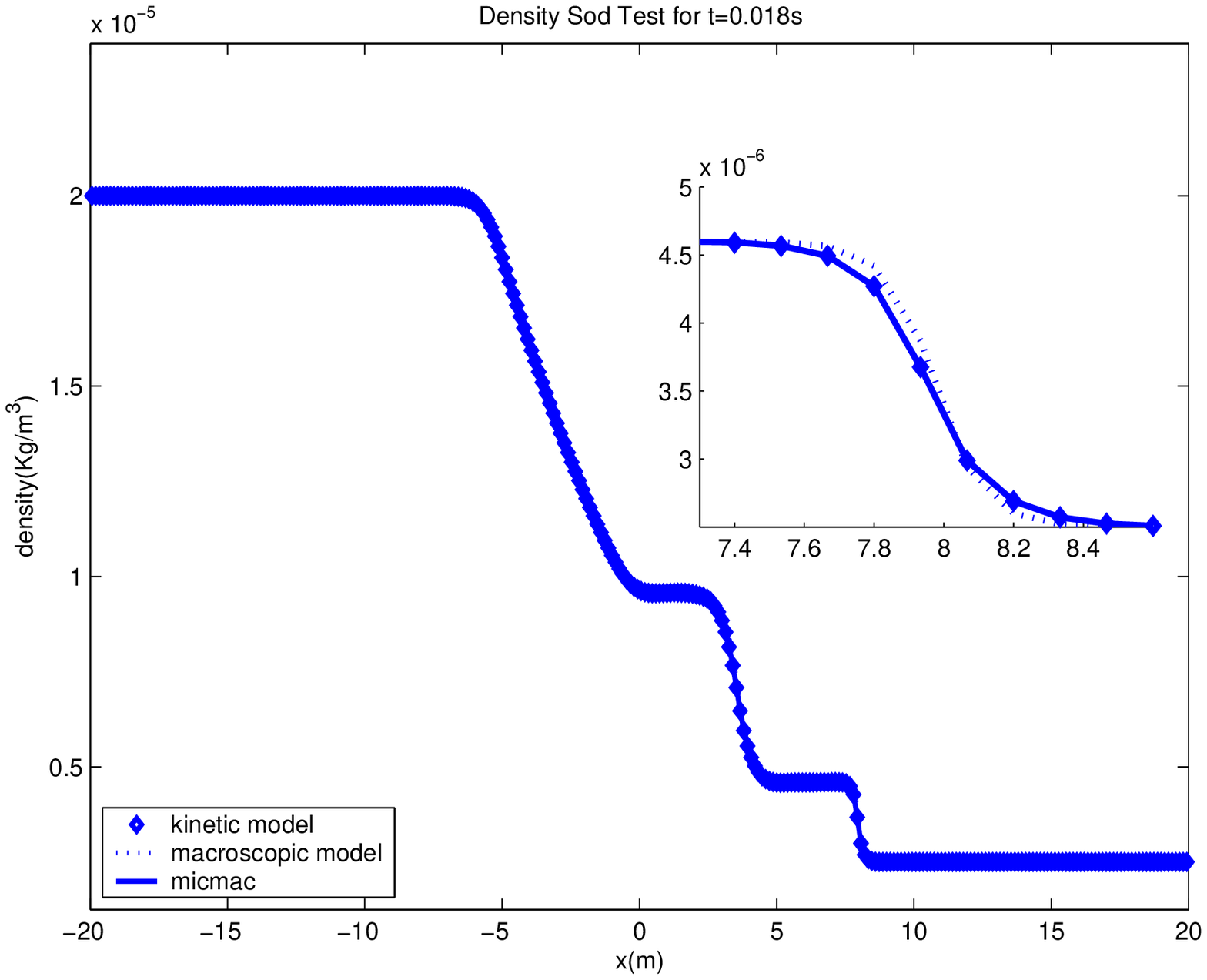}
\includegraphics[scale=0.34]{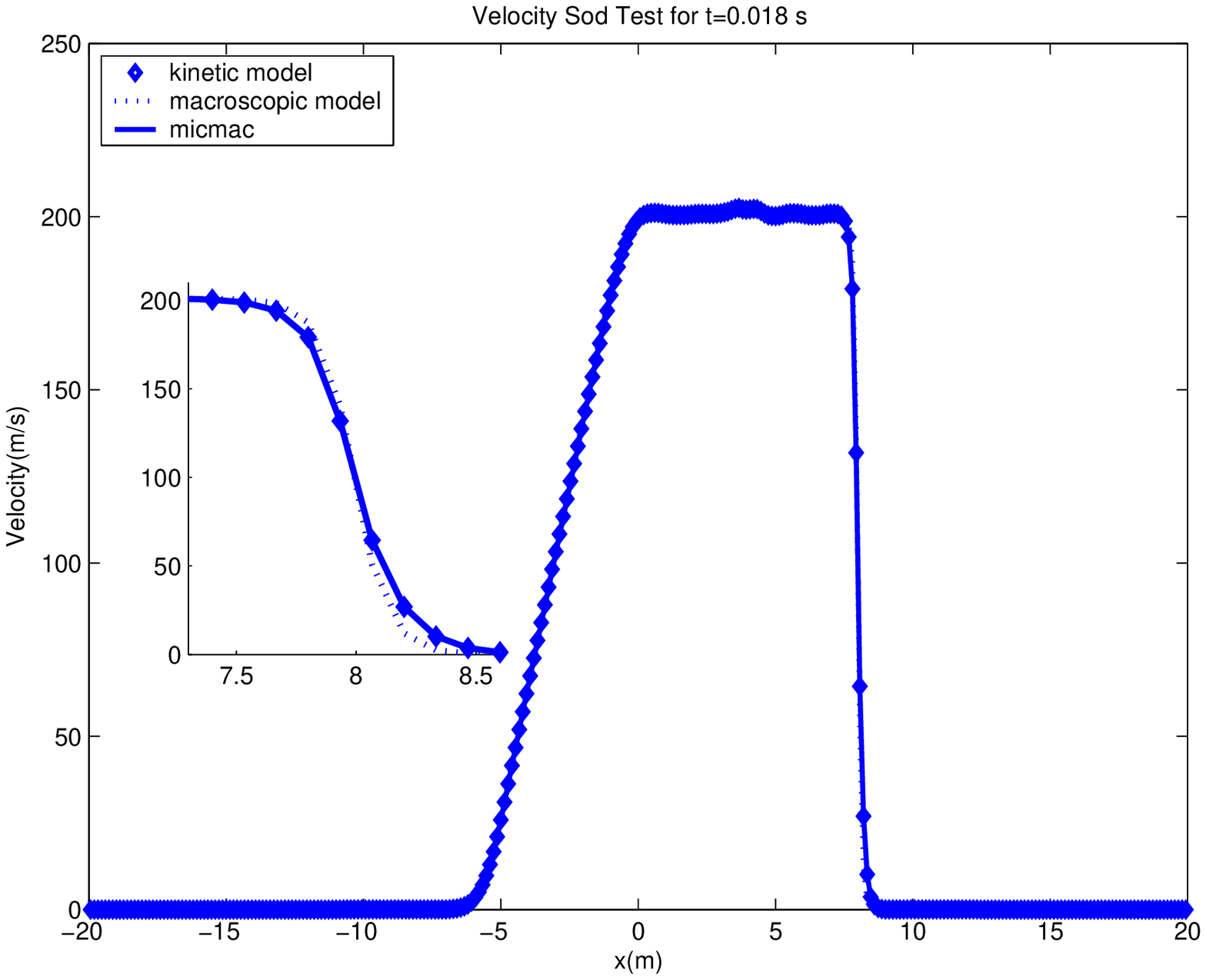}
\includegraphics[scale=0.34]{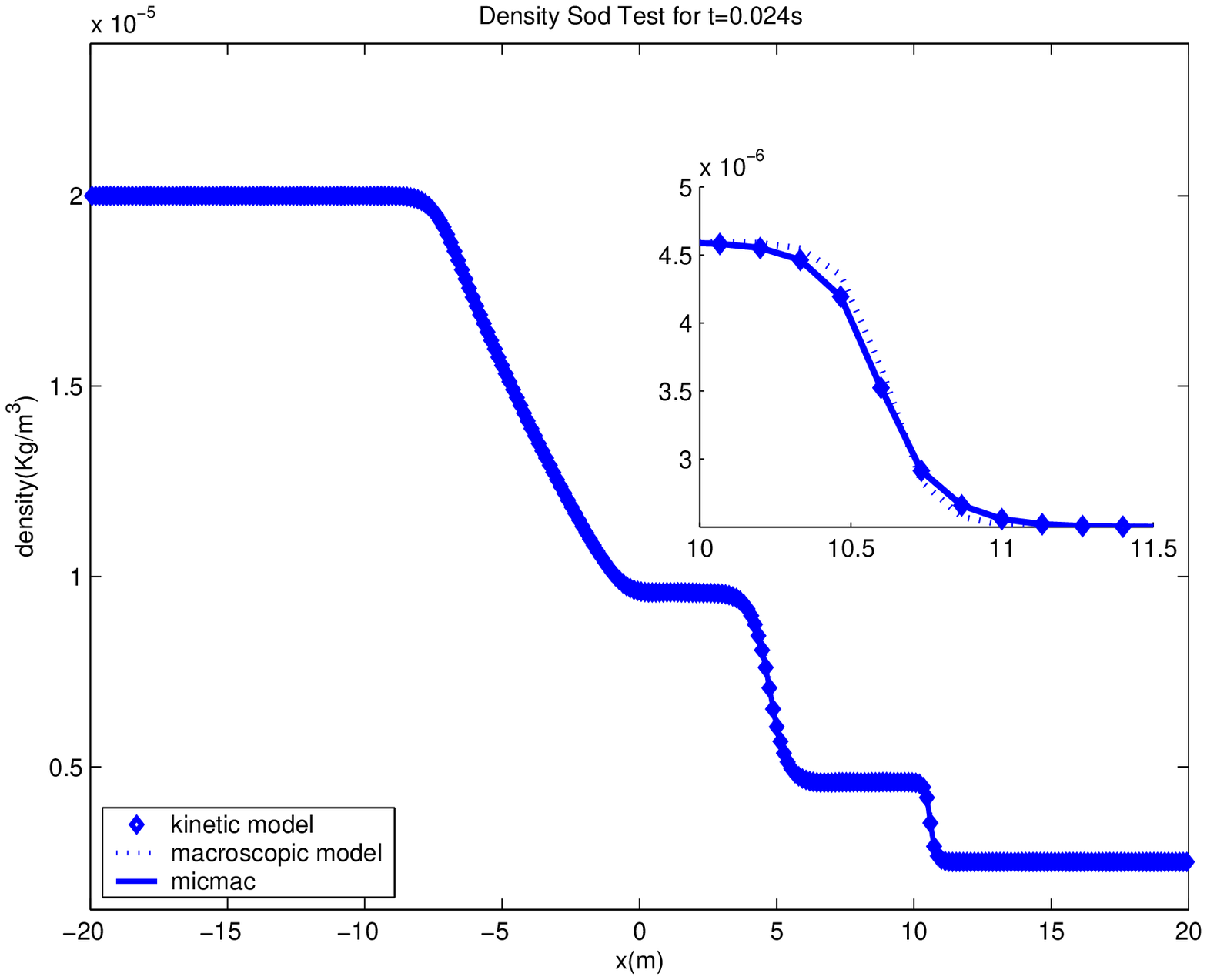}
\includegraphics[scale=0.34]{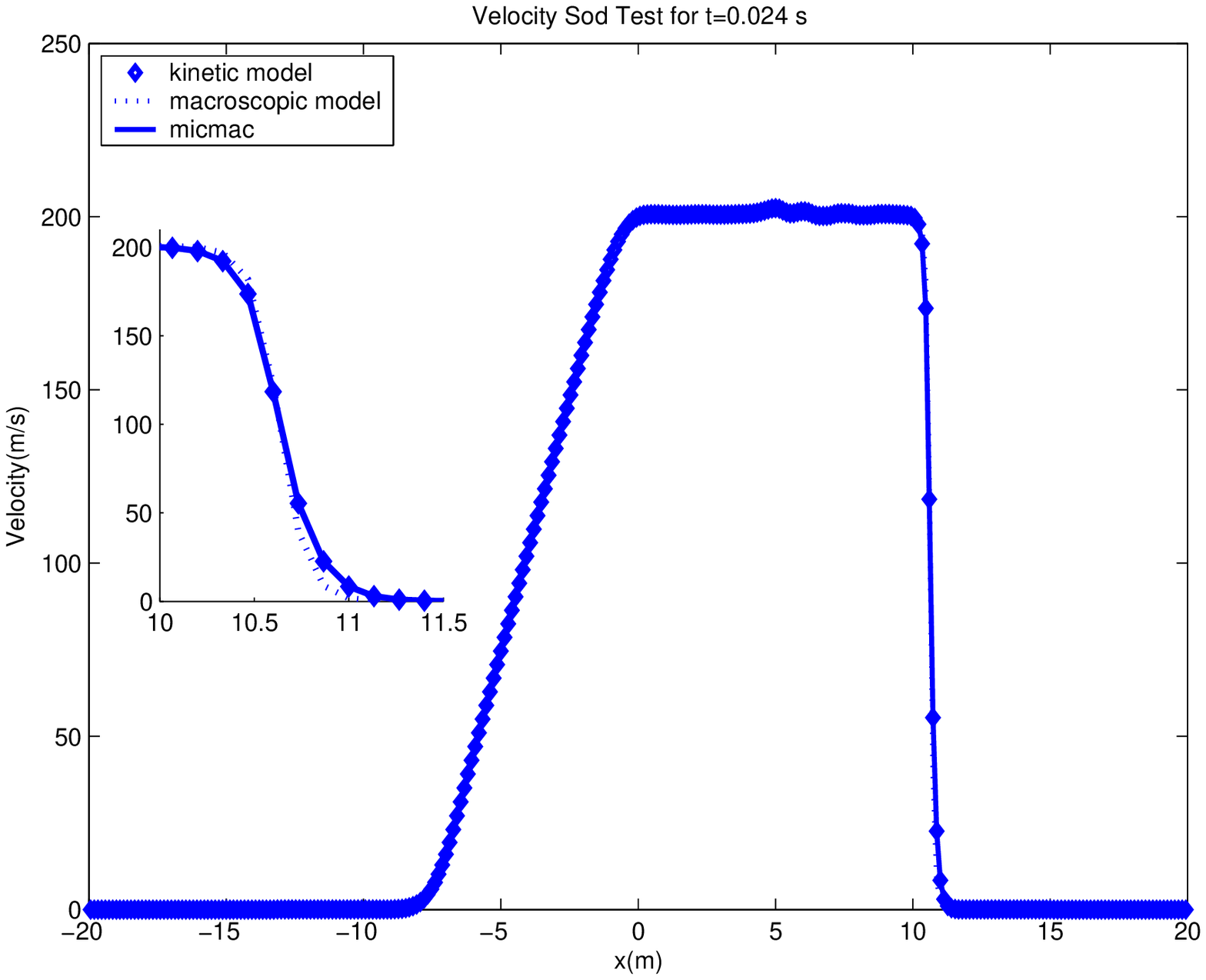}
\caption{Sod Test 1: Solution at $t=0.6\times 10^{-2}$ top,
$t=1.2\times 10^{-2}$ middle top, $t=1.8\times 10^{-2}$ middle
bottom, $t=2.4\times 10^{-2}$ bottom, density left, velocity  right.
The small panels are a magnification of the solution close to non
equilibrium regions. \label{sod1.1}}
\end{center}
\end{figure}

\begin{figure}
\begin{center}
\includegraphics[scale=0.34]{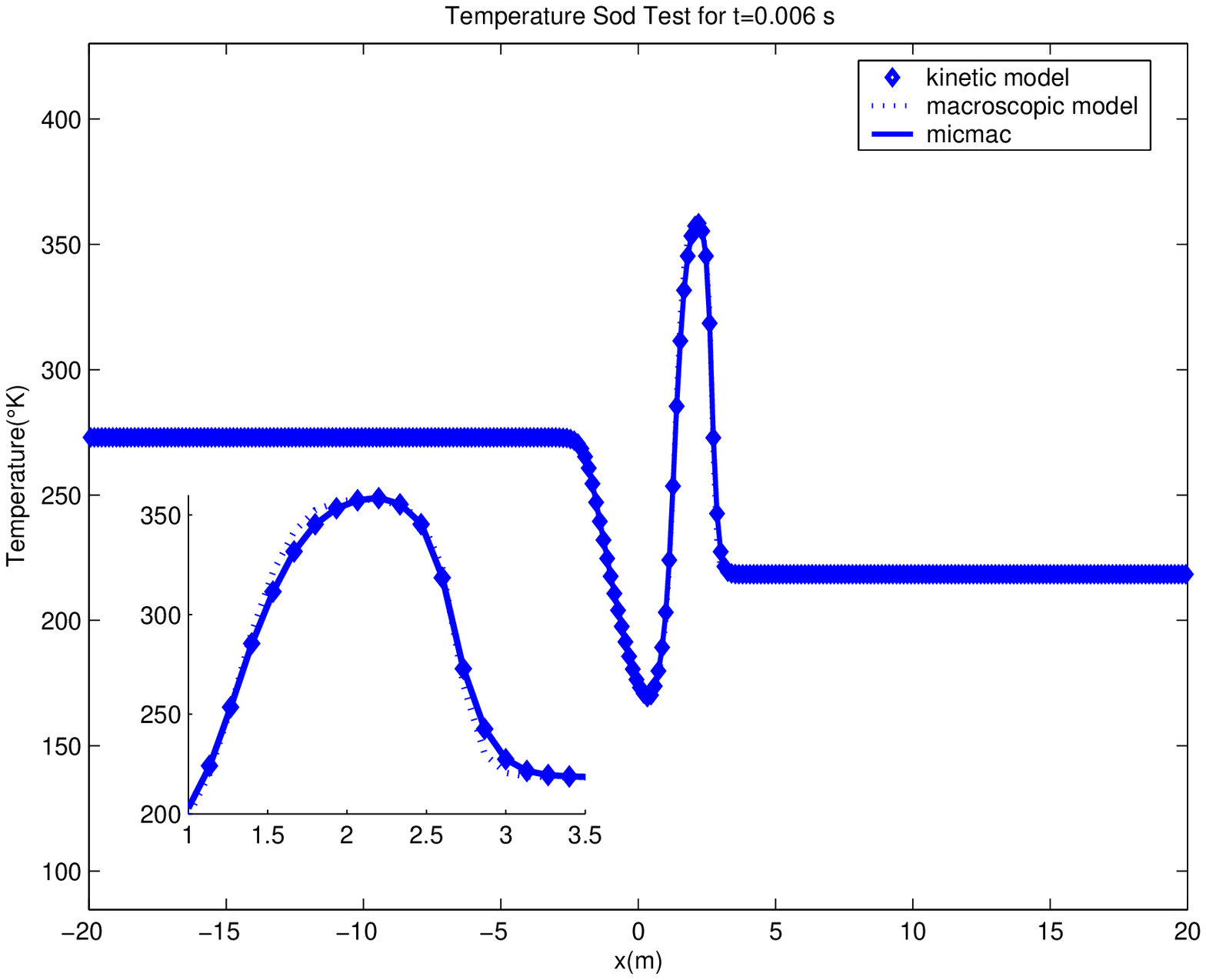}
\includegraphics[scale=0.34]{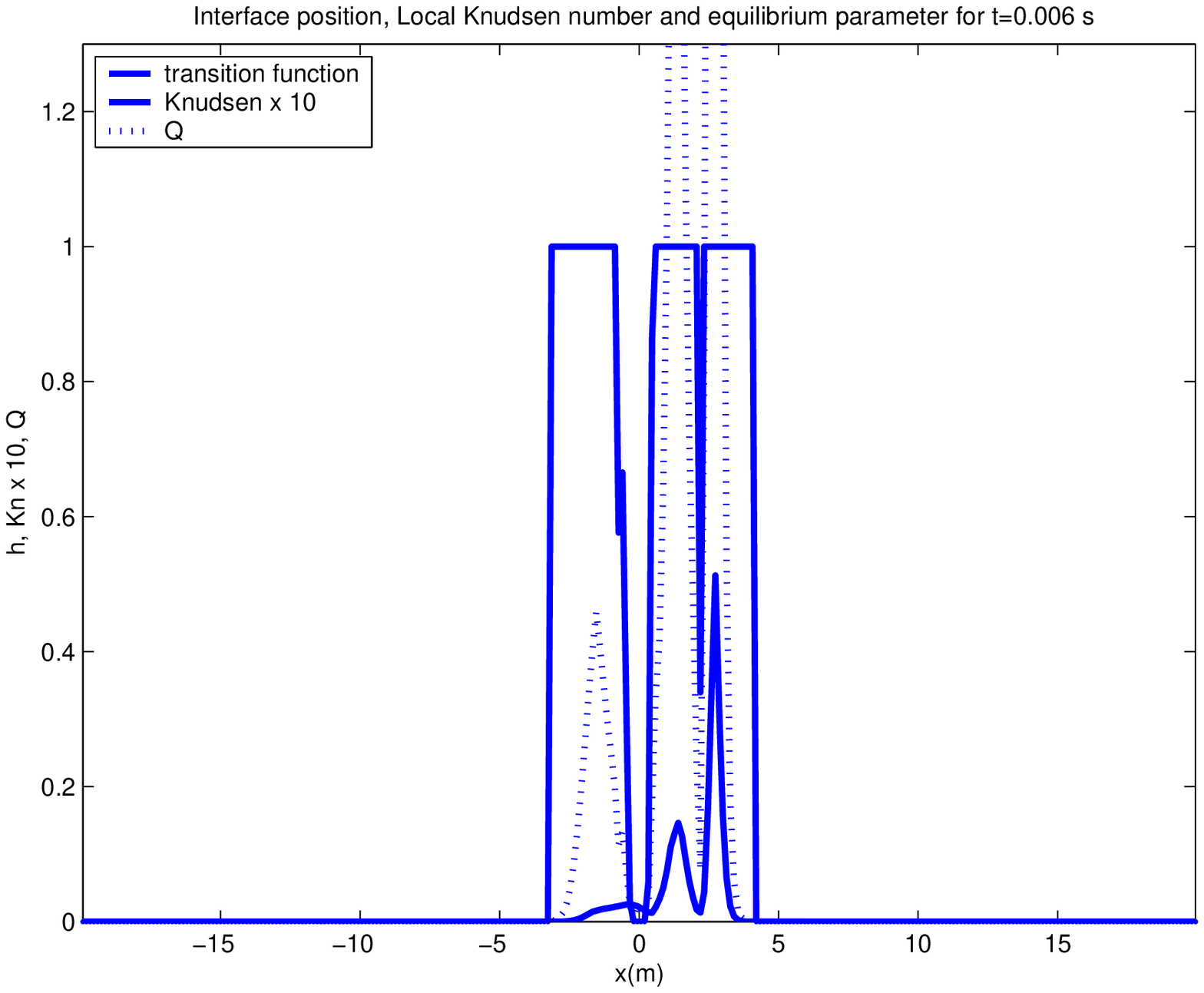}
\includegraphics[scale=0.34]{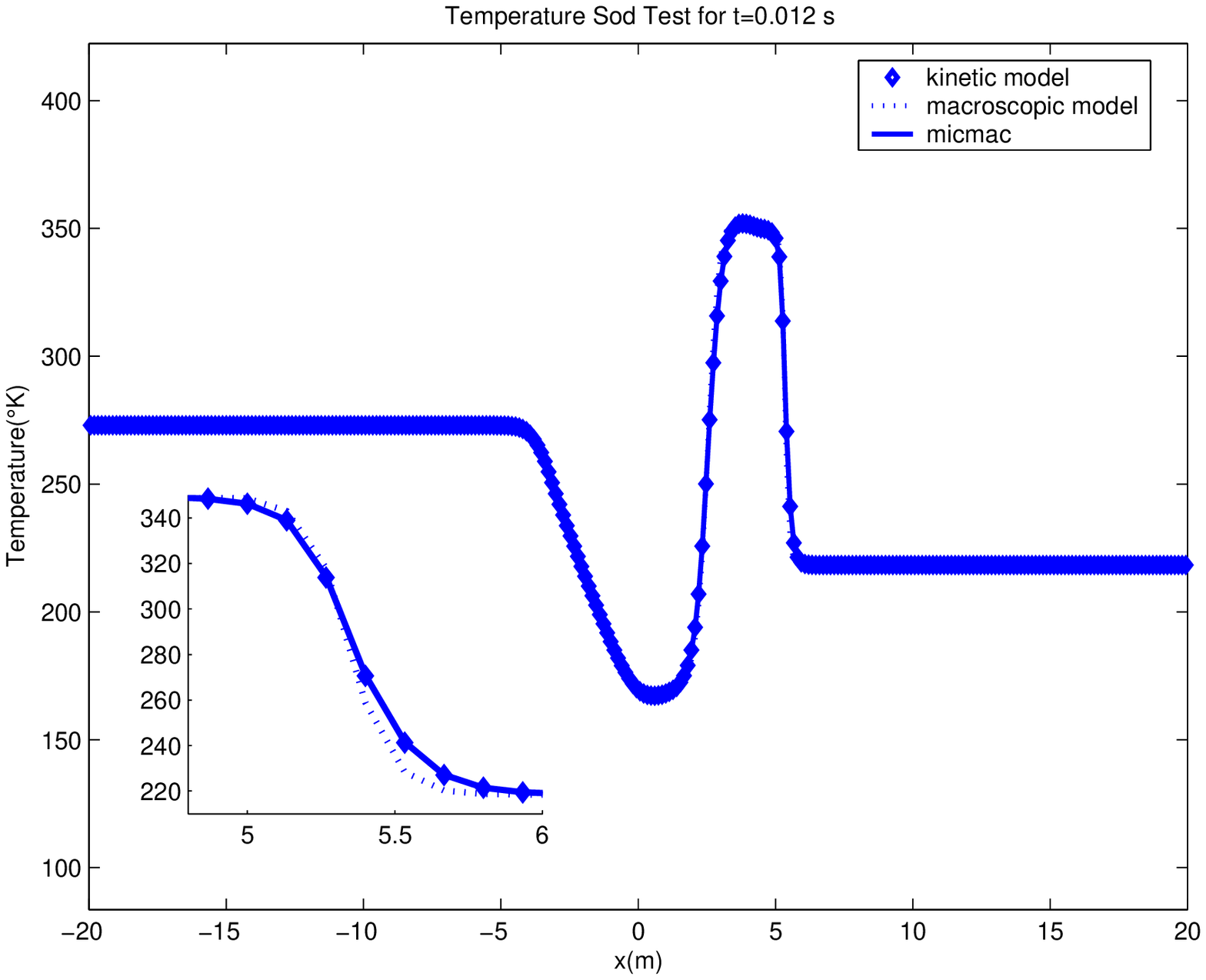}
\includegraphics[scale=0.34]{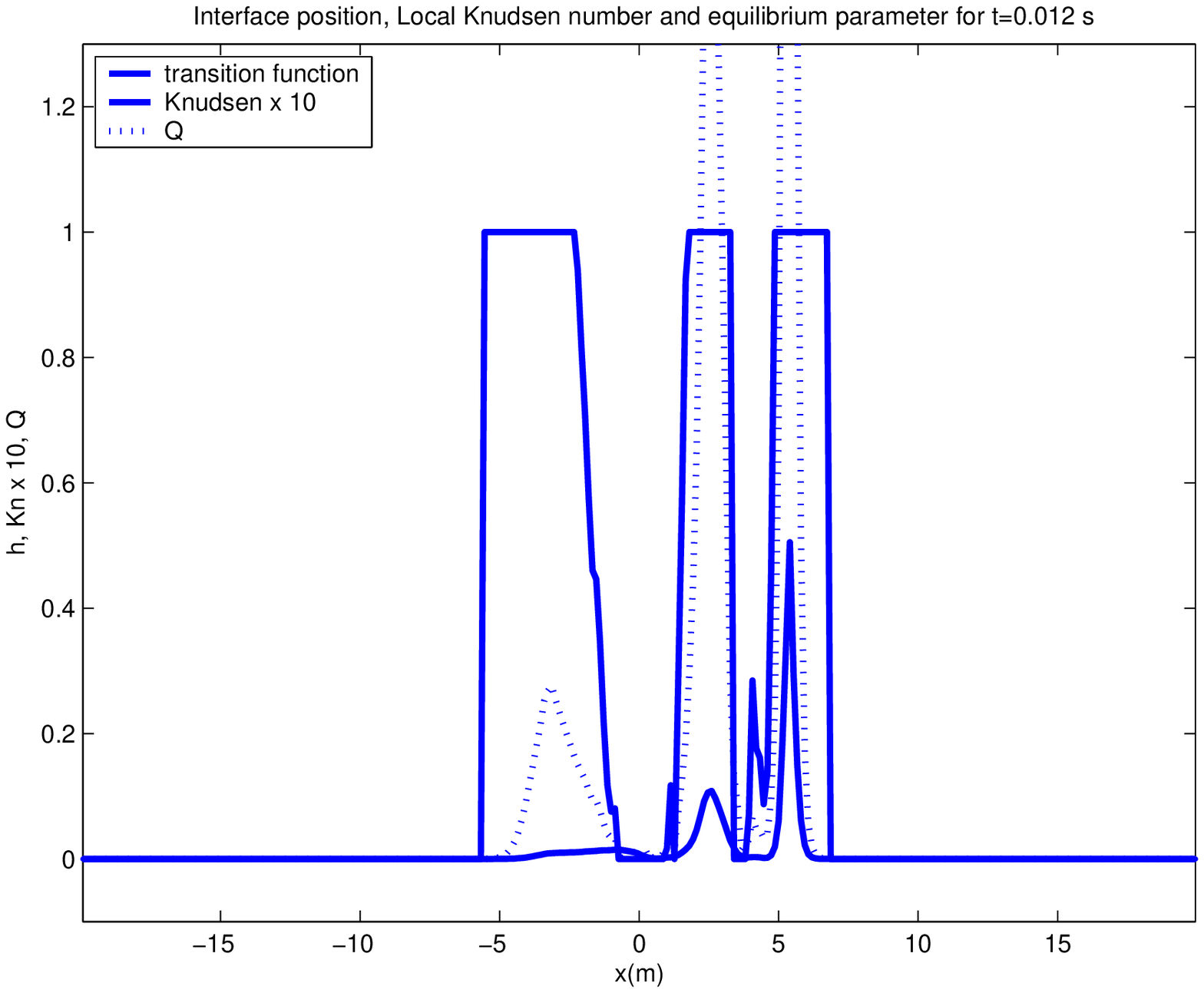}
\includegraphics[scale=0.34]{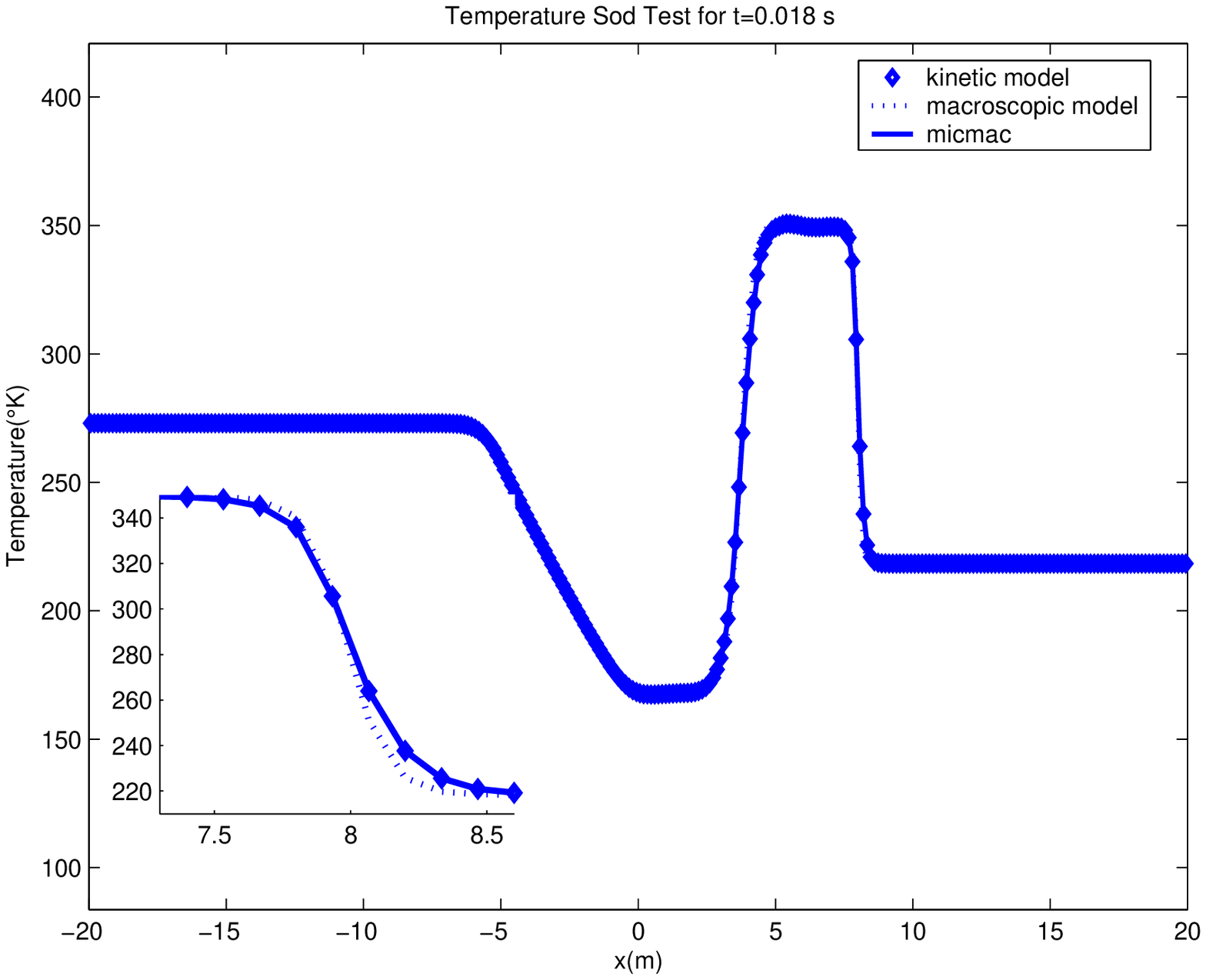}
\includegraphics[scale=0.34]{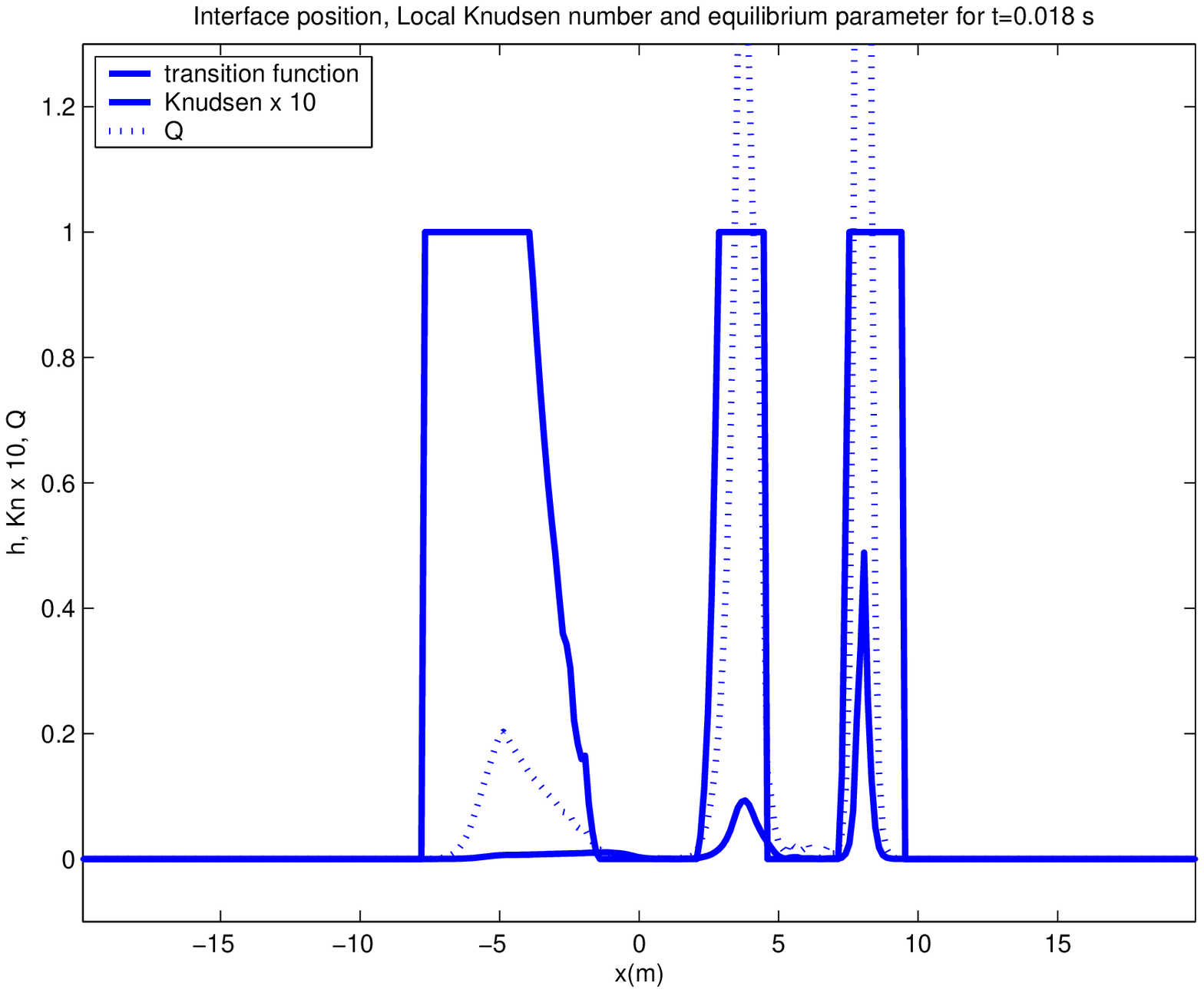}
\includegraphics[scale=0.34]{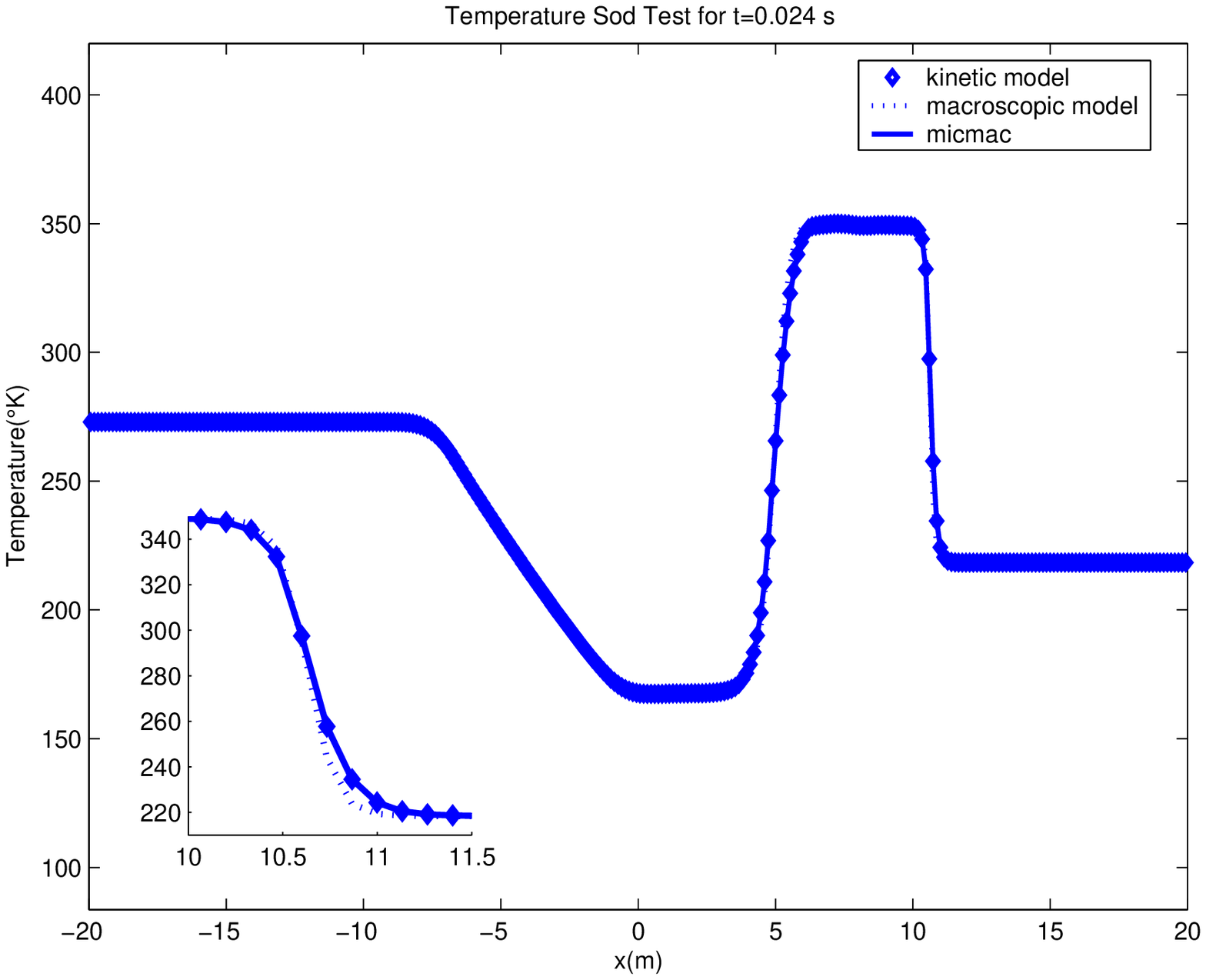}
\includegraphics[scale=0.34]{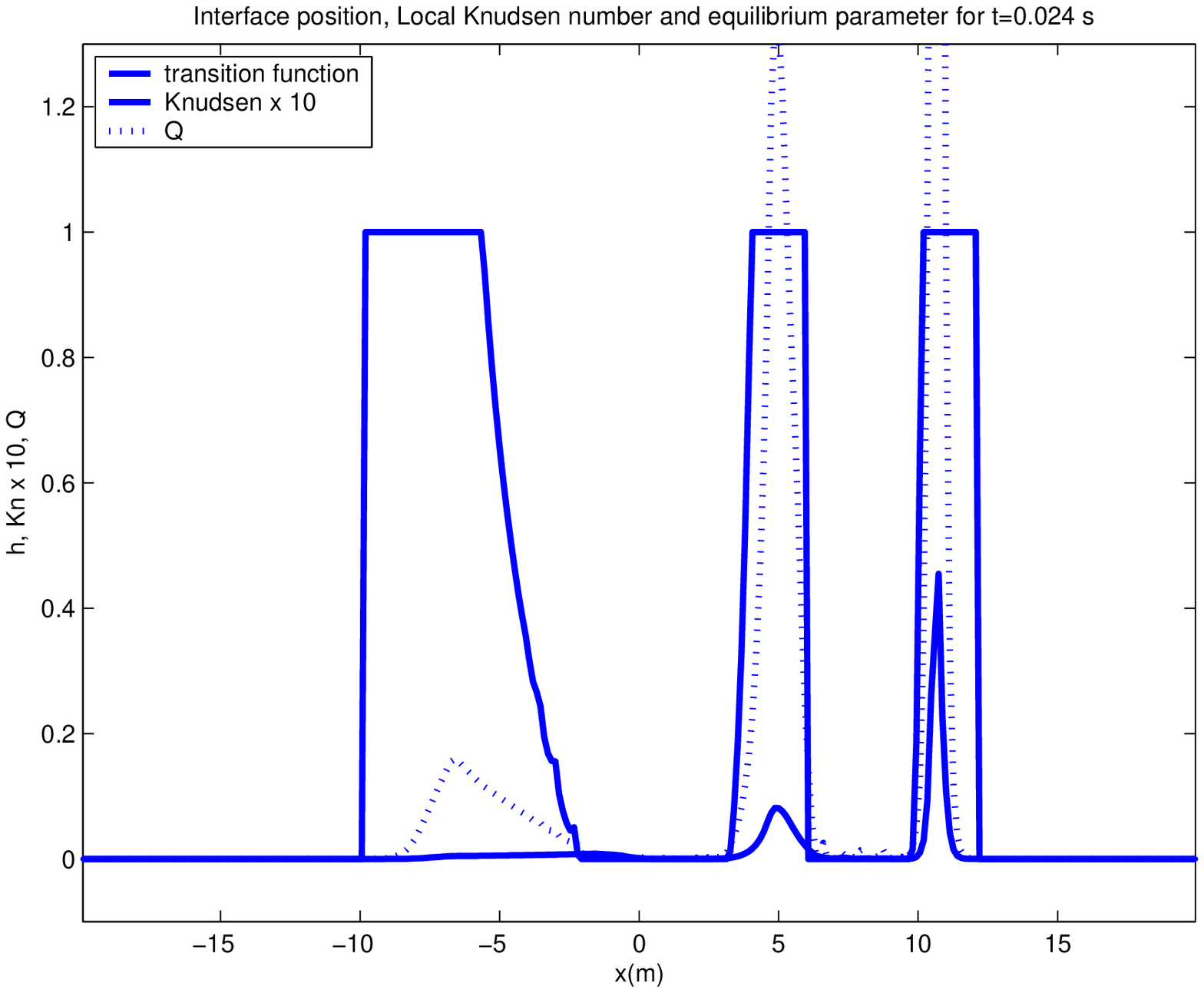}
\caption{Sod Test 1: Solution at $t=0.6\times 10^{-2}$ top,
$t=1.2\times 10^{-2}$ middle top, $t=1.8\times 10^{-2}$ middle
bottom, $t=2.4\times 10^{-2}$ bottom, temperature left, transition
function , Knudsen number and heat flux right. The small panels are
a magnification of the solution close to non equilibrium
regions.\label{sod1.2}}
\end{center}
\end{figure}

\begin{figure}
\begin{center}
\includegraphics[scale=0.34]{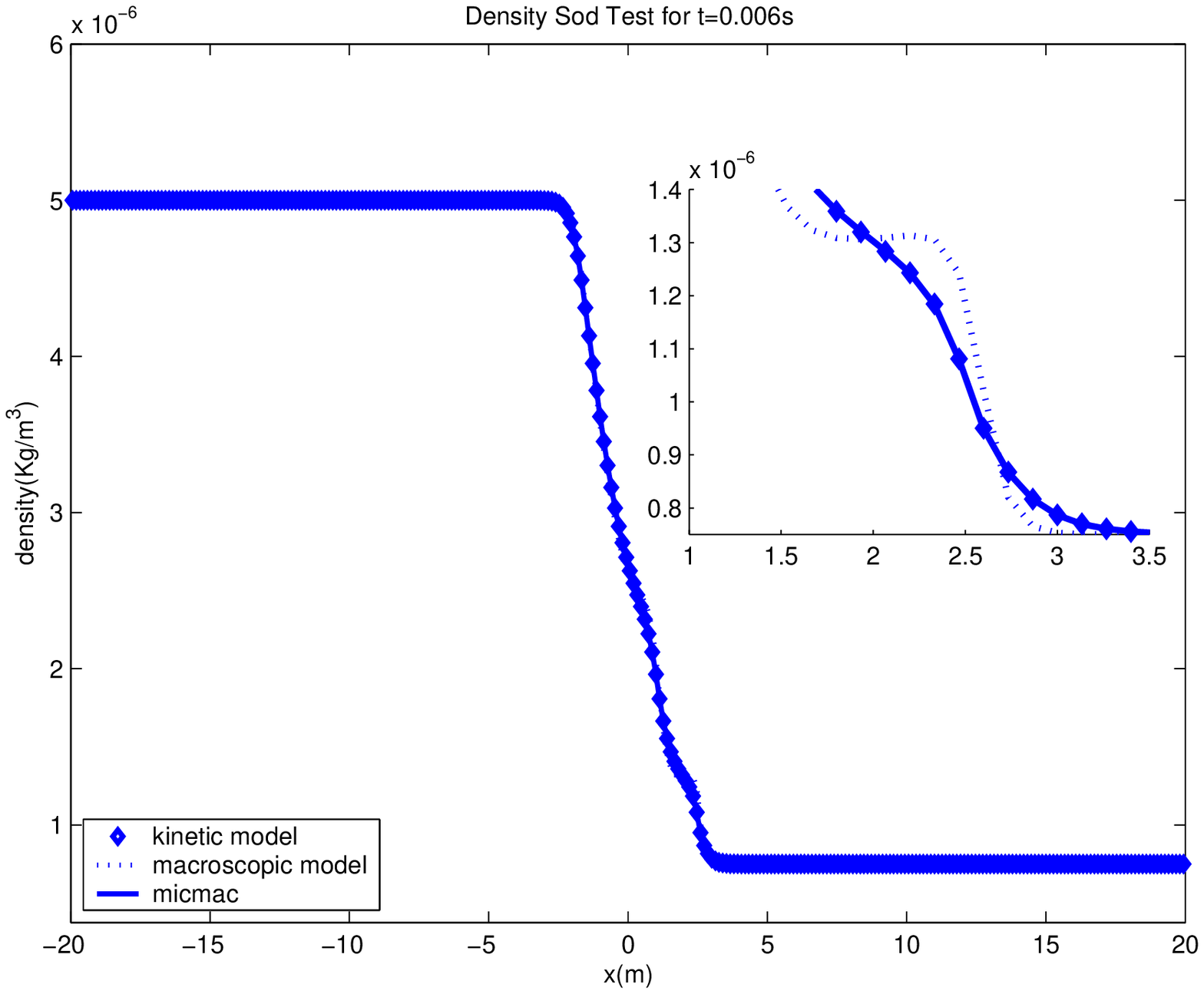}
\includegraphics[scale=0.34]{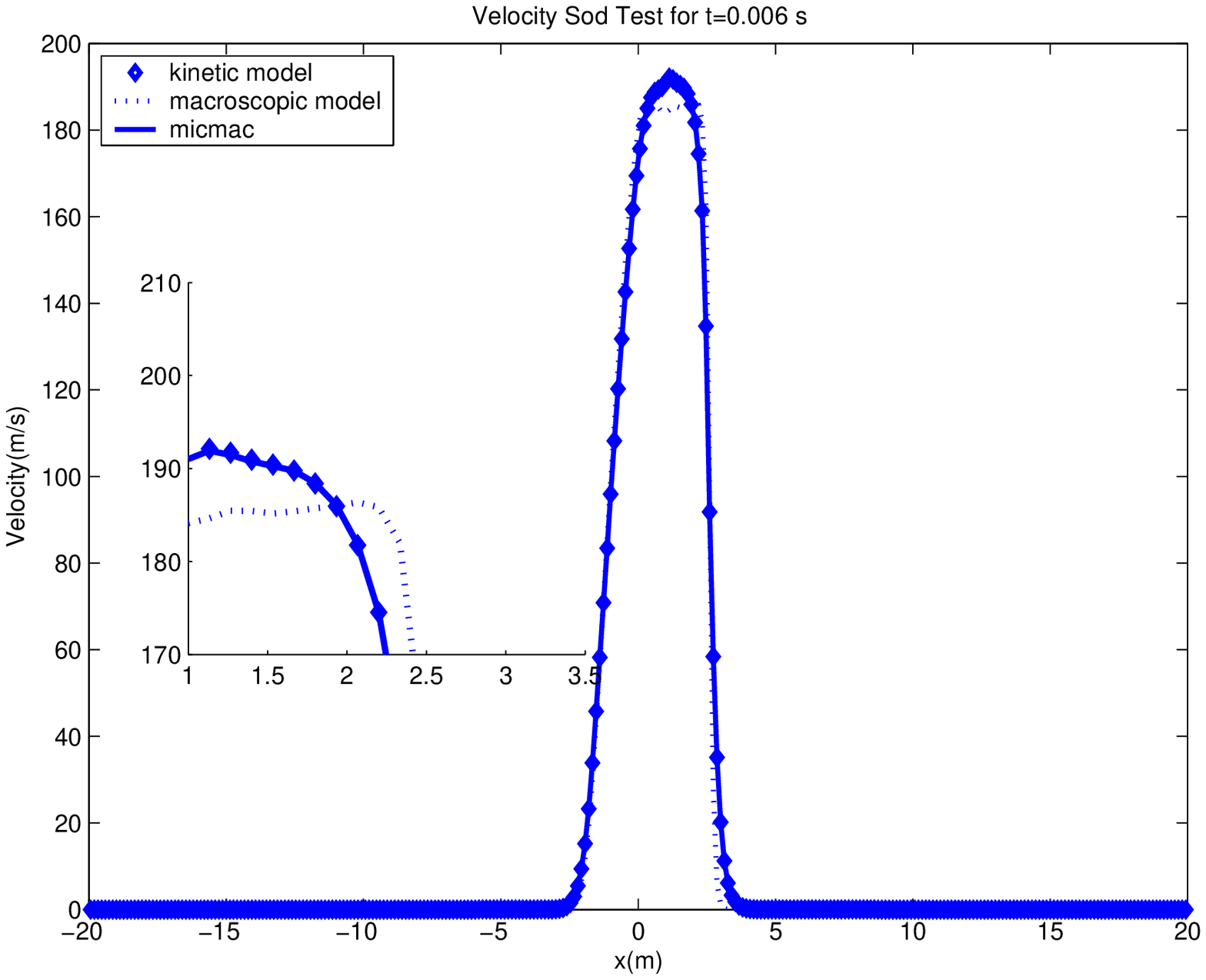}
\includegraphics[scale=0.34]{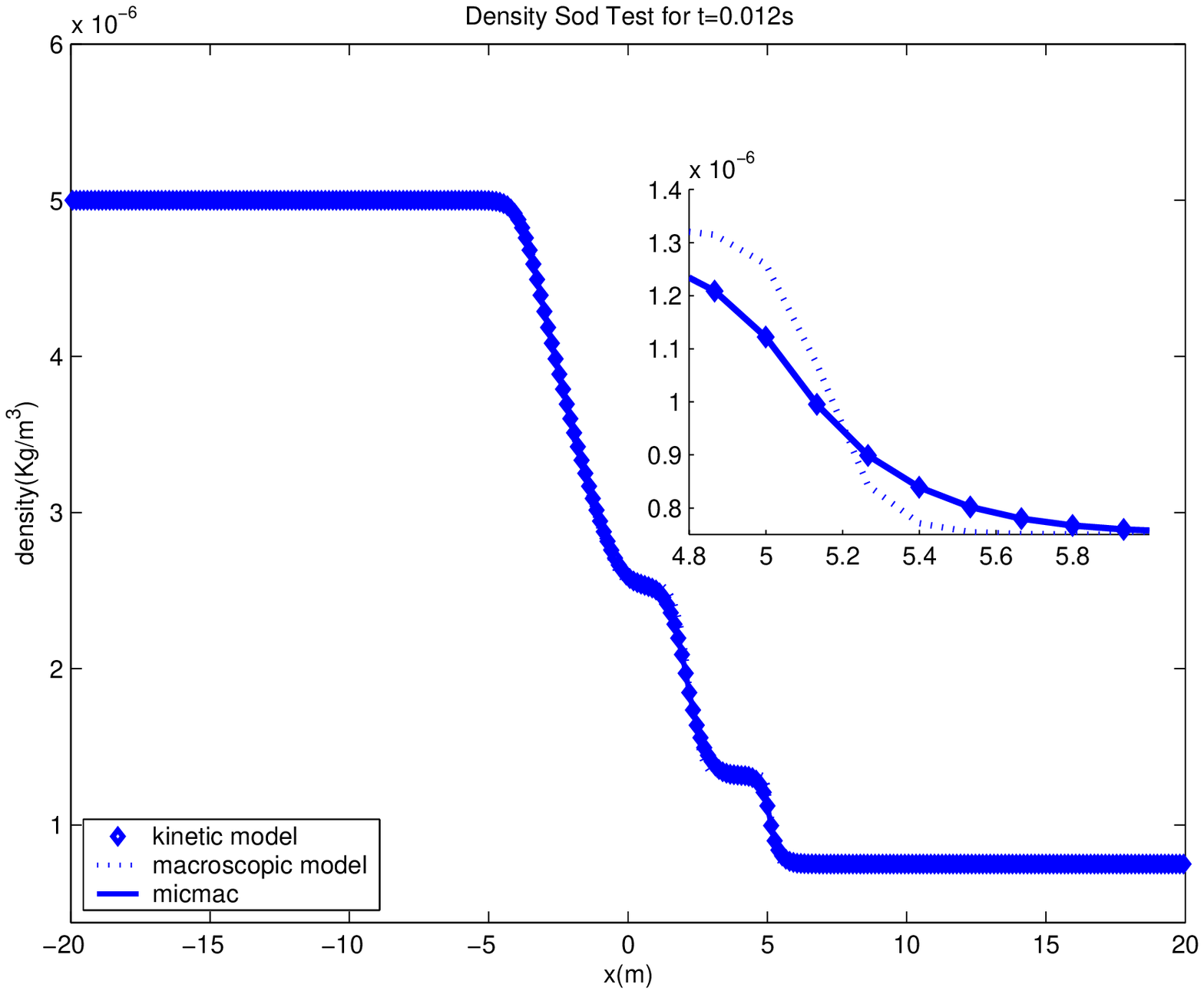}
\includegraphics[scale=0.34]{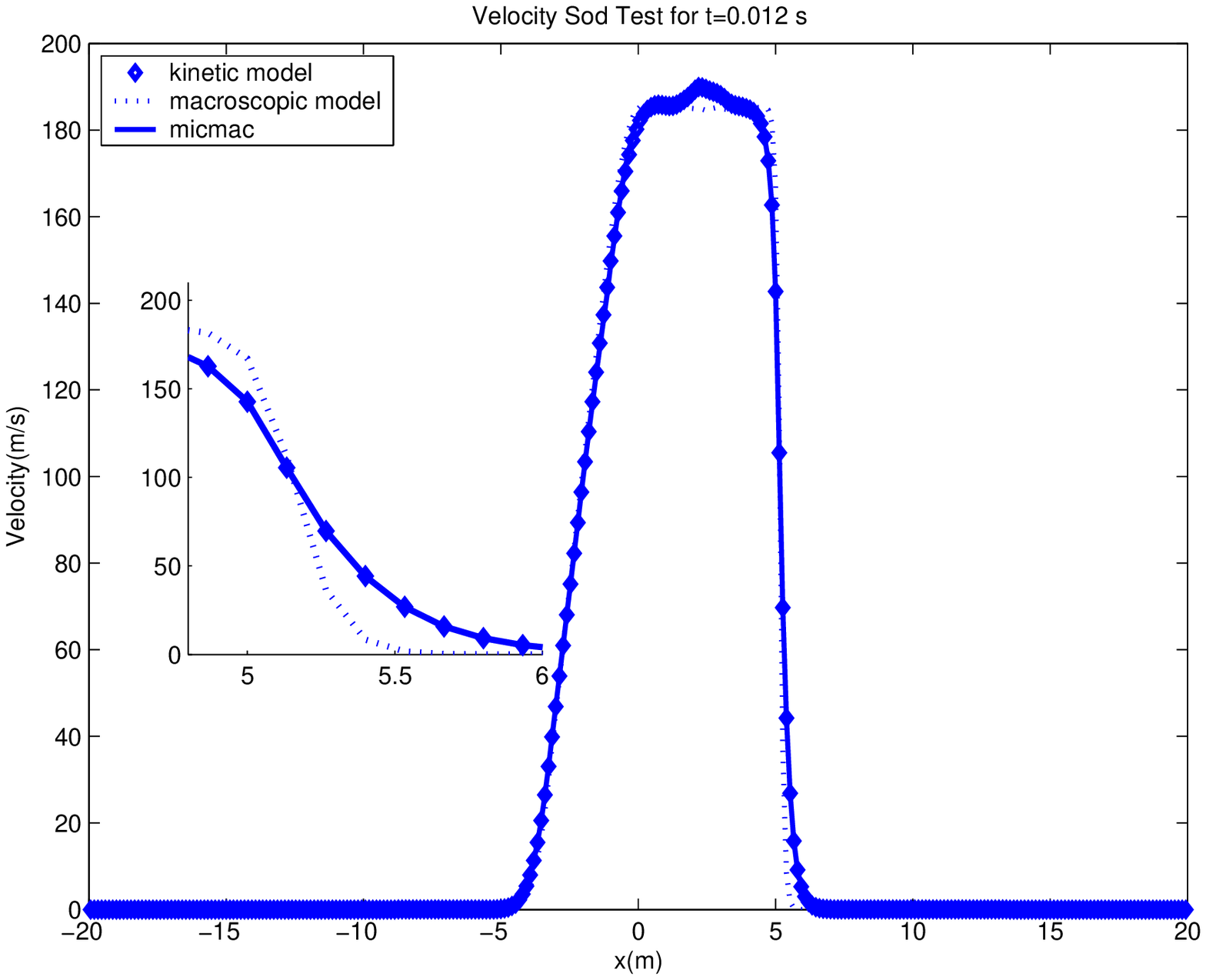}
\includegraphics[scale=0.34]{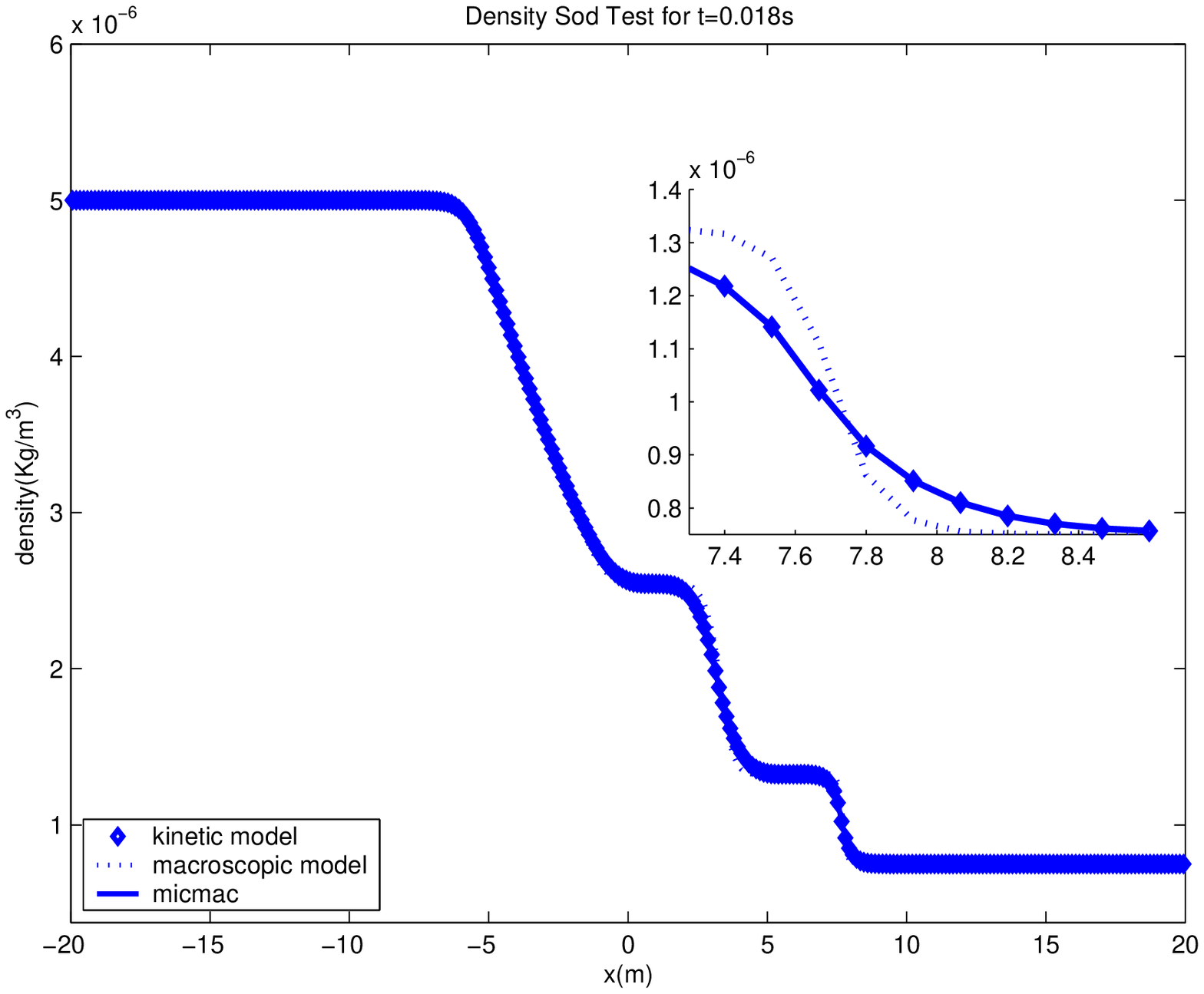}
\includegraphics[scale=0.34]{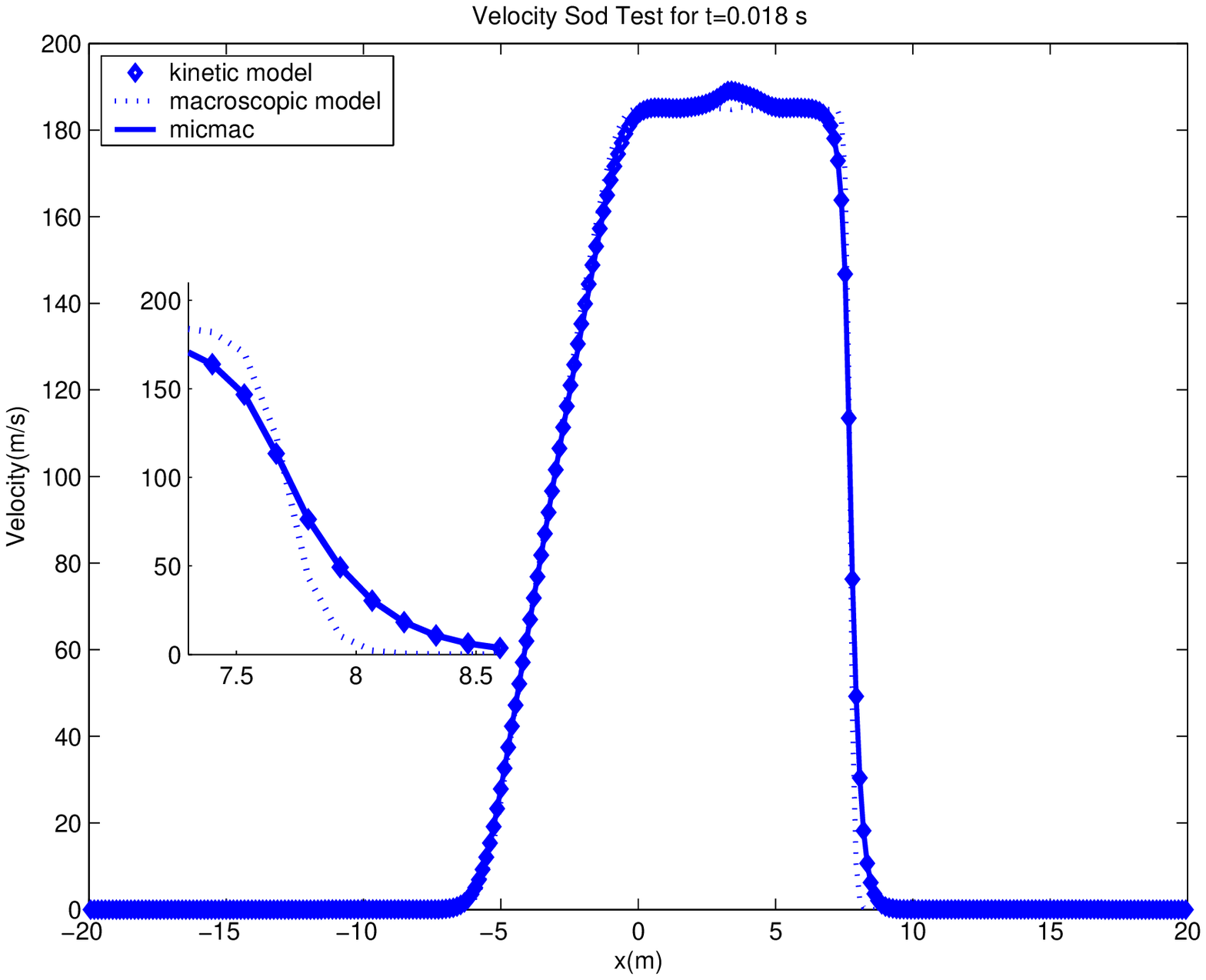}
\includegraphics[scale=0.34]{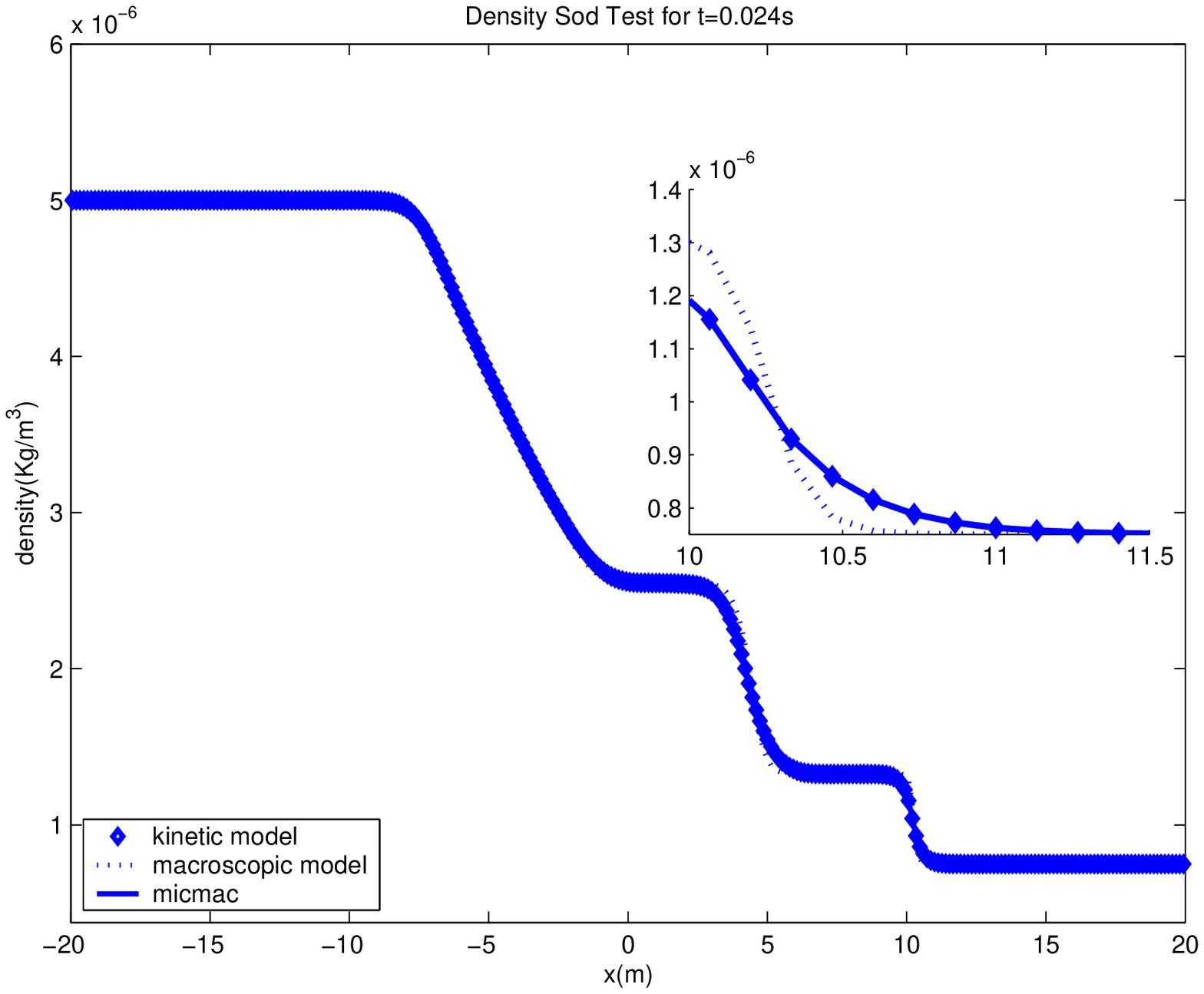}
\includegraphics[scale=0.34]{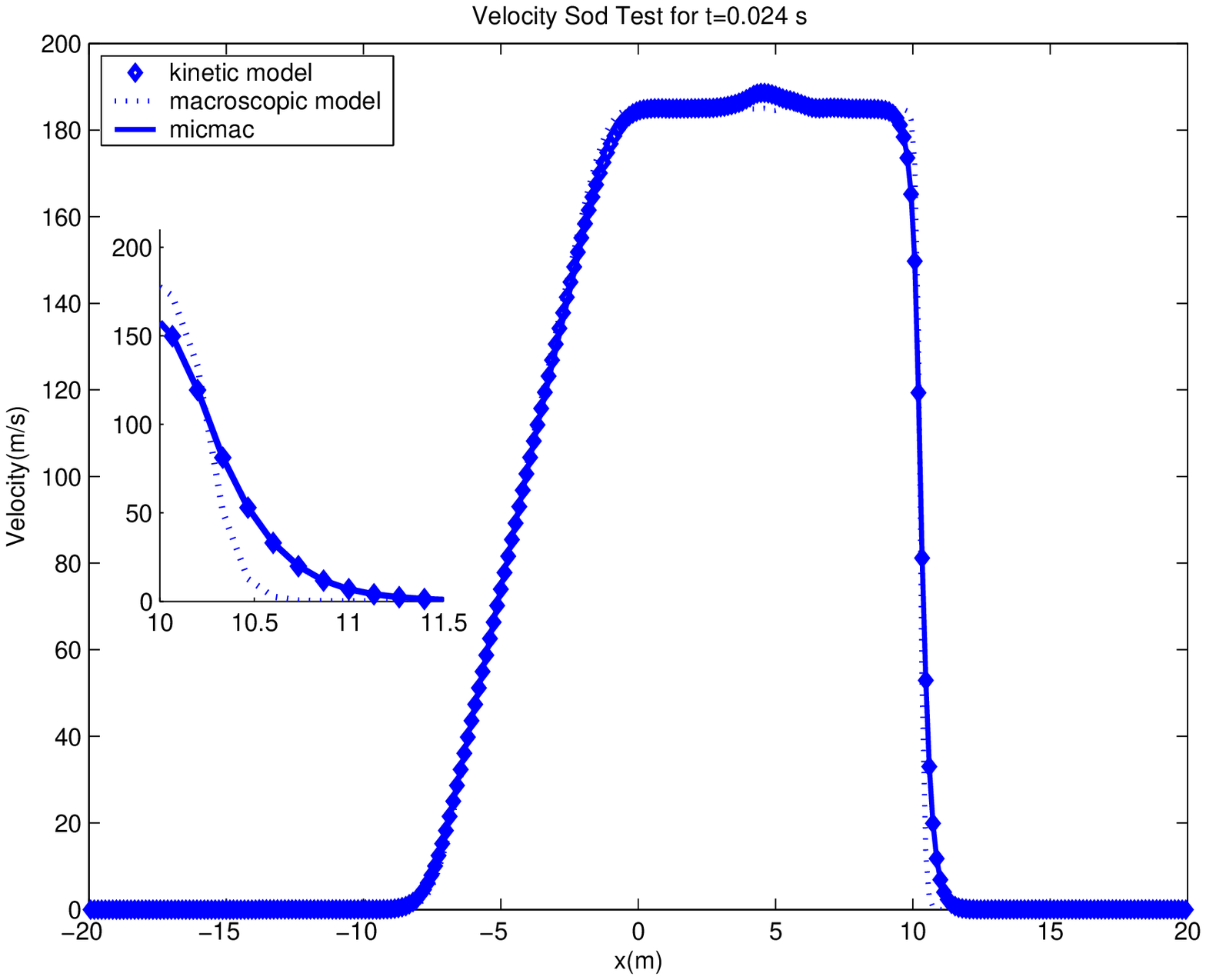}
\caption{Sod Test 2: Solution at $t=0.6\times 10^{-2}$ top,
$t=1.2\times 10^{-2}$ middle top, $t=1.8\times 10^{-2}$ middle
bottom, $t=2.4\times 10^{-2}$ bottom, density left, velocity  right.
The small panels are a magnification of the solution close to non
equilibrium regions.\label{sod2.1}}
\end{center}
\end{figure}

\begin{figure}
\begin{center}
\includegraphics[scale=0.34]{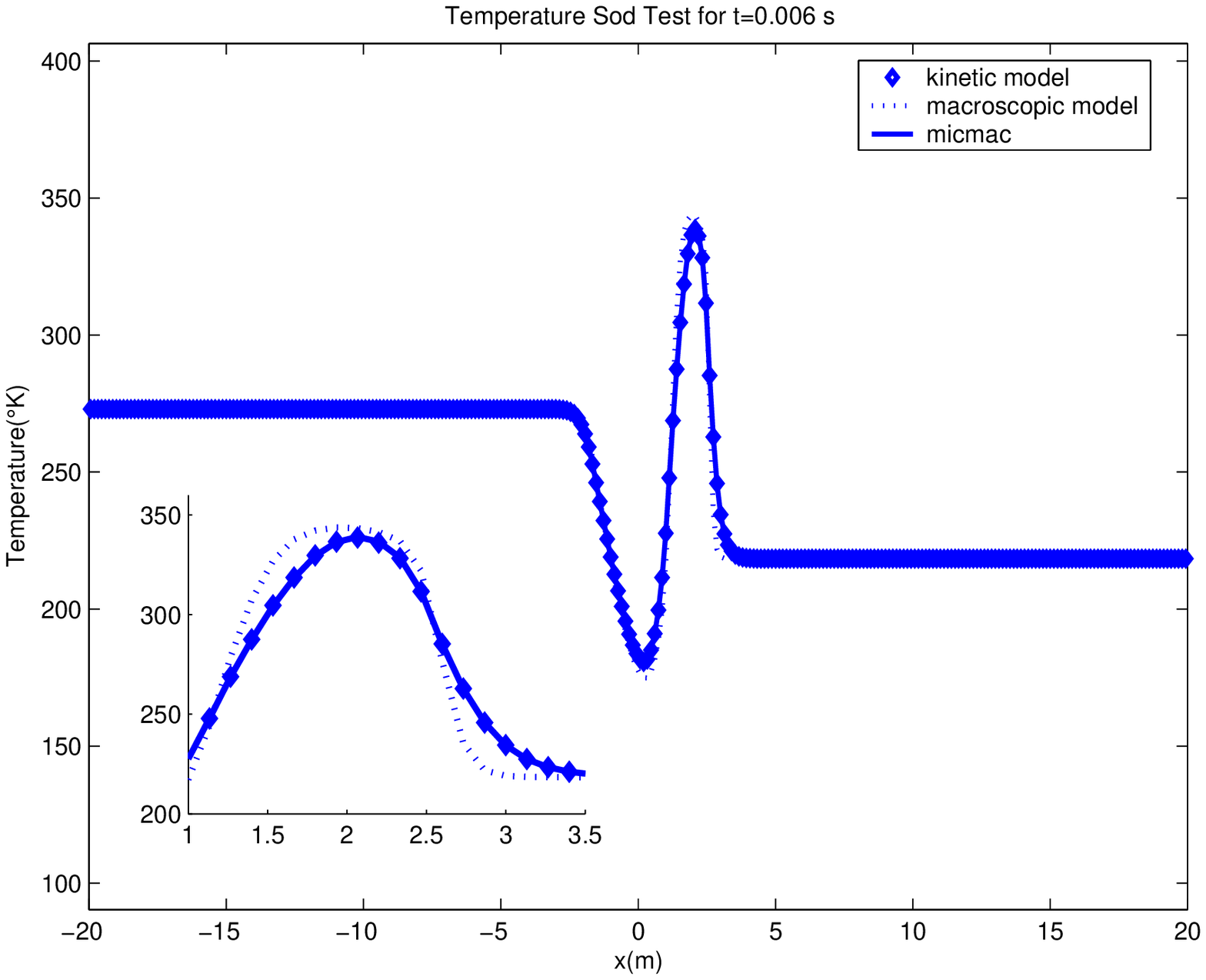}
\includegraphics[scale=0.34]{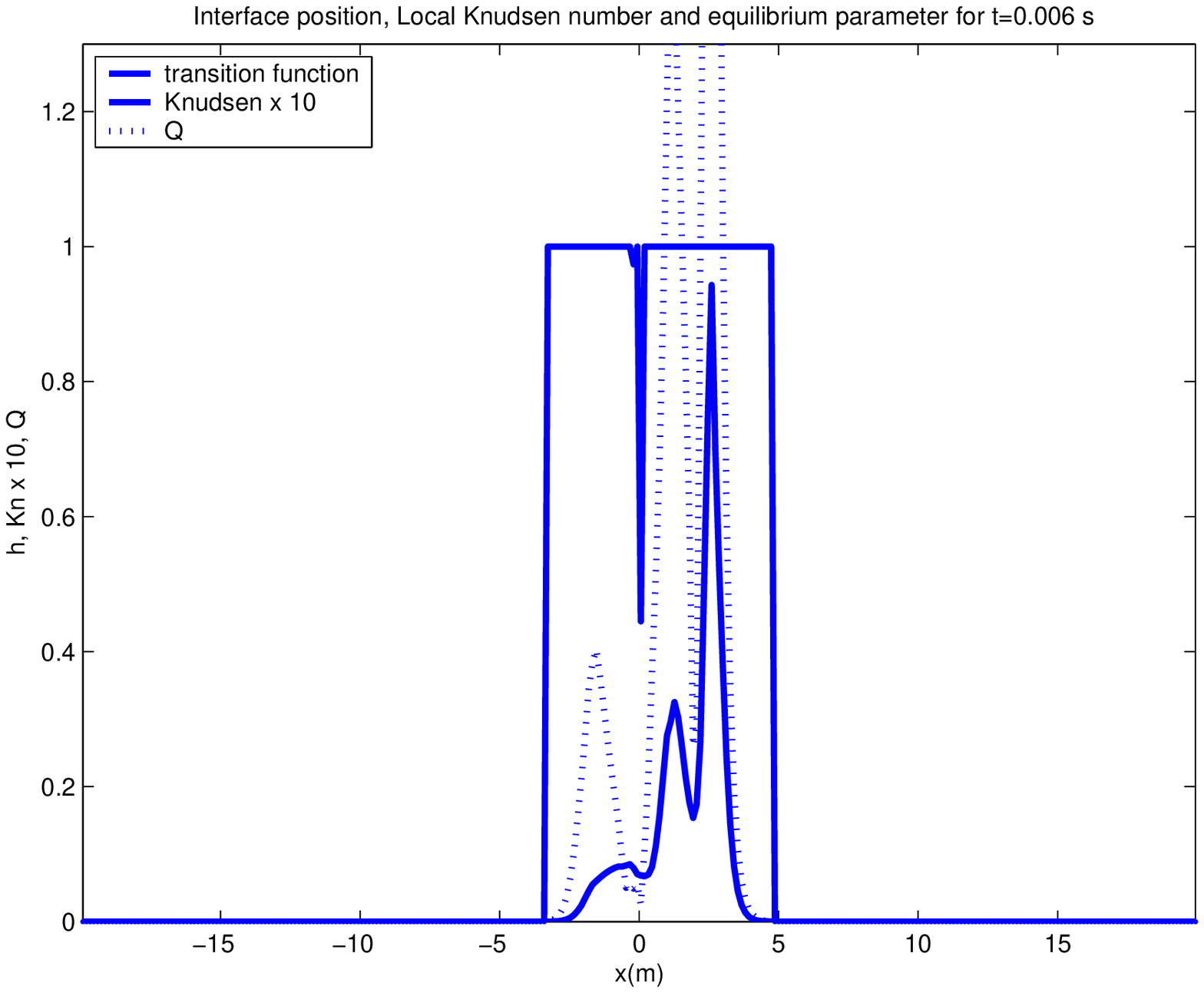}
\includegraphics[scale=0.34]{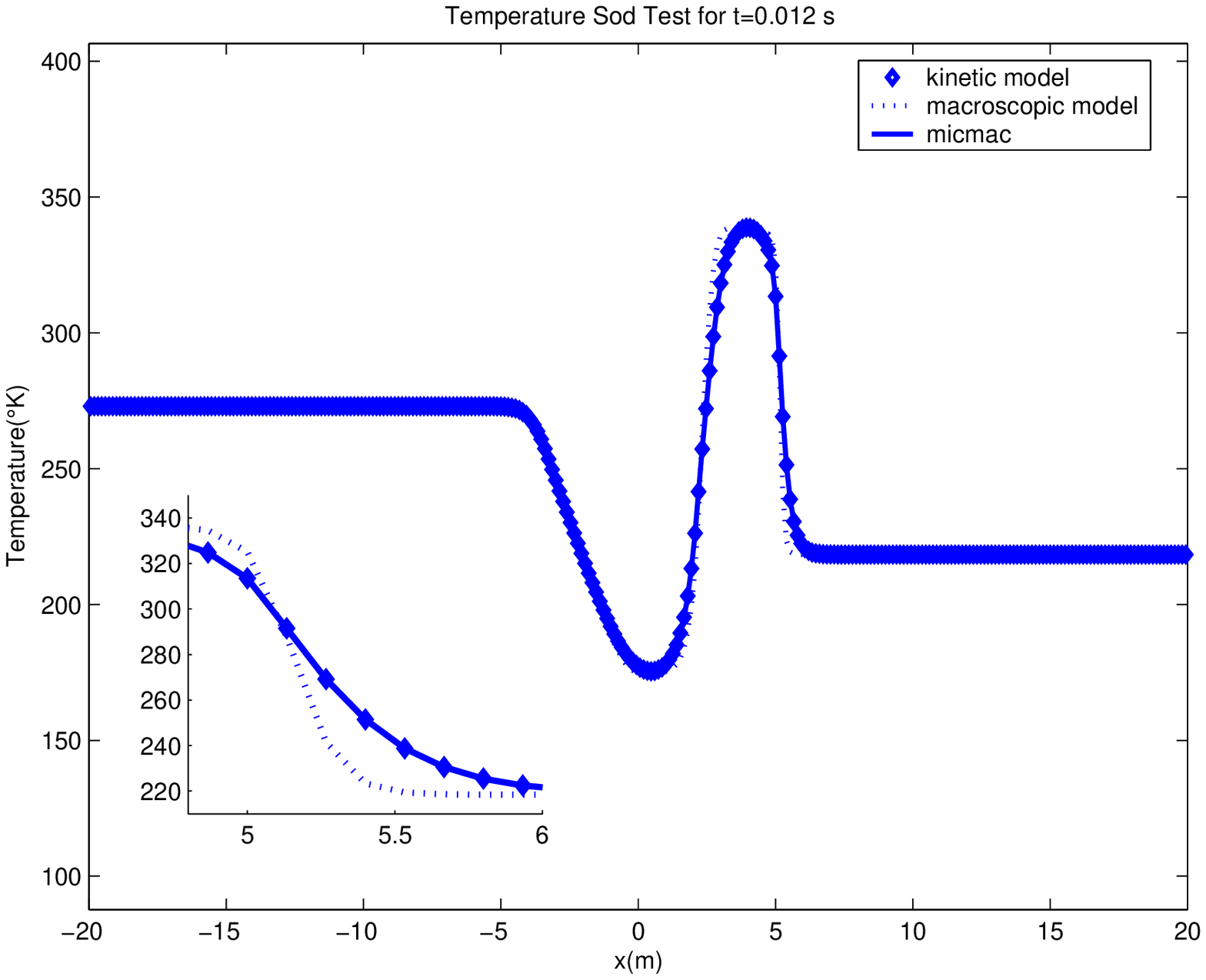}
\includegraphics[scale=0.34]{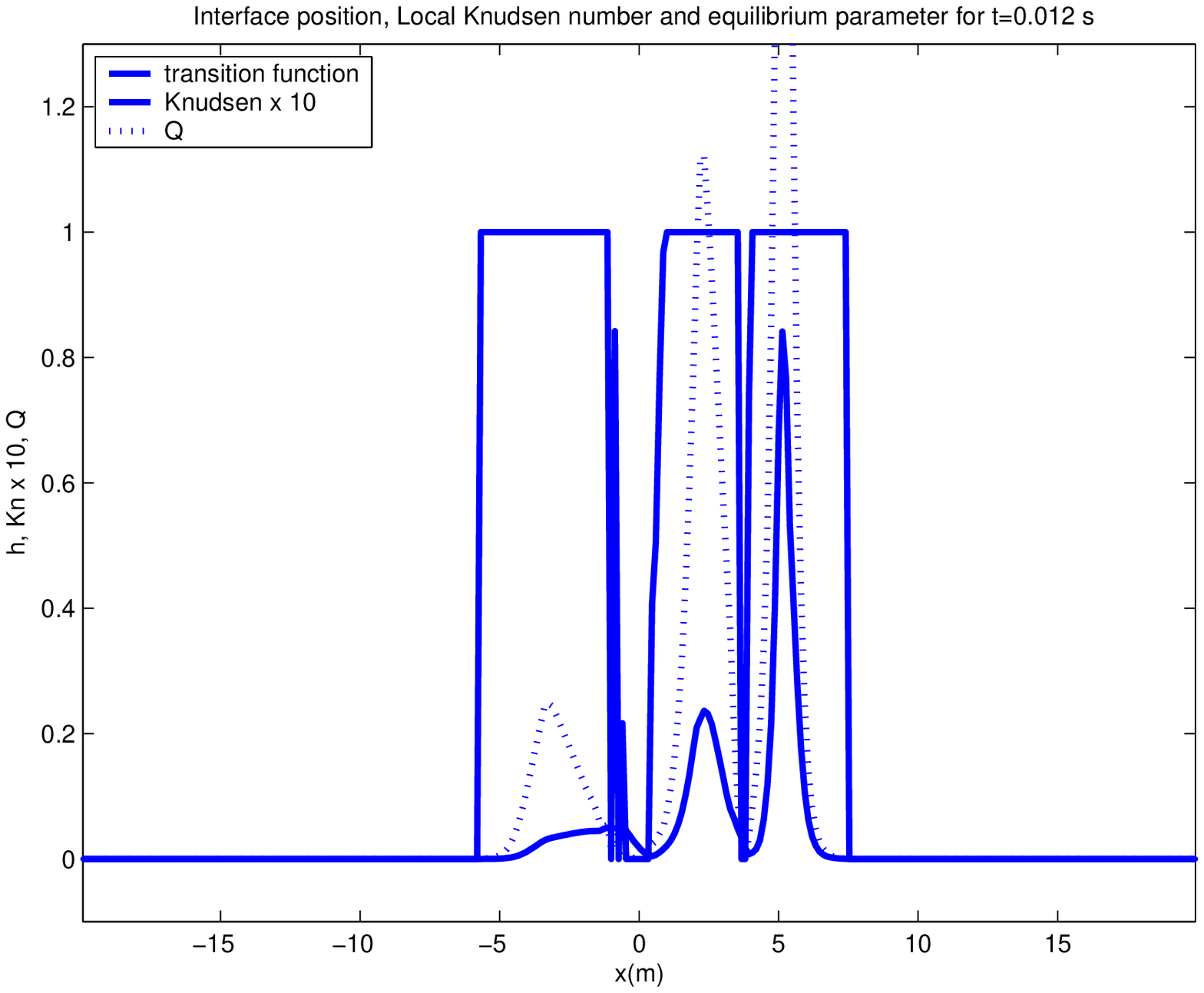}
\includegraphics[scale=0.34]{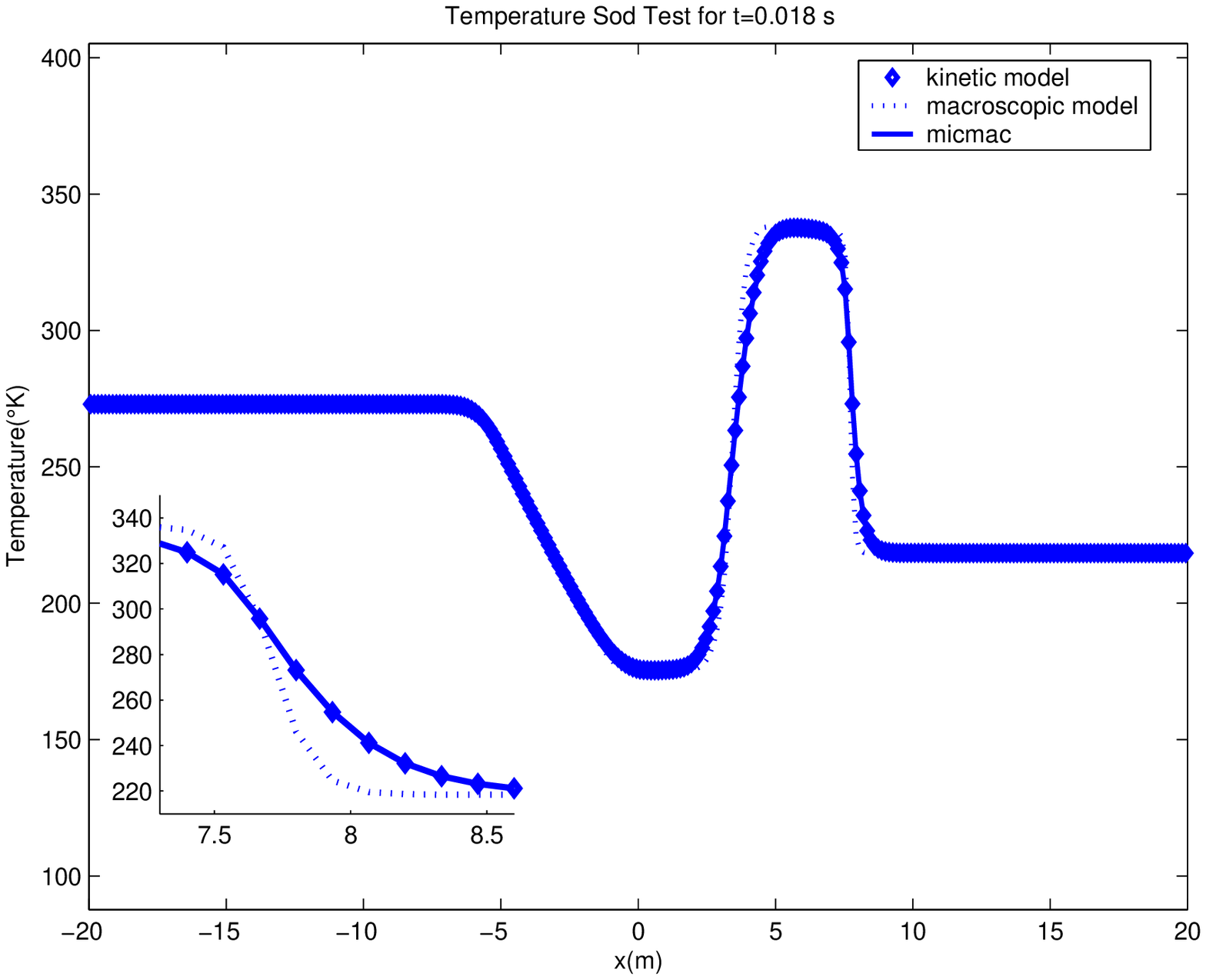}
\includegraphics[scale=0.34]{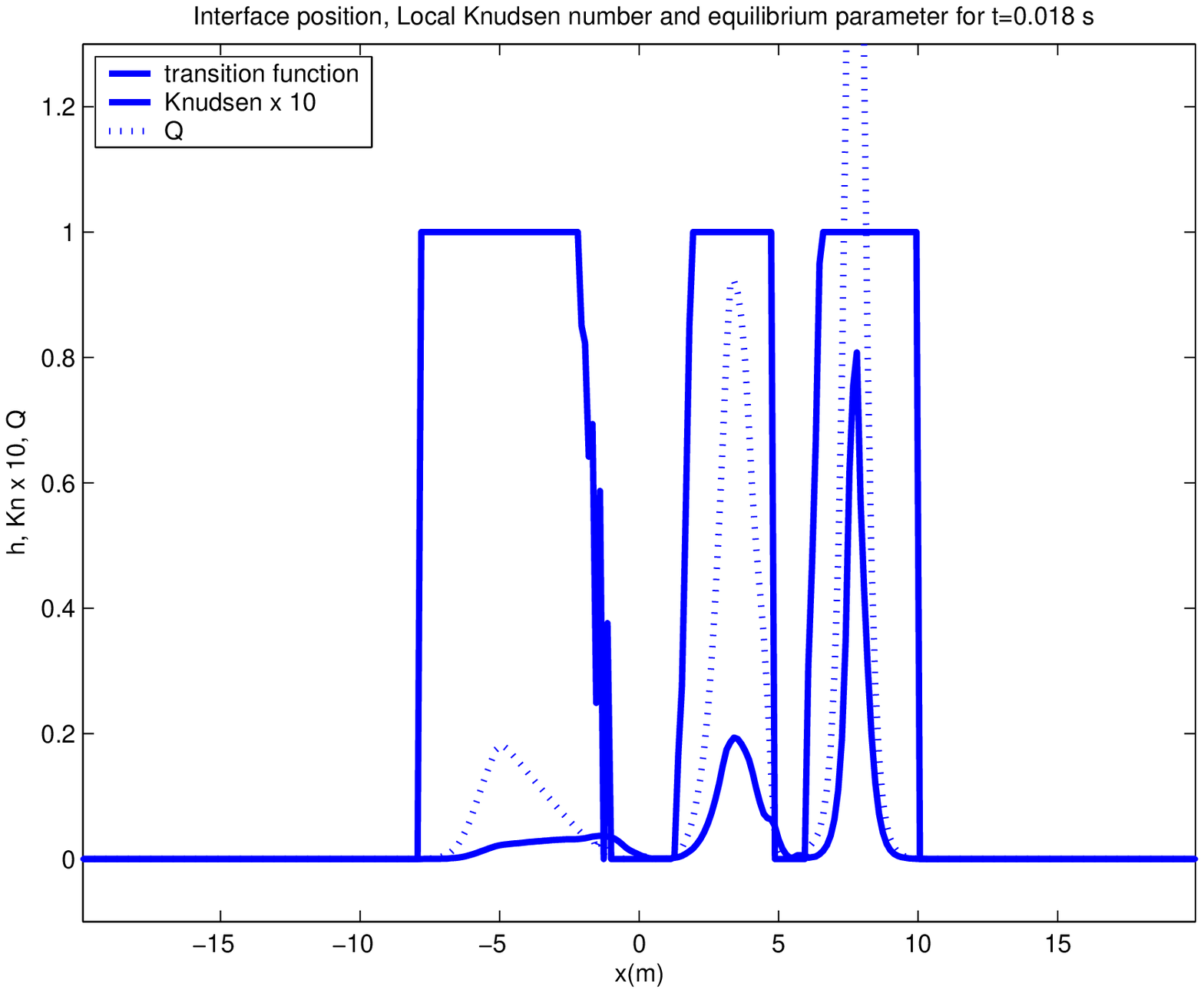}
\includegraphics[scale=0.34]{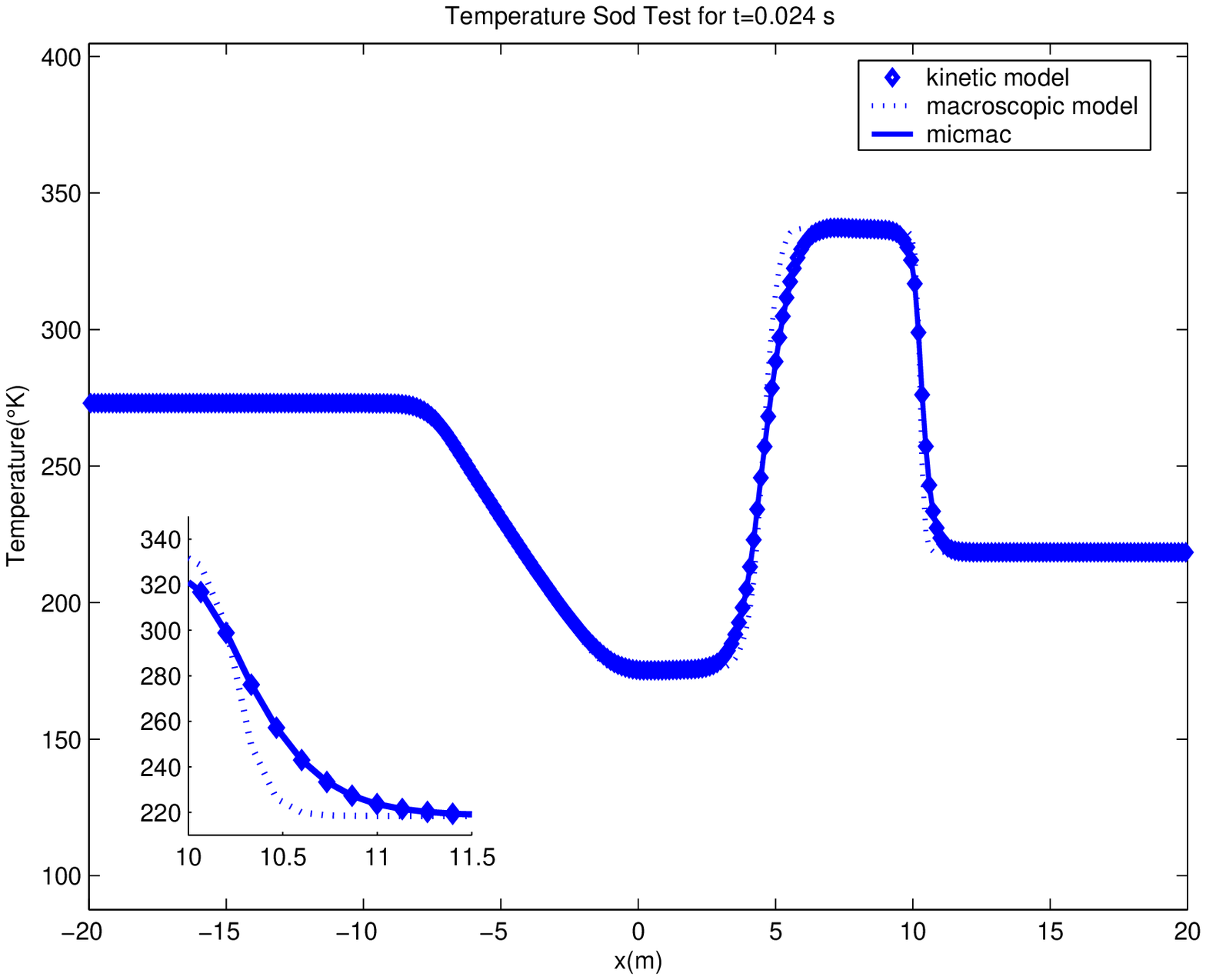}
\includegraphics[scale=0.34]{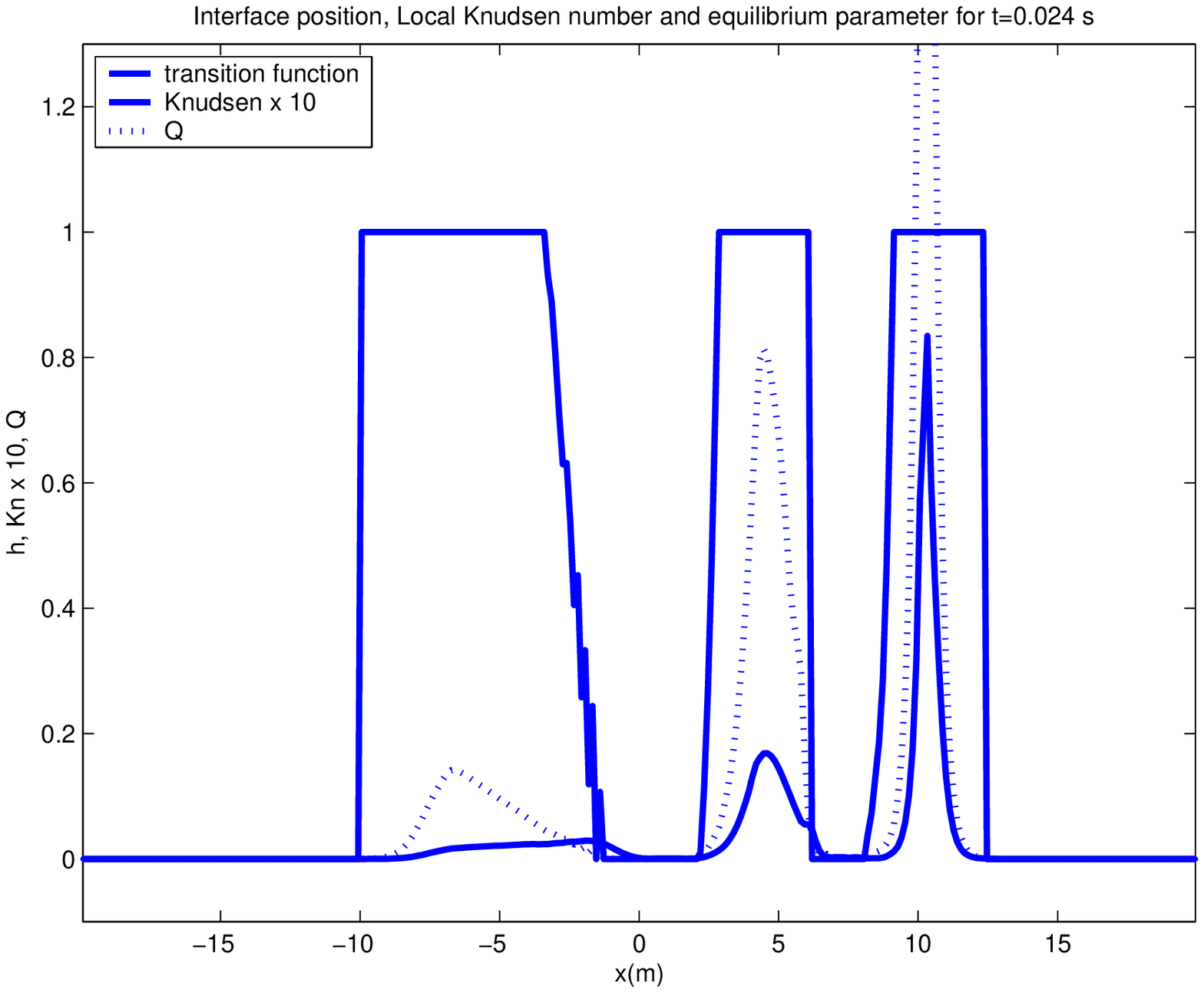}
\caption{Sod Test 2: Solution at $t=0.6\times 10^{-2}$ top,
$t=1.2\times 10^{-2}$ middle top, $t=1.8\times 10^{-2}$ middle
bottom, $t=2.4\times 10^{-2}$ bottom, temperature left, transition
function, Knudsen number and heat flux right. The small panels are a
magnification of the solution close to non equilibrium
regions.\label{sod2.2}}
\end{center}
\end{figure}


\begin{figure}
\begin{center}
\includegraphics[scale=0.34]{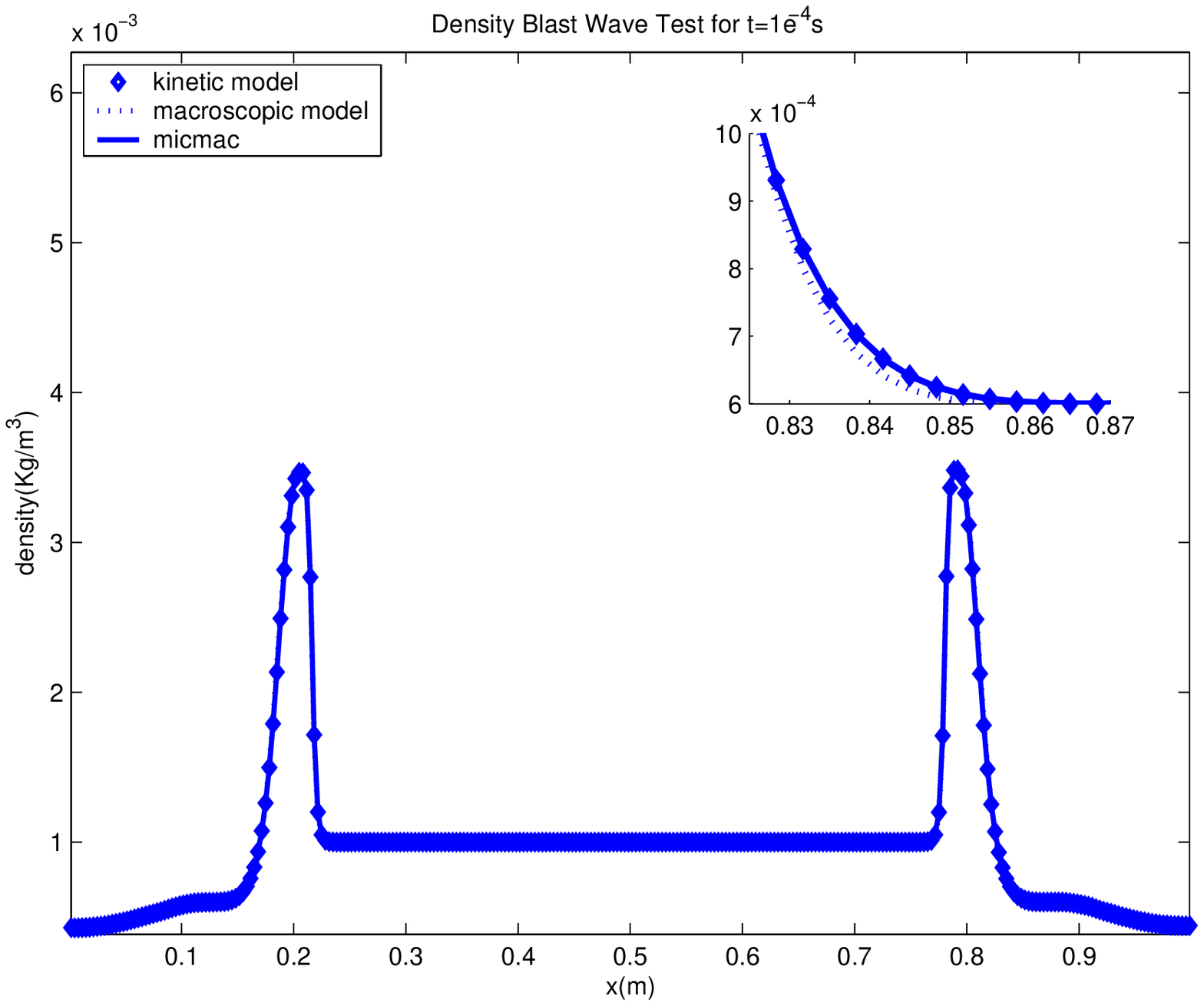}
\includegraphics[scale=0.34]{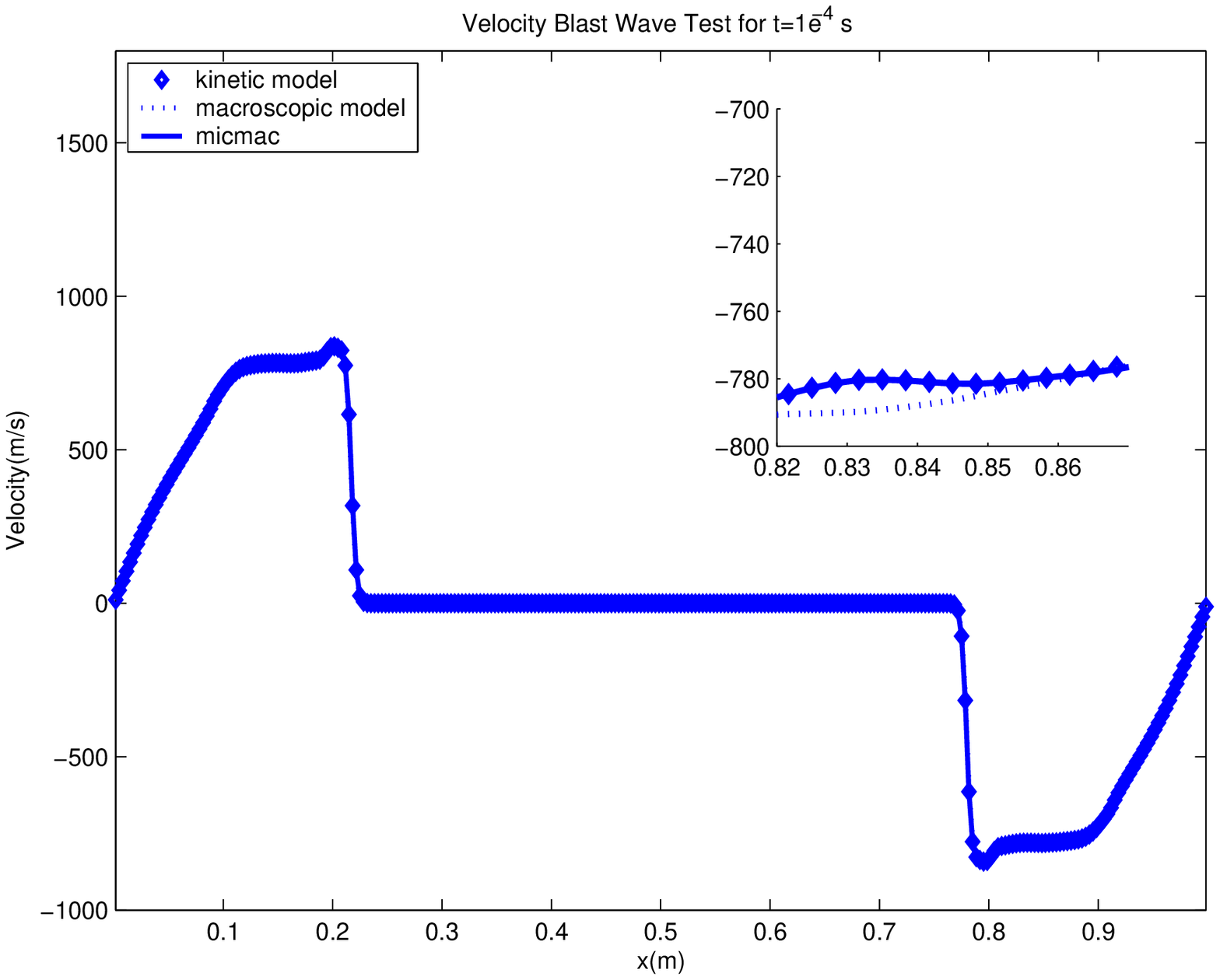}
\includegraphics[scale=0.34]{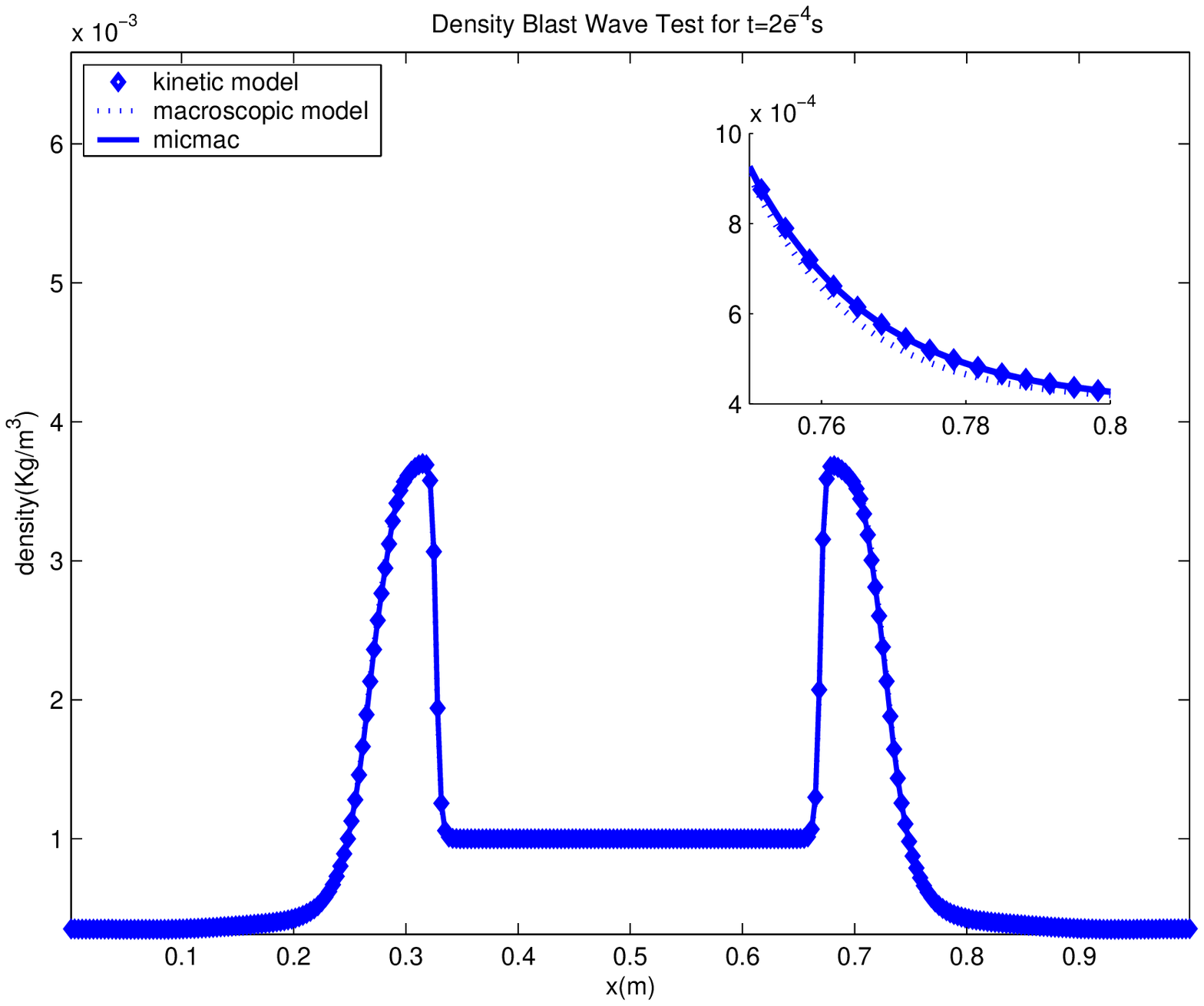}
\includegraphics[scale=0.34]{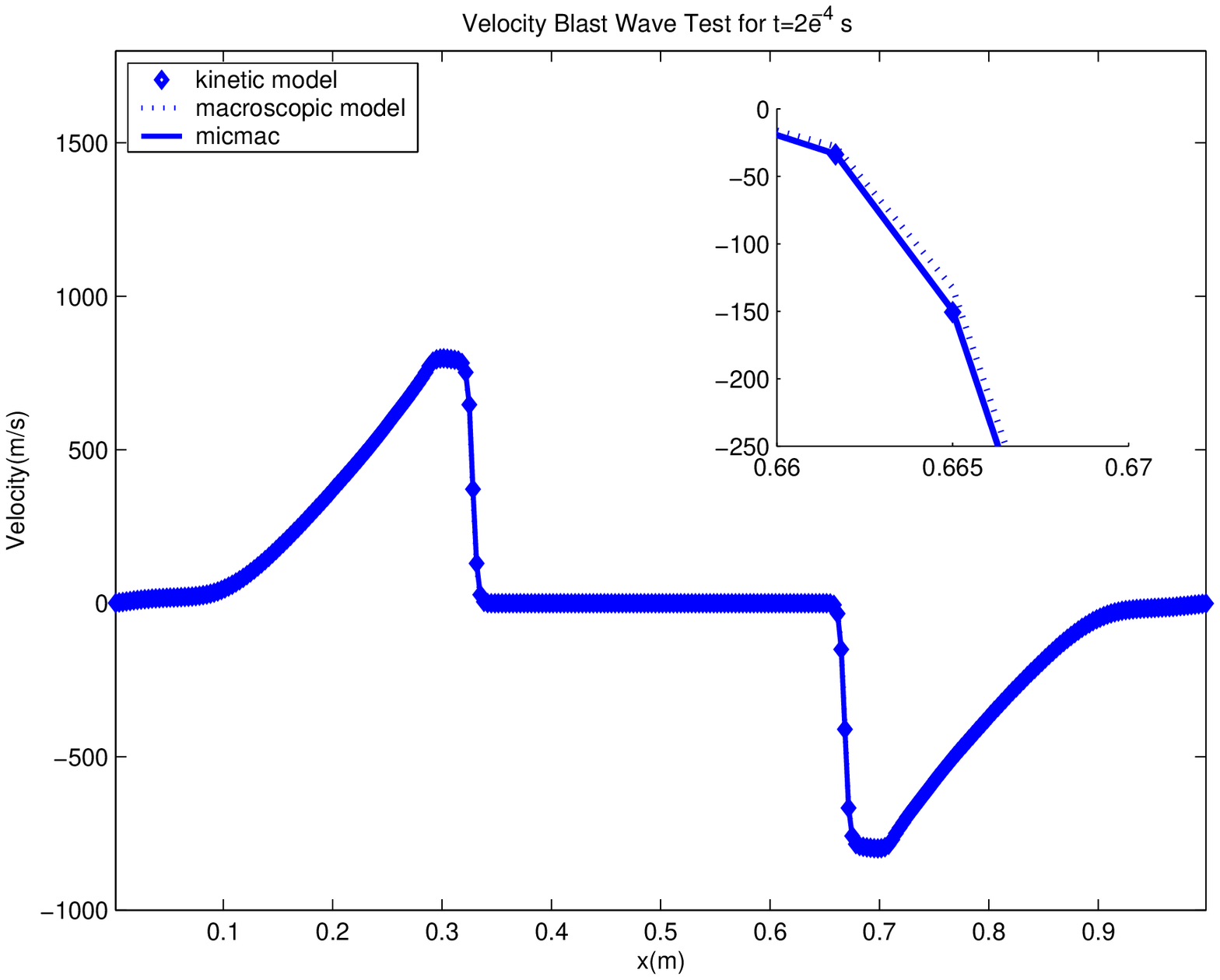}
\includegraphics[scale=0.34]{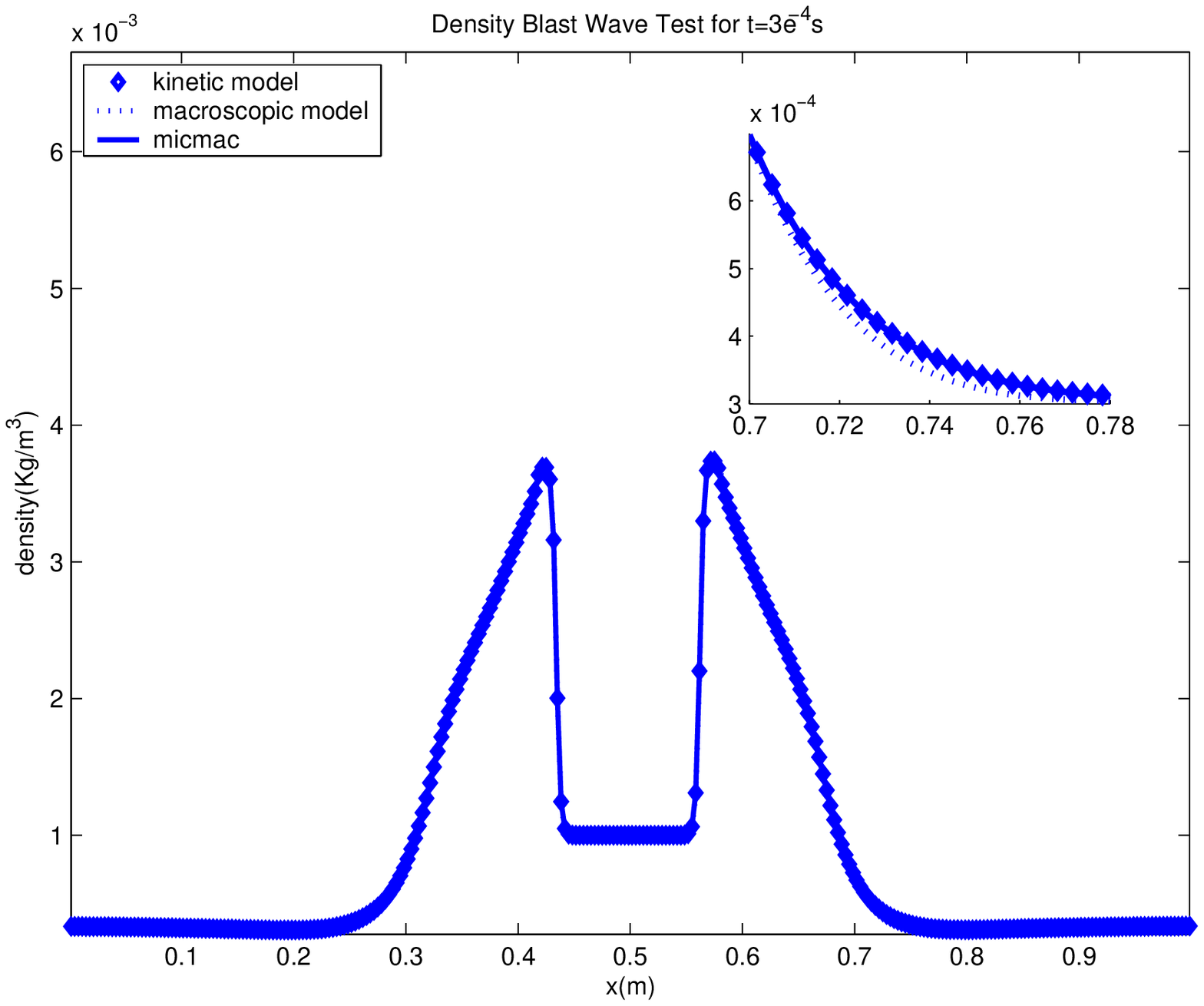}
\includegraphics[scale=0.34]{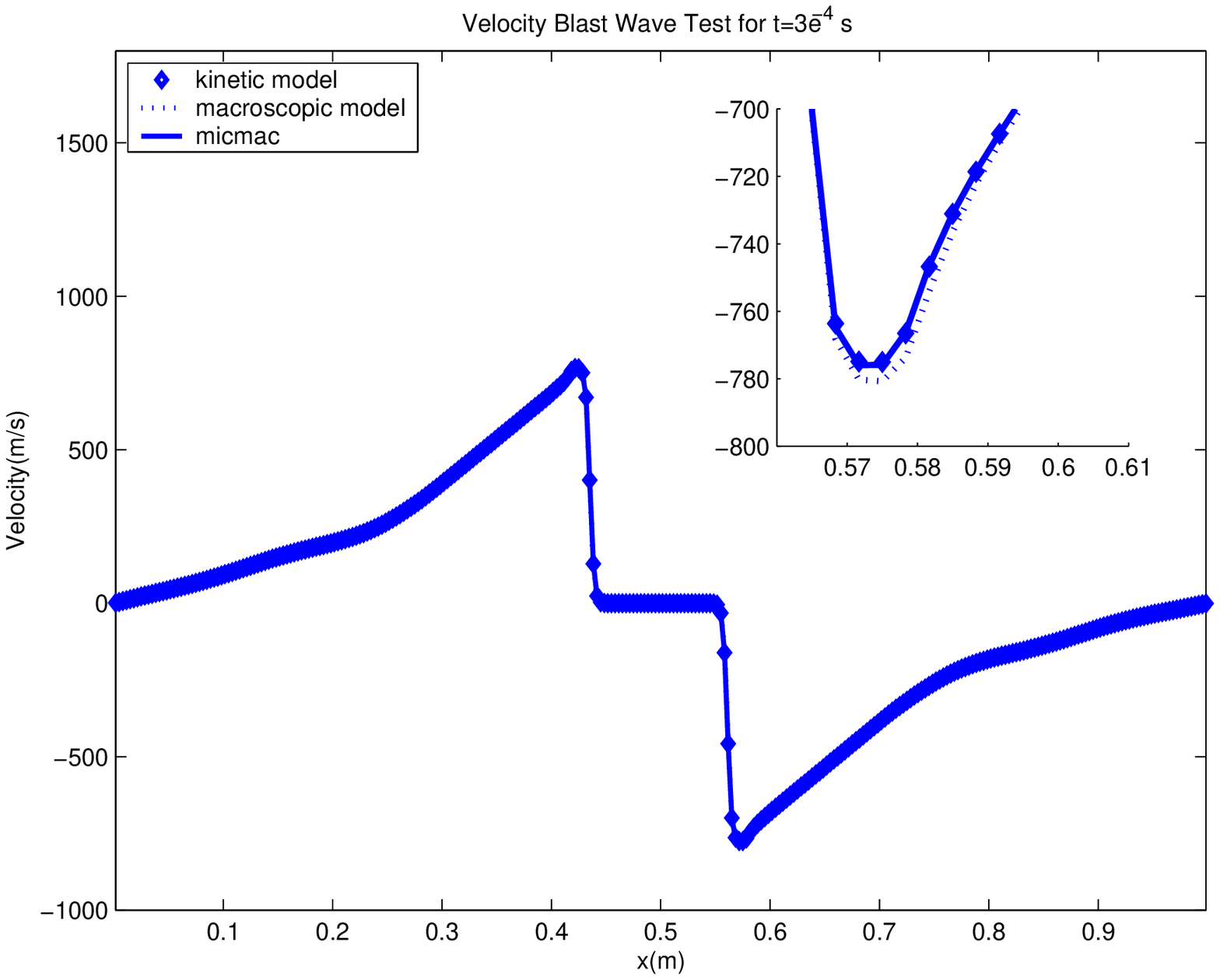}
\includegraphics[scale=0.34]{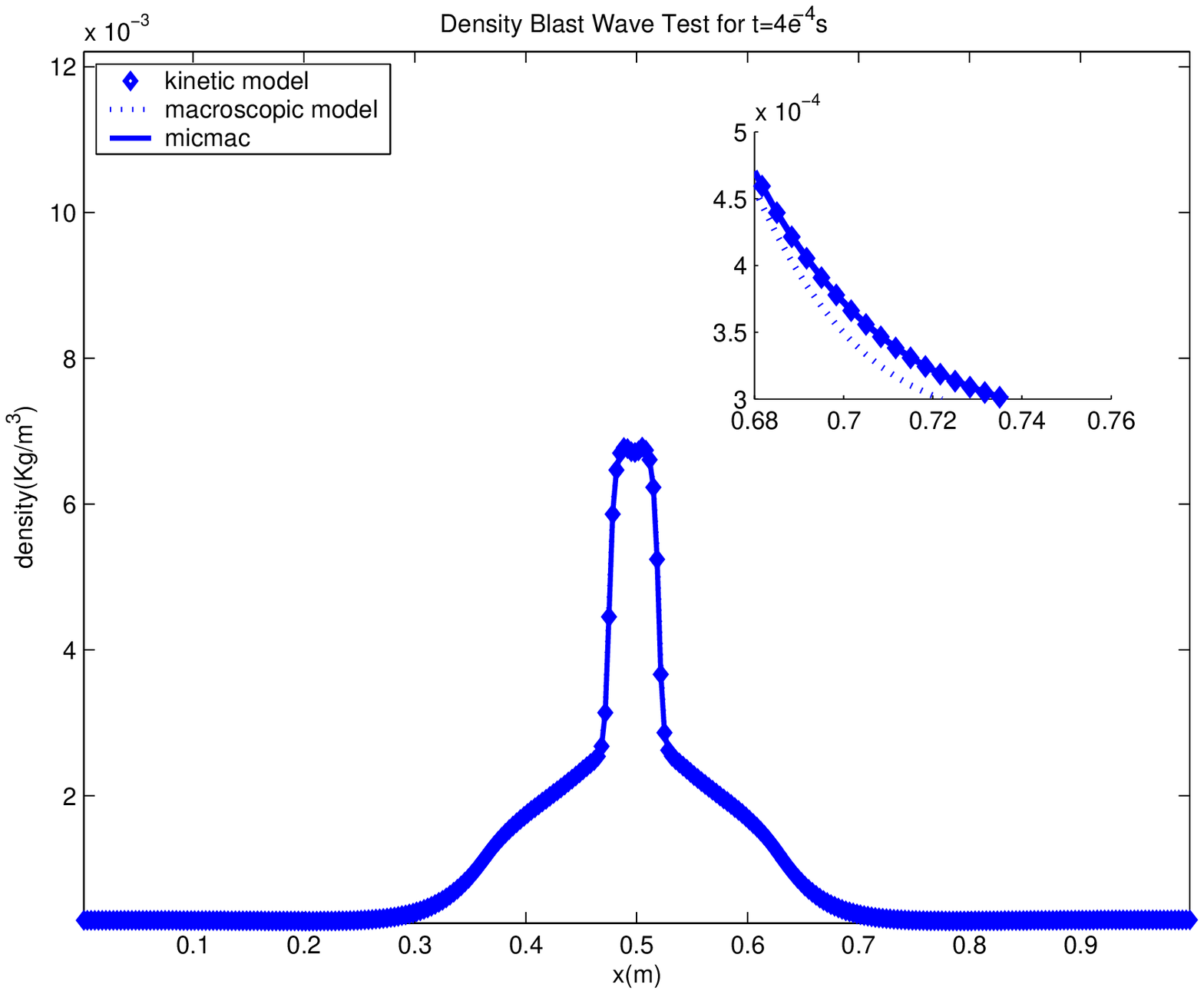}
\includegraphics[scale=0.34]{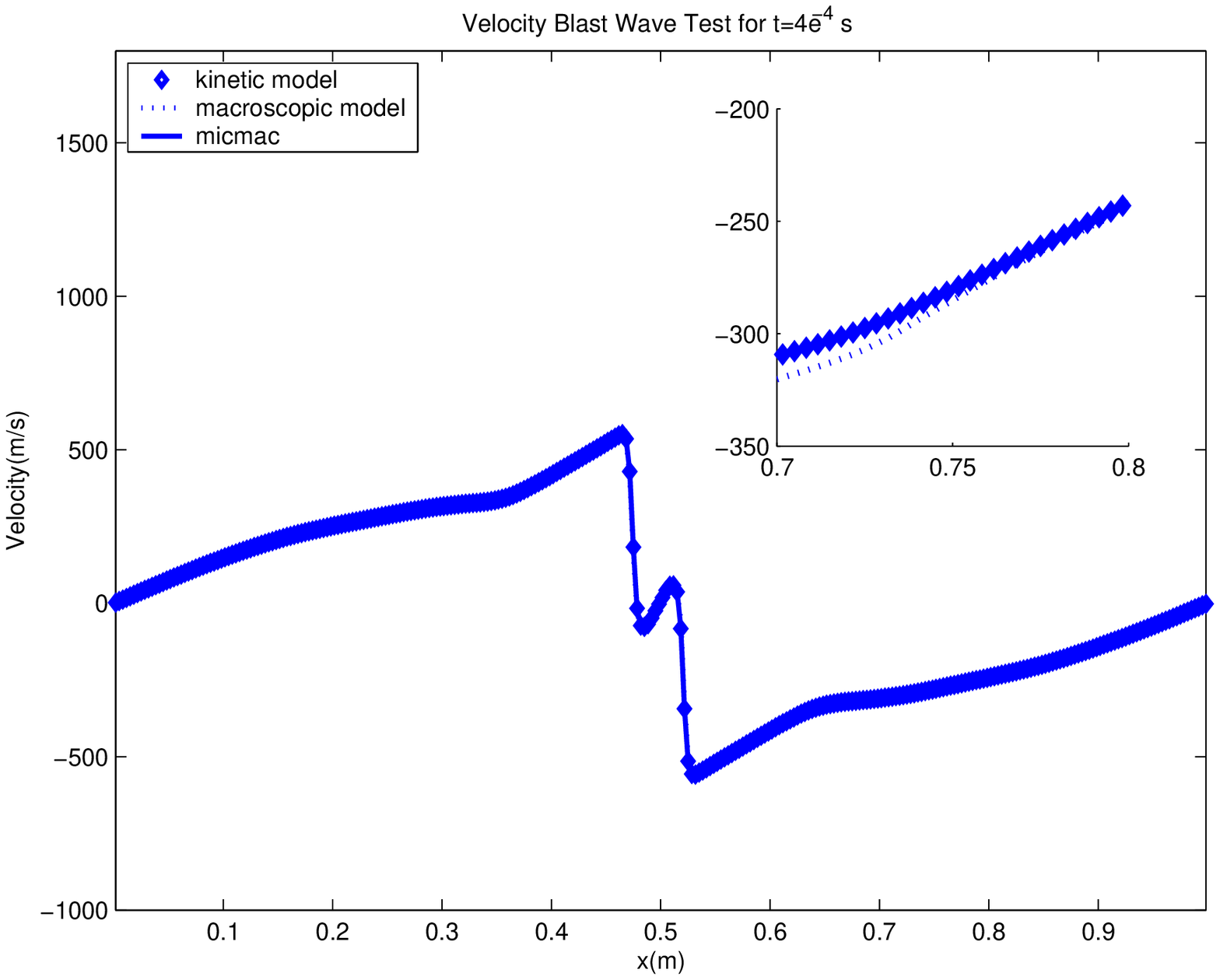}
\caption{Blast Wave Test 1: Solution at $t=1\times 10^{-4}$ top,
$t=2\times 10^{-4}$ middle top, $t=3\times 10^{-4}$ middle bottom,
$t=4\times 10^{-4}$ bottom, density left, velocity  right. The small
panels are a magnification of the solution close to non equilibrium
regions. \label{blast1.1}}
\end{center}
\end{figure}

\begin{figure}
\begin{center}
\includegraphics[scale=0.34]{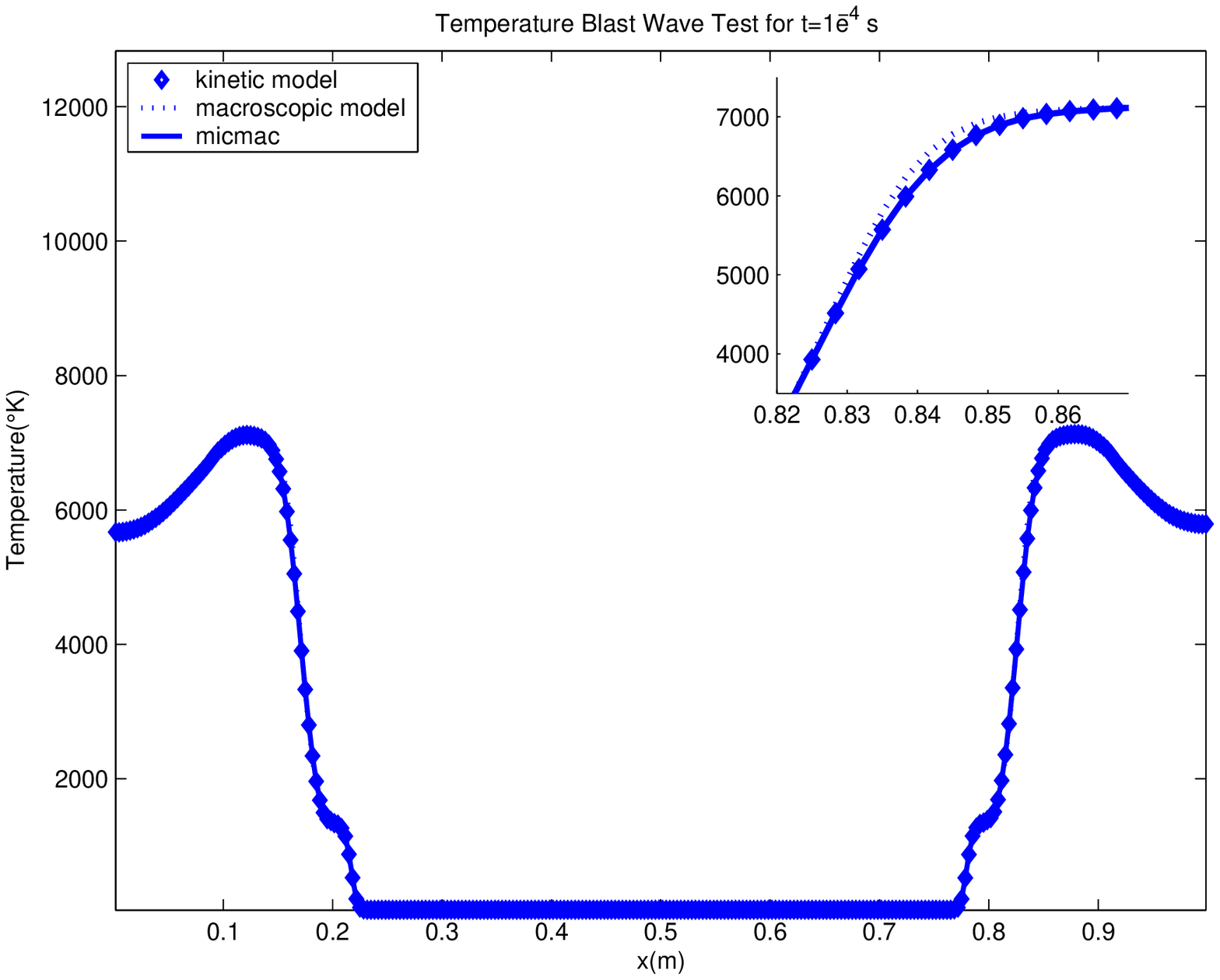}
\includegraphics[scale=0.34]{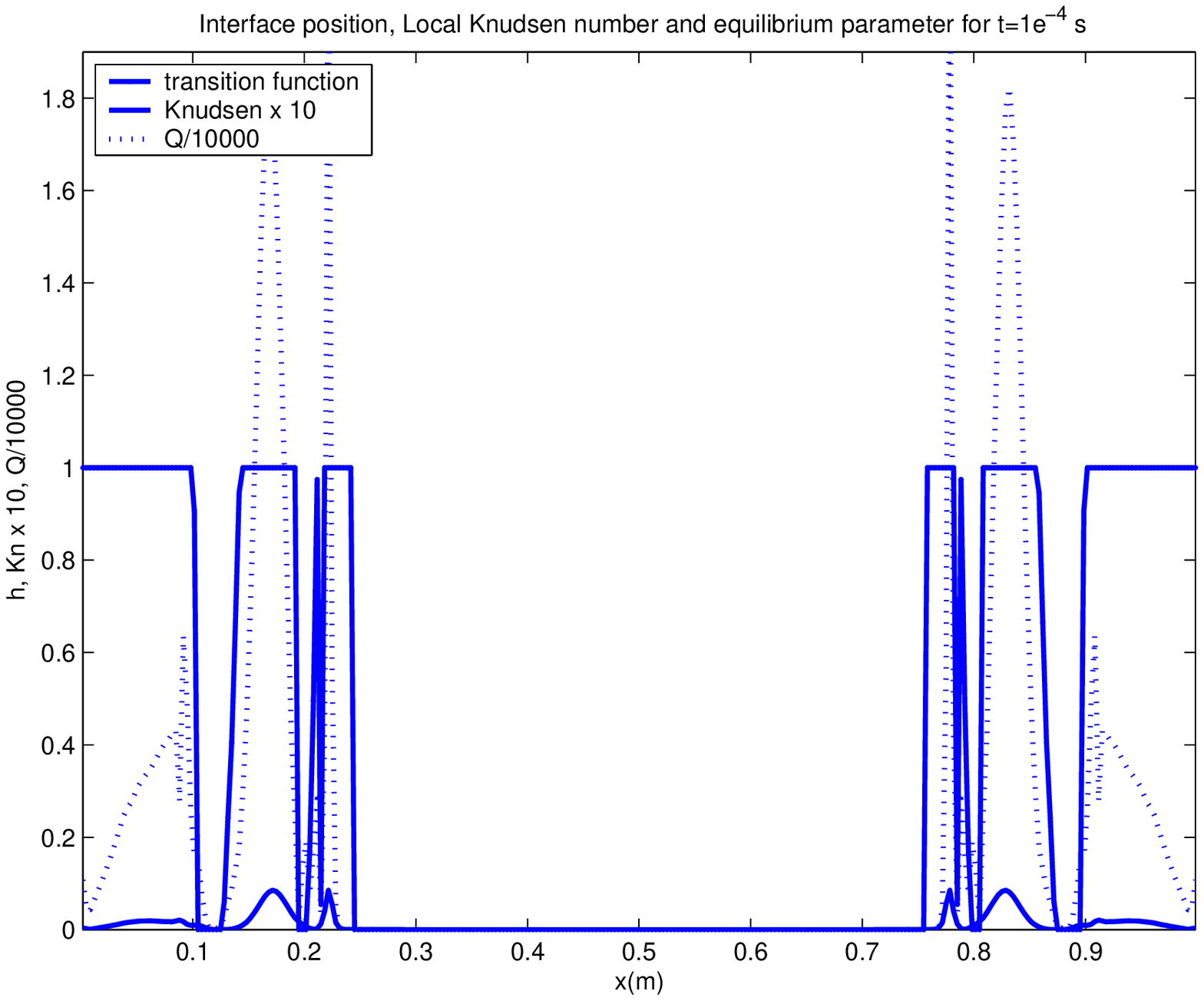}
\includegraphics[scale=0.34]{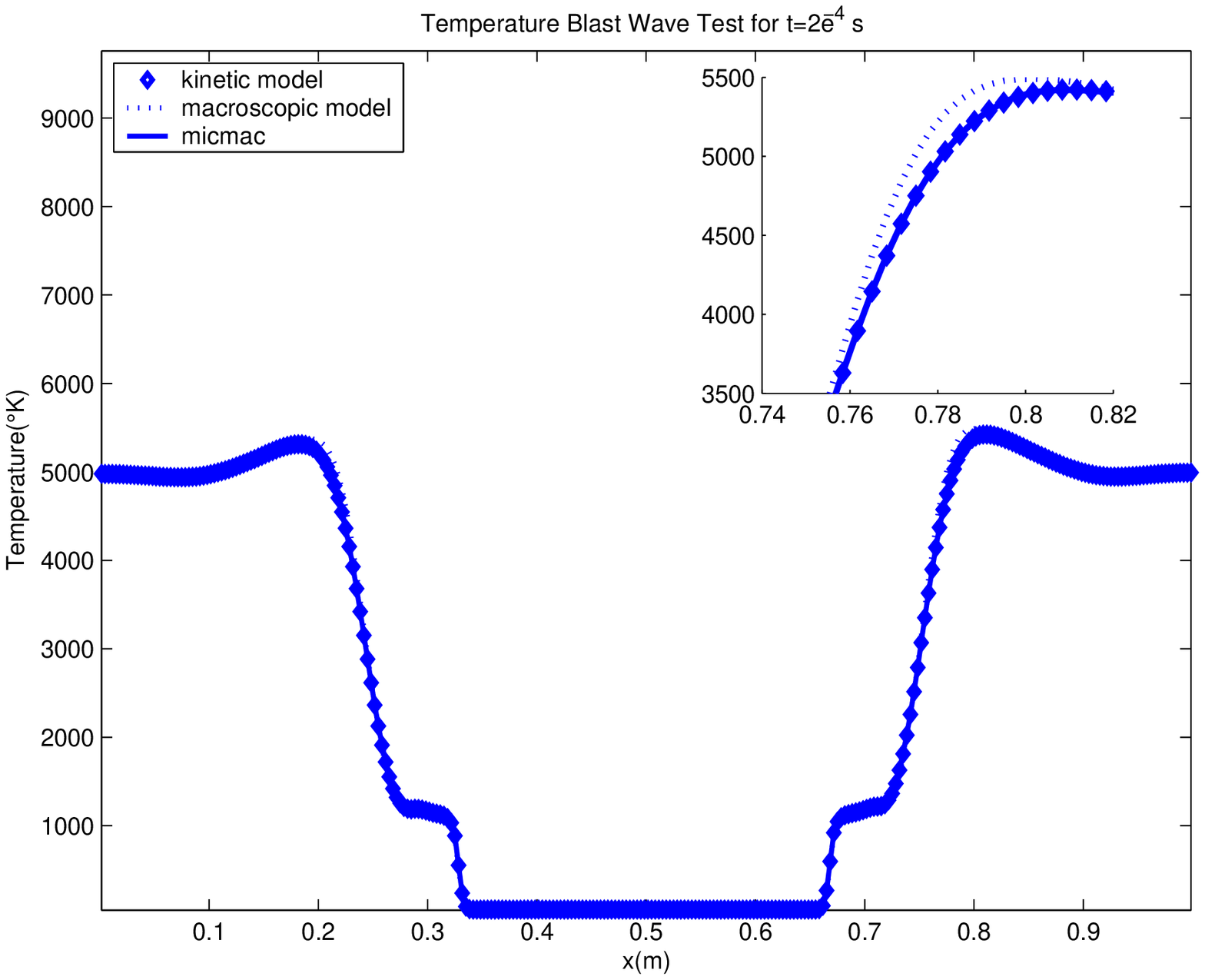}
\includegraphics[scale=0.34]{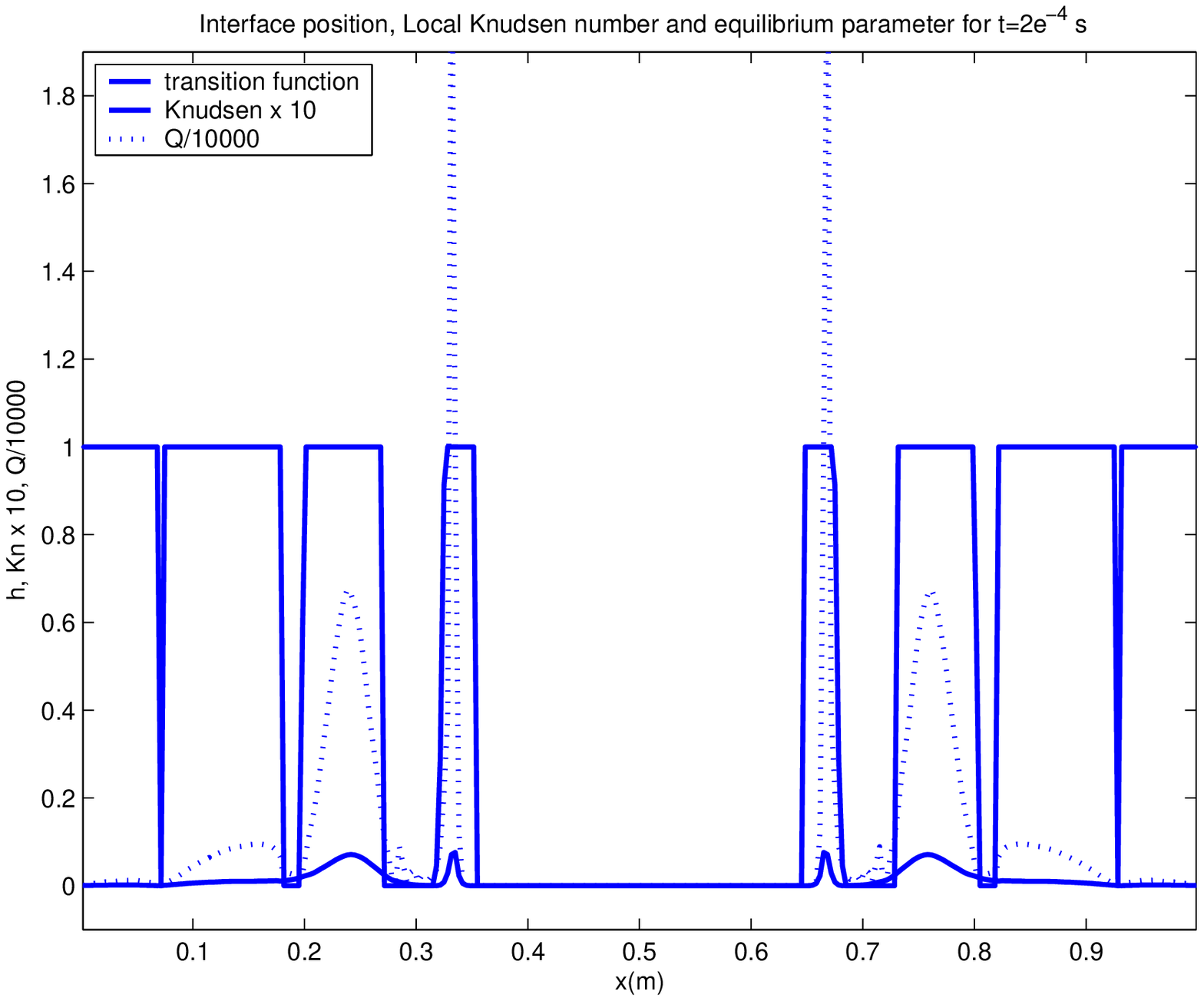}
\includegraphics[scale=0.34]{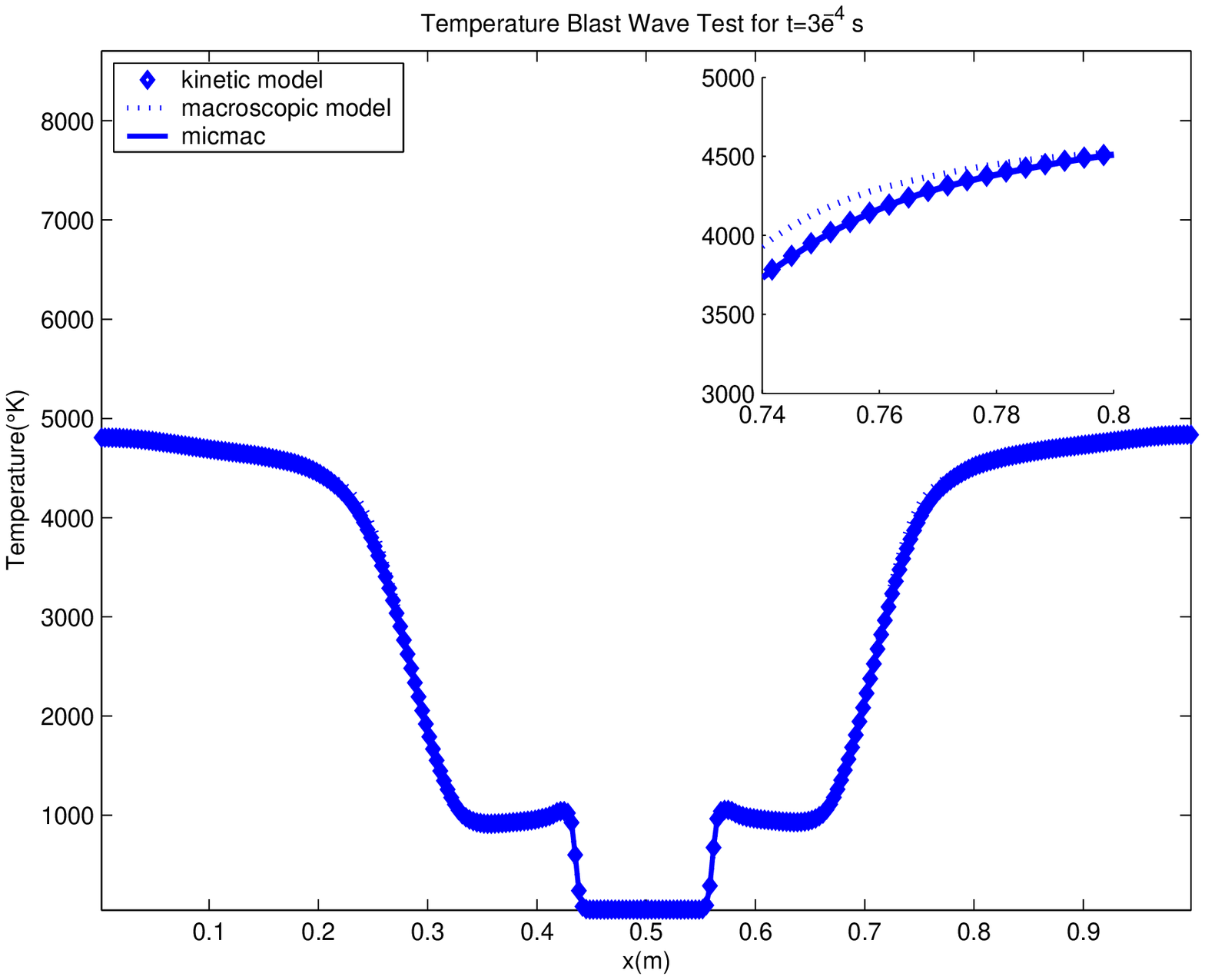}
\includegraphics[scale=0.34]{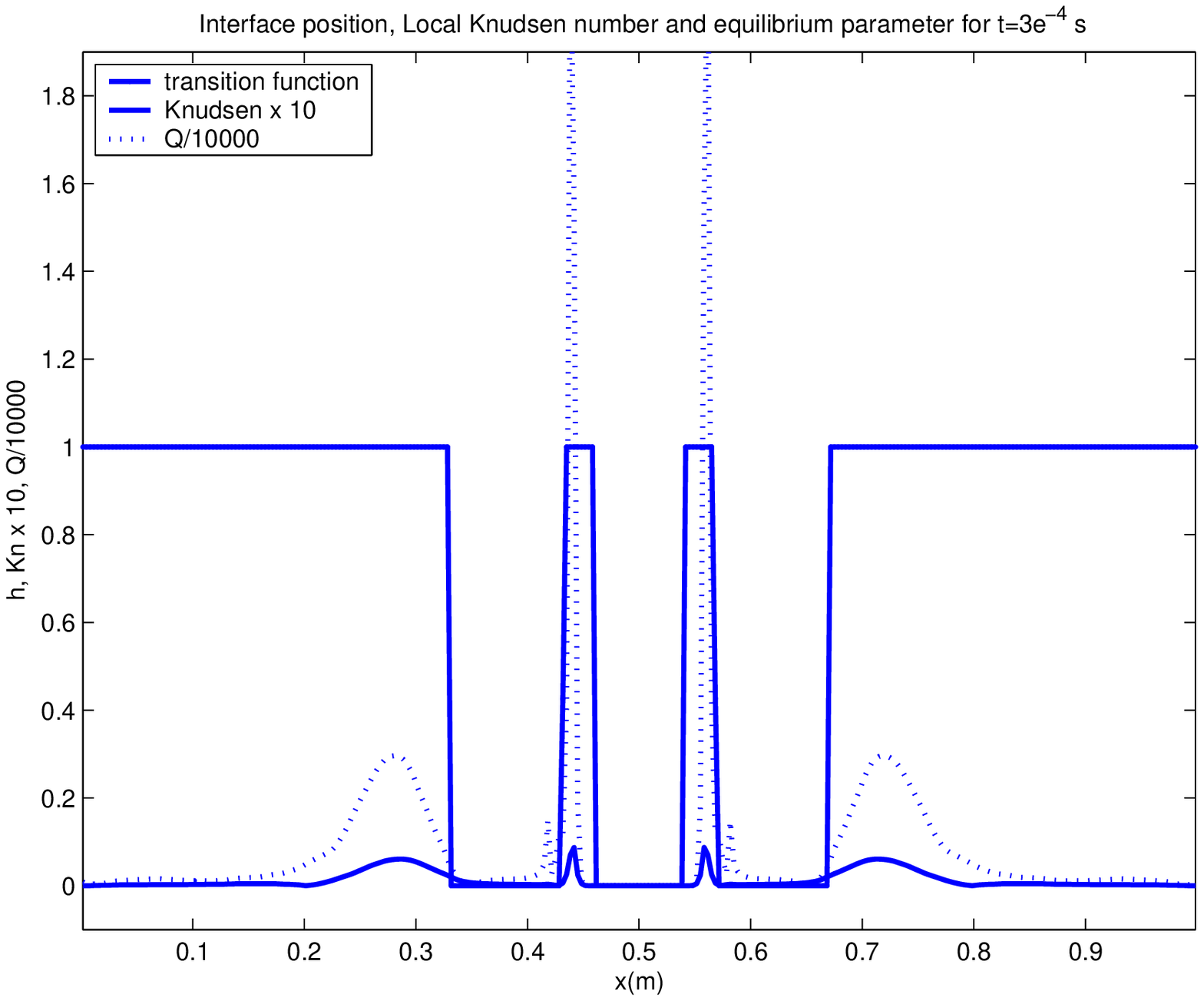}
\includegraphics[scale=0.34]{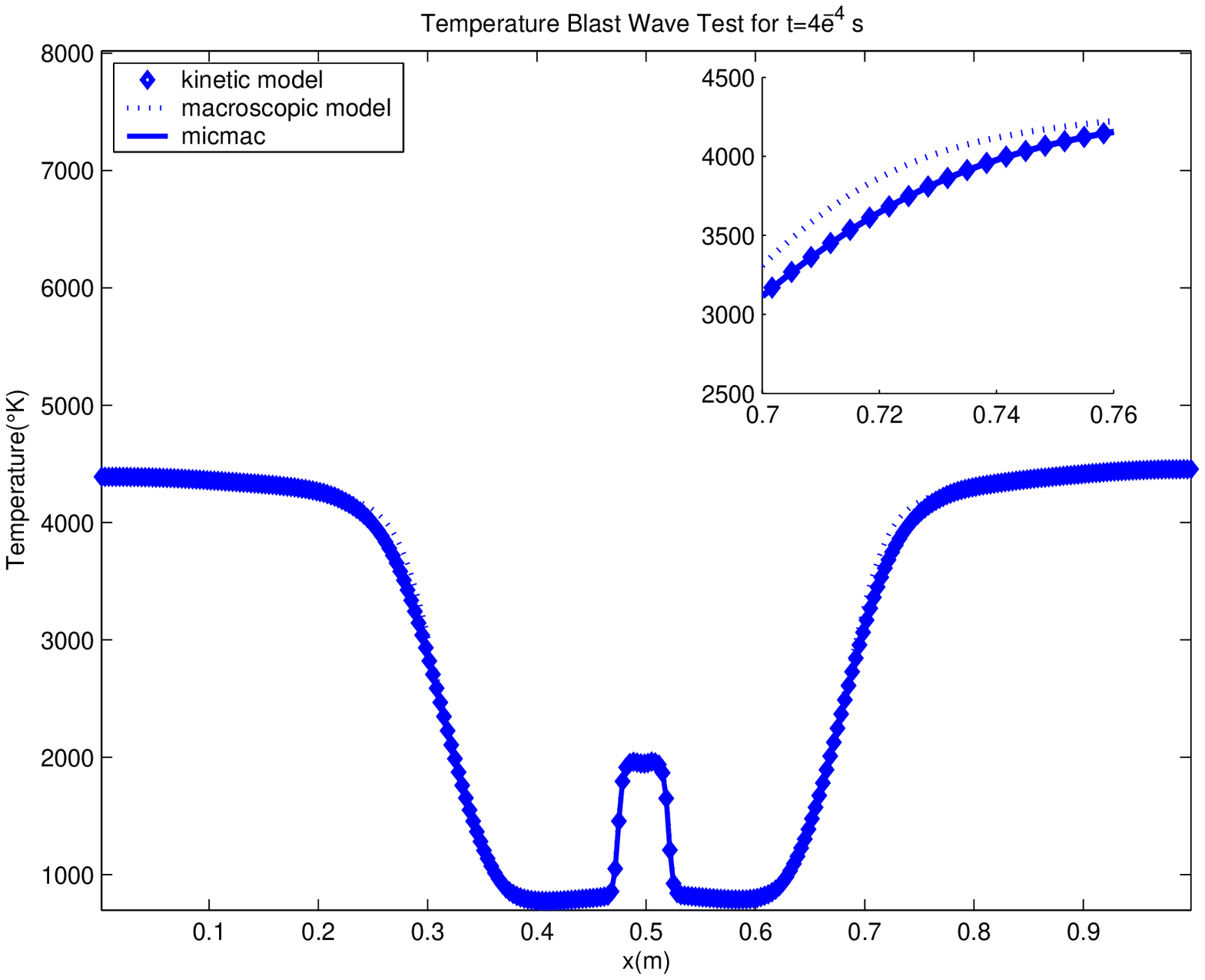}
\includegraphics[scale=0.34]{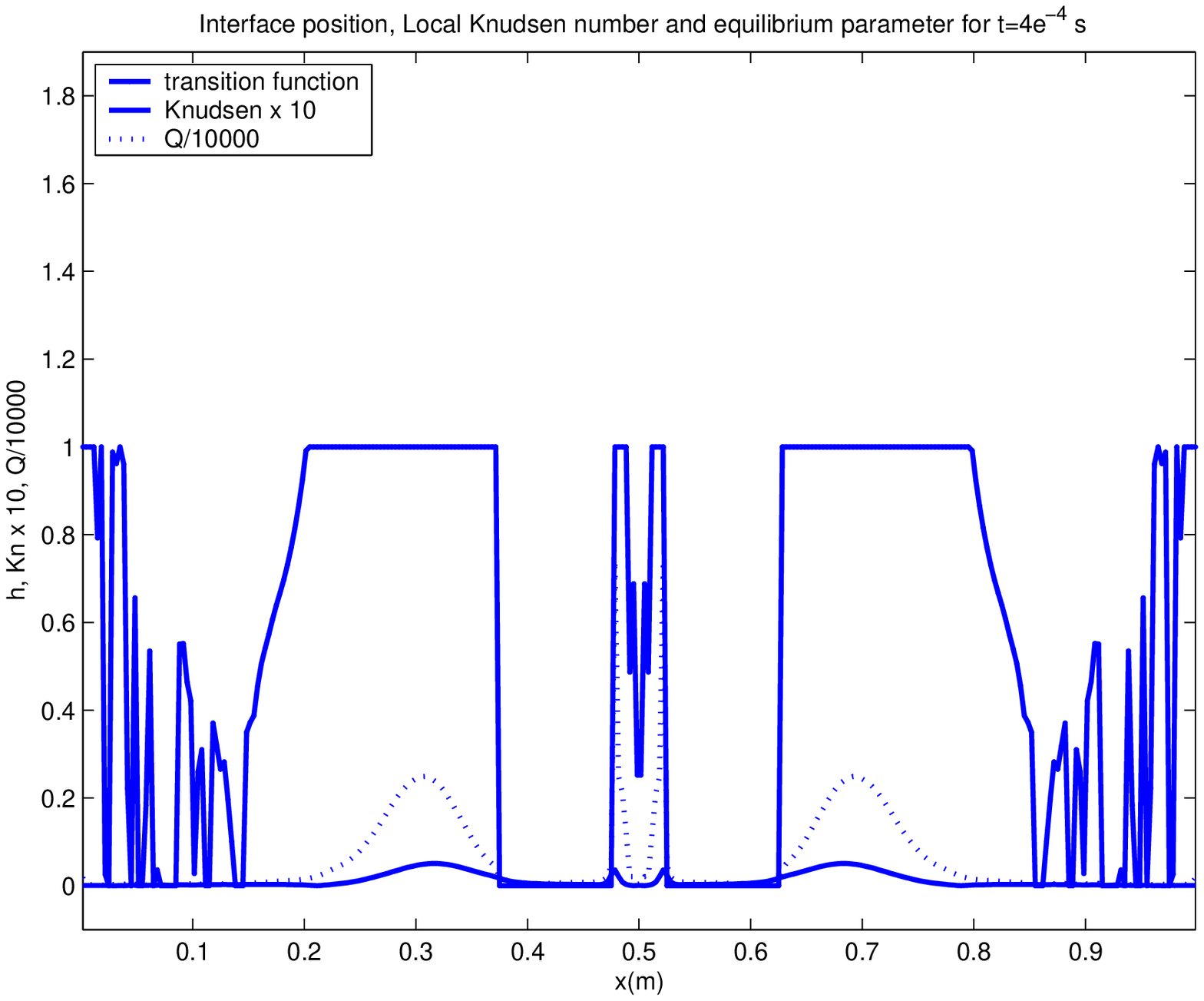}
\caption{Blast Wave Test 1: Solution at $t=1\times 10^{-4}$ top,
$t=2\times 10^{-4}$ middle top, $t=3\times 10^{-4}$ middle bottom,
$t=4\times 10^{-4}$ bottom, temperature left, transition function,
Knudsen number and heat flux right. The small panels are a
magnification of the solution close to non equilibrium regions.
\label{blast1.2}}
\end{center}
\end{figure}

\begin{figure}
\begin{center}
\includegraphics[scale=0.34]{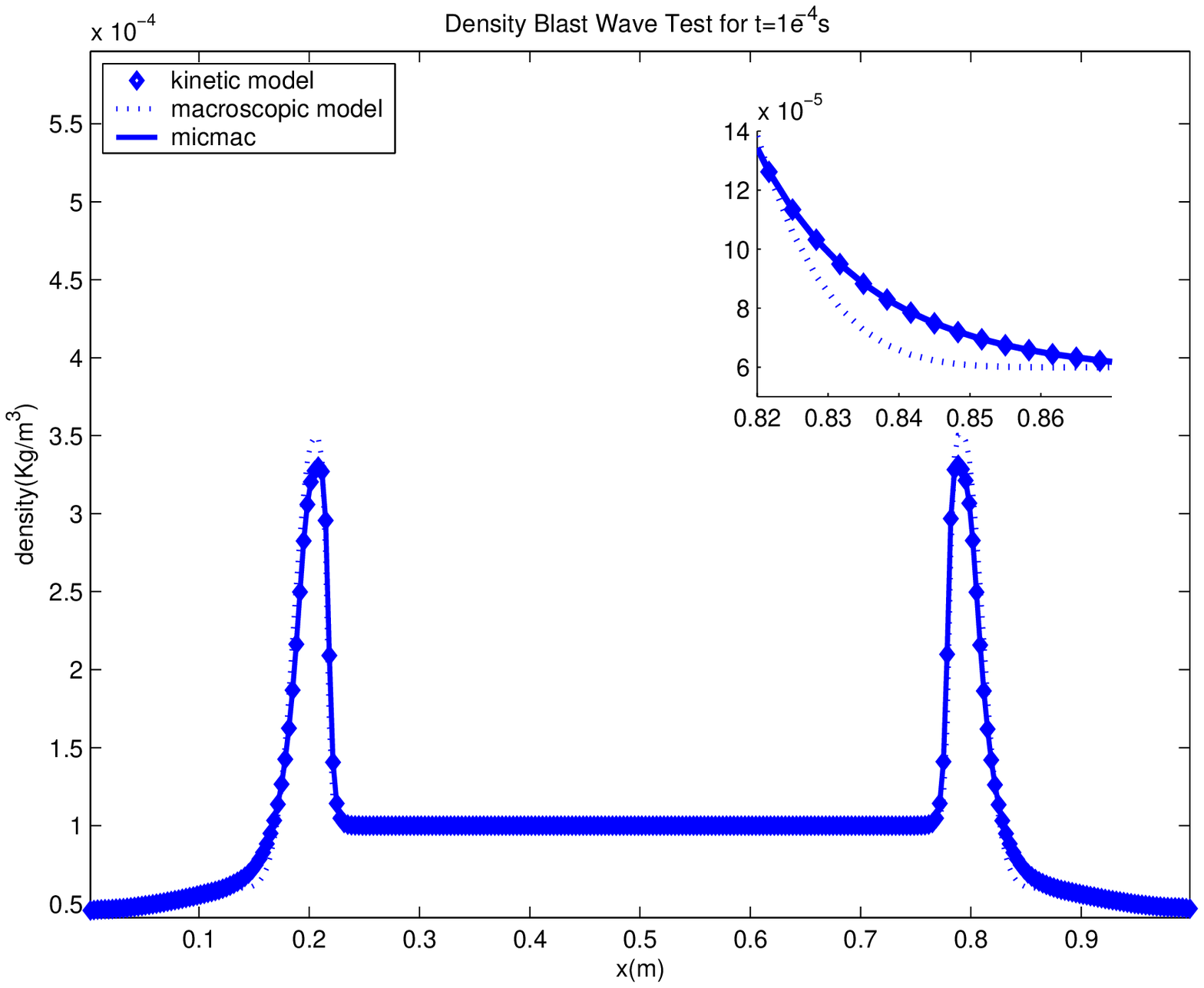}
\includegraphics[scale=0.34]{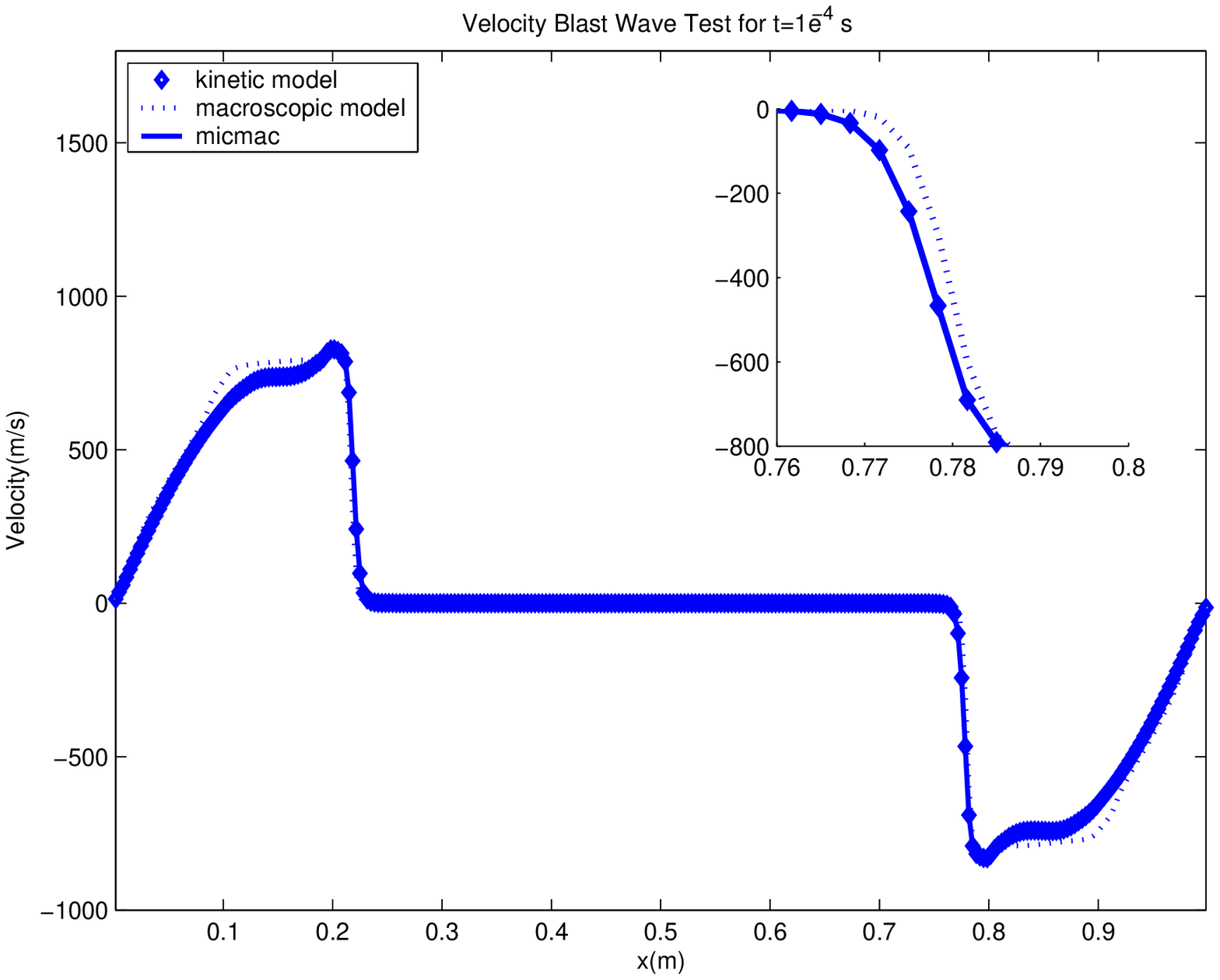}
\includegraphics[scale=0.34]{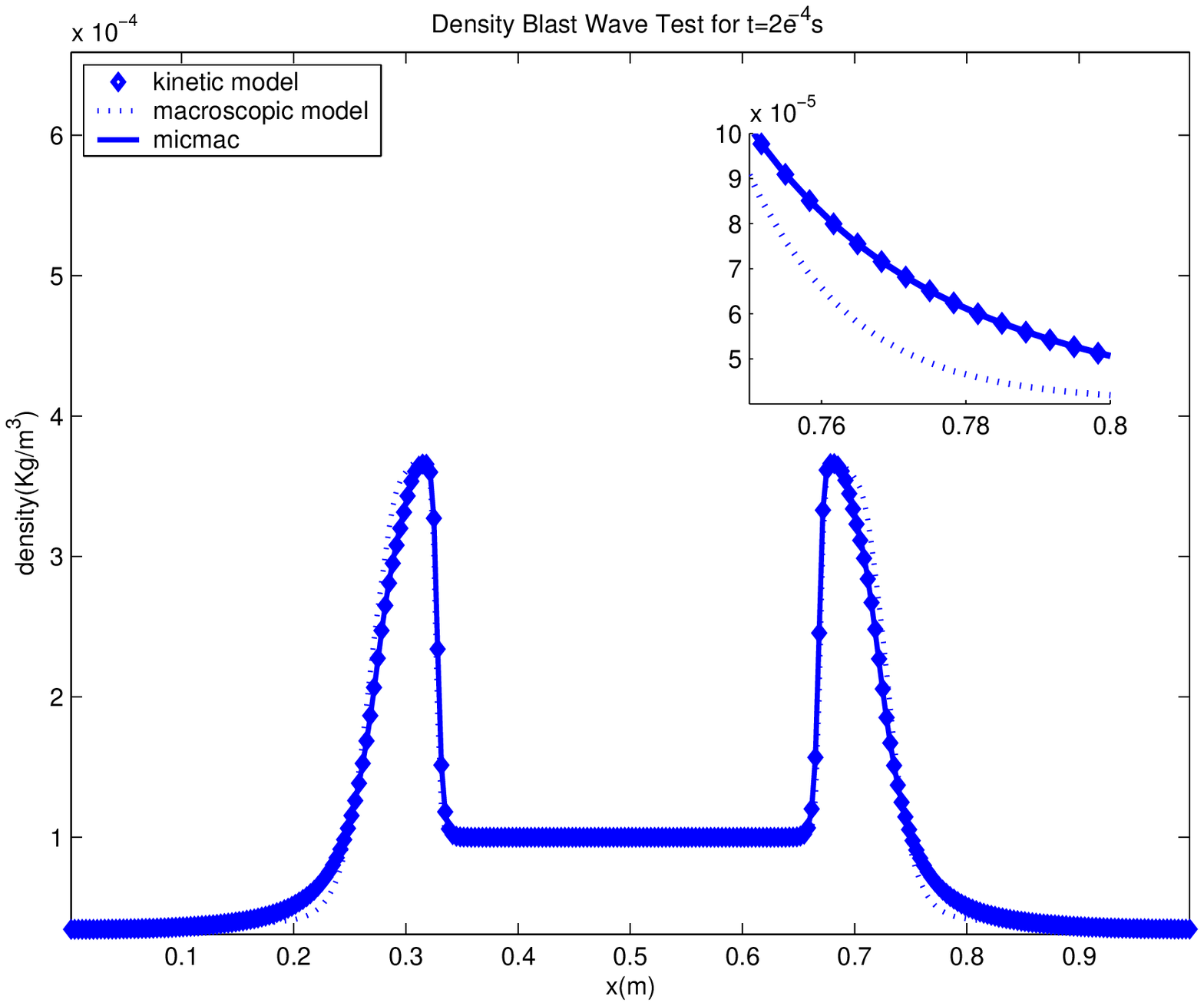}
\includegraphics[scale=0.34]{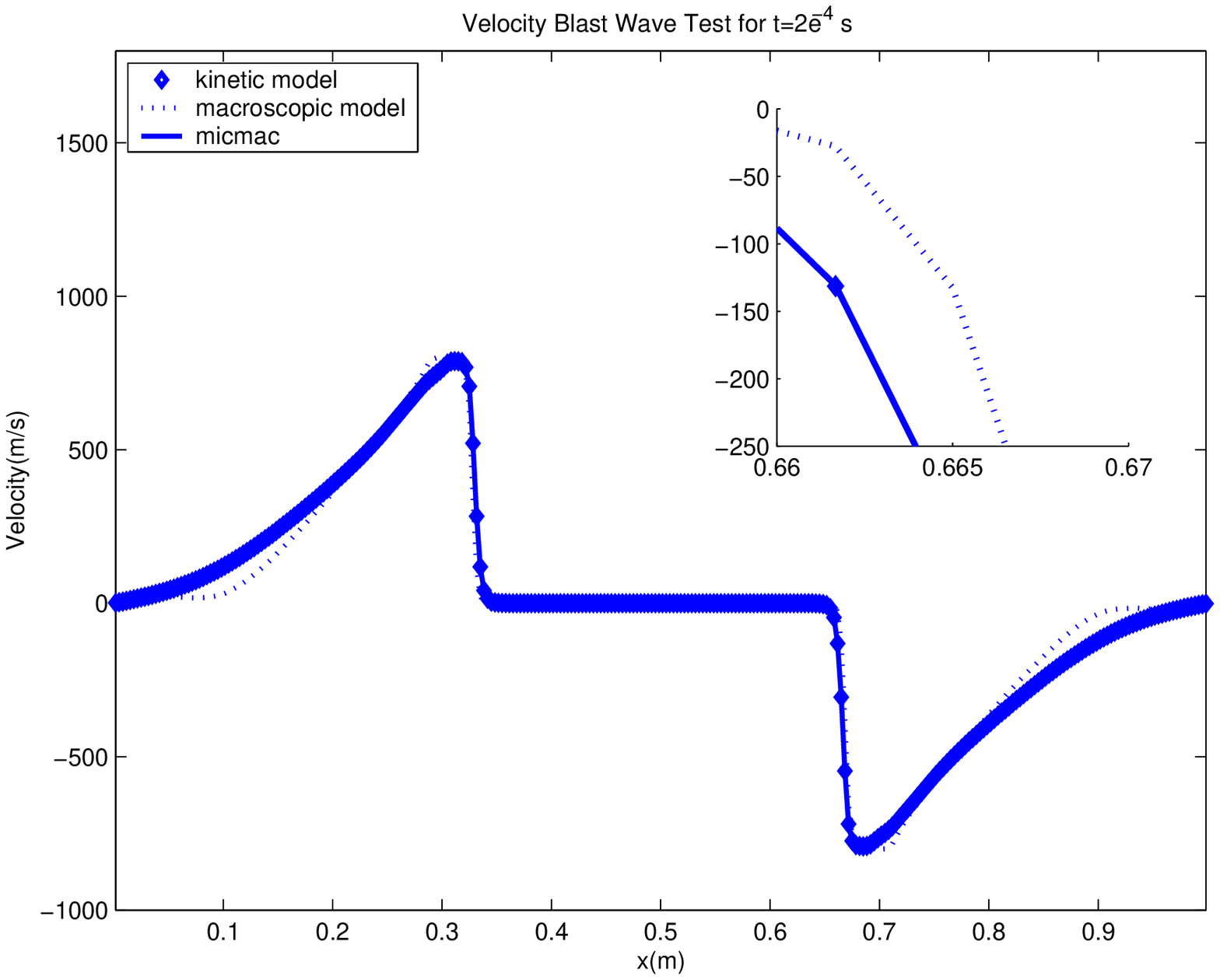}
\includegraphics[scale=0.34]{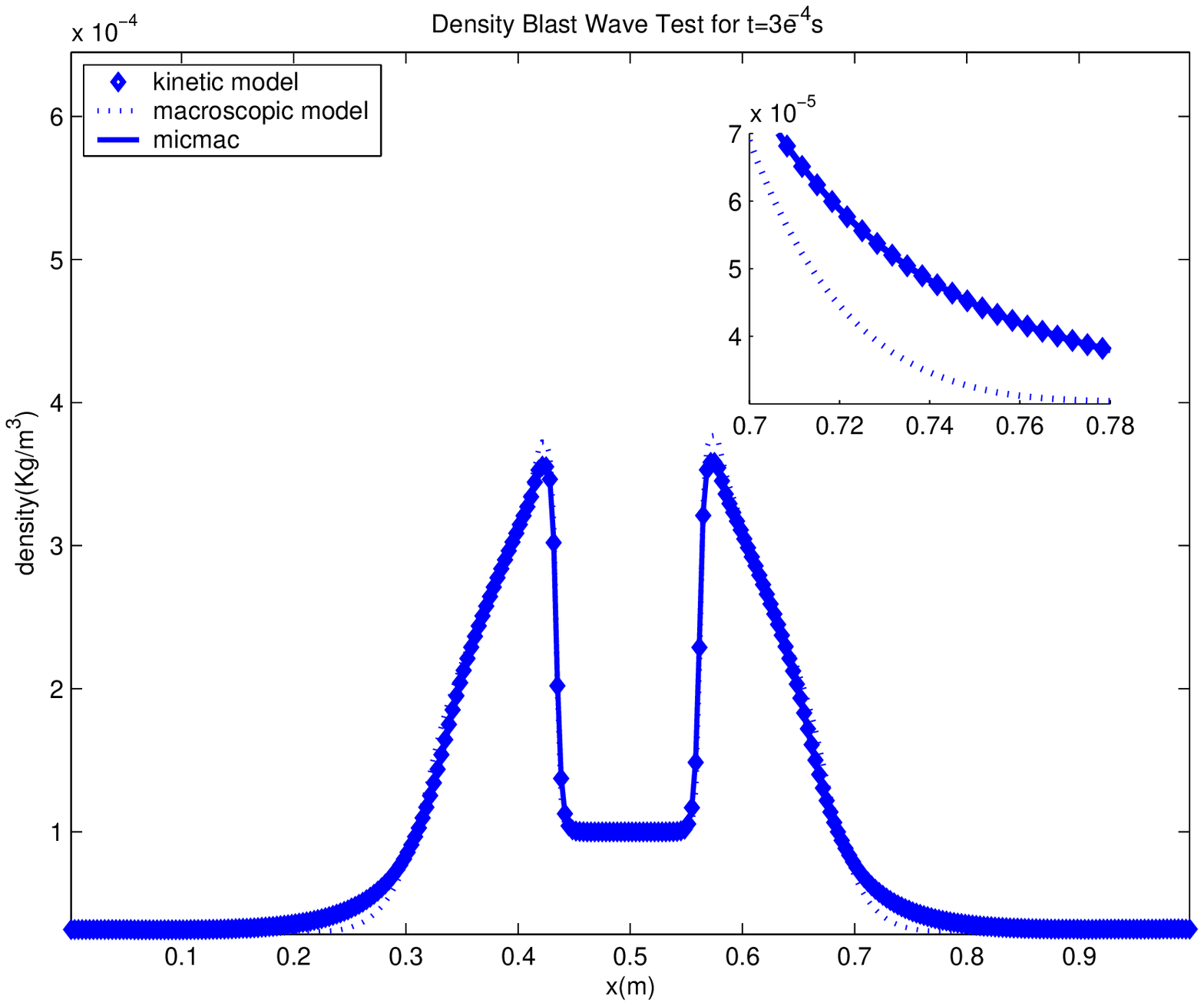}
\includegraphics[scale=0.34]{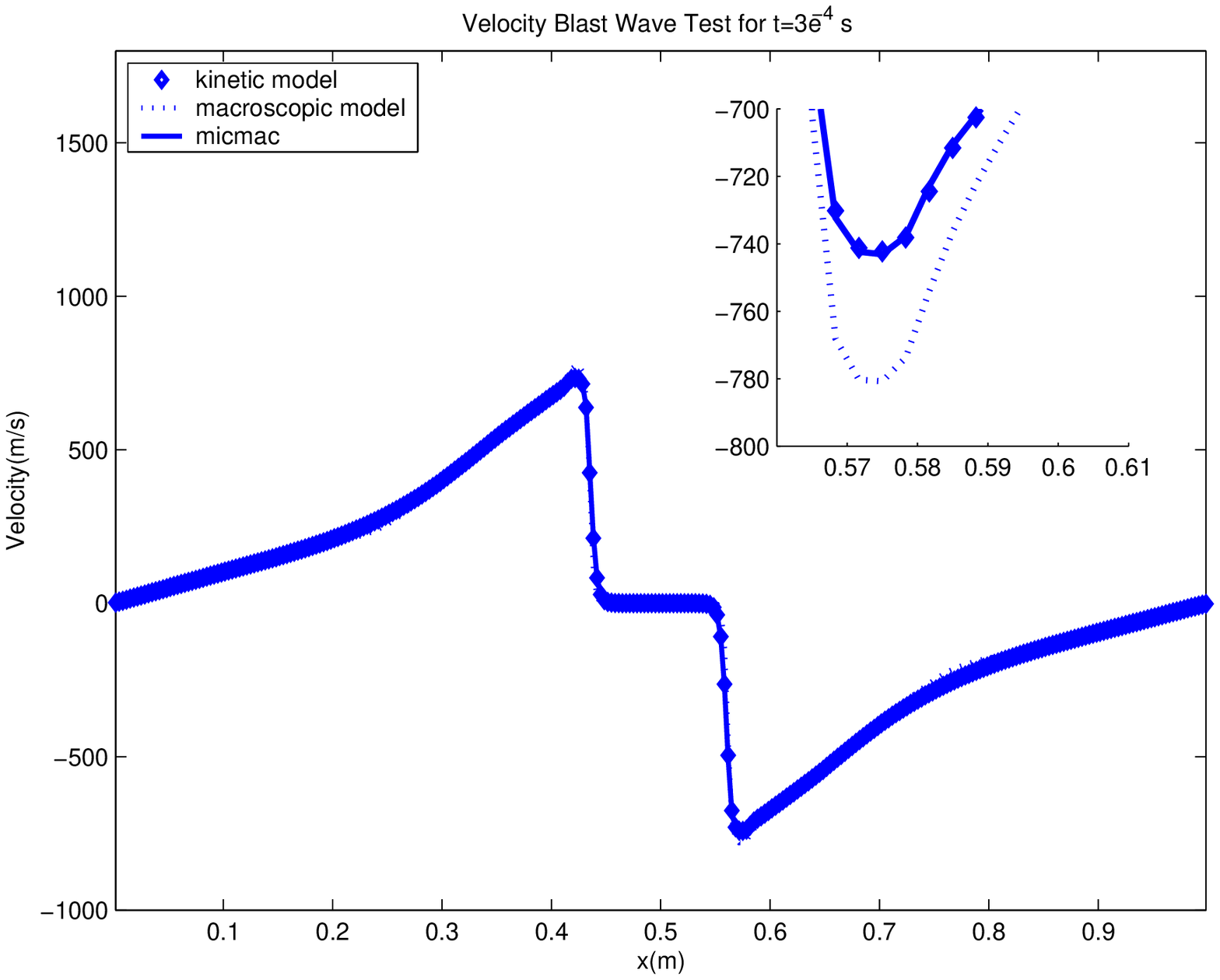}
\includegraphics[scale=0.34]{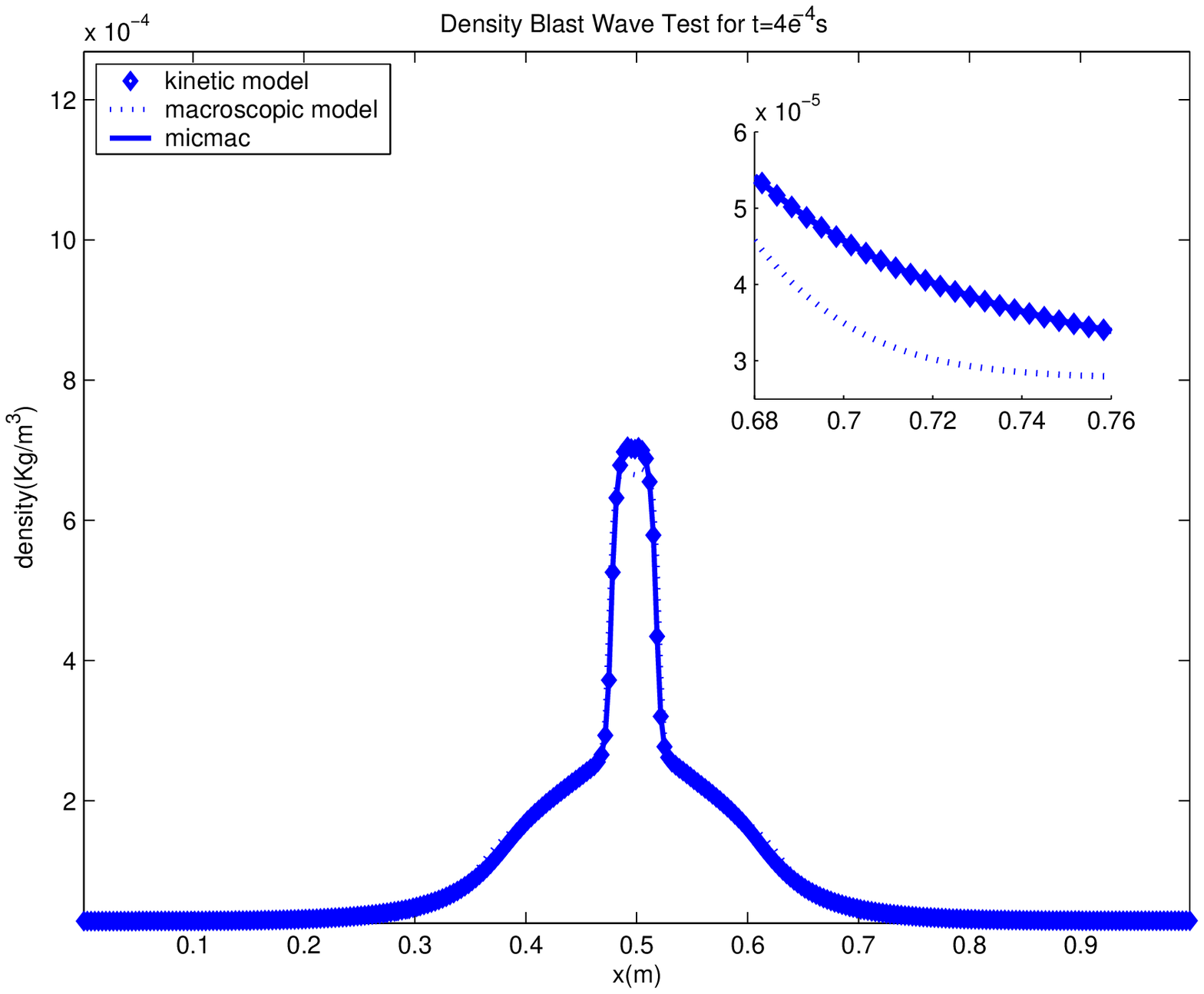}
\includegraphics[scale=0.34]{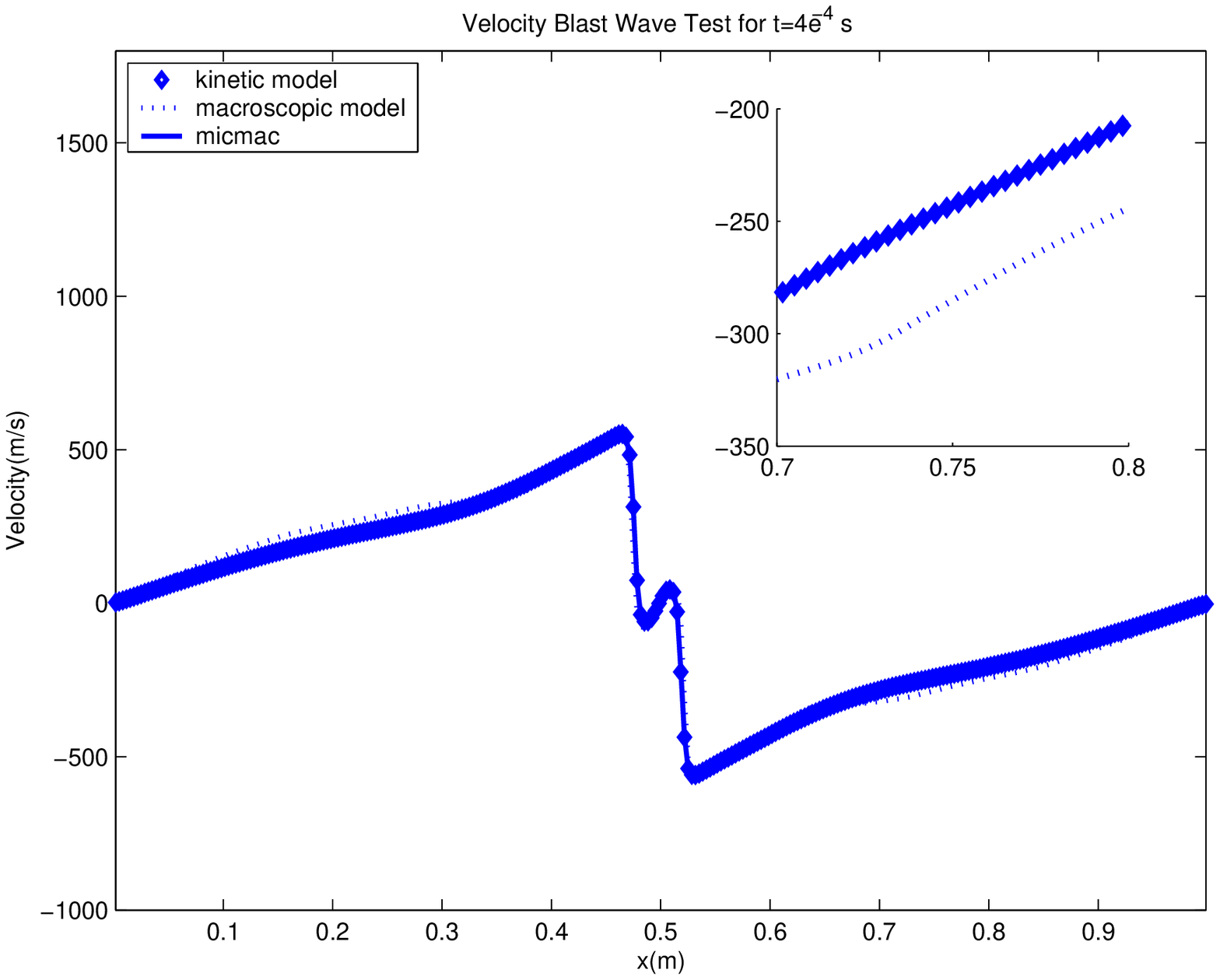}
\caption{Blast Wave Test 2: Solution at $t=1\times 10^{-4}$ top,
$t=2\times 10^{-4}$ middle top, $t=3\times 10^{-4}$ middle bottom,
$t=4\times 10^{-4}$ bottom, density left, velocity  right.
\label{blast2.1}}
\end{center}
\end{figure}

\begin{figure}
\begin{center}
\includegraphics[scale=0.34]{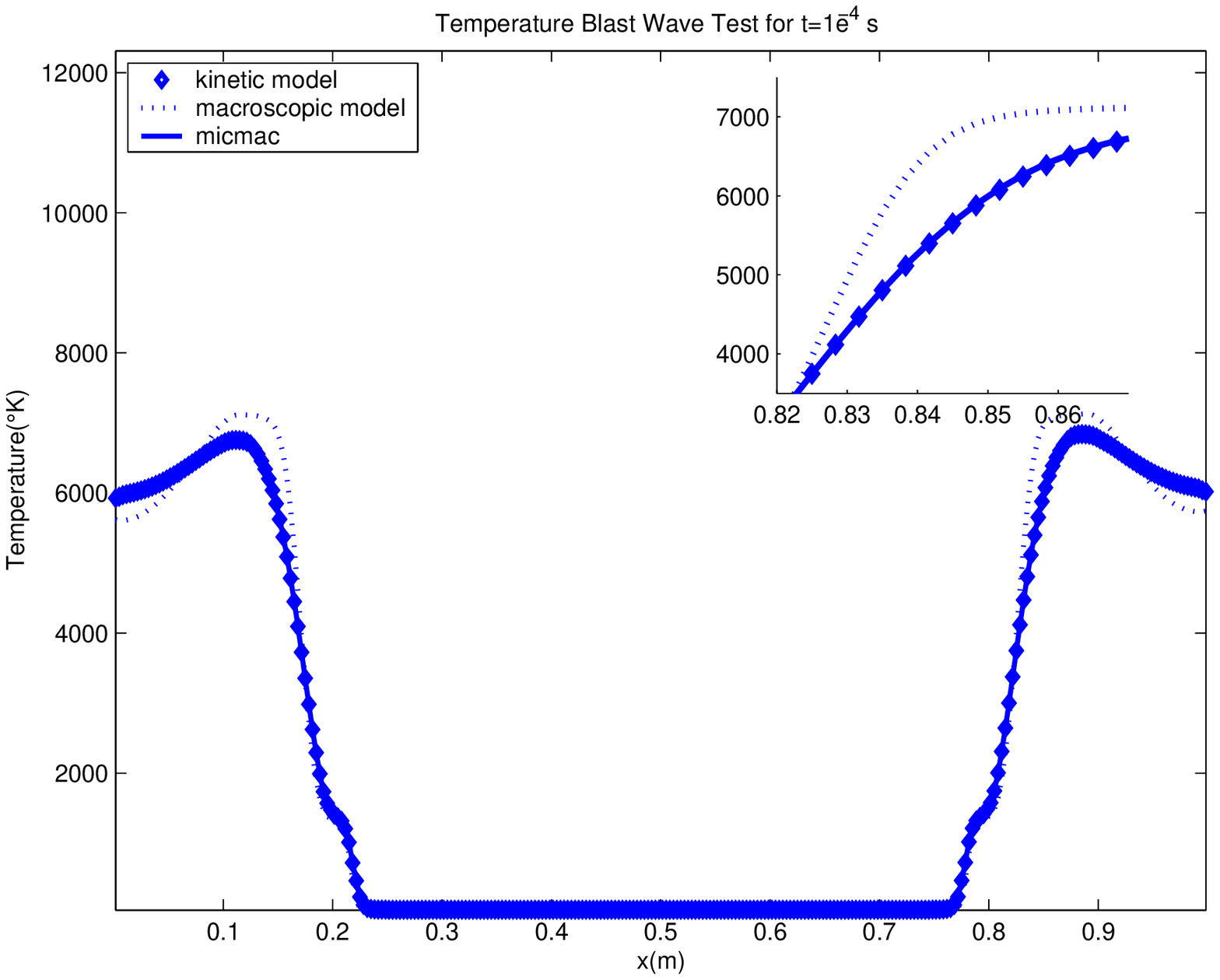}
\includegraphics[scale=0.34]{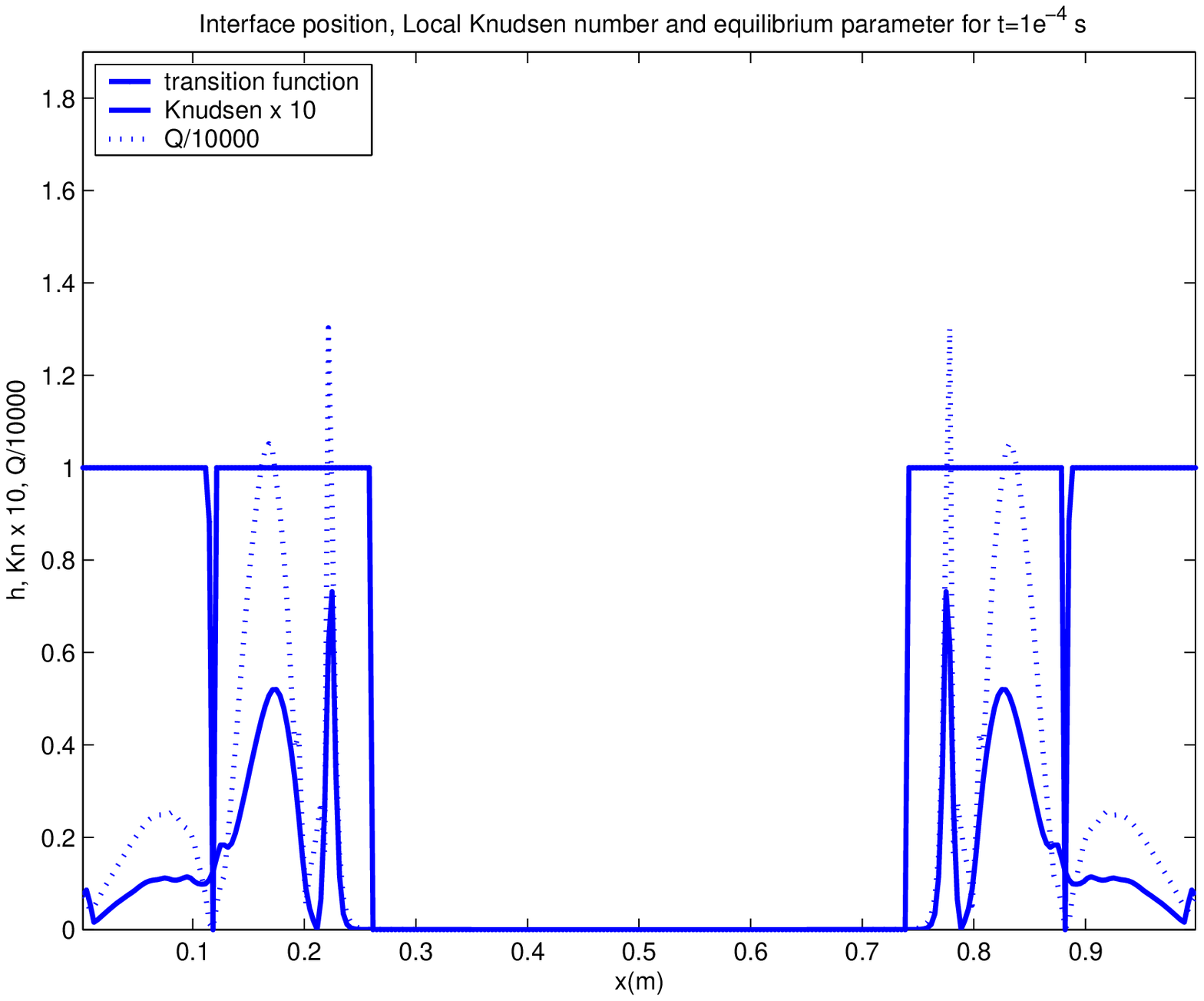}
\includegraphics[scale=0.34]{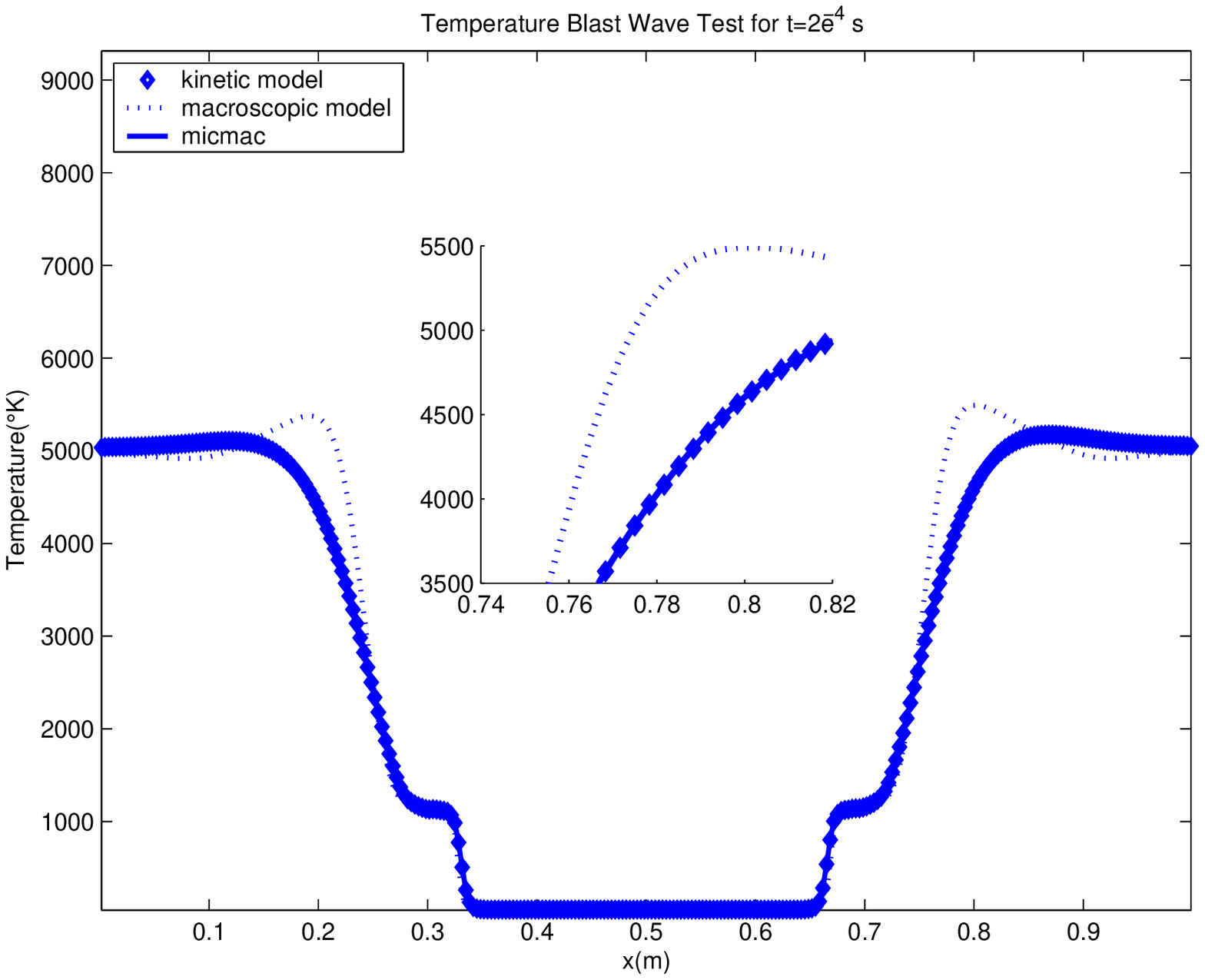}
\includegraphics[scale=0.34]{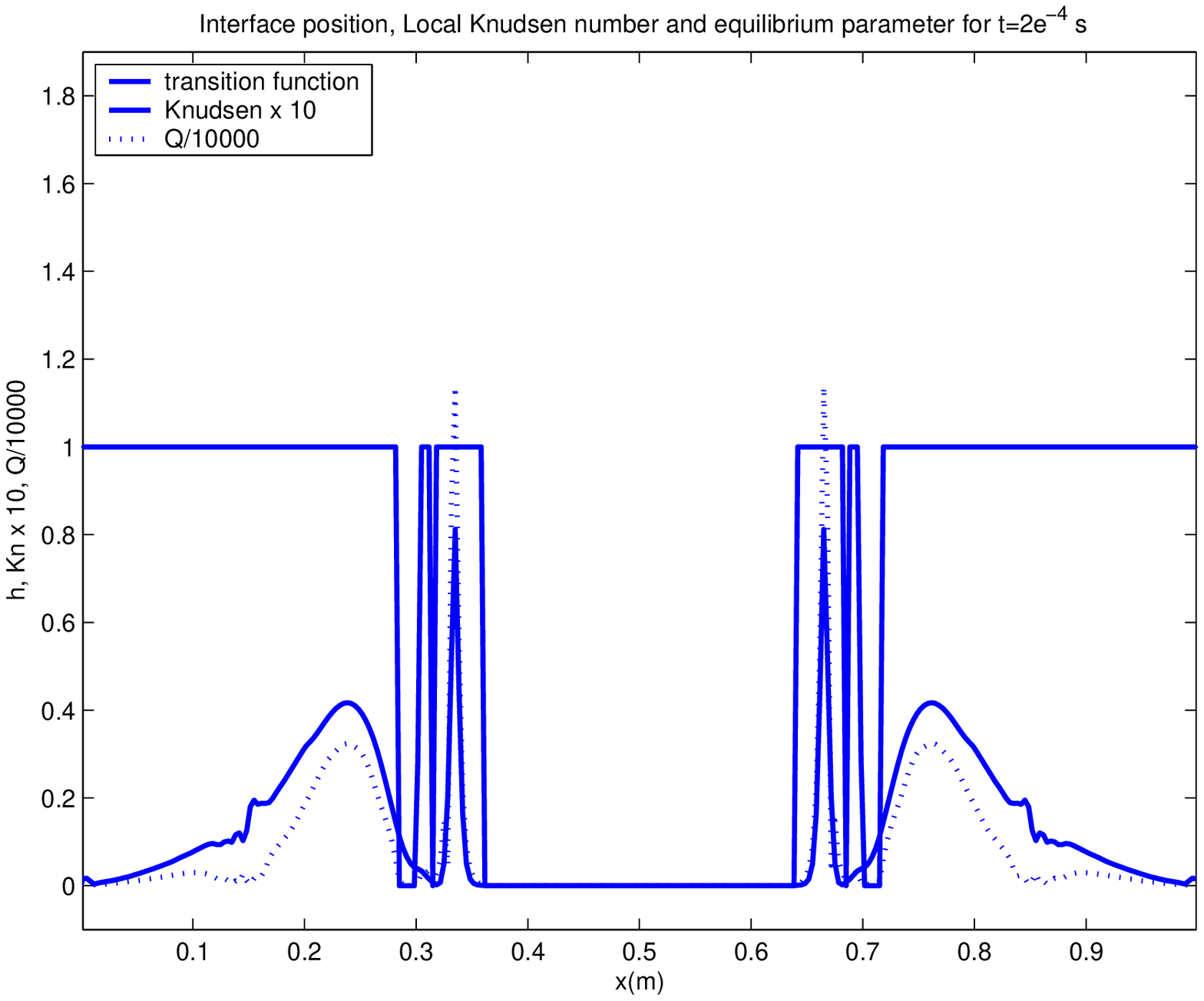}
\includegraphics[scale=0.34]{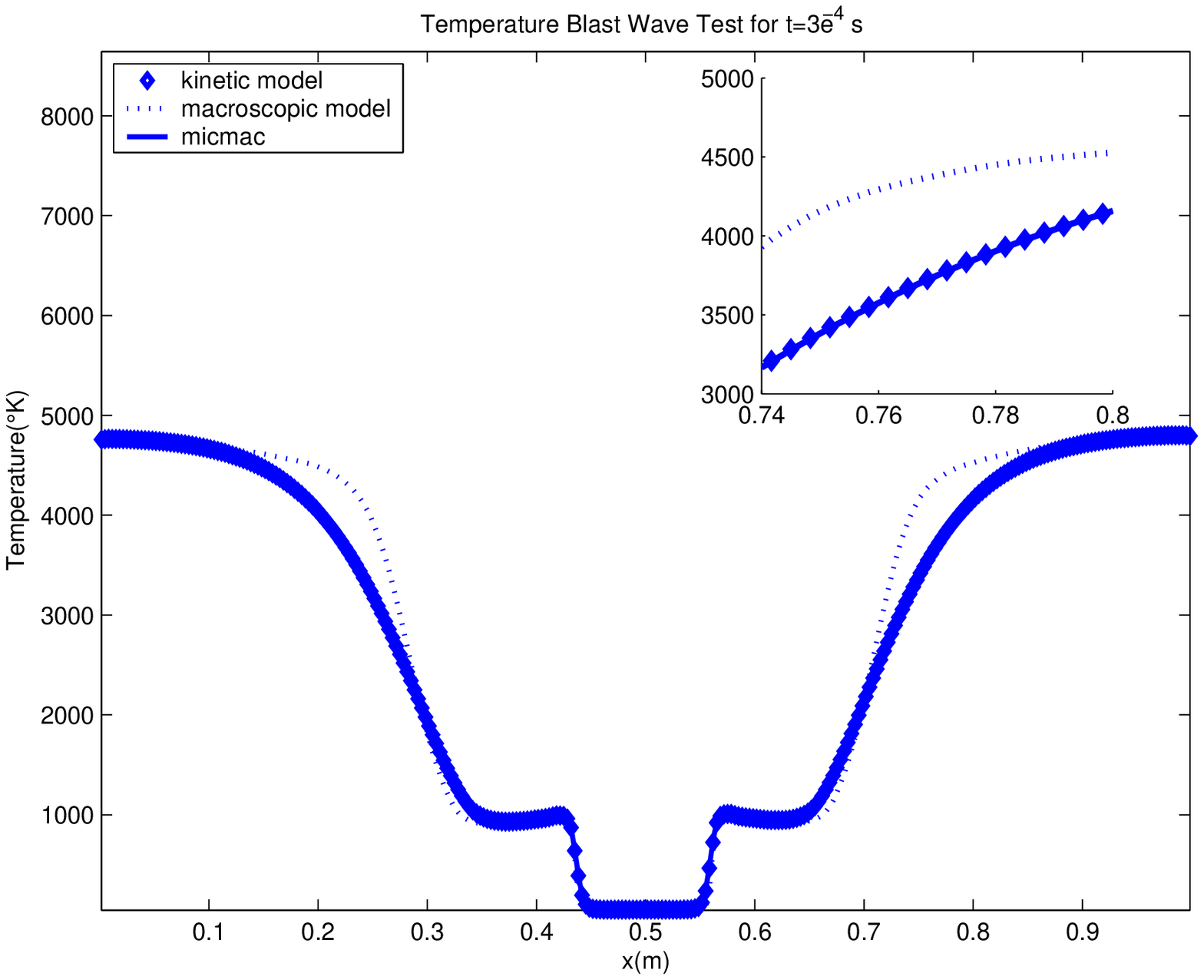}
\includegraphics[scale=0.34]{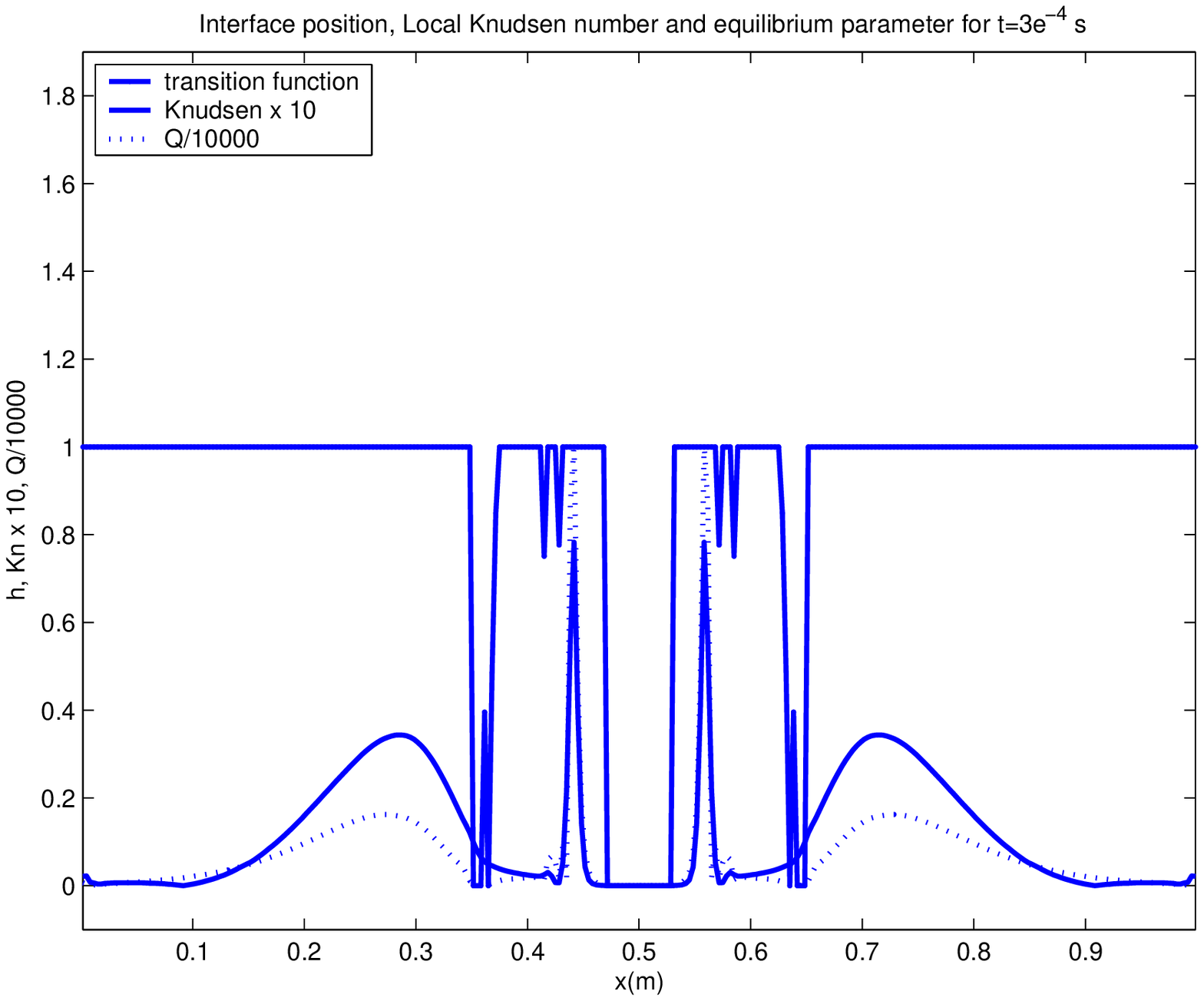}
\includegraphics[scale=0.34]{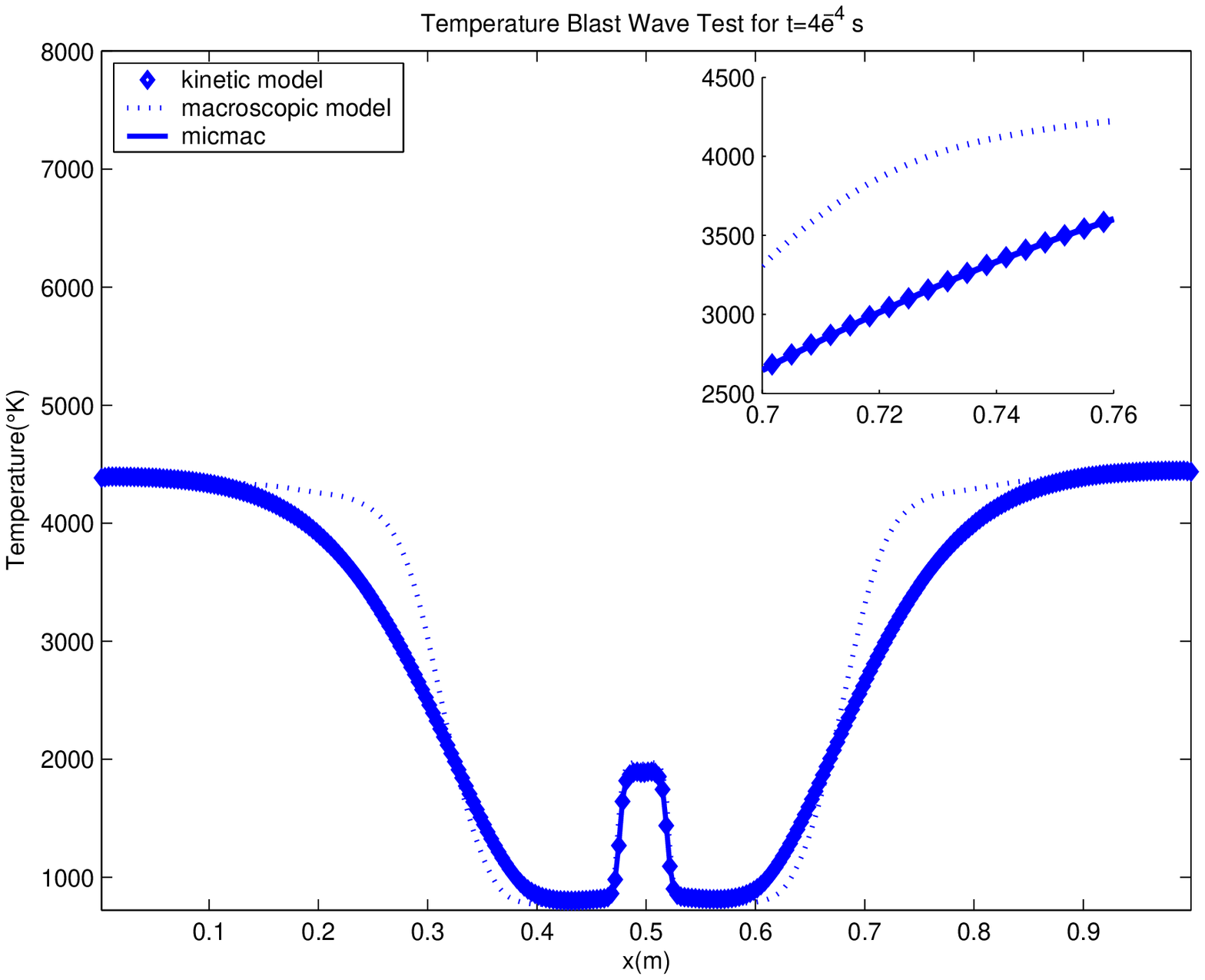}
\includegraphics[scale=0.34]{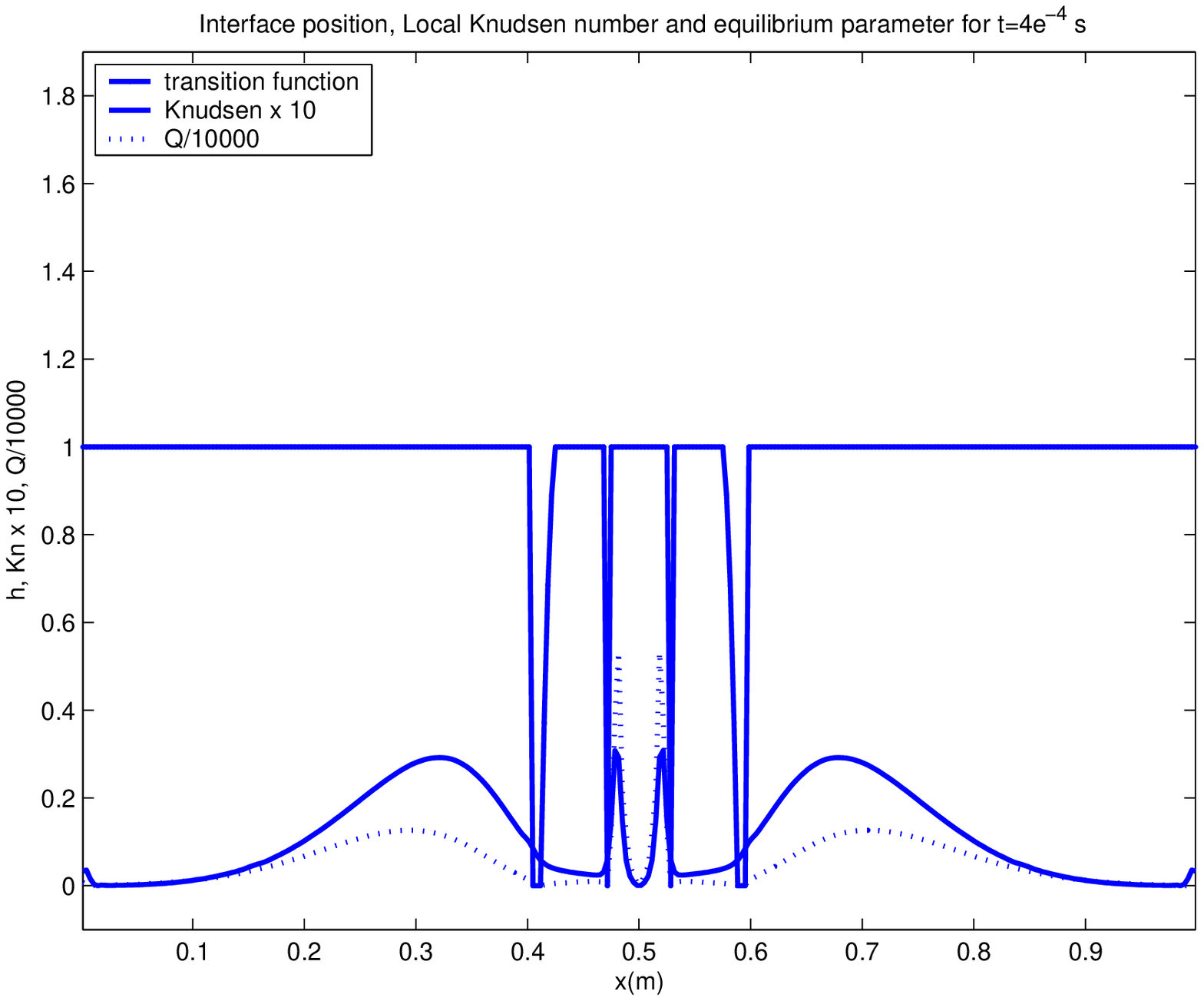}
\caption{Blast Wave Test 2: Solution at $t=1\times 10^{-4}$ top,
$t=2\times 10^{-4}$ middle top, $t=3\times 10^{-4}$ middle bottom,
$t=4\times 10^{-4}$ bottom, temperature left, transition function,
Knudsen number and heat flux. \label{blast2.2}}
\end{center}
\end{figure}


\end{document}